\documentclass[12pt]{iopart}

\usepackage[utf8]{inputenc}
\usepackage[margin = 1in]{geometry}
\usepackage{listings}
\expandafter\let\csname equation*\endcsname\relax

\expandafter\let\csname endequation*\endcsname\relax

\usepackage{amsmath}
\usepackage{amsfonts}
\usepackage{amssymb}
\usepackage{overpic}
\usepackage{epstopdf}
\usepackage{graphicx}
\usepackage{bm}
\usepackage[font = small]{caption}
\usepackage[font = small]{subcaption}
\usepackage{booktabs}
\usepackage{xcolor}

\usepackage{physics}
\usepackage{mathtools}
\usepackage[hidelinks]{hyperref}
\usepackage{amsthm}

\newcommand{\eps}{\ensuremath{\varepsilon}}

\renewcommand{\bs}{\boldsymbol}
\newcommand{\jac}{J_{\mathbf u} \, }
\newcommand{\mbf}{\mathbf}
\newcommand{\ds}{\displaystyle}
\newcommand{\evs}[1]{\textcolor{black}{#1}}
\newcommand{\evss}[1]{\textcolor{black}{#1}}

\newcommand{\ol}{\overline}

\begin{document}

\title[Beyond-all-order asymptotics for homoclinic snaking in reaction-transport systems]{Beyond-all-order asymptotics for homoclinic snaking of localised patterns in reaction-transport systems}

\author{Edgardo Villar-Sepúlveda}

\address{School of Engineering Mathematics and Technology, University of Bristol, Bristol BS8 1TR, UK}
\ead{edgardo.villar-sepulveda@bristol.ac.uk}
\vspace{10pt}
\begin{indented}
    \item Submitted for publication: \today
\end{indented}

\begin{abstract}
    Spatially localised stationary patterns of arbitrary wide spatial extent emerge from subcritical Turing bifurcations in one-dimensional reaction-diffusion systems. They lie on characteristic bifurcation curves that oscillate around a Maxwell point as part of a homoclinic snaking phenomenon. Here, a generalisation of the exponential asymptotics method by Chapman \& Kozyreff is developed to provide leading-order expressions for the width of the snaking region close to a Turing bifurcation's super/sub-critical transition in arbitrary $n$-component reaction-diffusion systems. First, general expressions are provided for the regular asymptotic approximation of the Maxwell point, which depends algebraically on the parametric distance from the codimension-two super/sub-critical Turing bifurcation. Then, expressions are derived for the width of the snaking, which is exponentially small in the same distance. The general theory is developed using a vectorised form of regular and exponential asymptotic expansions for localised patterns, which are matched at a Stokes' line. Detailed calculations for a particular example are algebraically cumbersome, depend sensitively on the form of the reaction kinetics, and rely on the summation of a weakly convergent series obtained via optimal truncation. Nevertheless, the process can be automated, and code is provided that carries out the calculations automatically, only requiring the user to input the model and the values of parameters at a codimension-two bifurcation \evs{point}. The general theory also applies to higher-order equations, including those that can be recast as a reaction-diffusion system. The theory is illustrated by comparing numerical computations of localised patterns in two versions of the Swift-Hohenberg equation with different nonlinearities, and versions of \evs{some classic} activator-inhibitor reaction-diffusion systems.
\end{abstract}

%
\vspace{2pc}

\section{Introduction} \label{sec:introduction}
    The spontaneous formation of spatially periodic patterns through a finite wavenumber instability of systems of reaction-diffusion equations arises in several areas of applied mathematics, beginning with Alan Turing's seminal work on the chemical basis of morphogenesis \cite{turing}. See \cite{dawes,krause_review}, for recent reviews and updates. If the Turing bifurcation is subcritical, giving rise to unstable periodic patterns that subsequently restabilise, then there should be a family of transverse heteroclinic connections between the background steady state and stable periodic patterns \cite{burke2007homoclinic,Chapman,Dean}. Unfolding this family leads to localised structures of arbitrarily wide spatial extent that are formed by homoclinic connections to the background with arbitrary oscillations near the periodic pattern \cite{Woods}. These connections organise themselves in what has been dubbed a snaking bifurcation diagram as depicted in Figure \ref{fig:fig1}. See, for example, \cite{FahadWoods,burke2007homoclinic,Woods_Intro,Knobloch1,villar2023degenerate} and references therein.

    Beck {\em et al} \cite{snakebeck} provided proof of the existence and properties of the snaking bifurcation diagram in a large class of reversible systems under generic assumptions on the existence of the heteroclinic connection. Computational methods can find much of the fine structure of homoclinic snaking in examples, but in general, analytical descriptions of the localised patterns remain problematic, even in one spatial dimension. See \cite{bramburger-review} for a review of the state of the art \evs{of localized patterns} in one or more spatial dimensions.

    One possibility to get an analytical handle on localised pattern formation due to homoclinic snaking is in the asymptotic limit of a codimension-two super/sub-critical Turing bifurcation. From such a bifurcation, there emerges a {\em Maxwell point} at which the background-to-periodic heteroclinic connection exists. Unfortunately, it is known that the normal form of such a bifurcation, up to any algebraic order, is integrable and so any heteroclinic connection is non-transverse and there can be no homoclinic snakes within the normal form (see e.g.~\cite{ioossbook,Woods}). Therefore, the snaking region must emerge beyond all algebraic orders in parameters that unfold the normal form. The computation of this beyond-all-orders width about the Maxwell point, for arbitrary reaction-diffusion systems posed on the real line forms the subject of the present paper.

    At present, this beyond-all-orders computation has only been attempted in a few special cases \cite{Chapman,HannesdeWitt,Dean}. A canonical example of pattern formation is the fourth-order Swift-Hohenberg partial differential equation (PDE). As we shall see, such higher-order equations can be brought into the framework of the class of reaction-diffusion equations studied in this paper, provided we allow for a singular temporal evolution operator, which is fine because we are primarily concerned with the existence of stationary patterns here, not their dynamics.

    \begin{figure}
        \centering
        \includegraphics[width=\textwidth]{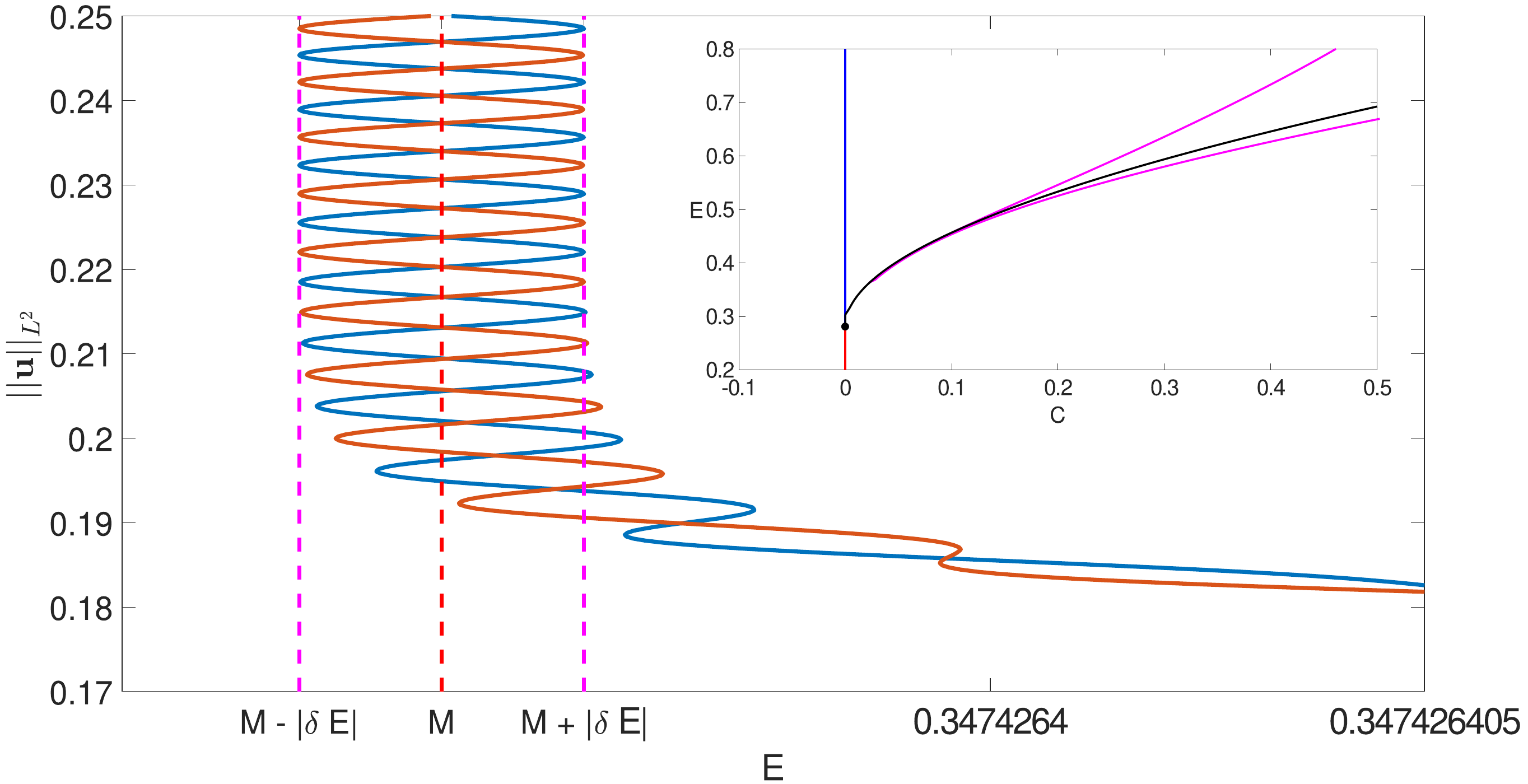}
        \caption{Qualitative description of the snake of localised patterns snaking around a Maxwell point in a one-parameter diagram. See text for details. The inset shows a two-parameter diagram where the magenta curves represent loci of the outermost folds of the snake, surrounding the Maxwell point (black curve) arising from a codimension-two Turing bifurcation (the point of transition between red and blue lines, which represent supercritical and subcritical Turing bifurcation lines, respectively). The asymptotic scaling of the width \evs{of the snake} between the two magenta curves as we approach the codimension-two point is what the analysis in the paper seeks to explain. While intended to be qualitative, the numerical curves are taken from the Swift-Hohenberg 2-3 \evs{example,} studied in Section \ref{sec:examples}.}
        \label{fig:fig1}
    \end{figure}

    This paper is motivated by the groundbreaking work by Chapman \& Kozyreff \cite{Chapman,Kozyreff}, who used \textit{exponential asymptotics} to study homoclinic snaking in the Swift-Hohenberg equation with quadratic-cubic nonlinearities posed on the line. They used matching of inner and outer expansions of exponentially growing and decaying solutions upon crossing a Stokes' line in the complex plane. The problem is reduced to the computation of a single coefficient of an exponentially small amplitude, which can be achieved through an iteration scheme. They found the iteration to be slowly converging and instead simply fit this coefficient to numerical results. Later, Dean \textit{et al.} \cite{Dean} applied the same technique to a version of the Swift-Hohenberg with cubic-quintic nonlinearities. The additional odd symmetry made the calculation slightly more straightforward, and they were able to make their iteration scheme converge to an amplitude that matched extremely well with numerical results. More recently, the same method was implemented by De Witt \cite{HannesdeWitt} for a modified version of the Schnakenberg system, which is a reaction-diffusion equation with two components. 
    
    \subsection{Problem description}
        We pose a general reaction-diffusion system on the line as:
        \begin{align}
            \mathbb M \, \partial_t \, \mbf u = \mbf f(\mbf u; a, b) + \hat D \, \partial_{xx} \mbf u, \qquad x \in \mathbb R, \label{geneq}
        \end{align}
        where $\mbf u = \mbf u(x, t) \in \mathbb R^n$, $n\geq 2$, is a vector of real functions that are assumed to be well-defined for all $(x, t) \in \mathbb R \times \mathbb R^+_0$, $\hat D$ is a diffusion matrix, which can be any square matrix as long as \eqref{geneq} is well-posed and the assumptions below hold, $a, b \in \mathbb R$ are parameters of the system, and $\mathbb M$ is the mass matrix of the system, which can be any square matrix, as long as \eqref{geneq} is well-posed. Note that writing the system in this way allows the analysis to be applied to systems of $n\geq 1$ equations containing higher, even order spatial derivatives. For example, a fourth-order equation $\partial_t u_1 = \partial_{xxxx} u_1$ can be written as \eqref{geneq} simply by adding the extra equation $\partial_{xx} u_1 - u_2 = 0$.
    
        In this article, we aim to generalize the method of \textit{exponential asymptotics} to study \textit{homoclinic snaking} near codimension-two Turing bifurcation points under the assumption that such a phenomenon occurs. Specifically, we will determine conditions for the existence of Maxwell points, where the homogeneous steady state has the same energy as a large-amplitude periodic orbit, and construct approximations of the localized structures that appear near such points.

        \evs{Specifically, assuming that localized solutions appear around a Maxwell point that exists close to a codimension-two Turing bifurcation point, where conditions for the latter will also be determined, then there exists real numbers, $\beta_1 > 0$, $\beta_3 < 0$, and $\eta$; a complex number, $K_2$; and a rational function $R(\varepsilon)$, such that homoclinic snaking exists for $b = \hat b + \delta b$, where $\hat b$ is the value of $b$ at the Maxwell point, and
        \begin{align}
            \abs{\delta b} \leq \abs{\frac{\sqrt{2} \, \beta_3}{2 \, \beta_1^2 \, R(\varepsilon)}} \, \frac{\pi \, \abs{K_2} \, e^{- \frac{\pi}{2} \, \left(\frac{1}{k \, \sqrt{\beta_1} \, \varepsilon^2} + \eta\right)}}{\varepsilon^6}, \label{bound_statement}
        \end{align}
        where $k$ is the critical wavenumber at the codimension-two Turing bifurcation point, and $\varepsilon$ is a measure of the separation in $a$ from the codimension-two Turing bifurcation point. Furthermore, $\beta_1$, $\beta_3$, $\eta$, $K_2$, and $R$ have closed-form definitions that will be provided in this article, and there is code available in \cite{beyond-code} to compute them for an arbitrary $n$-component reaction-diffusion system. Details on how to use the code are provided in \ref{app:code}.}
    
        \evs{To start developing ideas}, we assume that \eqref{geneq} has a homogeneous steady state $\mbf P = \mbf P(a, b)$ such that $\mbf f(\mbf P; a, b) = \mbf 0$ for all $a, b \in \mathbb R$, with $\jac \mbf f(\mbf P; a, b)$, the Jacobian matrix of $\mbf f(\mbf u; a, b)$ at $\mbf u = \mbf P$, being an invertible matrix.
        
        Furthermore, we assume that $\mbf P$ goes through a codimension-two Turing bifurcation point when $(a, b) = \left(a_0, b_0\right)$. That is, there exists $k^* > 0$, the \textit{wavenumber}, such that
        \begin{align*}
            \det\left(\jac \mbf f\left(\mbf P; a_0, b_0\right) - \left(k^*\right)^2 \, \hat D\right) = 0,
            \\
            \left. \frac{\partial}{\partial k} \left(\det\left(\jac \mbf f\left(\mbf P; a_0, b_0\right) - k^2 \, \hat D\right)\right)\right|_{k = k^*} = 0, \quad \text{and}
            \\
            \dim\left(\ker\left(\jac \mbf f\left(\mbf P; a_0, b_0\right) - \left(k^*\right)^2 \, \hat D\right)\right) = 1.
        \end{align*}
        Furthermore, we assume that
        \begin{align}
            \det\left(\jac \mbf f\left(\mbf P; a_0, b_0\right) - k^2 \, \hat D\right) \neq 0, \label{no-DT}
        \end{align}
        for all $k \in \mathbb R \setminus \left \{k^*\right\}$; and the existence of localized structures close to this codimension-two bifurcation point.
        
        Without loss of generality, we assume that $\mbf P = \mbf 0$ for all $a, b$. If not, we can make the change $\mbf u \to \mbf u - \mbf P$, which translates $\mbf P$ to the origin. Furthermore, for the same reason, without loss of generality, we assume that $\left(a_0, b_0\right) = (0, 0)$. Moreover, for simplicity of notation, we will omit the dependence on $a$ and $b$ and drop the star sign of the critical wavenumber $k^*$ when there is no place for confusion.
    
        The process we are about to follow comprises the steps shown in Figure \ref{fig:chart}, which comprise
        \begin{enumerate}
            \item Obtain an expression for the \evs{amplitude equation centred at a codimension-2 Turing bifurcation point and determine conditions for the existence of a Maxwell point.} (Section \ref{sec:regular_asymptotics}).
            \item Expand the system about the singularities of such amplitudes, which involves the computation of inner and outer solutions, which need to be matched (Sections \ref{sec:late_term} and \ref{sec:late_term2}).
            \item Determine the order of $n$ at which we truncate our asymptotic expansion optimally and study the equation for the remainder (Sections \ref{sec:first_residual} and \ref{sec:second_residual}).
            \item Find conditions that let us ensure that the remainder tends to zero as its corresponding independent variable tends to $\pm \infty$, \evs{which will require the study of Stokes' lines, where the key exponentially small contribution to our solution triggers, letting us} ensure that our expansion will be well-defined, providing us with an estimation for the width of the homoclinic snaking (end of Section \ref{sec:second_residual}).
        \end{enumerate}
        \begin{figure}
            \centering
            \includegraphics[width=\linewidth]{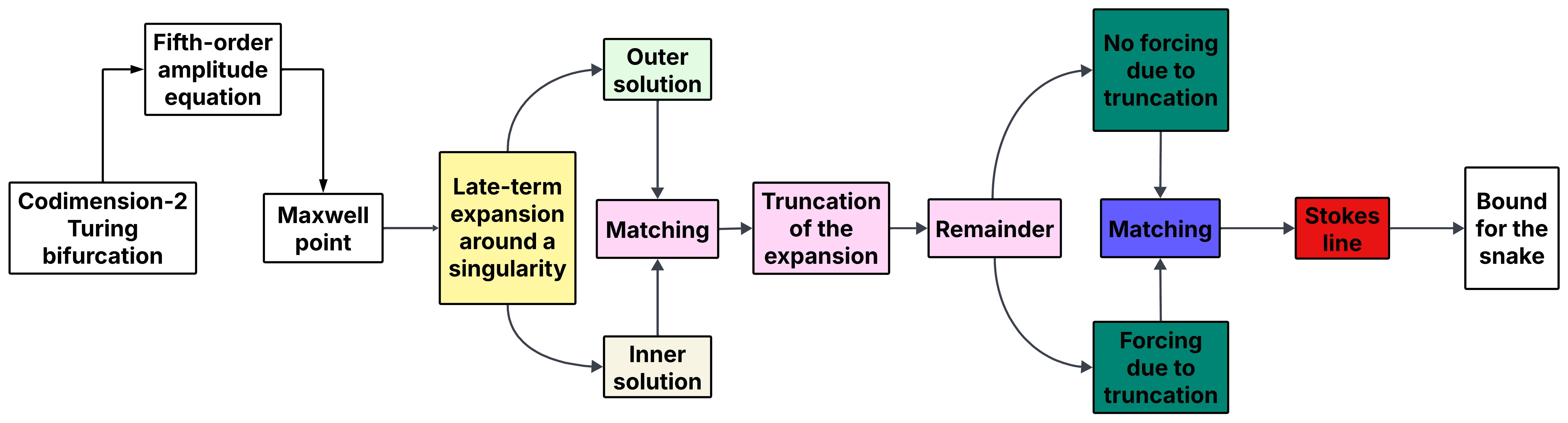}
            \caption{Diagram of the steps we are to follow to determine the bounds of the homoclinic snaking.}
            \label{fig:chart}
        \end{figure}

        \newtheorem{remark}{Remark}
  
    \subsection{Outline}
        The outline of this article is as follows. Section \ref{sec:regular_asymptotics} starts by introducing the standard notation used throughout the article, and computes a regular asymptotic expansion up to \evs{third} order of the kind of solutions we are interested in. \evs{The expansion up to order five is carried out in \ref{ap:order 5}}. Also, for later use, we will require going up to seventh order; those results are presented in \ref{app:7expansion}. Sections \ref{sec:late_term} and \ref{sec:late_term2} consider the behaviour of the asymptotic expansion as the order tends to $\infty$, with the two sections dealing respectively with late-term expansion in an inner and an outer region. Section \ref{sec:first_residual} then develops the equation for the remainder up to fifth order after truncating the expansion at a high order, without considering the forcing due to truncation. As for the amplitude equation, it will be required to develop the equation for the remainder up to seventh order; those results are presented in \ref{app:rem_order7}. The effect of this forcing is considered in Section \ref{sec:second_residual}, which gives a final condition to determine the width of the snaking. \ref{ap:joining_fronts} shows how to match solutions at the boundaries of the domains where different expansions are valid in order to join fronts. \ref{app:symmetric_case} shows what happens in the case in which a generic \evs{expansion is not suitable to determine the width of the snake}, and, to illustrate the theory, Section \ref{sec:examples} presents results from several examples in which this generic assumption holds and one in which it does not. In each case, we evaluate the specific regular and beyond-all-orders coefficients and compare the results with numerical computations. Note that this evaluation step is computationally cumbersome, but we provide the reader with code in an associated GitHub repository, \cite{beyond-code}, which can be used to replicate the results, as well as compute the snaking width for other examples. \evs{\ref{app:code} explains the logic of the code in \cite{beyond-code}, to let everyone compute the width of the snaking easily for any $n$-component reaction-diffusion systems.} Finally, Section \ref{sec:discussion} presents a brief discussion and suggests avenues for future work.

\section{Regular asymptotics} \label{sec:regular_asymptotics}
    We start the analysis of our system by performing a regular asymptotic expansion. This process has been carried out in different ways in the past in order to study the criticality of Turing bifurcations (see, e.g.~\cite{TOMS}). However, near a codimension-two Turing bifurcation point, localized solutions are expected to appear apart from the periodic states. To study such solutions, we need to obtain a general and accurate expression for the amplitude of different patterns that arise near a codimension-two Turing bifurcation. As we explain below, such an expression will have singularities that need to be studied to control the convergence of the solution we aim to build through this process.

    We start by introducing some notation to be used throughout this article.

    \subsection{Including parameters and kinetics in a general way}
        First, we scale $x \to x/k$ and expand \eqref{geneq} as:
        \begin{align}
            \mathbb M \, \frac{\partial \mbf u}{\partial t} = \sum_{i = 0}^{M_a} \sum_{j = 0}^{M_b} a^i \, b^j \, \mbf f_{i, j}(\mbf u) + k^2 \, \hat D \, \frac{\partial^2 \mbf u}{\partial x^2}, \label{general_equation}
        \end{align}
        where $M_a, M_b > 0$ are the degrees of the right-hand side of \eqref{general_equation} in terms of $a$ and $b$, respectively (they could be infinite) Here, we assume that for every $i, j$, $\mbf f_{i, j}$ is an analytic function on $\mbf u$, which implies that we can expand it using Taylor series up to order infinity. In particular, if we Taylor-expand each function $\mbf f_{i, j}$ around $\mbf u = \mbf 0$, the system becomes
    	\begin{align*}
    	    \mathbb M \, \frac{\partial \mbf u}{\partial t} = \sum_{i = 0}^{M_a} \sum_{j = 0}^{M_b} a^i \, b^j \sum_{m \geq 1} \frac{1}{m!} \, \left(\sum_{r = 1}^n u_r \, \frac{\partial}{\partial u_r}\right)^m \mbf f_{i, j}(\mbf 0) + k^2 \, \hat D \, \frac{\partial^2 \mbf u}{\partial x^2}.
    	\end{align*}
        We are interested in the study of stationary solutions, so we assume that $\dfrac{\partial \mbf u}{\partial t} = \mbf 0$. Let $0 < \varepsilon \ll 1$ be a small positive number. Define a long spatial variable by $X = \varepsilon^2 \, x$. Thus, the system becomes
        \begin{align}
            \mbf 0 = \sum_{i = 0}^{M_a} \sum_{j = 0}^{M_b} a^i \, b^j \sum_{m \geq 1} \frac{1}{m!} \, \left(\sum_{r = 1}^n u_r \, \frac{\partial}{\partial u_r}\right)^m \mbf f_{i, j}(\mbf 0) + k^2 \, \hat D \, \left(\frac{\partial^2 \mbf u}{\partial x^2} + 2 \, \varepsilon^2 \, \frac{\partial^2 \mbf u}{\partial x \partial X} + \varepsilon^4 \, \frac{\partial^2 \mbf u}{\partial X^2}\right). \label{eqtoexpand}
        \end{align}
    	Furthermore, consider the following expansions:
    	\begin{align*}
    	    \mbf u &\to \sum_{r = 1}^N \varepsilon^r \, \mbf u^{[r]} + \mbf R_N,
    	    \\
    	    a &\to \sum_{p \geq  1} a_p \, \varepsilon^p,
    	    \\
    	    b &\to \sum_{q \geq  1} b_q \, \varepsilon^q + \delta b,
    	\end{align*}
        where $\mbf R_N$ stands for the remainder of our approximation of the solution after truncating the asymptotic expansion at order $N$, for every pair of integers $p, q \geq 1$, $a_p, b_q$ are real numbers that approximate the Maxwell point, and $\delta b$ is the separation of $b$ from such a point. Our first goal is then to find an approximation of the Maxwell point that exists close to codimension-two Turing bifurcation points, together with an approximation of a solution that joins the homogeneous steady state $\mbf 0$ to a large-amplitude periodic orbit with the same minimal energy as $\mbf 0$. We do this by finding suitable parameters in the asymptotic expansion to find a Maxwell point, whilst making sure that said asymptotic expansion is valid.
    
        Now, for each integer $0 \leq i \leq M_a, 0 \leq j \leq M_b$, we introduce the following symmetric, multilinear vector functions
        \begin{equation}
            \begin{aligned}
                \mbf F_{2, i, j}\left(\mbf v_1, \mbf v_2\right) &= \frac{1}{2!} \, \left.\left(\sum_{1 \leq p_1, p_2 \leq n} v_{1, p_1} \, v_{2, p_2} \, \frac{\partial^2 \mbf f_{i, j}(\mbf u)}{\partial u_{p_1} \partial u_{p_2}}\right)\right|_{\mbf u = \mbf 0},
        		\\ 
        		\mbf F_{3, i, j}\left(\mbf v_1, \mbf v_2, \mbf v_3\right) &= \frac{1}{3!} \, \left.\left(\sum_{1 \leq p_1, p_2, p_3 \leq n} v_{1, p_1} \, v_{2, p_2} \, v_{3, p_3} \, \frac{\partial^3 \mbf f_{i, j}(\mbf u)}{\partial u_{p_1} \partial u_{p_2} \partial u_{p_3}}\right)\right|_{\mbf u = \mbf 0},
        		\\
        		\mbf F_{4, i, j}\left(\mbf v_1, \mbf v_2, \mbf v_3, \mbf v_4\right) &= \frac{1}{4!} \, \left.\left(\sum_{1 \leq p_1, p_2, p_3, p_4 \leq n} v_{1, p_1} \, v_{2, p_2} \, v_{3, p_3} \, v_{4, p_4} \, \frac{\partial^4 \mbf f_{i, j}}{\partial u_{p_1} \partial u_{p_2} \partial u_{p_3} \partial u_{p_4}}\right)\right|_{\mbf u = \mbf 0},
        		\\
        		\mbf F_{5, i, j}\left(\mbf v_1, \mbf v_2, \mbf v_3, \mbf v_4, \mbf v_5\right) &= \begin{multlined}[t]
        		    \frac{1}{5!} \, \left. \left(\sum_{1 \leq p_1, p_2, p_3, p_4, p_5 \leq n} v_{1, p_1} \, v_{2, p_2} \, v_{3, p_3} \, v_{4, p_4} \, v_{5, p_5}\right.\right.
                    \\
                    \times \left.\left. \frac{\partial^5 \mbf f_{i, j}}{\partial u_{p_1} \partial u_{p_2} \partial u_{p_3} \partial u_{p_4} \partial u_{p_5}}\right)\right|_{\mbf u = \mbf 0},
        		\end{multlined}
            \end{aligned} \label{symmetric_linear_vector_functions}
        \end{equation}
        where $\mbf v_\ell = \left(v_{\ell, 1}, \ldots, v_{\ell, n}\right)^\intercal$, for $\ell = 1, \ldots, 5$, and we note that these functions can be naturally generalized to higher orders.
    
        Also, for simplicity of notation, we define
        \begin{align*}
            \mathcal M_\ell = \jac \mbf f_{0, 0}(\mbf 0) - \ell^2 \, k^2 \, \hat D,
        \end{align*}
        for each non-negative integer $\ell$, and we note that $\mathcal M_\ell$ is invertible for every $\ell \neq 1$.

    \subsection{Construction of a regular amplitude equation}
        This section aims to show how to develop a regular asymptotic expansion for the amplitude of patterns arising at a Turing bifurcation, $A_1 = A_1(X) \in \mathbb C$, under the ansatz $u \sim A_1 \, e^{ix} \, \bs \phi_1^{[1]}$ where $\bs \phi_1^{[1]}$ is an eigenvector of $\mathcal M_1$. In particular, we want to show that $A_1$ satisfies an amplitude equation of the form
        \begin{align}
    	    \alpha_1 \, A_1'' + i \, \alpha_2 \, A_1' + i \, \alpha_3 \, \abs{A_1}^2 \, A_1' + i \, \alpha_4 \, A_1^2 \, \bar A_1' + \alpha_5 \, A_1 + \alpha_6 \, \abs{A_1}^2 \, A_1 + \alpha_7 \, \abs{A_1}^4 A_1 = 0, \label{firstamplitudeeq}
    	\end{align}
        where $\alpha_i \in \mathbb R$ for each $i = 1, \ldots, 7$, and $\alpha_1 \neq 0$. For this, we perform an asymptotic expansion of system \eqref{eqtoexpand} in terms of $\varepsilon$. We do this order by order as follows.
        
        \paragraph{Order $\mathcal O(\mathcal \varepsilon)$.} At this order, \eqref{eqtoexpand} becomes
        \begin{align*}
            \mbf 0 = \jac \mbf f_{0, 0}(\mbf 0) \, \mbf u^{[1]} + k^2 \, \hat D \, \frac{\partial^2 \mbf u^{[1]}}{\partial x^2},
        \end{align*}
        which implies
        \begin{align}
            \mbf u^{[1]} = A_1 \, e^{i\hat x} \, \bs \phi_1^{[1]} + c.c., \label{first_order_approximation}
        \end{align}
        where $c.c.$ stands for the complex-conjugate of the term to the left of it, $\bs \phi_1^{[1]} \neq \mbf 0$ fulfills
        \begin{align}
            \mathcal M_1 \, \bs \phi_1^{[1]} = \mbf 0, \label{phidef}
        \end{align}
        and $\hat x = x - \hat \chi$, where $\hat \chi$ is a real variable that stands for a spatial translation in $x$, which plays a key degree of freedom one needs for matching solutions in \ref{ap:joining_fronts}. However, for simplicity of notation, when there is no confusion, we will simply write $x$ instead of $\hat x$.
        
        \paragraph{Order $\mathcal O\left(\varepsilon^2\right)$.} At this order, \eqref{eqtoexpand} becomes
        \begin{align*}
            \mbf 0 = \jac \mbf f_{0, 0}(\mbf 0) \, \mbf u^{[2]} + a_1 \, \jac \mbf f_{1, 0}(\mbf 0) \, \mbf u^{[1]} + b_1 \, \jac \mbf f_{0, 1}(\mbf 0) \, \mbf u^{[1]} + \mbf F_{2, 0, 0}\left(\mbf u^{[1]}, \mbf u^{[1]}\right) + k^2 \, \hat D \, \frac{\partial^2 \mbf u^{[2]}}{\partial x^2},
        \end{align*}
        which implies
        \begin{align*}
            \left(\jac \mbf f_{0, 0}(\mbf 0) + k^2 \, \hat D \, \frac{\partial^2}{\partial x^2}\right) \, \mbf u^{[2]} &= - a_1 \, A_1 \, e^{ix} \, \jac \mbf f_{1, 0}(\mbf 0) \, \bs \phi_1^{[1]} - b_1 \, A_1 \, e^{ix} \, \jac \mbf f_{0, 1}(\mbf 0) \, \bs \phi_1^{[1]}
            \\
            & \quad - \mbf F_{2, 0, 0}\left(\bs \phi_1^{[1]}, \bs \phi_1^{[1]}\right) \left(\abs{A_1}^2 + A_1^2 \, e^{2ix}\right) + c.c.
        \end{align*}
        Therefore, as this equation is linear, we have that
        \begin{align*}
            \mbf u^{[2]} = \abs{A_1}^2 \, \mbf W_0^{[2]} + A_1 \, e^{ix} \, \mbf W_1^{[2]} + A_1^2 \, e^{2ix} \, \mbf W_2^{[2]} + A_2 \, e^{ix} \, \bs \phi_1^{[1]} + c.c.,
        \end{align*}
        where $A_2 = A_2(X)$,
        \begin{align}
            \mathcal M_0 \, \mbf W_0^{[2]} &= - \mbf F_{2, 0, 0}\left(\bs \phi_1^{[1]}, \bs \phi_1^{[1]}\right), \notag
            \\
            \mathcal M_1 \, \mbf W_1^{[2]} &= - a_1 \, \jac \mbf f_{1, 0}(\mbf 0) \, \bs \phi_1^{[1]} - b_1 \, \jac \mbf f_{0, 1}(\mbf 0) \, \bs \phi_1^{[1]}, \label{solvcond1}
        \end{align}
        and
        \begin{align*}
            \mathcal M_2 \, \mbf W_2^{[2]} &= - \mbf F_{2, 0, 0}\left(\bs \phi_1^{[1]}, \bs \phi_1^{[1]}\right)
        \end{align*}
        Here, we recall that $\det\left(\evs{\mathcal M}_1\right) = 0$. Therefore, \eqref{solvcond1} might not have a solution for every value of $a_1, b_1 \in \mathbb R$. For this expansion to be valid, we need all the equations to have a solution, so we make use of the Fredholm alternative, which states that for a generic Fredholm operator, $L$, we have that
        \begin{align}
            \mbox{im}(L)^\perp = \ker\left(L^*\right), \label{fredholm}
        \end{align}
        In particular, using the usual inner product of vectors in $\mathbb C^n$, we have that for a real matrix $M$, $M^* = M^\intercal$. Considering this, we define $\bs \psi_\pm = e^{\pm ix} \, \bs \psi$, where $\bs \psi \neq \mbf 0$ fulfills
        \begin{align}
            \mathcal M_1^\intercal \, \bs \psi = \mbf 0. \label{psidef}
        \end{align}
        With this, and using \eqref{fredholm}, we choose $a_1, b_1 \in \mathbb R$ such that
        \begin{align}
            \left \langle a_1 \, \jac \mbf f_{1, 0}(\mbf 0) \, \bs \phi_1^{[1]} + b_1 \, \jac \mbf f_{0, 1}(\mbf 0) \, \bs \phi_1^{[1]}, \bs \psi \right \rangle = 0, \label{a_1,b_1secondorder}
        \end{align}
        which ensures that \eqref{solvcond1} has a solution.
    
        \paragraph{Order $\mathcal O\left(\eps^3\right)$.} At this order, \eqref{eqtoexpand} becomes
    	\begin{align*}
    	    \mbf 0 = \jac \mbf f_{0, 0}(\mbf 0) \, \mbf u^{[3]} + a_1 \, \jac \mbf f_{1, 0}(\mbf 0) \, \mbf u^{[2]} + b_1 \, \jac \mbf f_{0, 1}(\mbf 0) \, \mbf u^{[2]} + a_2 \, \jac \mbf f_{1, 0}(\mbf 0) \, \mbf u^{[1]} + b_2 \, \jac \mbf f_{0, 1}(\mbf 0) \, \mbf u^{[1]}
            \\
            + a_1^2 \, \jac \mbf f_{2, 0}(\mbf 0) \, \mbf u^{[1]} + a_1 \, b_1 \, \jac \mbf f_{1, 1}(\mbf 0) \, \mbf u^{[1]} + b_1^2 \, \jac \mbf f_{0, 2}(\mbf 0) \, \mbf u^{[1]} + 2 \, \mbf F_{2, 0, 0}\left(\mbf u^{[1]}, \mbf u^{[2]}\right)
            \\
            + a_1 \, \mbf F_{2, 1, 0}\left(\mbf u^{[1]}, \mbf u^{[1]}\right) + b_1 \, \mbf F_{2, 0, 1}\left(\mbf u^{[1]}, \mbf u^{[1]}\right) + \mbf F_{3, 0, 0}\left(\mbf u^{[1]}, \mbf u^{[1]}, \mbf u^{[1]}\right)
    	    \\
            + k^2 \, \hat D \, \left(\frac{\partial^2 \mbf u^{[3]}}{\partial x^2} + 2 \, \frac{\partial^2 \mbf u^{[1]}}{\partial x \partial X}\right),
    	\end{align*}
    	which implies
    	\begin{align*}
    	    \left(\jac \mbf f_{0, 0}(\mbf 0) + k^2 \, \hat D \, \frac{\partial^2}{\partial x^2}\right) \, \mbf u^{[3]} &= - a_1 \, \jac \mbf f_{1, 0}(\mbf 0) \, \mbf u^{[2]} - b_1 \, \jac \mbf f_{0, 1}(\mbf 0) \, \mbf u^{[2]} - a_2 \, \jac \mbf f_{1, 0}(\mbf 0) \, \mbf u^{[1]}
            \\
            & \quad - b_2 \, \jac \mbf f_{0, 1}(\mbf 0) \, \mbf u^{[1]} - a_1^2 \, \jac \mbf f_{2, 0}(\mbf 0) \, \mbf u^{[1]} - a_1 \, b_1 \, \jac \mbf f_{1, 1}(\mbf 0) \, \mbf u^{[1]}
            \\
            & \quad - b_1^2 \, \jac \mbf f_{0, 2}(\mbf 0) \, \mbf u^{[1]} - 2 \, \mbf F_{2, 0, 0}\left(\mbf u^{[1]}, \mbf u^{[2]}\right) - a_1 \, \mbf F_{2, 1, 0}\left(\mbf u^{[1]}, \mbf u^{[1]}\right)
            \\
            & \quad - b_1 \, \mbf F_{2, 0, 1}\left(\mbf u^{[1]}, \mbf u^{[1]}\right) - \mbf F_{3, 0, 0}\left(\mbf u^{[1]}, \mbf u^{[1]}, \mbf u^{[1]}\right) - 2 \, k^2 \, \hat D \, \frac{\partial^2 \mbf u^{[1]}}{\partial x \partial X},
    	\end{align*}
    	Therefore, the solution to this equation is given by:
    	\begin{align*}
    	    \mbf u^{[3]} &= \abs{A_1}^2 \, \mbf W_0^{[3]} + A_1 \, e^{ix} \, \mbf W_1^{[3]} + \abs{A_1}^2 A_1 \, e^{ix} \, \mbf W_{1, 2}^{[3]} + i \, A_1' \, e^{ix} \, \mbf W_{1, 3}^{[3]} + A_1^2 \, e^{2ix} \, \mbf W_2^{[3]}
    	    \\
    	    & \quad + A_1^3 \, e^{3ix} \, \mbf W_3^{[3]} + 2 \, A_1 \, \bar A_2 \, \mbf W_0^{[2]} + A_2 \, e^{ix} \, \mbf W_1^{[2]} + A_3 \, e^{ix} \, \bs \phi_1^{[1]} + 2 \, A_1 \, A_2 \, e^{2ix} \, \mbf W_2^{[2]} + c.c.,
    	\end{align*}
    	where $A_3 = A_3(X)$, the sign $'$ denotes differentiation with respect to $X$,
    	\begin{align*}
    	    \mathcal M_0 \, \mbf W_0^{[3]} &= - a_1 \, \jac \mbf f_{1, 0}(\mbf 0) \, \mbf W_0^{[2]} - b_1 \, \jac \mbf f_{0, 1}(\mbf 0) \, \mbf W_0^{[2]} - 2 \, \mbf F_{2, 0, 0}\left(\bs \phi_1^{[1]}, \mbf W_1^{[2]}\right)
            \\
            & \quad - a_1 \, \mbf F_{2, 1, 0}\left(\bs \phi_1^{[1]}, \bs \phi_1^{[1]}\right) - b_1 \, \mbf F_{2, 0, 1}\left(\bs \phi_1^{[1]}, \bs \phi_1^{[1]}\right),
    	\end{align*}
        \begin{equation}
            \begin{aligned}
        	    \mathcal M_1 \, \mbf W_1^{[3]} &= - a_1 \, \jac \mbf f_{1, 0}(\mbf 0) \, \mbf W_1^{[2]} - b_1 \, \jac \mbf f_{0, 1}(\mbf 0) \, \mbf W_1^{[2]} - a_2 \, \jac \mbf f_{1, 0}(\mbf 0) \, \bs \phi_1^{[1]} - b_2 \, \jac \mbf f_{0, 1}(\mbf 0) \, \bs \phi_1^{[1]}
                \\
                & \quad  - a_1^2 \, \jac \mbf f_{2, 0}(\mbf 0) \, \bs \phi_1^{[1]} - a_1 \, b_1 \, \jac \mbf f_{1, 1}(\mbf 0) \, \bs \phi_1^{[1]} - b_1^2 \, \jac \mbf f_{0, 2}(\mbf 0) \, \bs \phi_1^{[1]},
            \end{aligned} \label{solvcond13}
        \end{equation}
        \begin{align}
    	    \mathcal M_1 \, \mbf W_{1, 2}^{[3]} = - 2 \, \mbf F_{2, 0, 0}\left(\bs \phi_1^{[1]}, 2 \, \mbf W_0^{[2]} + \mbf W_2^{[2]}\right) - 3 \, \mbf F_{3, 0, 0}\left(\bs \phi_1^{[1]}, \bs \phi_1^{[1]}, \bs \phi_1^{[1]}\right), \label{solvcond23}
    	\end{align}
    	\begin{align}
    	    \mathcal M_1 \, \mbf W_{1, 3}^{[3]} = - 2 \, k^2 \, \hat D \, \bs \phi_1^{[1]}, \label{solvcond33}
    	\end{align}
    	\begin{align*}
    	    \mathcal M_2 \, \mbf W_2^{[3]} &= - a_1 \, \jac \mbf f_{1, 0}(\mbf 0) \, \mbf W_2^{[2]} - b_1 \, \jac \mbf f_{0, 1}(\mbf 0) \, \mbf W_2^{[2]} - 2 \, \mbf F_{2, 0, 0}\left(\bs \phi_1^{[1]}, \mbf W_1^{[2]}\right)
            \\
            & \quad - a_1 \, \mbf F_{2, 1, 0}\left(\bs \phi_1^{[1]}, \bs \phi_1^{[1]}\right) - b_1 \, \mbf F_{2, 0, 1}\left(\bs \phi_1^{[1]}, \bs \phi_1^{[1]}\right),
    	\end{align*}
    	and
    	\begin{align*}
    	    \mathcal M_3 \, \mbf W_3^{[3]} = - 2 \, \mbf F_{2, 0, 0}\left(\bs \phi_1^{[1]}, \mbf W_2^{[2]}\right) - \mbf F_{3, 0, 0}\left(\bs \phi_1^{[1]}, \bs \phi_1^{[1]}, \bs \phi_1^{[1]}\right).
    	\end{align*}
    	Note that, as stated after \eqref{no-DT}, when $(a, b) = (0, 0)$, we are at a codimension-two Turing bifurcation point. Therefore, from \cite{TOMS}, we have that
    	\begin{align*}
    	    \left \langle 2 \, \mbf F_{2, 0, 0}\left(\bs \phi_1^{[1]}, 2 \, \mbf W_0^{[2]} + \mbf W_2^{[2]}\right) + 3 \,\mbf F_{3, 0, 0}\left(\bs \phi_1^{[1]}, \bs \phi_1^{[1]}, \bs \phi_1^{[1]}\right), \bs \psi\right \rangle = 0,
    	\end{align*}
        which implies that \eqref{solvcond23} has a solution.
        
    	On the other hand, to ensure that \eqref{solvcond13} has a solution, we need that $a_1, b_1, a_2, b_2 \in \mathbb R$ fulfill
    	\begin{align*}
    	    \left \langle a_1 \, \jac \mbf f_{1, 0}(\mbf 0) \, \mbf W_1^{[2]} + b_1 \, \jac \mbf f_{0, 1}(\mbf 0) \, \mbf W_1^{[2]} + a_2 \, \jac \mbf f_{1, 0}(\mbf 0) \, \bs \phi_1^{[1]} + b_2 \, \jac \mbf f_{0, 1}(\mbf 0) \, \bs \phi_1^{[1]}\right.
            \\
            \left. + a_1^2 \, \jac \mbf f_{2, 0}(\mbf 0) \, \bs \phi_1^{[1]} + a_1 \, b_1 \, \jac \mbf f_{1, 1}(\mbf 0) \, \bs \phi_1^{[1]} + b_1^2 \, \jac \mbf f_{0, 2}(\mbf 0) \, \bs \phi_1^{[1]}, \bs \psi \right \rangle = 0.
    	\end{align*}    	
    	Next, to see that \eqref{solvcond33} has a solution, we consider the following standard result:
    	\newtheorem{lemma}{Lemma}
    	\begin{lemma} \label{lemma:firstsolvcond}
    	    For $\bs \phi_1^{[1]}$ and $\bs \psi$ defined in \eqref{phidef} and \eqref{psidef}, respectively, we have that
    	    \begin{align*}
    	        \left \langle \hat D \, \bs \phi_1^{[1]}, \bs \psi \right \rangle = 0.
    	    \end{align*}
    	\end{lemma}
    	\begin{proof}
    	    For each $k \in \mathbb R$, we can think of $\bs \phi_1^{[1]} = \bs \phi_1^{[1]}(k)$ as a function of $k$ is defined by the equation
    	    \begin{align*}
    	        \left(\jac \, \mbf f_{0, 0}(\mbf 0) - k^2 \, \hat D\right) \, \bs \phi_1^{[1]} = \lambda(k) \, \bs \phi_1^{[1]},
    	    \end{align*}
    	    where $\lambda(k)$ represents the eigenvalue of $\jac \, \mbf f_{0, 0}(\mbf 0) - k^2 \, \hat D$ with the largest real part such that $\lambda\left(k^*\right) = 0$ for the value $k = k^*$ where the Turing bifurcation occurs. With this, after differentiating this expression with respect to $k$, we have
    	    \begin{align*}
    	        - 2 \, k \, \hat D \, \bs \phi_1^{[1]} + \left(\jac \, \mbf f_{0, 0}(\mbf 0) - k^2 \, \hat D\right) \, \frac{\dd \bs \phi_1^{[1]}}{\dd k} = \frac{\dd \lambda}{\dd k}(k) \, \bs \phi_1^{[1]} + \lambda(k) \, \frac{\dd \bs \phi_1^{[1]}}{\dd k}.
    	    \end{align*}
            Therefore, when evaluating \evs{$k$ at} $k^*$ ---and dropping the star sign---, we obtain
    	    \begin{align*}
    	        \mathcal M_1 \, \frac{\dd \bs \phi_1^{[1]}}{\dd k} = 2 \, k \, \hat D \, \bs \phi_1^{[1]}.
    	    \end{align*}
    	    Thus, when applying the inner product with respect to $\bs \psi$, we see that
    	    \begin{align*}
    	        2 \, k \, \left \langle \hat D \, \bs \phi_1^{[1]}, \bs \psi\right \rangle = \left \langle \mathcal M_1 \, \frac{\dd \bs \phi_1^{[1]}}{\dd k}, \bs \psi\right \rangle = \left \langle \frac{\dd \bs \phi_1^{[1]}}{\dd k}, \mathcal M_1^\intercal \, \bs \psi\right \rangle = 0,
    	    \end{align*}
    	    which concludes the proof.
    	\end{proof}

        \evs{The rest of the calculation needed to obtain the amplitude equation at order five, together with the definitions of its coefficients (see \eqref{firstamplitudeeq}), can be carried out, but they are cumbersome, so we move those details to \ref{ap:order 5}.}

    \subsection{Solving the amplitude equation} \label{sub:solving_amplitude}
        In this section, we look for heteroclinic orbits in \eqref{firstamplitudeeq}, between the trivial steady state $A_1 = 0$ and a finite-amplitude periodic pattern, whose amplitude is to be determined. We are interested in such fronts because they help us define the Maxwell point and, generically, give rise to homoclinic snaking, which we aim to study through \evs{an expansion} beyond all orders.
      
        To solve \eqref{firstamplitudeeq}, we write $A_1$ in polar coordinates as $A_1 = R_1 \, e^{i \, \varphi_1}$, where $R_1 = R_1(X) \in \mathbb R$ and $\varphi_1 = \varphi_1(X) \in \mathbb R$. With this, \eqref{firstamplitudeeq} becomes
        \begin{align*}
            \alpha_1 \, R_1'' - \alpha_1 \, R_1 \, \left(\varphi_1'\right)^2 - \alpha_2 \, R_1 \, \varphi_1' - \alpha_3 \, R_1^3 \, \varphi_1' + \alpha_4 \, R_1^3 \, \varphi_1' + \alpha_5 \, R_1 + \alpha_6 \, R_1^3 + \alpha_7 \, R_1^5
            \\
            + i \, \left(\alpha_1 \, R_1 \, \varphi_1'' + 2 \, \alpha_1 \, R_1' \, \varphi_1' + \alpha_2 \, R_1' + \left(\alpha_3 + \alpha_4\right) \, R_1^2 \, R_1'\right) = 0,
        \end{align*}
    	which can be split into two equations given by
    	\begin{align}
            \alpha_1 \, R_1'' - \alpha_1 \, R_1 \, \left(\varphi_1'\right)^2 - \alpha_2 \, R_1 \, \varphi_1' - \alpha_3 \, R_1^3 \, \varphi_1' + \alpha_4 \, R_1^3 \, \varphi_1' + \alpha_5 \, R_1 + \alpha_6 \, R_1^3 + \alpha_7 \, R_1^5 = 0, \label{realeq}
    	\end{align}
        and
        \begin{align}
            \alpha_1 \, R_1 \, \varphi_1'' + 2 \, \alpha_1 \, R_1' \, \varphi_1' + \alpha_2 \, R_1' + \left(\alpha_3 + \alpha_4\right) \, R_1^2 \, R_1' = 0. \label{imageq}
        \end{align}
    	Now, note that \eqref{imageq} is equivalent to
        \begin{align*}
             R_1^2 \, \varphi_1'' + 2 \, R_1 \, R_1' \, \varphi_1' &= - \frac{\alpha_3 + \alpha_4}{\alpha_1} \, R_1^3 \, R_1' - \frac{\alpha_2}{\alpha_1} \, R_1 \, R_1',
        \end{align*}
    	which can be integrated directly to obtain
        \begin{align}
            \varphi_1' &= - \frac{\alpha_3 + \alpha_4}{4 \, \alpha_1} \, R_1^2 - \frac{\alpha_2}{2 \, \alpha_1}. \label{varphi}
        \end{align}
    	On the other hand, when replacing \eqref{varphi} into \eqref{realeq}, we obtain
    	\begin{align}
    	    2 \, R_1'' = \frac{\dd V}{\dd R_1}, \label{diffusionequation}
    	\end{align}
    	where
    	\begin{align}
    	    \frac{\dd V}{\dd R_1} &= 2 \, \beta_1 \, R_1 + 4 \, \beta_3 \, R_1^3 + 6 \, \beta_5 \, R_1^5, \notag
    	    \\
    	    \beta_1 &= - \frac{1}{\alpha_1} \, \left(\alpha_5 + \frac{\alpha_2^2}{4 \, \alpha_1}\right), \notag
    	    \\
    	    \beta_3 &= - \frac{1}{2 \, \alpha_1} \, \left(\alpha_6 + \frac{\alpha_2 \, \left(\alpha_3 - \alpha_4\right)}{2 \, \alpha_1}\right), \label{betadef}
    	    \\
    	    \beta_5 &= - \frac{1}{3 \, \alpha_1} \left(\alpha_7 + \frac{\left(\alpha_3 + \alpha_4\right) \left(3 \, \alpha_3 - 5 \, \alpha_4\right)}{16 \, \alpha_1}\right), \notag
    	\end{align}
        \evs{which implies that these new coefficients, $\beta_1$, $\beta_3$, $\beta_5$ are shifted and scaled versions of $\alpha_5$, $\alpha_6$ and $\alpha_7$ according to the spatially-related terms present in \eqref{firstamplitudeeq}.}
        
    	Note that \eqref{diffusionequation} is invariant under the change $X \to \upsilon - X$ for every $\upsilon \in \mathbb R$. Furthermore, we have that $V$ represents an energy potential which, for simplicity, we subject to the condition $V(0) = 0$. This implies that
        \begin{align*}
            V\left(R_1\right) = \beta_1 \, R_1^2 + \beta_3 \, R_1^4 + \beta_5 \, R_1^6,
        \end{align*}
        which has critical points at $R_1 = 0$ and
    	\begin{align*}
    	    R_{1, \pm} &= \sqrt{\frac{- \beta_3 \pm \sqrt{\beta_3^2 - 3 \, \beta_1 \, \beta_5}}{3 \, \beta_5}}.
    	\end{align*}
    	We need these critical points to be real, so we require $\beta_1 \, \beta_5 > 0$, $\beta_3 \, \beta_5 < 0$, and $\beta_3^2 \geq 3 \, \beta_1 \, \beta_5$.
    	
    	Furthermore, note that
    	\begin{align}
    	    \frac{\dd^2 V}{\dd R_1^2}(0) &= 2 \, \beta_1, \quad \text{and} \notag
    	    \\
    	    \frac{\dd^2 V}{\dd R_1^2}\left(R_{1, +}\right) &= \frac{8}{3} \, \frac{\beta_3^2 - 3 \, \beta_1 \, \beta_5 - \beta_3 \, \sqrt{\beta_3^2 - 3 \, \beta_1 \, \beta_5}}{\beta_5}. \label{second-quantity}
    	\end{align}
    	To find a Maxwell point, we need to find the condition under which $R_1 = 0$ and $R_1 = R_{1, +}$ have the same minimal energy, which implies that we need $\beta_1 > 0$, $\beta_3 < 0$, and $\beta_5 > 0$. Furthermore, with these assumptions, we have that \eqref{second-quantity} is positive, which implies that $R_{1, +}$ is also a minimum.
        
    	In addition, we require $R_1 = 0$ to have the same minimal energy as $R_{1, +}$. Therefore, we need
    	\begin{align}
    	    V(0) = V\left(R_{1, +}\right) \qquad \iff \qquad \beta_3^2 = 4 \, \beta_1 \, \beta_5, \label{Second Maxwell}
    	\end{align}
        which defines a first approximation to the Maxwell point and is an assumption we make from now on.
        
        With this, if we multiply \eqref{diffusionequation} by $R_1'$ and integrate it once, we obtain
    	\begin{align}
    	    \left(R_1'\right)^2 = R_1^2 \, \left(\beta_1 + \beta_3 \, R_1^2 + \beta_5 \, R_1^4\right), \label{maxwell-consequence}
    	\end{align}
        where we note that the right-hand side is a positive multiple of a parabola in $R_1^2$ with a positive leading-order coefficient and a discriminant equal to zero due to \eqref{Second Maxwell}. This implies that the right-hand side of \eqref{maxwell-consequence} is non-negative, making \eqref{maxwell-consequence} well-defined.
        
    	Let us now consider the following change $R_1^2 = 1/\mathcal P$, which implies:
    	\begin{align*}
    		R_1' = -\frac{\mathcal P'}{2 \, \mathcal P^{3/2}} \qquad \text{and} \qquad \left(\frac{\mathcal P'}{2 \, \mathcal P^{3/2}}\right)^2= \frac{\beta_1}{\mathcal P} + \frac{\beta_3}{\mathcal P^2} + \frac{\beta_5}{\mathcal P^3}.
    	\end{align*}
    	Therefore,
    	\begin{align*}
    		\left(\mathcal P'\right)^2= 4 \, \left(\beta_1 \, \mathcal P^2 + \beta_3 \, \mathcal P + \beta_5\right).
    	\end{align*}
    	We proceed to solve this equation using the method of separation of variables. In particular, if we consider the negative branch of this equation when taking the square root, we have that
    	\begin{align*}
    		\frac{1}{2} \, \int \frac{\dd \mathcal P}{\sqrt{\beta_1 \, \mathcal P^2 +  \beta_3 \, \mathcal P + \beta_5}} = - X + C,
    	\end{align*}
        where $C \in \mathbb R$ is a constant of integration. Now, let $\mathcal Q = \sqrt{\beta_1 \, \mathcal P^2 + \beta_3 \, \mathcal P + \beta_5}$. Then,
    	\begin{align*}
    		\beta_1 \, \mathcal P^2 + \beta_3 \, \mathcal P + \beta_5 - \mathcal Q^2 = 0,
    	\end{align*}
    	which implies
    	\begin{align*}
    		\mathcal P = \frac{- \beta_3 \pm \sqrt{\beta_3^2 - 4 \, \beta_1 \, \left(\beta_5 - \mathcal Q^2\right)}}{2 \, \beta_1} = \frac{- \beta_3 \pm 2 \, \sqrt{\beta_1} \, \mathcal Q}{2 \, \beta_1} \qquad \text{and} \qquad \dd \mathcal P = \pm \, \frac{\dd \mathcal Q}{\sqrt{\beta_1}}.
    	\end{align*}
    	With this, we have
    	\begin{align*}
    		\frac{1}{2} \, \int \frac{\dd \mathcal P}{\sqrt{\beta_1 \, \mathcal P^2 + \beta_3 \, \mathcal P + \beta_5}} = \frac{1}{2} \, \int \frac{\dd \mathcal Q}{\mathcal Q \, \sqrt{\beta_1}} = \frac{1}{2 \, \sqrt{\beta_1}} \, \log\left(\sqrt{\beta_1 \, \mathcal P^2 + \beta_3 \, \mathcal P + \beta_5}\right),
    	\end{align*}
    	where we choose the positive sign for $\mathcal P$, as it needs to be non-negative, by definition.
    	
    	With this, we note that
        \begin{align*}
            \frac{1}{2 \, \sqrt{\beta_1}} \, \log\left(\sqrt{\beta_1 \, \mathcal P^2 + \beta_3 \, \mathcal P + \beta_5}\right) &= - X + C
            \\
            \iff \beta_1 \, \mathcal P^2 + \beta_3 \, \mathcal P + \beta_5 &= C^2 \, \exp\left(- 4 \, \sqrt{\beta_1} \, X\right),
        \end{align*}
        which yields
        \begin{align*}
            \mathcal P = \frac{1}{R_1^2} &= \frac{- \beta_3 + \sqrt{\beta_3^2 - 4 \, \beta_1 \, \left(\beta_5 - C^2 \, \exp\left(- 4 \, \sqrt{\beta_1} \, X\right)\right)}}{2 \, \beta_1}
            \\
            &= \frac{- \beta_3 + 2 \, \sqrt{\beta_1} \, C \, \exp\left(- 2 \, \sqrt{\beta_1} \, X\right)}{2 \, \beta_1}
        \end{align*}
        Therefore, we conclude that
        \begin{align}
            R_1^2 &= \frac{4 \, \beta_1}{- 2 \, \beta_3 + \exp\left(- 2 \, \sqrt{\beta_1} \, X\right)}, \label{R_1-up_front}
        \end{align}
        where we have taken $C = \dfrac{1}{4 \, \sqrt{\beta_1}}$, for simplicity. Note that \eqref{R_1-up_front} corresponds to an up-front since
        \begin{align*}
            \lim_{X \to - \infty} R_1^2 = 0, \qquad \text{and} \qquad \lim_{X \to \infty} R_1^2 = - \frac{2 \, \beta_1}{\beta_3} > 0.
        \end{align*}
        On the other hand, as equation \eqref{diffusionequation} is invariant under the change $X \to - X$, then the down-front
        \begin{align}
            \tilde R_1^2(X) = \frac{4 \, \beta_1}{- 2 \, \beta_3 +  \exp\left(2 \, \sqrt{\beta_1} \, X\right)}, \label{R_1-down_front}
        \end{align}
        is also a solution to \eqref{diffusionequation}. Figure \ref{fig:fronts_to_join} depicts these two fronts. Specifically, the red (respectively, blue) dashed line represents the up-front (respectively, down-front) given by \eqref{R_1-up_front} (respectively, \eqref{R_1-down_front}). On the other hand, the green continuous lines show an oscillatory pattern enclosed by these fronts. These two solutions' existence is key for studying late terms in the asymptotic expansion, as their existence and matching provide conditions on some parameters of the expansion (see \ref{ap:joining_fronts}).
        \begin{figure}
            \centering
            \includegraphics[width = 0.6\linewidth]{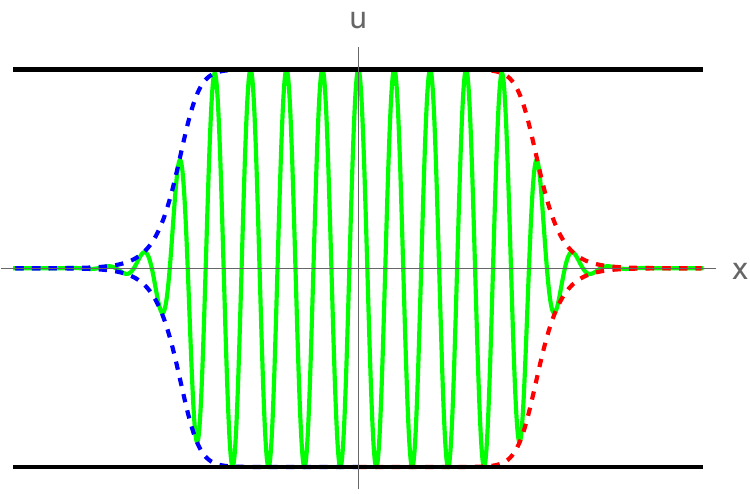}
            \caption{Form of solutions captured by the amplitude equation, \eqref{firstamplitudeeq}, at the Maxwell point. The red and blue dashed curves represent the envelope of oscillatory localized solutions to \eqref{general_equation}, represented by $R_1$. Specifically, the blue dashed line represents an up-front joining $R_1 = 0$ to $R_1 = \sqrt{- \dfrac{2 \, \beta_1}{\beta_3}}$ (represented by the top and bottom black horizontal lines) from left to right, whilst the red dashed curve represents a translated reflection of the blue curve, which turns out to be a down-front joining the same steady states in the opposite direction. Furthermore, the green continuous curve represents the oscillatory pattern enclosed by that envelope.}
            \label{fig:fronts_to_join}
        \end{figure}
        
        To keep consistency with notation, we focus \evs{in what comes next} on the up-front, \eqref{R_1-up_front}. The analysis for the down-front is analogous. This implies
        \begin{align*}
            \varphi_1' &= \frac{\beta_1 \, \left(\alpha_3 + \alpha_4\right)}{\alpha_1\left(2 \, \beta_3 - \exp\left(- 2 \, \sqrt{\beta_1} \, X\right)\right)} - \frac{\alpha_2}{2 \, \alpha_1},
        \end{align*}
        so $\varphi_1$ takes the following form:
        \begin{align}
            \varphi_1 = \eta \, \log\left(1 - 2 \, \beta_3 \, \exp\left(2 \sqrt{\beta_1} \, X\right)\right) - 2 \, \xi \, \sqrt{\beta_1} \, X - \xi \, \log\left(- 2 \, \beta_3\right) - \eta \, \log\left(2 \, \sqrt{\beta_1}\right), \label{varphi_1-up-front}
        \end{align}
        where
        \begin{align}
            \xi = \frac{\alpha_2}{4 \, \alpha_1 \, \sqrt{\beta_1}}, \qquad \text{and} \qquad \eta = \frac{\left(\alpha_3 + \alpha_4\right) \, \sqrt{\beta_1}}{4 \, \alpha_1 \, \beta_3}. \label{etaxi}
        \end{align}
        Now, to simplify notation, we consider the following translation:
        \begin{align*}
            X \to - \frac{\log\left(- 2 \, \beta_3\right)}{2 \, \sqrt{\beta_1}} + X,
        \end{align*}
        which transforms \eqref{R_1-up_front} and \eqref{varphi_1-up-front} into
        \begin{align*}
            R_1^2 &= - \frac{2 \, \beta_1}{\beta_3} \, \frac{\exp\left(2 \, \sqrt{\beta_1} \, X\right)}{1 + \exp\left(2 \, \sqrt{\beta_1} \, X\right)},
        \end{align*}
        which has singularities at $X_j = \dfrac{(2 \, j + 1) \, \pi \, i}{2 \, \sqrt{\beta_1}}$ for $j \in \mathbb Z$, and
        \begin{align*}
            \varphi_1 = \eta \, \log\left(1 + \exp\left(2 \sqrt{\beta _1} X\right)\right) - 2 \, \xi \, \sqrt{\beta_1} \, X - \eta \, \log\left(2 \, \sqrt{\beta_1}\right).
        \end{align*}
        respectively. With this, we note that
        \begin{align*}
            A_1 &= R_1 \, e^{i \, \varphi_1}
            \\
            &= \left(2 \, \sqrt{\beta_1}\right)^{1 - \eta \, i} \, \left(- 2 \, \beta_3\right)^{- \frac{1}{2}} \, \sqrt{\frac{1}{1 + \exp\left(- 2 \, \sqrt{\beta_1} \, X\right)}} \, \left(1 + \exp\left(2 \, \sqrt{\beta_1} \, X\right)\right)^{\eta \, i}
            \\
            & \quad \times \exp\left(- 2 \, i \, \sqrt{\beta_1} \, \xi \, X\right).
        \end{align*}
        Now, we note that the leading-order expansion of different functions associated with $A_1(X)$ around $X_0 = \dfrac{\pi \, i}{2 \, \sqrt{\beta_1}}$, which turns out to be one of the singularities of $R_1$ that is closest to the real axis in $X$ \big(together with $X_{- 1}$\big), is given by
        {\allowdisplaybreaks
        \begin{align*}
            A_1 \, \bs \phi_1^{[1]} &\sim - i \, C_A \, \sqrt{\frac{- 1}{X_0 - X}} \left(X_0 - X\right)^{\eta i} \, \bs \phi_1^{[1]},
            \\
            \bar A_1 \, \bs \phi_1^{[1]} &\sim - i \, C_{\bar A} \, \sqrt{\frac{- 1}{X_0 - X}} \left(X_0 - X\right)^{- \eta i} \, \bs \phi_1^{[1]},
            \\
            \abs{A_1}^2 \, \mbf W_0^{[2]} = R_1^2 \, \mbf W_0^{[2]} &\sim C_A \, C_{\bar A} \, \left(X_0 - X\right)^{- 1} \, \mbf W_0^{[2]},
            \\
            A_1^2 \, \mbf W_2^{[2]} &\sim C_A^2 \, \left(X_0 - X\right)^{- 1 + 2 \, \eta \, i} \, \mbf W_2^{[2]},
            \\
            \bar A_1^2 \, \mbf W_2^{[2]} &\sim C_{\bar A}^2 \, \left(X_0 - X\right)^{- 1 - 2 \, \eta \, i} \, \mbf W_2^{[2]},
            \allowdisplaybreaks
            \\
            \abs{A_1}^2 \, A_1 \, \mbf W_{1, 2}^{[3]} = R_1^2 \, A_1 \, \mbf W_{1, 2}^{[3]} &\sim i \, C_A^2 \, C_{\bar A} \, \left(\frac{- 1}{X_0 - X}\right)^{3/2} \left(X_0 - X\right)^{\eta i} \, \mbf W_{1, 2}^{[3]},
            \\
            \abs{A_1}^2 \, \bar A_1 \, \mbf W_{1, 2}^{[3]} = R_1^2 \, \bar A_1 \, \mbf W_{1, 2}^{[3]} &\sim i \, C_A \, C_{\bar A}^2 \, \left(\frac{- 1}{X_0 - X}\right)^{3/2} \left(X_0 - X\right)^{- \eta i} \, \mbf W_{1, 2}^{[3]},
            \\
            i \, A_1' \, \mbf W_{1, 3}^{[3]} &\sim - \frac{i}{2} \, (2 \, \eta + i) \, C_A \, \sqrt{\frac{- 1}{X_0 - X}} \, \left(X_0 - X\right)^{- 1 + \eta i} \, \mbf W_{1, 3}^{[3]},
            \\
            - i \, \bar A_1' \, \mbf W_{1, 3}^{[3]} &\sim - \frac{i}{2} \, (2 \, \eta - i) \, C_A \, \sqrt{\frac{- 1}{X_0 - X}} \, \left(X_0 - X\right)^{- 1 - \eta i} \, \mbf W_{1, 3}^{[3]},
            \\
            A_1^3 \, \mbf W_3^{[3]} &\sim i \, C_A^3 \, \left(\frac{- 1}{X_0 - X}\right)^{3/2} \left(X_0 - X\right)^{3 \, \eta \, i} \, \mbf W_3^{[3]},
            \\
            \bar A_1^3 \, \mbf W_3^{[3]} &\sim i \, C_{\bar A}^3 \, \left(\frac{- 1}{X_0 - X}\right)^{3/2} \left(X_0 - X\right)^{- 3 \, \eta \, i} \, \mbf W_3^{[3]},
            \allowdisplaybreaks
            \\
            \abs{A_1}^4 \, \mbf W_{0, 2}^{[4]} = R_1^4 \, \mbf W_{0, 2}^{[4]} &\sim C_A^2 \, C_{\bar A}^2 \, \left(X_0 - X\right)^{- 2} \, \mbf W_{0, 2}^{[4]},
            \\
            i \, \bar A_1 \, A_1' \, \mbf W_{0, 3}^{[4]} &\sim \frac{i}{2} \, (1 - 2 \, \eta \, i) \, C_A \, C_{\bar A} \, \left(X_0 - X\right)^{- 2} \, \mbf W_{0, 3}^{[4]},
            \\
            - i \, A_1 \, \bar A_1' \, \mbf W_{0, 3}^{[4]} &\sim - \frac{i}{2} \, (1 + 2 \, \eta \, i) \,  C_A \, C_{\bar A} \, \left(X_0 - X\right)^{- 2} \, \mbf W_{0, 3}^{[4]},
            \\
            \abs{A_1}^2 \, A_1^2 \, \mbf W_{2, 2}^{[4]} = R_1^2 \, A_1^2 \, \mbf W_{2, 2}^{[4]} &\sim C_A^3 \, C_{\bar A} \, \left(X_0 - X\right)^{- 2 + 2 \, \eta \, i} \, \mbf W_{2, 2}^{[4]},
            \\
            \abs{A_1}^2 \, \bar A_1^2 \, \mbf W_{2, 2}^{[4]} = R_1^2 \, \bar A_1^2 \, \mbf W_{2, 2}^{[4]} &\sim C_A \, C_{\bar A}^3 \, \left(X_0 - X\right)^{- 2 - 2 \, \eta \, i} \, \mbf W_{2, 2}^{[4]},
            \\
            i \, A_1 \, A_1' \, \mbf W_{2, 3}^{[4]} &\sim \frac{i}{2} \, (1 - 2 \, \eta \, i) \, C_A^2 \, \left(X_0 - X\right)^{- 2 + 2 \, \eta \, i} \, \mbf W_{2, 3}^{[4]},
            \\
            - i \, \bar A_1 \, \bar A_1' \, \mbf W_{2, 3}^{[4]} &\sim - \frac{i}{2} \, (1 + 2 \, \eta \, i) \, C_{\bar A}^2 \, \left(X_0 - X\right)^{- 2 - 2 \, \eta \, i} \, \mbf W_{2, 3}^{[4]}.
        \end{align*}
        }
        where
        {\allowdisplaybreaks
        \begin{align}
            C_A &= i \, e^{\pi \, \xi} \, \left(2 \, \sqrt{\beta_1}\right)^{1/2} \, \left(- 2 \, \beta_3\right)^{- 1/2}, \label{C_A}
            \\
            C_{\bar A} &= i \, e^{- \pi \, \xi} \, \left(2 \, \sqrt{\beta_1}\right)^{1/2} \, \left(- 2 \, \beta_3\right)^{- 1/2}. \notag
        \end{align}
        }
        \begin{remark}            
            First, note that
            \begin{align}
                C_{\bar A} &= - e^{- 2 \, \pi \, \xi} \, \bar C_A. \label{mostwonderfulidea}
            \end{align}            
            Furthermore, observe that the expression of the leading-order term of $A_1$ at $X_0$ can be simplified even further if we assume
            \begin{align*}
                \sqrt{\frac{- 1}{X_0 - X}} = i \, \sqrt{\frac{1}{X_0 - X}}.
            \end{align*}
            Nonetheless, as we are now working with a complex variable to carry out the analysis, we have to bear in mind that, in the complex plane, some functions are not uniquely defined. We will then state everything in this manner to maintain a general approach to development.
        \end{remark}
        
\section{Late-term expansion --- inner solution} \label{sec:late_term}
    In this section, we aim to study the asymptotic behaviour of our \evs{amplitude}, under the assumption that we carry out the asymptotic expansion up to a high order $N\gg 1$. In particular, as explained in \cite{Chapman}, to understand the leading-order asymptotics in this case, which allows us to obtain an approximation of the width of the homoclinic snaking, we must find two kinds of solutions: an \textit{inner} and an \textit{outer} solution. The inner solution refers to an approximation that is valid only sufficiently close to a singularity, and the outer solution lets us understand the behaviour of solutions away from it. Of course, as these solutions are intended to represent different approximations of the same function, they need to be matched as $X$ tends to the singularity.
    
    In particular, using the natural extensions of the symmetric multilinear vector functions we defined in \eqref{symmetric_linear_vector_functions}, let $M > 0$ be an integer and define $\mbf p = \left(p_1, p_2, \ldots \right)$, $\mbf q = \left(q_1, q_2, \ldots\right)$, $\mbf n = \left(n_1, \ldots, n_M\right)$, $\mbf r = \left(r_1, \ldots, r_M\right)$, $\mbf e_M = (1, \ldots, 1)$ (the vector with $M$ ones) and $\mbf e = (1, 2, \ldots, M)$. With this, assuming that $n$ is sufficiently large, and Taylor expanding \eqref{eqtoexpand}, we see that the $n$-th order equation is given by
    \begin{multline}
	    \mbf 0 = \sum_{\substack{p = 0 \\ \mbf p \cdot \mbf e = p}}^{n - 1} \, \sum_{\substack{q = 0 \\ \mbf q \cdot \mbf e = q}}^{n - 1 - p} C_{\mbf p, \mbf q} \, \left(\prod_{s \geq 1} a_s^{p_s}\right) \, \left(\prod_{s \geq 1} b_s^{q_s}\right) \, \jac \mbf f_{p, q}(\mbf 0) \, \mbf u^{[n - p - q]}
	    \\
	    + \sum_{\substack{p = 0 \\ \mbf p \cdot \mbf e = p}}^{n - 1} \, \sum_{\substack{q = 0 \\ \mbf q \cdot \mbf e = q}}^{n - 1 - p} \, \sum_{\substack{M = 2 \\ \mbf n \cdot \mbf e_M = n - p - q}}^{n - p - q} C_{\mbf p, \mbf q} \, \left(\prod_{s \geq 1} a_s^{p_s}\right) \, \left(\prod_{s \geq 1} b_s^{q_s}\right) \, \mbf F_{M, p, q}\left(\mbf u^{[n_1]}, \ldots , \mbf u^{[n_M]}\right)
	    \\
	    + k^2 \, \hat D \, \left(\frac{\partial^2 \mbf u^{[n]}}{\partial x^2} + 2 \, \frac{\partial^2 \mbf u^{[n - 2]}}{\partial x \partial X} + \frac{\partial^2 \mbf u^{[n - 4]}}{\partial X^2}\right), \label{late-term-expansion}
	\end{multline}
    where $C_{\mbf p, \mbf q}$ are the coefficients associated with the powers of $\mbf a$ and $\mbf b$ in the Taylor expansion of \evs{our vector field}, for each pair of vectors, $\mbf p$ and $\mbf q$, defined above. From the results we have obtained so far, we have
    \begin{align}
        \mbf u^{[n]} = \sum_{r = - n}^n \mbf W_{n, r} \, e^{i r x}, \label{general_ansatz_un}
    \end{align}
    where $n \geq 1$ is an integer, and $\mbf W_{n, r} = \mbf W_{n, r}(X)$ for every integers $n \geq 1, - n \leq r \leq n$. In particular, taking into account only the terms that are dominant at each mode as $X \to X_0$, we have $\mbf W_{n, - r} = \ol{\mbf W_{n, r}}$ for each pair of integers, $n \geq 1, - n \leq r \leq n$, $r \neq 0$, and we note that
    \begin{align*}
        \mbf W_{1, 1} \, e^{ix} + c.c. &= A_1 \, e^{ix} \, \bs \phi_1^{[1]} + c.c.,
        \\
        \mbf W_{2, 0} &= 2 \, \abs{A_1}^2 \, \mbf W_0^{[2]},
        \\
        \mbf W_{2, 2} \, e^{2ix} + c.c. &= A_1^2 \, e^{2ix} \, \mbf W_2^{[2]} + c.c.,
        \\
        \mbf W_{3, 1} \, e^{ix} + c.c. &\sim \abs{A_1}^2 \, A_1 \, e^{ix} \, \mbf W_{1, 2}^{[3]} + i \, A_1' \, e^{ix} \, \mbf W_{1, 3}^{[3]}  + c.c.,
        \\
        \mbf W_{3, 3} \, e^{3ix} + c.c. &=  A_1^3 \, e^{3ix} \, \mbf W_3^{[3]} + c.c.,
        \\
        \mbf W_{4, 0} &\sim \abs{A_1}^4 \, \mbf W_{0, 2}^{[4]} + i \, \bar A_1 \, A_1' \, \mbf W_{0, 3}^{[4]} + c.c.,
        \\
        \mbf W_{4, 2} \, e^{2ix} + c.c. &\sim \abs{A_1}^2 \, A_1^2 \, e^{2ix} \, \mbf W_{2, 2}^{[4]} + i \, A_1 \, A_1' \, e^{2ix} \, \mbf W_{2, 3}^{[4]} + c.c.
    \end{align*}
    Therefore, when replacing \eqref{general_ansatz_un} into \eqref{late-term-expansion} we obtain, for $n \geq 1$ large and each $r \in [- n, n]$,
    \begin{multline}
	    \mbf 0 = \sum_{\substack{p = 0 \\ \mbf p \cdot \mbf e = p}}^{n - 1} \, \sum_{\substack{q = 0 \\ \mbf q \cdot \mbf e = q}}^{n - 1 - p} C_{\mbf p, \mbf q} \, \left(\prod_{s \geq 1} a_s^{p_s}\right) \, \left(\prod_{s \geq 1} b_s^{q_s}\right) \, \jac \mbf f_{p, q}(\mbf 0) \, \mbf W_{n - p - q, r}
	    \\
	    + \sum_{\substack{p = 0 \\ \mbf p \cdot \mbf e = p}}^{n - 1} \, \sum_{\substack{q = 0 \\ \mbf q \cdot \mbf e = q}}^{n - 1 - p} \, \sum_{\substack{M = 2 \\ \mbf n \cdot \mbf e_M = n - p - q}}^{n - p - q} \sum_{\substack{\mbf r \cdot \mbf e_M = r \\ \abs{r_\ell} \leq n_\ell}} C_{\mbf p, \mbf q} \, \left(\prod_{s \geq 1} a_s^{p_s}\right) \, \left(\prod_{s \geq 1} b_s^{q_s}\right) \, \mbf F_{M, p, q}\left(\mbf W_{n_1, r_1}, \ldots , \mbf W_{n_M, r_M}\right)
	    \\
	    + k^2 \, \hat D \, \left(- r^2 \, \mbf W_{n, r} + 2 \, i \, r \, \mbf W_{n - 2, r}' + \mbf W_{n - 4, r}''\right), \label{firstansatz}
	\end{multline}
 
    \subsection{Asymptotic expansion}
        Again, what we need to do now is to solve a linear equation order by order. We highlight that the following analysis is valid for any singularity, but we focus on $X_0$, as it is one of the singularities closest to the real axis for $X$. In particular, based on \cite{Chapman} \evs{and the order we obtained for $A_1$, when expanded around $X_0$}, we use the ansatz
    	\begin{align}
    	    \mbf W_{n, r}(X) = \frac{\mbf V_{n, r}}{\left(X_0 - X\right)^{\frac{n}{2} - r \, \eta \, i}}, \label{ansatz-inner-sol}
    	\end{align}
        for our inner solution, which implies
        \begin{align*}
            \mbf W_{n - 2, r}' &= \frac{n - 2 - 2 \, r \, \eta \, i}{2} \, \frac{\mbf V_{n - 2, r}}{\left(X_0 - X\right)^{\frac{n}{2} - r \, \eta \, i}},
            \\
            \mbf W_{n - 4, r}'' &= \frac{(n - 2 - 2 \, r \, \eta \, i)(n - 4 - 2 \, r \, \eta \, i)}{4} \, \frac{\mbf V_{n - 4, r}}{\left(X_0-X\right)^{\frac{n}{2} - r \,  \eta \, i}}.
        \end{align*}
    	Therefore, when replacing \eqref{ansatz-inner-sol} into \eqref{firstansatz}, we obtain
        {\allowdisplaybreaks
        \begin{multline*}
    	    \mbf 0 = \sum_{\substack{p = 0 \\ \mbf p \cdot \mbf e = p}}^{n - 1} \, \sum_{\substack{q = 0 \\ \mbf q \cdot \mbf e = q}}^{n - 1 - p} C_{\mbf p, \mbf q} \, \left(\prod_{s \geq 1} a_s^{p_s}\right) \, \left(\prod_{s \geq 1} b_s^{q_s}\right) \, \jac \mbf f_{p, q}(\mbf 0) \, \frac{\mbf V_{n - p - q, r}}{\left(X_0 - X\right)^{\frac{n - p - q}{2} - r \, \eta \, i}}
    	    \\
    	    + \sum_{\substack{p = 0 \\ \mbf p \cdot \mbf e = p}}^{n - 1} \, \sum_{\substack{q = 0 \\ \mbf q \cdot \mbf e = q}}^{n - 1 - p} \, \sum_{\substack{M = 2 \\ \mbf n \cdot \mbf e_M = n - p - q}}^{n - p - q} \sum_{\substack{\mbf r \cdot \mbf e_M = r \\ \abs{r_\ell} \leq n_\ell}} C_{\mbf p, \mbf q} \, \left(\prod_{s \geq 1} a_s^{p_s}\right) \, \left(\prod_{s \geq 1} b_s^{q_s}\right)
             \\
             \times \mbf F_{M, p, q}\left(\frac{\mbf V_{n_1, r_1}}{\left(X_0 - X\right)^{\frac{n_1}{2} - r_1 \, \eta \, i}}, \ldots , \frac{\mbf V_{n_M, r_M}}{\left(X_0 - X\right)^{\frac{n_M}{2} - r_M \, \eta \, i}}\right)
    	    \\
    	    + k^2 \, \hat D \, \left(- r^2 \, \frac{\mbf V_{n, r}}{\left(X_0 - X\right)^{\frac{n}{2} - r \, \eta \, i}} + 2 \, i \, r \, \left(\frac{n - 2 - 2 \, r \, \eta \, i}{2} \, \frac{\mbf V_{n - 2, r}}{\left(X_0 - X\right)^{\frac{n}{2} - r \, \eta \, i}}\right)\right.
            \\
            \left. + \frac{(n - 2 - 2 \, r \, \eta \, i) \, (n - 4 - 2 \, r \, \eta \, i)}{4} \, \frac{\mbf V_{n - 4, r}}{\left(X_0-X\right)^{\frac{n}{2} - r \, \eta \, i}}\right),
    	\end{multline*}
        }
        which is equivalent to
        {\allowdisplaybreaks
        \begin{multline*}
    	    \mbf 0 = \sum_{\substack{p = 0 \\ \mbf p \cdot \mbf e = p}}^{n - 1} \, \sum_{\substack{q = 0 \\ \mbf q \cdot \mbf e = q}}^{n - 1 - p} C_{\mbf p, \mbf q} \, \left(\prod_{s \geq 1} a_s^{p_s}\right) \, \left(\prod_{s \geq 1} b_s^{q_s}\right) \, \jac \mbf f_{p, q}(\mbf 0) \, \frac{\mbf V_{n - p - q, r}}{\left(X_0 - X\right)^{\frac{- p - q}{2}}}
    	    \\
    	    + \sum_{\substack{p = 0 \\ \mbf p \cdot \mbf e = p}}^{n - 1} \, \sum_{\substack{q = 0 \\ \mbf q \cdot \mbf e = q}}^{n - 1 - p} \, \sum_{\substack{M = 2 \\ \mbf n \cdot \mbf e_M = n - p - q}}^{n - p - q} \sum_{\substack{\mbf r \cdot \mbf e_M = r \\ \abs{r_\ell} \leq n_\ell}} \left(\frac{C_{\mbf p, \mbf q}}{\left(X_0 - X\right)^{\frac{- p - q}{2}}} \, \left(\prod_{s \geq 1} a_s^{p_s}\right) \, \left(\prod_{s \geq 1} b_s^{q_s}\right)\right.
            \\
            \times \left. \mbf F_{M, p, q}\left(\mbf V_{n_1, r_1}, \ldots , \mbf V_{n_M, r_M}\right) \vphantom{\frac{C_{\mbf p, \mbf q}}{\left(X_0 - X\right)^{\frac{- p - q}{2}}}}\right)
    	    \\
    	    + k^2 \, \hat D \, \left(- r^2 \, \mbf V_{n, r} + 2 \, i \, r \, \left(\frac{n - 2 - 2 \, r \, \eta \, i}{2}\right) \, \mbf V_{n - 2, r}\right.
            \\
            \left. + \frac{(n - 2 - 2 \, r \, \eta \, i)(n - 4 - 2 \, r \, \eta \, i)}{4} \, \mbf V_{n - 4, r}\right),
    	\end{multline*}
        }
        and we observe that the only terms that play a role in the inner solution are those with no parameters. Therefore, we can simplify our equation for the inner solution as
        \begin{equation}
            \begin{aligned}
                \mbf 0 &= \left(\jac \mbf f_{0, 0}(\mbf 0) - r^2 \, k^2 \, \hat D\right) \mbf V_{n, r} + \sum_{\substack{M = 2 \\ \mbf n \cdot \mbf e_M = n}}^n \sum_{\substack{\mbf r \cdot \mbf e_M = r \\ \abs{r_\ell} \leq n_\ell}} \mbf F_{M, 0, 0}\left(\mbf V_{n_1, r_1}, \ldots , \mbf V_{n_M, r_M}\right)
        	    \\
        	    &\quad + k^2 \, \hat D \, \left(2 \, i \, r \, \left(\frac{n - 2 - 2 \, r \, \eta \, i}{2}\right) \, \mbf V_{n - 2, r} + \frac{(n - 2 - 2 \, r \, \eta \, i) \, (n - 4 - 2 \, r \, \eta \, i)}{4} \, \mbf V_{n - 4, r}\right),
            \end{aligned} \label{DM}
        \end{equation}
        Here, note that if we change $r$ into $- r$ in the previous expression, we have
        \begin{multline*}
    	    \mbf 0 = \left(\jac \mbf f_{0, 0}(\mbf 0) - r^2 \, k^2 \, \hat D\right) \mbf V_{n, - r} + \sum_{\substack{M = 2 \\ \mbf n \cdot \mbf e_M = n}}^n \sum_{\substack{\mbf r \cdot \mbf e_M = r \\ \abs{r_\ell} \leq n_\ell}} \mbf F_{M, 0, 0}\left(\mbf V_{n_1, - r_1}, \ldots , \mbf V_{n_M, - r_M}\right)
    	    \\
    	    + k^2 \, \hat D \, \left(- 2 \, i \, r \, \left(\frac{n - 2 + 2 \, r \, \eta \, i}{2}\right) \, \mbf V_{n - 2, - r} + \frac{(n - 2 + 2 \, r \, \eta \, i)(n - 4 + 2 \, r \, \eta \, i)}{4} \, \mbf V_{n - 4, - r}\right).
    	\end{multline*}
        which gives us the conjugate of \eqref{DM} if, based on \eqref{mostwonderfulidea}, we consider
        \begin{align}
            \mbf V_{n, - r} = (- 1)^n \, e^{- 2 \, r \, \pi \, \xi} \, \ol{\mbf V_{n, r}}, \label{conjrule}
        \end{align}
        for every integers $n \geq 1$ and $- n \leq r \leq n$.
        
        Now, in order to solve \eqref{DM}, we expand each vector $\mbf V_{n, r}$ as
        \begin{align}
            \mbf V_{n, r} = \kappa^{n - 1} \, \Gamma\left(\frac{n - 1}{2} + \gamma_r\right) \, \sum_{s\geq 0} \frac{\mbf v_r^{[s]}}{n^s}, \label{Vnr}
        \end{align}
        where $\kappa$ and $\gamma_r$ are complex variables yet to be determined.
        
        Now, noting that
        \begin{align*}
            \Gamma\left(\frac{n - 1}{2} + \gamma_r\right) &= \left(\frac{n - 3}{2} + \gamma_r\right) \Gamma\left(\frac{n - 3}{2} + \gamma_r\right)
            \\
            &= \left(\frac{n - 3}{2} + \gamma_r\right) \left(\frac{n - 5}{2} + \gamma_r\right) \Gamma\left(\frac{n - 5}{2} + \gamma_r\right),
        \end{align*}
        and replacing \eqref{Vnr} into \eqref{DM}, we obtain
        {\allowdisplaybreaks
        \begin{multline*}
    	    \mbf 0 = \left(\jac \mbf f_{0, 0}(\mbf 0) - r^2 \, k^2 \, \hat D\right) \kappa^{n - 1} \, \Gamma\left(\frac{n - 1}{2} + \gamma_r\right) \, \sum_{s\geq 0} \frac{\mbf v_r^{[s]}}{n^s}
            \\
            + \sum_{\substack{M = 2 \\ \mbf n \cdot \mbf e_M = n}}^n \sum_{\substack{\mbf r \cdot \mbf e_M = r \\ \abs{r_\ell} \leq n_\ell}} \mbf F_{M, 0, 0}\left(\mbf V_{n_1, r_1}, \ldots , \mbf V_{n_M, r_M}\right)
    	    \\
    	    + k^2 \, \hat D \, \left(2 \, i \, r \, \kappa^{n - 3} \, \left(\frac{n - 2 - 2 \, r \, \eta \, i}{2}\right) \, \Gamma\left(\frac{n - 3}{2} + \gamma_r\right) \, \sum_{s\geq 0} \frac{\mbf v_r^{[s]}}{(n - 2)^s}\right.
            \\
            \left. + \kappa^{n - 5} \, \frac{(n - 2 - 2 \, r \, \eta \, i)(n - 4 - 2 \, r \, \eta \, i)}{4} \, \Gamma\left(\frac{n - 5}{2} + \gamma_r\right) \, \sum_{s\geq 0} \frac{\mbf v_r^{[s]}}{(n - 4)^s}\right),
    	\end{multline*}
        }
        which is equivalent to
        \begin{multline*}
    	    \mbf 0 = \left(\jac \mbf f_{0, 0}(\mbf 0) - r^2 \, k^2 \, \hat D\right) \kappa^4 \, \sum_{s\geq 0} \frac{\mbf v_r^{[s]}}{n^s}
            \\
            + \kappa^{5 - n} \, \Gamma\left(\frac{n - 1}{2} + \gamma_r\right)^{-1} \, \sum_{\substack{M = 2 \\ \mbf n \cdot \mbf e_M = n}}^n \sum_{\substack{\mbf r \cdot \mbf e_M = r \\ \abs{r_\ell} \leq n_\ell}} \mbf F_{M, 0, 0}\left(\mbf V_{n_1, r_1}, \ldots , \mbf V_{n_M, r_M}\right)
    	    \\
    	    + k^2 \, \hat D \, \left(2 \, i \, r \, \kappa^2 \left(\frac{n - 2 - 2 \, r \, \eta \, i}{n - 3 + 2 \, \gamma_r}\right) \, \sum_{s\geq 0} \frac{\mbf v_r^{[s]}}{(n - 2)^s}\right.
            \\
            \left. + \frac{(n - 2 - 2 \, r \, \eta \, i)(n - 4 - 2 \, r \, \eta \, i)}{\left(n - 3 + 2 \, \gamma_r\right) \left(n - 5 + 2 \, \gamma_r\right)} \, \sum_{s\geq 0} \frac{\mbf v_r^{[s]}}{(n - 4)^s}\right),
    	\end{multline*}
        Furthermore, by expanding the previous expression in powers of $1/n$, assuming that $n$ is large, we have
        \begin{multline}
            \mbf 0 = \left(\jac \mbf f_{0, 0}(\mbf 0) - r^2 \, k^2 \, \hat D\right) \, \kappa^4 \, \sum_{s \geq 0} \frac{\mbf v_r^{[s]}}{n^s}
            \\
            + \kappa^{5 - n} \, \Gamma\left(\frac{n - 1}{2} + \gamma_r\right)^{-1} \, \sum_{\substack{M = 2 \\ \mbf n \cdot \mbf e_M = n}}^n \sum_{\substack{\mbf r \cdot \mbf e_M = r \\ \abs{r_\ell} \leq n_\ell}} \mbf F_{M, 0, 0}\left(\mbf V_{n_1, r_1}, \ldots , \mbf V_{n_M, r_M}\right)
    	    \\
            + k^2 \, \hat D \, \left(2 \, i \, r \, \kappa^2 \left(1 + \frac{- 2 \, \gamma_r - 2 \, r \, \eta \, i + 1}{n} + \frac{\left(2 \, \gamma_r - 3\right) \left(2 \, \gamma_r + 2 \, r \, \eta \, i - 1\right)}{n^2}\right)\right.
            \\
            \times \left(\mbf v_r^{[0]} + \left(\frac{1}{n} + \frac{2}{n^2}\right) \, \mbf v_r^{[1]} + \frac{\mbf v_r^{[2]}}{n^2}\right)
            \\
            + \left(1 + \frac{- 4 \, \gamma_r - 4 \, r \, \eta \, i + 2}{n} - \frac{\left(- 2 \, i \, \gamma_r + 2 \, \eta \, r + i\right) \left(- 6 \, i \, \gamma_r + 2 \, r \, \eta + 9 \, i\right)}{n^2}\right)
            \\
            \times \left.\left(\mbf v_r^{[0]} + \left(\frac{1}{n} + \frac{4}{n^2}\right) \, \mbf v_r^{[1]} + \frac{\mbf v_r^{[2]}}{n^2}\right)\right) + \mathcal O\left(\frac{1}{n^3}\right). \label{nicetoexpand}
        \end{multline}
        Now, as $n$ is large, then powers of $1/n$ can be considered small numbers, which lets us follow the usual approach to study this equation following an asymptotic \evs{procedure on $1/n$}.

    \subsection{Solution at \evs{order} $\mathcal O(1)$} \label{sub:kappa_explanation}
        At this order, \eqref{nicetoexpand} becomes
    	\begin{align}
    	    M(\kappa, r) \, \mbf v_r^{[0]} = - \mbf{NL}(1), \label{kernel}
    	\end{align}
    	where
    	\begin{align*}
    	    M(\kappa, r) = \kappa^4 \left(\jac \mbf f_{0, 0}(\mbf 0) - r^2 \, k^2 \, \hat D\right) + \kappa^2 \left(2 \, i \, r \, k^2 \, \hat D\right) + k^2 \, \hat D,
    	\end{align*}
        and $\mbf{NL}(1)$, are terms of order $\mathcal O(1)$ associated with the nonlinearities of system \eqref{general_equation}, which depend on $\gamma_r$. We highlight that we will deal with this problem using mode dominance as follows. We define $\gamma$ as the value of $\gamma_r$ with the highest real part as $r$ takes integer values between $- \infty$ and $\infty$. With this, we can note that $\mbf{NL}(1) = \mbf 0$ for the values of $r$ such that $\gamma_r = \gamma$, assuming they are finite. To see why this is not an assumption, consider the following result:
        \begin{lemma}
            The implicit function $\kappa = \abs{\kappa(r)}$ determined by $\det(M(\kappa, r)) = 0$ is even and bounded away from $r = \pm 1$. Furthermore, it is finite for $r = \pm 1$ and $\ds{\lim_{\abs{r} \to \infty} \kappa(r) = 0}$.
        \end{lemma}
        \begin{proof}
            Note that
            \begin{align*}
                \ol{M(\kappa, r)} = M(\bar \kappa, - r),
            \end{align*}
            which proves that $\kappa = \abs{\kappa(r)}$ is an even function.

            Now, note that $\det(M(\kappa, r))$ is a polynomial of degree $4 \, n$ in $\kappa$ for all $\abs{r} \neq 1$. On the other hand, when $\abs{r} = 1$, the degree of $\det(M(\kappa, r))$ decreases to $4 \, n - 4$. However, due to the fundamental theorem of algebra, these polynomials will have well-defined finite complex roots for all values of $r$.
    
            On the other hand, when $\abs{r}$ gets large,
            \begin{align*}
                \det(M(\kappa, r)) \sim \det\left(r \, k^2 \, \kappa^2 \, \left(- r \, \kappa^2 + 2 \, i\right) \, \hat D\right).
            \end{align*}
            which implies that $\kappa$ must tend to zero as $\abs{r}\to \infty$ in order to fulfill the equation $\det(M(\kappa, r)) = 0$, which proves the result.
        \end{proof}
    	Now, to solve \eqref{kernel}, we want $\mbf v_r^{[0]}$ to be non-zero for the values of $r$ such that $\gamma_r = \gamma$. To determine those values of $r$, we need to find the values of $\kappa$ with the largest amplitude such that $\det(M(\kappa, r)) = 0$, together with the values of $r$ that \evs{lead} to those values of $\kappa$, as these will generate the leading-order solution.

        Observe that
        \begin{align*}
            M(\kappa, r) &= \kappa^4 \left(\jac \mbf f_{0, 0}(\mbf 0) + \left(- r^2 + \frac{2 \, i \, r}{\kappa^2} + \frac{1}{\kappa^4}\right) \, k^2 \, \hat D\right).
        \end{align*}
        Next, note that there might be complex roots in $\ell$ to the equation $\det\left(\jac \mbf f_{0, 0}(\mbf 0) - \ell^2 \, \hat D\right) = 0$. Nevertheless, as we are focusing on the patterns that arise at a codimension-two Turing bifurcation point, we only focus on the case $\ell \in \mathbb R$, which implies that $\ell = k$ is the only solution \evs{when we assume that we are working with a codimension-two point at a primary Turing bifurcation}. Thus, $\det(M(\kappa, r)) = 0$ if        
        \begin{align}
            - r^2 + \frac{2 \, i \, r}{\kappa^2} + \frac{1}{\kappa^4} = - 1 \qquad \iff \qquad \kappa^2 = \frac{i}{\left(r \pm 1\right)}, \label{kappa}
        \end{align}
        for $r \neq \pm 1$, whilst $\kappa^2 = \pm i/2$ when $r = \pm 1$. This implies that the dominant modes for $\kappa$ are attained when $r \in \{0, \pm 2\}$. There is one case that needs extra care, though. That is the case in which system \eqref{general_equation} is \evs{odd in $\mbf u$}. In such a case, the modes $r = 0, \pm 2$ lead to trivial solutions, making the coming analysis unfeasible \cite{Dean}. We will consider the general case for now and will develop the changes needed in the symmetric case in \ref{app:symmetric_case}.
        
        Now, we define $\kappa_\pm^2 = \pm i$, and take the following values of $\gamma$ in order of dominance, taking into account that an increase of 1 in the value of $n$ increases the power of $X_0 - X$ in the denominator of \eqref{ansatz-inner-sol} by 1/2:
        \begin{align}
            \gamma_{0, \pm 2} = \gamma, \qquad \gamma_{\pm 1} = \gamma - \frac{1}{2}, \qquad \text{and} \qquad \gamma_{\pm p} = \gamma - \frac{p - 2}{2}, \quad \text{for } p \geq 3. \label{gammavals}
        \end{align}
        \subsection{The case $\kappa^2 = \kappa_+^2$}
            In this case, we have
            \begin{align*}
                M\left(\pm \sqrt{i}, r\right) &= - \left(\jac \mbf f_{0, 0}(\mbf 0) - (r - 1)^2 \, k^2 \, \hat D\right).
            \end{align*}
            In particular,
            \begin{align*}
                M\left(\pm \sqrt{i}, 0\right) &= M\left(\pm \sqrt{i}, 2\right) = - \left(\jac \mbf f_{0, 0}(\mbf 0) - k^2 \, \hat D\right),
            \end{align*}
            which implies that $\mbf v_0^{[0]} = c_0^{[0]} \, \bs \phi_1^{[1]}$, and $\mbf v_2^{[0]} = c_2^{[0]} \, \bs \phi_1^{[1]}$, where $c_0^{[0]}, c_2^{[0]} \in \mathbb R$, and
            \begin{align*}
                \vdots
                \\
                \mbf v_{- 4}^{[0]} & = \mbf 0,
                \\
                \mbf v_{- \evs{3}}^{[0]} & = \mbf 0,
                \\
                \mbf v_{- 2}^{[0]} & = \mbf 0,
                \\
                \mbf v_{- 1}^{[0]} &= - 2 \, \kappa^3 \, C_{\bar A} \, c_0^{[0]} \, \mbf W_2^{[2]},
                \\
                \mbf v_1^{[0]} &= - 2 \, \kappa^3 \, \left(C_{\bar A} \, c_2^{[0]} + C_A \, c_0^{[0]}\right) \, \mbf W_0^{[2]},
                \\
                \mbf v_3^{[0]} &= - 2 \, \kappa^3 \, C_A \, c_2^{[0]} \, \mbf W_2^{[2]},
                \\
                \mbf v_4^{[0]} &= - 3 \, i \, C_A^2 \, c_2^{[0]} \, \mbf W_3^{[3]},
                \\
                \vdots
            \end{align*}
            With this, we continue to the following order.
            
            \paragraph{Order $\mathcal O(1/n)$.} Following the same procedure at this order, we obtain
            \begin{align*}
                \vdots
                \\
                \mbf v_{- 2}^{[1]} &= - 6 \, i \, \mbf W_3^{[3]} \, c_0^{[0]} \, C_{\bar A}^2,
                \\
                \mbf v_{- 1}^{[1]} &= - 2 \, \kappa \, C_{\bar A} \, \left(C_{\bar A}^2 \, c_2^{[0]} \, \mbf W_{2, 2}^{[4]} + c_0^{[0]} \, \left(3 \, C_A \, C_{\bar A} \, \mbf W_{2, 2}^{[4]} + (\eta + i \, (\gamma - 1)) \, \mbf W_{2, 3}^{[4]}\right)\right),
                \\
                \mbf v_0^{[1]} &= - \left(c_0^{[0]} \, \left((1 - 2 \, \gamma) \, \mbf W_{1, 3}^{[3]} + 4 \, i \, C_A \, C_{\bar A} \, \mbf W_{1, 2}^{[3]}\right) + 2 \, i \, C_{\bar A}^2 \, c_2^{[0]} \, \mbf W_{1, 2}^{[3]}\right),
                \\
                \mbf v_1^{[1]} &= \begin{multlined}[t]
                    - 2 \, \kappa \, \left(C_A \, c_0^{[0]} \, \left(4 \, C_A \, C_{\bar A} \, \mbf W_{0, 2}^{[4]} + (\eta + \gamma \, i) \, \mbf W_{0, 3}^{[4]}\right)\right.
                    \\
                    \left. + C_{\bar A} \, c_2^{[0]} \, \left(4 \, C_A \, C_{\bar A} \, \mbf W_{0, 2}^{[4]} + (3 \, \eta - \gamma \, i) \, \mbf W_{0, 3}^{[4]}\right)\right),
                \end{multlined}
                \\
                \mbf v_2^{[1]} &= - i \, \left(2 \, C_A \, \left(2 \, c_2^{[0]} \, C_{\bar A} + C_A \, c_0^{[0]}\right) \, \mbf W_{1, 2}^{[3]} + c_2^{[0]} \, (- 2 \, \gamma \, i + 4 \, \eta + i) \, \mbf W_{1, 3}^{[3]}\right),
                \\
                \mbf v_3^{[1]} &= - 2 \, \kappa \, C_A \, \left(C_A^2 \, c_0^{[0]} \, \mbf W_{2, 2}^{[4]} + c_2^{[0]} \, \left(3 \, C_A \, C_{\bar A} \, \mbf W_{2, 2}^{[4]} + (- \gamma \, i + 3 \, \eta + i) \, \mbf W_{2, 3}^{[4]}\right)\right),
                \\
                \vdots
            \end{align*}
            With this, we are ready to determine a solvability condition at the following order.
            
            \paragraph{Order $\mathcal O(1/n^2)$} \evs{A}t this order, we obtain two solvability conditions given by
            \begin{multline*}
                S_{c, 0} = - \left(\vphantom{C_A^2 \, C_{\bar A}^2} (4 \, \gamma \, (\gamma - 2) + 3) \, \alpha_1 + 4 \, C_A \, C_{\bar A} \, \left((\eta + (\gamma - 1) \, i) \, \alpha_3 - (2 \, \eta + i) \, \alpha_4\right)\right.
                \\
                \left. + 12 \, C_A^2 \, C_{\bar A}^2 \, \alpha_7\right) \, c_0^{[0]}
                \\
                + 2 \, C_{\bar A}^2 \left((- 2 \, \eta + i) \, \alpha_3 + (4 \, \eta + i \, (1 - 2 \, \gamma)) \, \alpha_4 - 4 \, C_A \, C_{\bar A} \, \alpha_7\right) \, c_2^{[0]} = 0,
            \end{multline*}
            and
            \begin{multline*}
                S_{0, 2} = - 2 \, C_A^2 \, \left((2 \, \eta + i) \, \alpha_3 + i \, (1 - 2 \, \gamma) \, \alpha_4 + 4 \, C_A \, C_{\bar A} \, \alpha_7\right) \, c_0^{[0]}
                \\
                - \left((2 \, \gamma - 3 + 4 \, \eta \, i) (2 \, \gamma - 1 + 4 \, \eta \, i) \, \alpha_1\right.
                \\
                \left. + 4 \, C_A \, C_{\bar A} \left((3 \, \eta + i \, (1 - \gamma)) \, \alpha_3 + (- 2 \, \eta + i) \, \alpha_4\right) + 12 \, C_A^2 \, C_{\bar A}^2 \, \alpha_7\right) \, c_2^{[0]} = 0.
            \end{multline*}
            Note that this is a linear system of equations for $c_0^{[0]}$ and $c_2^{[0]}$, which has nontrivial solutions when $\beta_3^2 = 4 \, \beta_1 \, \beta_5$ if and only if \evs{its corresponding determinant equals zero, which occurs if and only if}
            \begin{align*}
                \gamma \in \{- 1 - \eta \, i, - \eta \, i, 2 - \eta \, i, 3 - \eta \, i\}.
            \end{align*}
            Therefore, as $\gamma$ was defined to be the value of $\gamma_r$ with the highest real part, then we take $\gamma = 3 - \eta \, i$, which yields
            \begin{align}
                c_2^{[0]} = \frac{h_1}{k_1} \, c_0^{[0]}, \label{h1-k1ratio}
            \end{align}
            where
            \begin{align*}
                h_1 &= (4 \, \gamma \, (\gamma - 2) + 3) \, \alpha_1 + 4 \, C_A \, C_{\bar A} \, \left((\eta + (\gamma - 1) \, i) \, \alpha_3 - (2 \, \eta + i) \, \alpha_4\right) + 12 \, C_A^2 \, C_{\bar A}^2 \, \alpha_7,
                \\
                k_1 &= 2 \, C_{\bar A}^2 \left((- 2 \, \eta + i) \, \alpha_3 + (4 \, \eta + i \, (1 - 2 \, \gamma)) \, \alpha_4 - 4 \, C_A \, C_{\bar A} \, \alpha_7\right).
            \end{align*}
            
        \subsection{The case $\kappa^2 = \kappa_-^2$}
            In this case, we have
            \begin{align*}
                M\left(\pm \sqrt{- i}, r\right) = - \left(\jac \mbf f_{0, 0}(\mbf 0) - (r + 1)^2 \, k^2 \, \hat D\right).
            \end{align*}
            Therefore, $\mbf v_{- 2}^{[0]} = c_{- 2}^{[0]} \, \bs \phi_1^{[1]}$, and $\mbf v_0^{[0]} = c_0^{[0]} \, \bs \phi_1^{[1]}$, where $c_{- 2}^{[0]}, c_0^{[0]} \in \mathbb R$, and
            \begin{align*}
                \vdots
                \\
                \mbf v_{- 4}^{[0]} &= 3 \, i \, C_{\bar A}^2 \, c_{- 2}^{[0]} \, \mbf W_3^{[3]},
                \\
                \mbf v_{- 3}^{[0]} &= - 2 \, \kappa^3 \, C_{\bar A} \, c_{- 2}^{[0]} \, \mbf W_2^{[2]},
                \\
                \mbf v_{- 1}^{[0]} &= - 2 \, \kappa^3 \, \mbf W_0^{[2]} \, \left(C_{\bar A} \, c_0^{[0]} + C_A \, c_{- 2}^{[0]}\right),
                \\
                \mbf v_1^{[0]} &= - 2 \, \kappa^3 \, C_A \, c_0^{[0]} \, \mbf W_2^{[2]},
                \\
                \mbf v_2^{[0]} &= \mbf 0,
                \\
                \mbf v_3^{[0]} &= \mbf 0,
                \\
                \mbf v_4^{[0]} &= \mbf 0,
                \\
                \vdots
            \end{align*}
            \paragraph{Order $\mathcal O(1/n)$} At this order, we obtain
            \begin{align*}
                \vdots
                \\
                \mbf v_{- 3}^{[1]} &= - 2 \, \kappa \, C_{\bar A} \, \left(C_{\bar A}^2 \, c_0^{[0]} \, \mbf W_{2, 2}^{[4]} + c_{- 2}^{[0]} \, \left(3 \, C_A \, C_{\bar A} \, \mbf W_{2, 2}^{[4]} + (3 \, \eta + i \, (\gamma - 1)) \, \mbf W_{2, 3}^{[4]}\right)\right),
                \\
                \mbf v_{- 2}^{[1]} &= - \left(c_{- 2}^{[0]} \, \left((2 \, \gamma - 1 - 4 \, \eta \, i) \, \mbf W_{1, 3}^{[3]} - 4 \, i \, C_A \, C_{\bar A} \, \mbf W_{1, 2}^{[3]}\right) - 2 \, i \, C_{\bar A}^2 \, c_0^{[0]} \, \mbf W_{1, 2}^{[3]}\right),
                \\
                \mbf v_{- 1}^{[1]} &= \begin{multlined}[t]
                    - 2 \, \kappa \, \left(C_A \, c_{- 2}^{[0]} \, \left(4 \, C_A \, C_{\bar A} \, \mbf W_{0, 2}^{[4]} + (3 \, \eta + \gamma \, i) \, \mbf W_{0, 3}^{[4]}\right)\right.
                    \\
                    \left. + c_0^{[0]} \, C_{\bar A} \, \left(4 \, C_A \, C_{\bar A} \, \mbf W_{0, 2}^{[4]} + (\eta - \gamma \, i) \, \mbf W_{0, 3}^{[4]}\right)\right),
                \end{multlined}
                \\
                \mbf v_0^{[1]} &= - \left((1 - 2 \, \gamma) \, c_0^{[0]} \, \mbf W_{1, 3}^{[3]} - 2 \, i \, C_A \, \left(2 \, C_{\bar A} \, c_0^{[0]} + C_A \, c_{- 2}^{[0]}\right) \, \mbf W_{1, 2}^{[3]}\right),
                \\
                \mbf v_1^{[1]} &= - 2 \, C_A \, \kappa \, \left(C_A^2 \, c_{- 2}^{[0]} \, \mbf W_{2, 2}^{[4]} + c_0^{[0]} \, \left(3 \, C_A \, \mbf W_{2, 2}^{[4]} \, C_{\bar A} + (\eta + (1 - \gamma) \, i) \, \mbf W_{2, 3}^{[4]}\right)\right),
                \\
                \mbf v_2^{[1]} &= 6 \, i \, C_A^2 \, c_0^{[0]} \, \mbf W_3^{[3]},
                \\
                \vdots
            \end{align*}
            which implies that the solvability conditions at order $\mathcal O\left(1/n^2\right)$ are given by
            \begin{multline*}
                S_{c , - 2} = \left(- (2 \, \gamma - 3 - 4 \, i \, \eta) (2 \, \gamma - 1 - 4 \, i \, \eta) \, \alpha_1\right.
                \\
                \left. + 4 \, C_A \, C_{\bar A} \, \left(- i \, (\gamma - 3 \, \eta \, i - 1) \, \alpha_3 + (2 \, \eta + i) \, \alpha_4\right) - 12 \, C_A^2 \, C_{\bar A}^2 \, \alpha_7\right) \, c_{- 2}^{[0]}
                \\
                - 2 \, C_{\bar A}^2 \left((2 \, \eta - i) \, \alpha_3 + i \, (2 \, \gamma - 1) \, \alpha_4 + 4 \, C_A \, C_{\bar A} \, \alpha_7\right) \, c_0^{[0]} = 0,
            \end{multline*}
            and
            \begin{multline*}
                S_{c, 0} = - 2 \, C_A^2 \, \left((2 \, \eta + i) \, \alpha_3 + (- 4 \, \eta + i \, (1 - 2 \, \gamma)) \, \alpha_4 + 4 \, C_A \, C_{\bar A} \, \alpha_7\right) \, c_{- 2}^{[0]}
                \\
                - \left(\vphantom{C_A^2 \, C_{\bar A}^2 \, \alpha_7} (4 \, \gamma \, (\gamma - 2) + 3) \, \alpha_1 + 4 \, C_A \, C_{\bar A} \, \left((\eta - i \, (\gamma - 1)) \, \alpha_3 - (2 \, \eta - i) \, \alpha_4\right)\right.
                \\
                \left. + 12 \, C_A^2 \, C_{\bar A}^2 \, \alpha_7\right) \, c_0^{[0]} = 0.
            \end{multline*}
            As before, it can be proven that this system of equations has nontrivial solutions when $\beta_3^2 = 4 \, \beta_1 \, \beta_5$ if and only if
            \begin{align*}
                \gamma \in \{- 1 + \eta \, i, \eta \, i, 2 + \eta \, i, 3 + \eta \, i\}.
            \end{align*}
            Furthermore, when $\gamma$ takes one of these values, we have
            \begin{align*}
                c_{- 2}^{[0]} = \frac{\bar h_1}{\bar k_1} \, c_0^{[0]}.
            \end{align*}
            Moreover, from \eqref{Vnr}, we conclude that $c_0^{[0]}$ and $c_{\pm 2}^{[0]}$ are key terms of our inner solution, which can be obtained via the following limit:
            \begin{align}
                c_r^{[0]} &= \frac{1}{\left \langle \bs \phi_1^{[1]}, \bs \phi_1^{[1]}\right \rangle} \, \lim_{n \to \infty} \frac{\kappa^{- (n - 1)} \, \left \langle \mbf V_{n, r}, \bs \phi_1^{[1]} \right \rangle}{\Gamma\left(\frac{n - 1}{2} + \gamma\right)}, \label{c0-def}
            \end{align}
            for $r = - 2, 0, 2$ and the corresponding value of $\kappa$, provided that these limits exist.
            
        \subsection{Combining eigenvectors}
            Now, we note that the solution comprising the dominant values of $\kappa$ is given by
            {\allowdisplaybreaks
            \begin{align}
                \mbf u^{[n]} &= \sum_{r = - n}^n \mbf W_{n, r}(X) \, e^{irx} = \sum_{r = - n}^n \frac{\mbf V_{n, r}}{\left(X_0 - X\right)^{\frac{n}{2} - r \, \eta \, i}} \, e^{irx} \notag
                \\
                &= \sum_{r = - n}^n \frac{\kappa^{n - 1} \, \Gamma\left(\frac{n - 1}{2} + \gamma_r\right)}{\left(X_0 - X\right)^{\frac{n}{2} - r \, \eta \, i}} \, \sum_{s \geq 0} \frac{\mbf v_r^{[s]}}{n^s} \, e^{irx} \notag
                \\
                & \underset{\substack{n \to \infty \\[0.7ex] X \to X_0}}{\sim} \sum_{r = 0, 2}\begin{multlined}[t]
                    \left(\left(\sqrt{i}\right)^{n - 1} \, \frac{\Gamma\left(\frac{n - 1}{2} + \gamma\right)}{\left(X_0 - X\right)^{\frac{n}{2} - r \, \eta \, i}} \, \mbf v_r^{[0]} \, e^{irx}\right.
                    \\
                    \left. + \left(- \sqrt{i}\right)^{n - 1} \, \frac{\Gamma\left(\frac{n - 1}{2} + \gamma\right)}{\left(X_0 - X\right)^{\frac{n}{2} - r \, \eta \, i}} \, \mbf v_r^{[0]} \, e^{irx}\right)
                \end{multlined} \notag
                \\
                & \qquad + \sum_{r = - 2, 0}\begin{multlined}[t]
                    \left(\left(\sqrt{- i}\right)^{n - 1} \, \frac{\Gamma\left(\frac{n - 1}{2} + \gamma\right)}{\left(X_0 - X\right)^{\frac{n}{2} - r \, \eta \, i}} \, \mbf v_r^{[0]} \, e^{irx}\right.
                    \\
                    \left. + \left(- \sqrt{- i}\right)^{n - 1} \, \frac{\Gamma\left(\frac{n - 1}{2} + \gamma\right)}{\left(X_0 - X\right)^{\frac{n}{2} - r \, \eta \, i}} \, \mbf v_r^{[0]} \, e^{irx}\right)
                \end{multlined} \notag
                \\
                &= \frac{\Gamma\left(\frac{n - 1}{2} + \gamma\right)}{\left(X_0 - X\right)^{\frac{n}{2}}} \, e^{i (n - 1) \pi/4} \, \left(\lambda_1 + (- 1)^n \, \lambda_2\right) \, \left(1 + \frac{h_1}{k_1} \, \frac{e^{2ix}}{\left(X_0 - X\right)^{- 2 \, \eta \, i}}\right) \, \bs \phi_1^{[1]} \notag
                \\
                & \qquad + \frac{\Gamma\left(\frac{n - 1}{2} + \bar \gamma\right)}{\left(X_0 - X\right)^{\frac{n}{2}}} \, e^{- i (n - 1) \pi/4} \, \left(\lambda_3 + (- 1)^n \, \lambda_4\right) \, \left(1 + \frac{\bar h_1}{\bar k_1} \, \frac{e^{- 2ix}}{\left(X_0 - X\right)^{2 \, \eta \, i}}\right) \, \bs \phi_1^{[1]}, \label{final-inner}
            \end{align}
            }
            where $\lambda_i$ correspond to the corresponding coefficients $c_r^{[0]}$ we found before for the corresponding dominant value of $\kappa$.
            
            Last but not least, for the inner solution we have
            \begin{align*}
                \mbf W_{2 \, n, 2 \, r - 1} = \mbf 0, \qquad \text{and} \qquad \mbf W_{2 \, n - 1, 2 \, r} = \mbf 0, \quad \forall \, n \in \mathbb N, r \in \mathbb Z.
            \end{align*}
            Therefore, as $\mbf u^{[n]}$, has even harmonics, we need $\lambda_2 = \lambda_1$, and $\lambda_4 = \lambda_3$, which implies
            \begin{equation}
                \begin{aligned}
                    \mbf u^{[n]} &\sim i^{\frac{n - 1}{2}} \, \lambda_1 \, \frac{\Gamma\left(\frac{n - 1}{2} + \gamma\right)}{\left(X_0 - X\right)^{\frac{n}{2}}} \, \left(1 + (- 1)^n\right) \, \left(1 + \frac{h_1}{k_1} \, \frac{e^{2ix}}{\left(X_0 - X\right)^{- 2 \, \eta \, i}}\right) \, \bs \phi_1^{[1]}
                    \\
                    & \qquad + (- i)^{\frac{n - 1}{2}} \, \lambda_3 \, \frac{\Gamma\left(\frac{n - 1}{2} + \bar \gamma\right)}{\left(X_0 - X\right)^{\frac{n}{2}}} \, \left(1 + (- 1)^n\right) \, \left(1 + \frac{\bar h_1}{\bar k_1} \, \frac{e^{- 2ix}}{\left(X_0 - X\right)^{2 \, \eta \, i}}\right) \, \bs \phi_1^{[1]}.
                \end{aligned} \label{innerapprox}
            \end{equation}
            Thus, from now on, we assume that $n$ is even (when \evs{assumed to be large}).
            \begin{remark} \label{rem:factorjustification}
                Note that, for each dominant value of $\kappa$ we consider, there are two associated dominant values of $r$. Therefore, as there are four dominant values of $\kappa$, the summations in \eqref{final-inner} and \eqref{innerapprox} comprise eight terms in total. Nevertheless, something similar will happen to the other four terms, so they are omitted.
            \end{remark}
            
\section{Late term expansion --- outer solution}
\label{sec:late_term2}    
    The previous procedure allowed us to approximate the solution when close to singularities. Nevertheless, we note that the ansatz \eqref{ansatz-inner-sol} is not so accurate away from them as it lacks the information provided by the \evs{expansion of the} parameters, $a$ and $b$. For that reason, we need to consider a more general ansatz away from the singularities. Let
    \begin{align}
        \mbf W_{n, r} = \frac{\kappa^{n - 1} \, \Gamma\left(\frac{n - 1}{2} + \gamma_r\right)}{\left(X_0 - X\right)^{\frac{n - 1}{2} + \gamma_r}} \, \sum_{s\geq 0} \frac{\mbf g_r^{[s/2]}(X)}{n^{s/2}}, \label{ansatz-os}
    \end{align}
    where we have used the summation over natural halved orders to take into account all the terms we obtained in our expansion for the amplitude up to order five (see Section \ref{sec:regular_asymptotics}). With this, note that
    \begin{align*}
        \mbf W_{n - 2, r}' &= \begin{multlined}[t]
            \frac{1}{2} \, \kappa^{n - 3} \, \frac{\Gamma\left(\frac{n - 3}{2} + \gamma_r\right)}{\left(X_0 - X\right)^{\frac{n - 1}{2} + \gamma_r}} \sum_{s\geq 0} \frac{1}{(n - 2)^{s/2}} \left(\vphantom{\left(X_0 - X\right) \, \left(\mbf g_r^{[s/2]}\right)'} \left(n - 3 + 2 \, \gamma_r\right) \, \mbf g_r^{[s/2]}\right.
            \\
            \left. + 2 \, \left(X_0 - X\right) \, \left(\mbf g_r^{[s/2]}\right)'\right),
        \end{multlined}
        \\
        \mbf W_{n - 4, r}'' &= \frac{1}{4} \, \kappa^{n - 5} \, \frac{\Gamma\left(\frac{n - 5}{2} + \gamma_r\right)}{\left(X_0 - X\right)^{\frac{n - 1}{2} + \gamma_r}} \sum_{s\geq 0} \frac{1}{(n - 4)^{s/2}} \left(\vphantom{\left(X_0 - X\right)^2 \, \left(\mbf g_r^{[s/2]}\right)''} \left(n - 5 + 2 \, \gamma_r\right) \left(n - 3 + 2 \, \gamma_r\right) \, \mbf g_r^{[s/2]}\right.
        \\
        & \qquad \left. + 4 \, \left(n - 5 + 2 \, \gamma_r\right) \, \left(X_0 - X\right) \, \left(\mbf g_r^{[s/2]}\right)' + 4 \, \left(X_0 - X\right)^2 \, \left(\mbf g_r^{[s/2]}\right)''\right).
    \end{align*}
    Thus, when replacing \eqref{ansatz-os} into \eqref{firstansatz}, we obtain
    {\allowdisplaybreaks
    \begin{multline*}
        \mbf 0 = \sum_{\substack{p = 0 \\ \mbf p \cdot \mbf e = p}}^{n - 1} \, \sum_{\substack{q = 0 \\ \mbf q \cdot \mbf e = q}}^{n - 1 - p} \left(C_{\mbf p, \mbf q} \, \left(\prod_{s \geq 1} a_s^{p_s}\right) \, \left(\prod_{s \geq 1} b_s^{q_s}\right) \, \jac \mbf f_{p, q}(\mbf 0) \, \kappa^{n - 1 - p - q}\right.
        \\
        \left. \times \frac{\Gamma\left(\frac{n - p - q - 1}{2} + \gamma_r\right)}{\left(X_0 - X\right)^{\frac{n - p - q - 1}{2} + \gamma_r}} \, \sum_{s\geq 0} \frac{\mbf g_r^{[s/2]}}{(n - p - q)^{s/2}}\right)
        \\
        + \sum_{\substack{p = 0 \\ \mbf p \cdot \mbf e = p}}^{n - 1} \, \sum_{\substack{q = 0 \\ \mbf q \cdot \mbf e = q}}^{n - 1 - p} \, \sum_{\substack{M = 2 \\ \mbf n \cdot \mbf e_M = n - p - q}}^{n - p - q} \sum_{\substack{\mbf r \cdot \mbf e_M = r \\ - n_i \leq r_i \leq n_i}} \left(C_{\mbf p, \mbf q} \, \left(\prod_{s \geq 1} a_s^{p_s}\right) \, \left(\prod_{s \geq 1} b_s^{q_s}\right)\right.
        \\
        \left. \times \mbf F_{M, p, q}\left(A_{n_1, r_1} \, \mbf W_{n_1, r_1}, \ldots , A_{n_M, r_M} \, \mbf W_{n_M, r_M}\right) \vphantom{C_{\mbf p, \mbf q} \, \left(\prod_{s \geq 1} a_s^{p_s}\right) \, \left(\prod_{s \geq 1} b_s^{q_s}\right)}\right)
        \\
        + k^2 \, \hat D \, \left(- r^2 \, \kappa^{n - 1} \, \frac{\Gamma\left(\frac{n - 1}{2} + \gamma_r\right)}{\left(X_0 - X\right)^{\frac{n - 1}{2} + \gamma_r}} \, \sum_{s\geq 0} \frac{\mbf g_r^{[s/2]}}{n^{s/2}}\right.
        \\
        + 2 \, i \, r \, \begin{multlined}[t]
            \left(\frac{1}{2} \, \kappa^{n - 3} \, \frac{\Gamma\left(\frac{n - 3}{2} + \gamma_r\right)}{\left(X_0 - X\right)^{\frac{n - 1}{2} + \gamma_r}}\right.
            \\
            \left. \times \sum_{s\geq 0} \frac{1}{(n - 2)^{s/2}} \left(\left(n - 3 + 2 \, \gamma_r\right) \, \mbf g_r^{[s/2]} + 2 \, \left(X_0 - X\right) \, \left(\mbf g_r^{[s/2]}\right)'\right)\right)
        \end{multlined}
        \\
        + \frac{1}{4} \, \kappa^{n - 5} \, \frac{\Gamma\left(\frac{n - 5}{2} + \gamma_r\right)}{\left(X_0 - X\right)^{\frac{n - 1}{2} + \gamma_r}} \sum_{s\geq 0} \frac{1}{(n - 4)^{s/2}} \left(\left(n - 5 + 2 \, \gamma_r\right) \left(n - 3 + 2 \, \gamma_r\right) \, \mbf g_r^{[s/2]}\right.
        \\
        \left. \left. + 4 \, \left(n - 5 + 2 \, \gamma_r\right) \left(X_0 - X\right) \left(\mbf g_r^{[s/2]}\right)' + 4 \, \left(X_0 - X\right)^2 \left(\mbf g_r^{[s/2]}\right)''\right)\right),
    \end{multline*}
    }
    which is equivalent to
    {\allowdisplaybreaks
    \begin{multline}
        \mbf 0 = \Gamma\left(\frac{n - 1}{2} + \gamma_r\right)^{- 1} \, \sum_{\substack{p = 0 \\ \mbf p \cdot \mbf e = p}}^{n - 1} \, \sum_{\substack{q = 0 \\ \mbf q \cdot \mbf e = q}}^{n - 1 - p} \left(C_{\mbf p, \mbf q} \, \left(\prod_{s \geq 1} a_s^{p_s}\right) \, \left(\prod_{s \geq 1} b_s^{q_s}\right) \, \jac \mbf f_{p, q}(\mbf 0) \, \kappa^{4 - p - q}\right.
        \\
        \left. \times \frac{\Gamma\left(\frac{n - p - q}{2} + \gamma_r\right)}{\left(X_0 - X\right)^{\frac{- p - q}{2}}} \, \sum_{s \geq 0} \frac{\mbf g_r^{[s/2]}}{(n - p - q)^{s/2}}\right)
        \\
        + \kappa^{5 - n} \, \Gamma\left(\frac{n - 1}{2} + \gamma_r\right)^{- 1} \, \left(X_0 - X\right)^{\frac{n - 1}{2} + \gamma_r} \sum_{\substack{p = 0 \\ \mbf p \cdot \mbf e = p}}^{n - 1} \, \sum_{\substack{q = 0 \\ \mbf q \cdot \mbf e = q}}^{n - 1 - p} \, \sum_{\substack{M = 2 \\ \mbf n \cdot \mbf e_M = n - p - q}}^{n - p - q} \sum_{\substack{\mbf r \cdot \mbf e_M = r \\ - n_i \leq r_i \leq n_i}} \left( \vphantom{\left(\prod_{s \geq 1} a_s^{p_s}\right) \, \left(\prod_{s \geq 1} b_s^{q_s}\right)} C_{\mbf p, \mbf q}\right.
        \\
        \left. \times \left(\prod_{s \geq 1} a_s^{p_s}\right) \, \left(\prod_{s \geq 1} b_s^{q_s}\right) \, \mbf F_{M, p, q}\left(A_{n_1, r_1} \, \mbf W_{n_1, r_1}, \ldots , A_{n_M, r_M} \, \mbf W_{n_M, r_M}\right)\right)
        \\
        + k^2 \, \hat D \, \left(- r^2 \, \kappa^4 \, \sum_{s\geq 0} \frac{\mbf g_r^{[s/2]}}{n^{s/2}}\right.
        \\
        + 2 \, i \, r \, \kappa^2 \, \sum_{s\geq 0} \frac{1}{(n - 2)^{s/2}}  \left(\mbf g_r^{[s/2]} + \frac{2}{n - 3 + 2 \, \gamma_r} \, \left(X_0 - X\right) \, \left(\mbf g_r^{[s/2]}\right)'\right)
        \\
        + \sum_{s\geq 0} \frac{1}{(n - 4)^{s/2}} \left(\mbf v_r^{[s/2]} + \frac{4}{n - 3 + 2 \, \gamma_r} \, \left(X_0 - X\right) \left(\mbf g_r^{[s/2]}\right)'\right.
        \\
        \left. \left. + \frac{4}{\left(n - 5 + 2 \, \gamma_r\right) \left(n - 3 + 2 \, \gamma_r\right)} \, \left(X_0 - X\right)^2 \left(\mbf g_r^{[s/2]}\right)''\right)\right), \label{notnicetoexpand}
    \end{multline}
    }
    Now, as usual, we go forward by following an order-by-order expansion.
    
    \subsection{The case $\kappa^2 = \kappa_+^2$}
        \paragraph{Order $O(1)$.} At this order, as for the inner solution, \eqref{notnicetoexpand} becomes
        \begin{align*}
            M(\kappa, r) \, \mbf g_r^{[0]} = \mbf{NL}(1).
        \end{align*}
        Now, as $\kappa^2 = \kappa_+^2$, then $\mbf g_0^{[0]} = h_0^{[0]} \, \bs \phi_1^{[1]}$, and $\mbf g_2^{[0]} = h_2^{[0]} \, \bs \phi_1^{[1]}$, where $h_0^{[0]} = h_0^{[0]}(X)$ and $h_2^{[0]} = h_2^{[0]}(X)$, which implies that
        \begin{align*}
            \vdots
            \\
            \mbf g_{- 4}^{[0]} &= \mbf 0,
            \\
            \mbf g_{- 3}^{[0]} &= \mbf 0,
            \\
            \mbf g_{- 2}^{[0]} &= \mbf 0,
            \\
            \mbf g_{- 1}^{[0]} &= - 2 \, \kappa^3 \, \bar A_1 \, h_0^{[0]} \, \mbf W_2^{[2]},
            \\
            \mbf g_1^{[0]} &= - 2 \, \kappa^3 \, \left(\bar A_1 \, h_2^{[0]} + A_1 \, h_0^{[0]}\right) \, \mbf W_0^{[2]},
            \\
            \mbf g_3^{[0]} &= - 2 \, \kappa^3 \, A_1 \, h_2^{[0]} \, \mbf W_2^{[2]},
            \\
            \mbf g_4^{[0]} &= - 3 \, i \, A_1^2 \, h_2^{[0]} \, \mbf W_3^{[3]},
            \\
            \vdots
        \end{align*}
        \paragraph{Order $\mathcal O\left(1/n^{1/2}\right)$} At this order, \evs{we obtain}
        \begin{align*}
            \vdots
            \\
            \mbf g_{- 3}^{[1/2]} &= \mbf 0,
            \\
            \mbf g_{- 2}^{[1/2]} &= \mbf 0,
            \\
            \mbf g_{- 1}^{[1/2]} &= - 2 \, \sqrt{2} \, i \, \sqrt{X_0 - X} \, \bar A_1 \, h_0^{[0]} \, \mbf W_2^{[3]},
            \\
            \mbf g_0^{[1/2]} &= - \sqrt{2} \, \kappa^3 \, \sqrt{X_0 - X} \, h_0^{[0]} \, \mbf W_1^{[2]},
            \\
            \mbf g_1^{[1/2]} &= - 2 \, \sqrt{2} \, i \, \sqrt{X_0 - X} \, \left(\bar A_1 \, h_2^{[0]} + A_1 \, h_0^{[0]}\right) \, \mbf W_0^{[3]},
            \\
            \mbf g_2^{[1/2]} &= - \sqrt{2} \, \kappa^3 \, \sqrt{X_0 - X} \, h_2^{[0]} \, \mbf W_1^{[2]},
            \\
            \mbf g_3^{[1/2]} &= - 2 \, \sqrt{2} \, i \, \sqrt{X_0 - X} \, A_1 \, h_2^{[0]} \, \mbf W_2^{[3]},
            \\
            \vdots
        \end{align*}
        
        \paragraph{Order $\mathcal O(1/n)$.} At this order, \evs{we obtain}
        {\allowdisplaybreaks
        \begin{align*}
            \vdots
            \\
            \mbf g_{- 2}^{[1]} &= 6 \, i \, \left(X - X_0\right) \, \bar A_1^2 \, h_0^{[0]} \, \mbf W_3^{[3]},
            \\
            \mbf g_{- 1}^{[1]} &= 2 \, \kappa \, \left(X - X_0\right) \, \left(\vphantom{\left(2 \, h_0^{[0]} \, \mbf W_2^{[4]} - i \, \left(h_0^{[0]}\right)' \, \mbf W_{2, 3}^{[4]}\right)} \bar A_1^3 \, h_2^{[0]} \, \mbf W_{2, 2}^{[4]} + 3 \, \abs{A_1}^2 \, \bar A_1 \, h_0^{[0]} \, \mbf W_{2, 2}^{[4]}\right.
            \\
            &\qquad \left. + \bar A_1 \, \left(2 \, h_0^{[0]} \, \mbf W_2^{[4]} - i \, \left(h_0^{[0]}\right)' \, \mbf W_{2, 3}^{[4]}\right) - i \, \bar A_1' \, h_0^{[0]} \, \mbf W_{2, 3}^{[4]}\right),
            \\
            \mbf g_0^{[1]} &= - 2 \, \left(X - X_0\right) \, \left(- 2 \, i \, \abs{A_1}^2 \, h_0^{[0]} \, \mbf W_{1, 2}^{[3]} - i \, \bar A_1^2 \, h_2^{[0]} \, \mbf W_{1, 2}^{[3]} - i \, h_0^{[0]} \, \mbf W_1^{[3]} - \left(h_0^{[0]}\right)' \, \mbf W_{1, 3}^{[3]}\right),
            \\
            \mbf g_1^{[1]} &= 2 \, \kappa \, \left(X - X_0\right) \, \left(A_1 \, \left(2 \, h_0^{[0]} \, \left(2 \, \abs{A_1}^2 \, \mbf W_{0, 2}^{[4]} + \mbf W_0^{[4]}\right) + 4 \, \bar A_1^2 \, h_2^{[0]} \, \mbf W_{0, 2}^{[4]} - i \, \left(h_0^{[0]}\right)' \, \mbf W_{0, 3}^{[4]}\right)\right.
            \\
            &\qquad \left. + \bar A_1 \, \left(2 \, h_2^{[0]} \, \mbf W_0^{[4]} + i \, \left(h_2^{[0]}\right)' \, \mbf W_{0, 3}^{[4]}\right) - i \, \left(\bar A_1' \, h_2^{[0]} - A_1' \, h_0^{[0]}\right) \, \mbf W_{0, 3}^{[4]}\right),
            \\
            \mbf g_2^{[1]} &= 2 \, i \, \left(X - X_0\right) \, \left(h_2^{[0]} \, \left(2 \, \abs{A_1}^2 \, \mbf W_{1, 2}^{[3]} + \mbf W_1^{[3]}\right) + A_1^2 \, h_0^{[0]} \, \mbf W_{1, 2}^{[3]} + i \, \left(h_2^{[0]}\right)' \, \mbf W_{1, 3}^{[3]}\right),
            \\
            \mbf g_3^{[1]} &= \begin{multlined}[t]
                2 \, \kappa \, \left(X - X_0\right) \, \left(A_1 \, \left(h_2^{[0]} \, \left(3 \, \abs{A_1}^2 \, \mbf W_{2, 2}^{[4]} + 2 \, \mbf W_2^{[4]}\right) + i \, \left(h_2^{[0]}\right)' \, \mbf W_{2, 3}^{[4]}\right)\right.
                \\
                \left. + A_1^3 \, h_0^{[0]} \, \mbf W_{2, 2}^{[4]} + i \, A_1' \, h_2^{[0]} \, \mbf W_{2, 3}^{[4]} \vphantom{\left(h_2^{[0]}\right)'}\right),
            \end{multlined}
            \\
            \vdots
        \end{align*}
        }
    
        \paragraph{Order $\mathcal O\left(1/n^{3/2}\right)$} At this order, \evs{we obtain}
        \begin{align*}
            \vdots
            \\
            \mbf g_0^{[3/2]} &= \begin{multlined}[t]
                \frac{\kappa \, \sqrt{X_0 - X}}{2 \, \sqrt{2}} \, \left(\vphantom{\left(h_0^{[0]}\right)'} h_0^{[0]} \, \left(16 \, \left(X - X_0\right) \, \abs{A_1}^2 \, \mbf W_{1, 2}^{[4]} + i \, (4 \, \gamma - 5) \, \mbf W_1^{[2]} + 8 \, \left(X - X_0\right) \, \mbf W_1^{[4]}\right)\right.
                \\
                \left. + 8 \, \left(X - X_0\right) \, \left(\bar A_1^2 \, h_2^{[0]} \, \mbf W_{1, 2}^{[4]} - i \, \left(h_0^{[0]}\right)' \, \mbf W_{1, 3}^{[4]}\right)\right),
            \end{multlined}
            \\
            \mbf g_2^{[3/2]} &= \begin{multlined}[t]
                \frac{\kappa \, \sqrt{X_0 - X}}{2 \, \sqrt{2}} \, \left(\vphantom{\left(h_2^{[0]}\right)'} h_2^{[0]} \, \left(16 \, \left(X - X_0\right) \, \abs{A_1}^2 \, \, \mbf W_{1, 2}^{[4]} + i \, (4 \, \gamma - 5) \, \mbf W_1^{[2]} + 8 \, \left(X - X_0\right) \, \mbf W_1^{[4]}\right)\right.
                \\
                \left. + 8 \, \left(X - X_0\right) \, A_1^2 \, h_0^{[0]} \, \mbf W_{1, 2}^{[4]} + 8 \, i \, \left(X - X_0\right) \, \left(h_2^{[0]}\right)' \, \mbf W_{1, 3}^{[4]}\right).
            \end{multlined}
            \\
            \vdots
        \end{align*}
        With this, we are ready to get the solvability conditions at the next order.
    
        \paragraph{Order $\mathcal O\left(1/n^2\right)$} At this order, we obtain two solvability conditions given by
        \begin{multline}
            S_{c, 0} = \alpha_1 \, \left(h_0^{[0]}\right)'' - i \left(\alpha_2 + \alpha_3 \, \abs{A_1}^2\right) \, \left(h_0^{[0]}\right)'
            \\
            + \left(\alpha_5 + 2 \, \alpha_6 \, \abs{A_1}^2 + 3 \, \alpha_7 \, \abs{A_1}^4 - i \left(\alpha_3 \, A_1 \, \bar A_1' + 2 \, \alpha_4 \, \bar A_1 \, A_1'\right)\right) \, h_0^{[0]}
            \\
            - i \, \alpha_4 \, \bar A_1^2 \, \left(h_2^{[0]}\right)' + \left(\alpha_6 \, \bar A_1^2 + 2 \, \alpha_7 \, \bar A_1^2 \, \abs{A_1}^2 - i \, \alpha_3 \, \bar A_1 \, \bar A_1'\right) h_2^{[0]} = 0, \label{sc11os}
        \end{multline}
        and
        \begin{multline}
            S_{c, 2} = \alpha_1 \, \left(h_2^{[0]}\right)'' + i \, \left(\alpha_2 + \alpha_3 \, \abs{A_1}^2\right) \, \left(h_2^{[0]}\right)'
            \\
            + \left(\alpha_5 + 2 \, \alpha_6 \, \abs{A_1}^2 + 3 \, \alpha_7 \, \abs{A_1}^4 + i \, \left(\alpha_3 \, \bar A_1 \, A_1' + 2 \, \alpha_4 \, A_1 \, \bar A_1'\right)\right) \, h_2^{[0]}
            \\
            + i \, \alpha_4 \, A_1^2 \, \left(h_0^{[0]}\right)' + \left(\alpha_6 \, A_1^2 + 2 \, \alpha_7 \, A_1^2 \, \abs{A_1}^2 + i \, \alpha_3 \, A_1 \, A_1'\right) \, h_0^{[0]} = 0. \label{sc12os}
        \end{multline}

    \subsection{The case $\kappa^2 = \kappa_-^2$}
        In this case, we have $\mbf g_{- 2}^{[0]} = h_{- 2}^{[0]} \, \bs \phi_1^{[1]}$ and $\mbf g_0^{[0]} = h_0^{[0]} \, \bs \phi_1^{[1]}$, where $h_{- 2}^{[0]} = h_{- 2}^{[0]}(X)$ and $h_0^{[0]} = h_0^{[0]}(X)$. Furthermore, as before, following an asymptotic expansion using an order-by-order approach, we have that.
    
        \paragraph{Order $\mathcal O(1)$.} At this order, \evs{we obtain}
        \begin{align*}
            \vdots
            \\
            \mbf g_{- 4}^{[0]} &= 3 \, i \, \bar A_1^2 \, h_{- 2}^{[0]} \, \mbf W_3^{[3]},
            \\
            \mbf g_{- 3}^{[0]} &= - 2 \, \kappa^3 \, \bar A_1 \, h_{- 2}^{[0]} \, \mbf W_2^{[2]},
            \\
            \mbf g_{- 1}^{[0]} &= - 2 \, \kappa^3 \, \left(\bar A_1 \, h_0^{[0]} + A_1 \, h_{- 2}^{[0]}\right) \, \mbf W_0^{[2]},
            \\
            \mbf g_1^{[0]} &= - 2 \, \kappa^3 \, A_1 \, h_0^{[0]} \, \mbf W_2^{[2]},
            \\
            \mbf g_2^{[0]} &= \mbf 0,
            \\
            \mbf g_3^{[0]} &= \mbf 0,
            \\
            \mbf g_4^{[0]} &= \mbf 0,
            \\
            \vdots
        \end{align*}
    
        \paragraph{Order $\mathcal O\left(1/n^{1/2}\right)$.} At this order, \evs{we obtain}
        {\allowdisplaybreaks
        \begin{align*}
            \vdots
            \\
            \mbf g_{- 3}^{[1/2]} &= 2 \, i \, \sqrt{2} \, \sqrt{X_0 - X} \, \bar A_1 \, h_{- 2}^{[0]} \, \mbf W_2^{[3]},
            \\
            \mbf g_{- 2}^{[1/2]} &= - \sqrt{2} \, \kappa^3 \, \sqrt{X_0 - X} \, h_{- 2}^{[0]} \, \mbf W_1^{[2]},
            \\
            \mbf g_{- 1}^{[1/2]} &= 2 \, \sqrt{2} \, i \, \mbf W_0^{[3]} \, \sqrt{X_0 - X} \, \left(\bar A_1 \, h_0^{[0]} + A_1 \, h_{- 2}^{[0]}\right),
            \\
            \mbf g_0^{[1/2]} &= - \sqrt{2} \, \kappa^3 \, \sqrt{X_0 - X} \, h_0^{[0]} \, \mbf W_1^{[2]},
            \\
            \mbf g_1^{[1/2]} &= 2 \, \sqrt{2} \, i \, \sqrt{X_0 - X} \, A_1 \, h_0^{[0]} \, \mbf W_2^{[3]},
            \\
            \mbf g_2^{[1/2]} &= \mbf 0,
            \\
            \mbf g_3^{[1/2]} &= \mbf 0,
            \\
            \vdots
        \end{align*}
        }
        \paragraph{Order $\mathcal O(1/n)$.} At this order, \evs{we obtain}
        {\allowdisplaybreaks
        \begin{align*}
            \vdots
            \\
            \mbf g_{- 3}^{[1]} &= 2 \, \kappa \, \left(X - X_0\right) \, \left(\bar A_1^3 \, h_0^{[0]} \, \mbf W_{2, 2}^{[4]} + 3 \, \abs{A_1}^2 \, \bar A_1 \, h_{- 2}^{[0]} \, \mbf W_{2, 2}^{[4]}\right.
            \\
            &\qquad \left. + \bar A_1 \, \left(2 \, h_{- 2}^{[0]} \, \mbf W_2^{[4]} - i \, \left(h_{- 2}^{[0]}\right)' \, \mbf W_{2, 3}^{[4]}\right) - i \, \bar A_1' \, h_{- 2}^{[0]} \, \mbf W_{2, 3}^{[4]}\right),
            \\
            \mbf g_{- 2}^{[1]} &= - 2 \, \left(X - X_0\right) \, \left(2 \, i \, \abs{A_1}^2 \, h_{- 2}^{[0]} \, \mbf W_{1, 2}^{[3]} + i \, \left(\bar A_1^2 \, h_0^{[0]} \, \mbf W_{1, 2}^{[3]} + h_{- 2}^{[0]} \, \mbf W_1^{[3]}\right) + \left(h_{- 2}^{[0]}\right)' \, \mbf W_{1, 3}^{[3]}\right),
            \\
            \mbf g_{- 1}^{[1]} &= 2 \, \kappa \, \left(X - X_0\right) \, \left(A_1 \, \left(2 \, h_{- 2}^{[0]} \, \left(2 \, \abs{A_1}^2 \, \mbf W_{0, 2}^{[4]} + \mbf W_0^{[4]}\right) + 4 \, \bar A_1^2 \, h_0^{[0]} \, \mbf W_{0, 2}^{[4]} - i \, \left(h_{- 2}^{[0]}\right)' \, \mbf W_{0, 3}^{[4]}\right)\right.
            \\
            &\qquad \left. + \bar A_1 \, \left(2 \, h_0^{[0]} \, \mbf W_0^{[4]} + i \, \left(h_0^{[0]}\right)' \, \mbf W_{0, 3}^{[4]}\right) - i \, \left(h_0^{[0]} \, \bar A_1' - h_{- 2}^{[0]} \, A_1'\right) \, \mbf W_{0, 3}^{[4]}\right),
            \\
            \mbf g_0^{[1]} &= - 2 \, i \, \left(X - X_0\right) \, \left(h_0^{[0]} \, \left(2 \, \abs{A_1}^2 \, \mbf W_{1, 2}^{[3]} + \mbf W_1^{[3]}\right) + A_1^2 \, h_{- 2}^{[0]} \, \mbf W_{1, 2}^{[3]} + i \, \left(h_0^{[0]}\right)' \, \mbf W_{1, 3}^{[3]}\right),
            \\
            \mbf g_1^{[1]} &= 2 \, \kappa \, \left(X - X_0\right) \, \left(A_1 \, \left(h_0^{[0]} \, \left(3 \, \abs{A_1}^2 \, \mbf W_{2, 2}^{[4]} + 2 \, \mbf W_2^{[4]}\right) + i \, \left(h_0^{[0]}\right)' \, \mbf W_{2, 3}^{[4]}\right)\right.
            \\
            &\qquad \left. + A_1^3 \, h_{- 2}^{[0]} \, \mbf W_{2, 2}^{[4]} + i \, A_1' \, h_0^{[0]} \, \mbf W_{2, 3}^{[4]}\right),
            \\
            \mbf g_2^{[1]} &= - 6 \, i \, \left(X - X_0\right) \, A_1^2 \, h_0^{[0]} \, \mbf W_3^{[3]},
            \\
            \vdots
        \end{align*}
        }
    
        \paragraph{Order $\mathcal O(1/n^{3/2})$.} At this order, \evs{we obtain}
        \begin{align*}
            \vdots
            \\
            \mbf g_{- 2}^{[3/2]} &= \begin{multlined}[t]
                \frac{\kappa \, \sqrt{X_0 - X}}{2 \, \sqrt{2}} \, \left(\vphantom{\left(\bar A_1^2 \, h_0^{[0]} \, \mbf W_{1, 2}^{[4]} - i \, \left(h_{- 2}^{[0]}\right)' \, \mbf W_{1, 3}^{[4]}\right)} h_{- 2}^{[0]} \, \left(16 \, \left(X - X_0\right) \, \abs{A_1}^2 \, \mbf W_{1, 2}^{[4]} - i \, (4 \, \gamma - 5) \, \mbf W_1^{[2]} + 8 \, \left(X - X_0\right) \, \mbf W_1^{[4]}\right)\right.
                \\
                \left. + 8 \, \left(X - X_0\right) \, \left(\bar A_1^2 \, h_0^{[0]} \, \mbf W_{1, 2}^{[4]} - i \, \left(h_{- 2}^{[0]}\right)' \, \mbf W_{1, 3}^{[4]}\right)\right),
            \end{multlined}
            \\
            \mbf g_0^{[3/2]} &= \begin{multlined}[t]
                \frac{\kappa \, \sqrt{X_0 - X}}{2 \, \sqrt{2}} \, \left(h_0^{[0]} \, \left(16 \, \left(X - X_0\right) \, \abs{A_1}^2 \, \mbf W_{1, 2}^{[4]} - i \, (4 \, \gamma - 5) \, \mbf W_1^{[2]} + 8 \, \left(X - X_0\right) \, \mbf W_1^{[4]}\right)\right.
                \\
                \left. + 8 \, \left(X - X_0\right) \, A_1^2 \, h_{- 2}^{[0]} \, \mbf W_{1, 2}^{[4]} + 8 \, i \, \left(X - X_0\right) \, \left(h_0^{[0]}\right)' \, \mbf W_{1, 3}^{[4]}\right).
            \end{multlined}
            \\
            \vdots
        \end{align*}
        \paragraph{Order $\mathcal O\left(1/n^2\right)$.} Finally, at this order, we obtain two solvability conditions given by
        \begin{multline}
            S_{c, - 2} = \alpha_1 \, \left(h_{- 2}^{[0]}\right)'' - i \, \left(\alpha_2 + \alpha_3 \, \abs{A_1}^2\right) \, \left(h_{- 2}^{[0]}\right)'
            \\
            + \left(\alpha_5 + 2 \, \alpha_6 \, \abs{A_1}^2 + 3 \, \alpha_7 \, \abs{A_1}^4 - i \, \left(\alpha_3 \, A_1 \, \bar A_1' + 2 \, \alpha_4 \, \bar A_1 \, A_1'\right)\right) \, h_{- 2}^{[0]}
            \\
            - i \, \alpha_4 \, \bar A_1^2 \, \left(h_0^{[0]}\right)' + \left(\alpha_6 \, \bar A_1^2 + 2 \, \alpha_7 \, \abs{A_1}^2 \, \bar A_1^2 - i \, \alpha_3 \, \bar A_1 \, \bar A_1'\right) \, h_0^{[0]} = 0, \label{sc21os}
        \end{multline}
        and
        \begin{multline}
            S_{c, 0} = \alpha_1 \, \left(h_0^{[0]}\right)'' + i \, \left(\alpha_2 + \alpha_3 \, \abs{A_1}^2\right) \, \left(h_0^{[0]}\right)'
            \\
            + \left(\alpha_5 + 2 \, \alpha_6 \, \abs{A_1}^2 + 3 \, \alpha_7 \, \abs{A_1}^4 + i \, \left(\alpha_3 \, \bar A_1 \, A_1' + 2 \, \alpha_4 \, A_1 \, \bar A_1'\right)\right) \, h_0^{[0]}
            \\
            + i \, \alpha_4 \, A_1^2 \, \left(h_{- 2}^{[0]}\right)' + \left(\alpha_6 \, A_1^2 + 2 \, \alpha_7 \, \abs{A_1}^2 \, A_1^2 + i \, \alpha_3 \, A_1 \, A_1'\right) \, h_{- 2}^{[0]} = 0. \label{sc22os}
        \end{multline}
        
    \subsection{Solving a coupled system of amplitude equations}
        Note that the pairs of equations \eqref{sc11os}--\eqref{sc12os} and \eqref{sc21os}--\eqref{sc22os} are symmetric. Thus, it is enough to solve one pair of them in order to know the solutions for both. In particular, we focus on solving \evs{the pair \eqref{sc11os}--\eqref{sc12os}, which is given by}
        \begin{multline}
            \alpha_1 \, \left(h_0^{[0]}\right)'' - i \left(\alpha_2 + \alpha_3 \, \abs{A_1}^2\right) \, \left(h_0^{[0]}\right)'
            \\
            + \left(\alpha_5 + 2 \, \alpha_6 \, \abs{A_1}^2 + 3 \, \alpha_7 \, \abs{A_1}^4 - i \left(\alpha_3 \, A_1 \, \bar A_1' + 2 \, \alpha_4 \, \bar A_1 \, A_1'\right)\right) \, h_0^{[0]}
            \\
            - i \, \alpha_4 \, \bar A_1^2 \, \left(h_2^{[0]}\right)' + \left(\alpha_6 \, \bar A_1^2 + 2 \, \alpha_7 \, \bar A_1^2 \, \abs{A_1}^2 - i \, \alpha_3 \, \bar A_1 \, \bar A_1'\right) \, h_2^{[0]} = 0, \label{actualeq1os}
        \end{multline}
        and
        \begin{multline}
            \alpha_1 \, \left(h_2^{[0]}\right)'' + i \, \left(\alpha_2 + \alpha_3 \, \abs{A_1}^2\right) \, \left(h_2^{[0]}\right)'
            \\
            + \left(\alpha_5 + 2 \, \alpha_6 \, \abs{A_1}^2 + 3 \, \alpha_7 \, \abs{A_1}^4 + i \, \left(\alpha_3 \, \bar A_1 \, A_1' + 2 \, \alpha_4 \, A_1 \, \bar A_1'\right)\right) \, h_2^{[0]}
            \\
            + i \, \alpha_4 \, A_1^2 \, \left(h_0^{[0]}\right)' + \left(\alpha_6 \, A_1^2 + 2 \, \alpha_7 \, A_1^2 \, \abs{A_1}^2 + i \, \alpha_3 \, A_1 \, A_1'\right) \, h_0^{[0]} = 0. \label{actualeq2os}
        \end{multline}
        To solve these equations, we note that \eqref{actualeq2os} is \evs{a} complex conjugate version of \eqref{actualeq1os}. Furthermore, based on \cite{Chapman,Dean}, and by looking at the equations $R_1$ and $\varphi_1$ solve, \eqref{realeq} and \eqref{imageq}, we set
        \begin{align*}
            h_0^{[0]} = \left(R_2 - i \, R_1 \, \varphi_2\right) \, e^{- i \, \varphi_1} \qquad \text{and} \qquad h_2^{[0]} = \left(R_2 + i \, R_1 \, \varphi_2\right) \, e^{i \, \varphi_1},
        \end{align*}
        where $R_2 \evs{= R_2(X)}$ and $\varphi_2 \evs{= \varphi_2(X)}$ are real functions. Therefore, when we replace this into \eqref{actualeq1os}, we obtain
        \begin{multline*}
            - 48 \, \alpha_1^2 \, \beta_5 \, R_1^5 \, \varphi_2 - 32 \, \alpha_1^2 \, \beta_3 \, R_1^3 \, \varphi_2 + 16 \, \alpha_1^2 \, \varphi_2 \, R_1'' + 8 \, \alpha_1 \, \alpha_3 \, R_1^2 \, R_2' + 8 \, \alpha_1 \, \alpha_4 \, R_1^2 \, R_2'
            \\
            + 32 \, \alpha_1^2 \, R_1' \, \varphi_2' + 24 \, \alpha_1 \, \alpha_3 \, R_2 \, R_1 \, R_1' + 24 \, \alpha_1 \, \alpha_4 \, R_2 \, R_1 \, R_1'  - 16 \, \alpha_1^2 \, \beta_1 \, R_1 \, \varphi_2 + 16 \, \alpha_1^2 \, R_1 \, \varphi_2''
            \\
            + i \, \left(- 240 \, \alpha_1^2 \, \beta_5 \, R_2 \, R_1^4 - 4 \, \alpha_3^2 \, R_2 \, R_1^4 + 12 \, \alpha_4^2 \, R_2 \, R_1^4 + 8 \, \alpha_3 \, \alpha_4 \, R_2 \, R_1^4 - 8 \, \alpha_1 \, \alpha_3 \, R_1^3 \, \varphi_2'\right.
            \\
            \left. + 24 \, \alpha_1 \, \alpha_4 \, R_1^3 \, \phi_2' - 96 \, \alpha_1^2 \, \beta_3 \, R_2 \, R_1^2 - 16 \, \alpha_1^2 \, \beta_1 \, R_2 + 16 \, \alpha_1^2 \, R_2''\right) = 0.
        \end{multline*}
        Here, as all the coefficients and functions are real, we have that the real and imaginary parts of the previous equation must equal zero to fulfill it. Then, after splitting the equation and simplifying it using \eqref{realeq}, \eqref{imageq}, and $\beta_3^2 = 4 \, \beta_1 \, \beta_5$, we obtain
        \begin{align}
            4 \, \alpha_1 \, R_1' \, \varphi_2' + 3 \, \left(\alpha_3 + \alpha_4\right) \, R_1 \, R_2 \, R_1' + R_1 \, \left(2 \, \alpha_1 \, \varphi_2'' + \left(\alpha_3 + \alpha_4\right) \, R_1 \, R_2'\right) &= 0, \label{simplest}
        \end{align}
        and
        \begin{multline}
            4 \, \alpha_1^2 \, \left(R_2'' - \beta_1 \, R_2\right)
            \\
            - R_1^2 \, \left(R_2 \, \left(24 \, \alpha_1^2 \, \beta_3 + R_1^2 \, \left(60 \, \alpha_1^2 \, \beta_5 + \left(\alpha_3 + \alpha_4\right) \, \left(\alpha_3 - 3 \, \alpha_4\right)\right)\right)\right.
            \\
            \left. + 2 \, \alpha_1 \, \left(\alpha_3 - 3 \, \alpha_4\right) \, R_1 \, \varphi_2'\right) = 0. \label{mostcomplicated}
        \end{multline}
        Note that \eqref{simplest} can be directly solved for $\varphi_2'$. Specifically, we have
        \begin{align*}
            \varphi_2' = - \frac{\left(\alpha_3 + \alpha_4\right) \, R_1 \, R_2}{2 \, \alpha_1} + \frac{4 \, \alpha_1 \, K_3}{R_1^2},
        \end{align*}
        where $K_3 \in \mathbb R$ is a constant of integration. Now, when replacing this expression into \eqref{mostcomplicated}, we obtain
        \begin{align}
            R_2'' - \left(\beta_1 + 6 \, \beta_3 \, R_1^2 + 15 \, \beta_5 \, R_1^4\right) \, R_2 = 2 \, \left(\alpha_3 - 3 \, \alpha_4\right) \, K_3 \, R_1, \label{eqforR2}
        \end{align}
        which is a linear non-homogeneous equation in $R_2$. To solve it, note that the choice $K_3 = 0$ makes the equation homogeneous, with a solution given by
        \begin{align*}
            R_{2, h} = R_1' \, \left(K_1 + K_2 \, \int_{X_0}^X \frac{\dd s}{\left(R_1'\right)^2}\right).
        \end{align*}
        Now, observe that if $K_1 = K_2 = K_3 = 0$, then $\varphi_2 = K_4 \in \mathbb R$ is a constant, which leads to
        \begin{align*}
            h_0^{[0]} &= - i \, K_4 \, R_1 \, e^{- i \, \varphi_1} = - i \, K_4 \, \bar A_1,
            \\
            h_2^{[0]} &= i \, K_4 \, R_1 \, e^{i \, \varphi_1} = i \, K_4 \, A_1.
        \end{align*}
        On the other hand, when $K_2 = K_3 = 0$, we have that $R_{2, h} = K_1 \, R_1'$, which implies
        \begin{align*}
            \varphi_2' = - K_1 \, \frac{\left(\alpha_3 + \alpha_4\right) \, \left(R_1^2\right)'}{4 \, \alpha_1} = - K_1 \, \varphi_1''.
        \end{align*}
        This yields $\varphi_2 = - K_1 \, \varphi_1'$ (setting the constant of integration to zero), which implies
        \begin{align*}
            h_0^{[0]} &= \left(K_1 \, R_1' - i \, K_1 \, R_1 \, \varphi_1'\right) \, e^{- i \, \varphi_1} = K_1 \, \bar A_1',
            \\
            h_2^{[0]} &= \left(K_1 \, R_1' + i \, K_1 \, R_1 \, \varphi_1'\right) \, e^{i \, \varphi_1} = K_1 \, A_1'.
        \end{align*}
        Now, if $K_3 \neq 0$, then we can use the method of variation of parameters to find the solution to the non-homogeneous equation, \eqref{eqforR2}. In particular, we have
        \begin{align*}
            R_{2, p} = \left(\alpha_3 - 3 \, \alpha_4\right) \, K_3 \, R_1' \, \ds{\int_{X_0}^X \frac{R_1^2}{\left(R_1'\right)^2} \, \dd s}.
        \end{align*}
        Therefore, we conclude that, when $K_1 = 0$,
        \begin{align*}
            R_2 = R_1' \, \int_{X_0}^X \frac{K_2 + \left(\alpha_3 - 3 \, \alpha_4\right) \, K_3 \, R_1^2}{\left(R_1'\right)^2} \, \dd s,
        \end{align*}
        which implies
        \begin{align*}
            \varphi_2 = - \frac{\alpha_3 + \alpha_4}{2 \, \alpha_1} \, \int_{X_0}^X \left(R_1 \, R_1' \, \int_{X_0}^s \frac{K_2 + \left(\alpha_3 - 3 \, \alpha_4\right) \, K_3 \, R_1^2}{\left(R_1'\right)^2} \, \dd \sigma\right) \, \dd s + 4 \, \alpha_1 \, K_3 \, \int_{X_0}^X \frac{\dd s}{R_1^2}.
        \end{align*}
        Thus, when gathering all the information we have obtained, we conclude that
        \begin{align*}
            h_2^{[0]} &= \begin{multlined}[t]
                K_1 \, A_1' + i \, K_4 \, A_1 + \left(R_1' \, \int_{X_0}^X \frac{K_2 + \left(\alpha_3 - 3 \, \alpha_4\right) \, K_3 \, R_1^2}{\left(R_1'\right)^2} \, \dd s\right.
                \\
                \left. + i \, R_1 \, \left(- \frac{\alpha_3 + \alpha_4}{2 \, \alpha_1} \, \int_{X_0}^X \left(R_1 \, R_1' \, \int_{X_0}^s \frac{K_2 + \left(\alpha_3 - 3 \, \alpha_4\right) \, K_3 \, R_1^2}{\left(R_1'\right)^2} \, \dd \sigma\right) \, \dd s\right.\right.
                \\
                \left.\left.+ 4 \, \alpha_1 \, K_3 \, \int_{X_0}^X \frac{\dd s}{R_1^2}\right)\right) \, e^{i \, \varphi_1},
            \end{multlined}
        \end{align*}
        and
        \begin{align*}
            h_0^{[0]} &= \begin{multlined}[t]
                K_1 \, \bar A_1' - i \, K_4 \, \bar A_1 + \left(R_1' \, \int_{X_0}^X \frac{K_2 + \left(\alpha_3 - 3 \, \alpha_4\right) \, K_3 \, R_1^2}{\left(R_1'\right)^2} \, \dd s\right.
                \\
                \left. - i \, R_1 \, \left(- \frac{\alpha_3 + \alpha_4}{2 \, \alpha_1} \, \int_{X_0}^X \left(R_1 \, R_1' \, \int_{X_0}^s \frac{K_2 + \left(\alpha_3 - 3 \, \alpha_4\right) \, K_3 \, R_1^2}{\left(R'\right)^2} \, \dd \sigma\right) \, \dd s\right.\right.
                \\
                \left.\left. + 4 \, \alpha_1 \, K_3 \, \int_{X_0}^X \frac{\dd s}{R_1^2}\right)\right) \, e^{- i \, \varphi_1}.
            \end{multlined}
        \end{align*}

        \begin{remark} \label{remark:non-real}
            We highlight that we solved the amplitude equations in this section by assuming that $R_2$ and $\varphi_2$ were real functions. Nevertheless, this is not the case when making an integration involving $X_0$ in the limits of integration. In reality, we are only interested in the case where $X$ is real. However, as we have been focusing on the study of singularities, we need to centre our solutions on them in every step to \evs{be consistent with} our analysis.
        \end{remark}

    \subsection{Matching inner and outer solutions}
        We have obtained an inner and an outer solution for our late-term expansion. To match these, we need
        \begin{align*}
            \kappa^{n - 1} \, \Gamma\left(\frac{n - 1}{2} + \gamma_r\right) \, \frac{\mbf g_r^{[0]}(X)}{\left(X_0 - X\right)^{\frac{n - 1}{2} + \gamma_r}} \sim \kappa^{n - 1} \, \Gamma\left(\frac{n - 1}{2} + \gamma_r\right) \, \frac{\mbf v_r^{[0]}}{\left(X_0 - X\right)^{\frac{n}{2} - r \, \eta \, i}},
        \end{align*}
        which implies
        \begin{align*}
            \mbf g_r^{[0]}(X) \sim \frac{\mbf v_r^{[0]}}{\left(X_0 - X\right)^{\frac{1}{2} - r \, \eta \, i - \gamma_r}}.
        \end{align*}
        In particular, we note that
        \begin{align*}
            \mbf g_{- 2}^{[0]}(X) \sim \frac{\mbf v_{- 2}^{[0]}}{\left(X_0 - X\right)^{\frac{1}{2} + 2 \, \eta \, i - \gamma}}, \quad \mbf g_0^{[0]}(X) \sim \frac{\mbf v_0^{[0]}}{\left(X_0 - X\right)^{\frac{1}{2} - \gamma}}, \quad \mbf g_2^{[0]}(X) \sim \frac{\mbf v_2^{[0]}}{\left(X_0 - X\right)^{\frac{1}{2} - 2 \, \eta \, i - \gamma}}.
        \end{align*}
        Furthermore, when $\gamma = 3 - \eta \, i$, we have
        \begin{align*}
            \mbf g_{- 2}^{[0]}(X) \sim \frac{\mbf v_{- 2}^{[0]}}{\left(X_0 - X\right)^{- \frac{5}{2} + 3 \, \eta \, i}}, \qquad \mbf g_0^{[0]}(X) \sim \frac{\mbf v_0^{[0]}}{\left(X_0 - X\right)^{- \frac{5}{2} + \eta \, i}}, \qquad \mbf g_2^{[0]}(X) \sim \frac{\mbf v_2^{[0]}}{\left(X_0 - X\right)^{- \frac{5}{2} - \eta \, i}}.
        \end{align*}
        Now, we observe that
        \begin{align*}
            h_2^{[0]} \sim \begin{multlined}[t]
                K_1 \, \left(\frac{i}{2} \, (2 \, \eta + i) \, e^{\pi \, \xi} \, \left(2 \, \sqrt{\beta_1}\right)^{1/2} \, \left(- 2 \, \beta_3\right)^{- 1/2} \, \left(X_0 - X\right)^{\eta \, i} \, \left(\frac{- 1}{X_0 - X}\right)^{3/2}\right)
                \\
                + K_2 \, \left(- \frac{1}{6} \, \frac{(3 + 2 \, \eta \, i) \, e^{\pi \, \xi} \, \left(2 \, \sqrt{\beta_1}\right)^{- 1/2} \, \left(- 2 \, \beta_3\right)^{1/2} \, \left(X_0 - X\right)^{2 + \eta \, i}}{\sqrt{\frac{- 1}{X_0 - X}}}\right)
                \\
                + K_3 \, \left(\frac{i \, \left(\alpha_3 \, (3 + 4 \, \eta \, (\eta - i)) - 3 \, \alpha_4 \, (2 \, \eta - i)^2\right) \, e^{\pi \, \xi} \, \left(2 \, \sqrt{\beta_1}\right)^{1/2} \, \left(X_0 - X\right)^{1 + \eta \, i}}{6 \, \eta \, \left(- 2 \, \beta_3\right)^{1/2} \, \sqrt{\frac{- 1}{X_0 - X}}}\right)
                \\
                + K_4 \, \left(i \, e^{\pi \, \xi} \, \left(2 \, \sqrt{\beta_1}\right)^{1/2} \, \left(- 2 \, \beta_3\right)^{- 1/2} \, \sqrt{\frac{- 1}{X_0 - X}} \left(X_0 - X\right)^{\eta \, i}\right),
            \end{multlined}
        \end{align*}
        which implies that the analysis done previously for the inner solution is matched only by the dominant term associated with $K_2$, which lets us conclude that
        \begin{align*}
            h_2^{[0]} \sim K_2 \, \left(R_1' \int_{X_0}^X \frac{\dd s}{\left(R_1'\right)^2} - \frac{\alpha_3 + \alpha_4}{2 \, \alpha_1} \, i \, R_1 \, \int_{X_0}^X \left(R_1 \, R_1' \, \int_{X_0}^s \frac{\dd \sigma}{\left(R_1'\right)^2}\right) \, \dd s\right) \, e^{i \, \varphi_1}.
        \end{align*}
        Furthermore, by matching this solution with \eqref{innerapprox} as $X \to X_0$, we obtain
        \begin{align}
            K_2 = - \frac{6 \, i \, e^{- \pi \, \xi} \, \left(2 \, \sqrt{\beta_1}\right)^{1/2} \, \left(- 2 \, \beta_3\right)^{- 1/2}}{(3 + 2 \, \eta \, i)} \, \lambda_1. \label{K_2}
        \end{align}
        \evs{Moreover}, we conclude that the contribution to the solution provided by $X_0$ is given by
        \begin{equation}
            \begin{aligned}
                \mbf u^{[n]} &\sim i^{\frac{n - 1}{2}} \, \frac{\Gamma\left(\frac{n - 1}{2} + \gamma\right)}{\left(X_0 - X\right)^{\frac{n - 1}{2} + \gamma}} \, \left(1 + (- 1)^n\right) \, \left(h_0^{[0]} + h_2^{[0]} \, e^{2ix}\right) \, \bs \phi_1^{[1]}
                \\
                &\quad + (- i)^{\frac{n - 1}{2}} \, \frac{\Gamma\left(\frac{n - 1}{2} + \bar \gamma\right)}{\left(X_0 - X\right)^{\frac{n - 1}{2} + \bar \gamma}} \, \left(1 + (- 1)^n\right) \, \left(\ol{h_0^{[0]}} + \ol{h_{- 2}^{[0]}} \, e^{- 2ix}\right) \, \bs \phi_1^{[1]},
            \end{aligned} \label{forrhs+}
        \end{equation}
        as $n \to \infty$. Similarly, due to symmetry, we have that the contribution to the solution provided by $\bar X_0 = X_{- 1}$ is given by
        \begin{align*}
            \mbf u^{[n]} &\sim (- i)^{\frac{n - 1}{2}} \, \frac{\Gamma\left(\frac{n - 1}{2} + \bar \gamma\right)}{\left(\bar X_0 - X\right)^{\frac{n - 1}{2} + \bar \gamma}} \, \left(1 + (- 1)^n\right) \, \left(\ol{h_0^{[0]}} + \ol{h_2^{[0]}} \, e^{-2ix}\right) \, \bs \phi_1^{[1]}
            \\
            &\quad + i^{\frac{n - 1}{2}} \, \frac{\Gamma\left(\frac{n - 1}{2} + \bar \gamma\right)}{\left(\bar X_0 - X\right)^{\frac{n - 1}{2} + \bar \gamma}} \, \left(1 + (- 1)^n\right) \, \left(h_0^{[0]} + h_{- 2}^{[0]} \, e^{2ix}\right) \, \bs \phi_1^{[1]},
        \end{align*}
        as $n \to \infty$.
        \begin{remark}
            We highlight that the last few expressions correspond to an abuse of notation. Let us recall that $R_2$ and $\varphi_2$ should be real functions (see Remark \ref{remark:non-real}), which would make $h_0^{[2]} = \ol{h_2^{[0]}}$. Nevertheless, as we are working at complex singularities, this is no longer the case. One technically needs to be extra careful with the terms that are being complex-conjugated and consider the fact that there are more constants of integration fulfilling different rules that can be deduced in the same way as when we matched solutions in this section. In any case, the development shown here is enough for our purposes, and it is consistent with \cite{Chapman,Dean}.
        \end{remark}

\section{Estimation of the residual} \label{sec:first_residual}
    In this section, we start the study of the residual that arises after truncating the asymptotic expansion we have been carrying out. Specifically, if we assume that we have performed the expansion up to a sufficiently large order, then the remainder is expected to be small. Therefore, we are only interested in studying its linearized equation, which turns out to approximate it well and will let us obtain more conditions to ensure that such a term tends to zero as $X \to \pm \infty$.

    We begin by introducing new notation for the linearized equation of the residual. In particular, we define the following symmetric multilinear vector functions in terms of the Jacobian of \eqref{general_equation}:
    {\allowdisplaybreaks
    \begin{align*}
        \jac \mbf F_{1, i, j}\left(\mbf v_1\right) &= \left.\left(\sum_{p_1 = 1}^n v_{1, p_1} \frac{\partial \jac \mbf f_{i, j}(\mbf u)}{\partial u_{p_1}}\right)\right|_{\mbf u = \mbf 0},
        \\
        \jac \mbf F_{2, i, j}\left(\mbf v_1, \mbf v_2\right) &= \frac{1}{2!} \, \left.\left(\sum_{p_1, p_2 = 1}^n v_{1, p_1} \, v_{2, p_2} \, \frac{\partial^2 \jac \mbf f_{i, j}(\mbf u)}{\partial u_{p_1} \partial u_{p_2}}\right)\right|_{\mbf u = \mbf 0},
        \\
        \jac \mbf F_{3, i, j}\left(\mbf v_1, \mbf v_2, \mbf v_3\right) &= \frac{1}{3!} \, \left.\left(\sum_{p_1, p_2, p_3 = 1}^n v_{1, p_1} \, v_{2, p_2} \, v_{3, p_3} \, \frac{\partial^3 \jac \mbf f_{i, j}(\mbf u)}{\partial u_{p_1} \partial u_{p_2} \partial u_{p_3}}\right)\right|_{\mbf u = \mbf 0},
        \\
        \jac \mbf F_{4, i, j}\left(\mbf v_1, \mbf v_2, \mbf v_3, \mbf v_4\right) &= \frac{1}{4!} \, \left.\left(\sum_{p_1, p_2, p_3, p_4 = 1}^n v_{1, p_1} \, v_{2, p_2} \, v_{3, p_3} \, v_{4, p_4} \, \frac{\partial^4 \jac \mbf f_{i, j}(\mbf u)}{\partial u_{p_1} \partial u_{p_2} \partial u_{p_3} \partial u_{p_4}}\right)\right|_{\mbf u = \mbf 0}.
    \end{align*}
    }
    However, although using these functions is useful for stating the equations for the remainder at each order, they will not let us spot key relationships between the vectors that define the remainder and the vectors obtained in Section \ref{sec:regular_asymptotics}. To spot these relationships, we consider the following result:
    \begin{lemma}
        For every pair of non-negative integers $i, j$, and every vectors $\mbf v_1, \mbf v_2, \mbf v_3, \mbf v_4, \mbf v_5$, we have
        \begin{align*}
            \jac \mbf F_{1, i, j}\left(\mbf v_1\right) \, \mbf v_2 &= 2 \, \mbf F_{2, i, j}\left(\mbf v_1, \mbf v_2\right),
            \\
            \jac \mbf F_{2, i, j}\left(\mbf v_1, \mbf v_2\right) \, \mbf v_3 &= 3 \, \mbf F_{3, i, j}\left(\mbf v_1, \mbf v_2, \mbf v_3\right),
            \\
            \jac \mbf F_{3, i, j}\left(\mbf v_1, \mbf v_2, \mbf v_3\right) \, \mbf v_4 &= 4 \, \mbf F_{4, i, j}\left(\mbf v_1, \mbf v_2, \mbf v_3, \mbf v_4\right),
            \\
            \jac \mbf F_{4, i, j}\left(\mbf v_1, \mbf v_2, \mbf v_3, \mbf v_4\right) \, \mbf v_5 &= 5 \, \mbf F_{5, i, j}\left(\mbf v_1, \mbf v_2, \mbf v_3, \mbf v_4, \mbf v_5\right).
        \end{align*}
    \end{lemma}
    \begin{proof}
        The proof is similar for each equality. Therefore, for simplicity, we give explicit proof of the first equality only. Observe that
        \begin{align*}
            \jac \mbf F_{1, i, j}\left(\mbf v_1\right) \, \mbf v_2 &= \left.\left(\sum_{p = 1}^n v_{1, p} \, \frac{\partial \jac \mbf f_{i, j}(\mbf u)}{\partial u_p}\right)\right|_{\mbf u = \mbf 0} \, \mbf v_2
            \\
            &= \left.\left(\sum_{p = 1}^n \sum_{q = 1}^n v_{1, p} \, v_{2, q} \, \frac{\partial^2 \mbf f_{i, j}(\mbf u)}{\partial u_q \partial u_p}\right)\right|_{\mbf u = \mbf 0} = 2 \, \mbf F_{2, i, j}\left(\mbf v_1, \mbf v_2\right),
        \end{align*}
        which proves the first equality and, therefore, as explained above, the result.
    \end{proof}
	The next step is to estimate the remainder \evs{due to the truncation of the asymptotic expansion} we are carrying out. To do that, we consider that our solution can be written as
    \begin{align*}
        \mbf u = \mbf S + \mbf R_N,
    \end{align*}
    where $\mbf S = \ds{\sum_{n = 1}^N \varepsilon^n \, \mbf u^{[n]}}$ is the approximation of the solution we have developed so far, and $\mbf R_N$ is the remainder, which is assumed to be small. We can therefore linearize \eqref{general_equation} about $\mbf S$, considering that $\mbf S$ solves the equation up to order $N$, obtaining
    \begin{equation}
        \begin{aligned}
            \mbf 0 &= \sum_{p = 0}^{M_a} \, \sum_{q = 0}^{M_b} a^p \, b^q \, \jac \mbf f_{p, q}(\mbf S) \, \mbf R_N + k^2 \, \hat D \, \left(\frac{\partial^2 \mbf R_N}{\partial x^2} + 2 \, \varepsilon^2 \, \frac{\partial^2 \mbf R_N}{\partial x \partial X} + \varepsilon^4 \, \frac{\partial^2 \mbf R_N}{\partial X^2}\right) + \mbf E(\delta b)
            \\
            & \quad + \text{forcing due to truncation},
        \end{aligned} \label{remaindereq}
    \end{equation}
    where `forcing due to truncation' corresponds to terms of order $\mathcal O\left(\varepsilon^{N + 1}\right)$ that come out of the evaluation of $\mbf S$ in the linear operator due to the truncation of the expansion, and $\mbf E(\delta b)$ corresponds to the term that appears when $b$ gets away from the Maxwell point. Now, as usual, we need to solve \eqref{remaindereq} going order by order by expanding $\mbf R_N$ as
    \begin{align*}
        \mbf R_N = \sum_{p = 1}^{\infty} \varepsilon^p \, \mbf R_N^{[p]}.
    \end{align*}
    Fortunately, the corresponding solutions at each order can be written in terms of the solutions we already obtained in Section \ref{sec:regular_asymptotics}, as \evs{it can be seen} below.
    \paragraph{Order $\mathcal O(\varepsilon)$:} At this order, \eqref{remaindereq} becomes
    \begin{align*}
        \left(\jac \mbf f_{0, 0}(\mbf 0) + k^2 \, \frac{\partial^2}{\partial x^2}\right) \mbf R_N^{[1]} = \mbf 0,
    \end{align*}
    which has a solution
    \begin{align*}
        \mbf R_N^{[1]} = B_1 \, e^{ix} \, \bs \phi_1^{[1]} + c.c.,
    \end{align*}
    where $B_1 = B_1(X)$.
    \paragraph{Order $\mathcal O\left(\varepsilon^2\right)$:} At this order, \eqref{remaindereq} becomes
    \begin{align*}
        \mbf 0 = \jac \mbf f_{0, 0}(\mbf 0) \, \mbf R_N^{[2]} + a_1 \, \jac \mbf f_{1, 0}(\mbf 0) \, \mbf R_N^{[1]} + b_1 \, \jac \mbf f_{0, 1}(\mbf 0) \, \mbf R_N^{[1]} + \jac \mbf F_{1, 0, 0}\left(\mbf u^{[1]}\right) \, \mbf R_N^{[1]} + k^2 \, \hat D \, \frac{\partial^2 \mbf R_N^{[2]}}{\partial x^2}.
    \end{align*}
    Therefore,
    \begin{align*}
        \mbf R_N^{[2]} = 2 \, A_1 \, \bar B_1 \, \mbf W_0^{[2]} + B_1 \, e^{ix} \, \mbf W_1^{[2]} + 2 \, A_1 \, B_1 \, e^{2ix} \, \mbf W_2^{[2]} + B_2 \, e^{ix} \, \bs \phi_1^{[1]} + c.c.,
    \end{align*}
    where $B_2 = B_2(X)$.

    \paragraph{Order $\mathcal O\left(\varepsilon^3\right)$:} At this order, \eqref{remaindereq} becomes
    \begin{align*}
        \mbf 0 &= \jac \mbf f_{0, 0}(\mbf 0) \, \mbf R_N^{[3]} + a_1 \, \jac \mbf f_{1, 0}(\mbf 0) \, \mbf R_N^{[2]} + b_1 \, \jac \mbf f_{0, 1}(\mbf 0) \, \mbf R_N^{[2]} + a_2 \, \jac \mbf f_{1, 0}(\mbf 0) \, \mbf R_N^{[1]} + b_2 \, \jac \mbf f_{0, 1}(\mbf 0) \, \mbf R_N^{[1]}
        \\
        & \quad + a_1^2 \, \jac \mbf f_{2, 0}(\mbf 0) \, \mbf R_N^{[1]} + a_1 \, b_1 \, \jac \mbf f_{1, 1}(\mbf 0) \, \mbf R_N^{[1]} + b_1^2 \, \jac \mbf f_{0, 2}(\mbf 0) \, \mbf R_N^{[1]} + \jac \mbf F_{1, 0, 0}\left(\mbf u^{[1]}\right) \, \mbf R_N^{[2]}
        \\
        & \quad + \jac \mbf F_{1, 0, 0}\left(\mbf u^{[2]}\right) \, \mbf R_N^{[1]} + a_1 \, \jac \mbf F_{1, 1, 0}\left(\mbf u^{[1]}\right) \, \mbf R_N^{[1]} + b_1 \, \jac \mbf F_{1, 0, 1} \left(\mbf u^{[1]}\right) \, \mbf R_N^{[1]}
        \\
        & \quad + \jac \mbf F_{2, 0, 0}\left(\mbf u^{[1]}, \mbf u^{[1]}\right) \, \mbf R_N^{[1]} + k^2 \, \hat D \, \left(\frac{\partial^2 \mbf R_N^{[3]}}{\partial x^2} + 2 \, \frac{\partial^2 \mbf R_N^{[1]}}{\partial x \partial X}\right),
    \end{align*}
    which implies that
    \begin{align*}
        \mbf R_N^{[3]} &= 2 \, A_1 \, \bar B_1 \, \mbf W_0^{[3]} + B_1 \, e^{ix} \, \mbf W_1^{[3]} + A_1^2 \, \bar B_1 \, e^{ix} \, \mbf W_{1, 2}^{[3]} + 2 \, \abs{A_1}^2 \, B_1 \, e^{ix} \, \mbf W_{1, 2}^{[3]} + i \, B_1' \, e^{ix} \, \mbf W_{1, 3}^{[3]}
        \\
        & + 2 \, A_1 \, B_1 \, e^{2ix} \, \mbf W_2^{[3]} + 2 \, A_2 \, B_1 \, e^{2ix} \, \mbf W_2^{[2]} + 3 \, A_1^2 \, B_1 \, e^{3ix} \, \mbf W_3^{[3]} + 2 \, A_1 \, \bar B_2 \, \mbf W_0^{[2]} + 2 \, A_2 \, \bar B_1 \, \mbf W_0^{[2]}   
        \\
        & \quad + B_2 \, e^{ix} \, \mbf W_1^{[2]} + B_3 \, e^{ix} \, \bs \phi_1^{[1]} + 2 \, A_1 \, B_2 \, e^{2ix} \, \mbf W_2^{[2]} + c.c.,
    \end{align*}
    where $B_3 = B_3(X)$.
    
    \paragraph{Order $\mathcal O\left(\varepsilon^4\right)$:} At this order, \eqref{remaindereq} becomes
    \begin{align*}
        \mbf 0 &= \jac \mbf f_{0, 0}(\mbf 0) \, \mbf R_N^{[4]} + a_1 \, \jac \mbf f_{1, 0}(\mbf 0) \, \mbf R_N^{[3]} + b_1 \, \jac \mbf f_{0, 1}(\mbf 0) \, \mbf R_N^{[3]} + a_2 \, \jac \mbf f_{1, 0}(\mbf 0) \, \mbf R_N^{[2]}
        \\
        & \quad + b_2 \, \jac \mbf f_{0, 1}(\mbf 0) \, \mbf R_N^{[2]} + a_3 \, \jac \mbf f_{1, 0}(\mbf 0) \, \mbf R_N^{[1]} + b_3 \, \jac \mbf f_{0, 1}(\mbf 0) \, \mbf R_N^{[1]} + a_1^2 \, \jac \mbf f_{2, 0}(\mbf 0) \, \mbf R_N^{[2]}
        \\
        & \quad + a_1 \, b_1 \, \jac \mbf f_{1, 1}(\mbf 0) \, \mbf R_N^{[2]} + b_1^2 \, \jac \mbf f_{0, 2}(\mbf 0) \, \mbf R_N^{[2]} + 2 \, a_1 \, a_2 \, \jac \mbf f_{2, 0}(\mbf 0) \, \mbf R_N^{[1]} + a_1 \, b_2 \, \jac \mbf f_{1, 1}(\mbf 0) \, \mbf R_N^{[1]}
        \\
        & \quad + a_2 \, b_1 \, \jac \mbf f_{1, 1}(\mbf 0) \, \mbf R_N^{[1]} + 2 \, b_1 \, b_2 \, \jac \mbf f_{0, 2}(\mbf 0) \, \mbf R_N^{[1]} + a_1^3 \, \jac \mbf f_{3, 0}(\mbf 0) \, \mbf R_N^{[1]} + a_1^2 \, b_1 \, \jac \mbf f_{2, 1}(\mbf 0) \, \mbf R_N^{[1]}
        \\
        & \quad + a_1 \, b_1^2 \, \jac \mbf f_{1, 2}(\mbf 0) \, \mbf R_N^{[1]} +  b_1^3 \, \jac \mbf f_{0, 3}(\mbf 0) \, \mbf R_N^{[1]} + \jac \mbf F_{1, 0, 0}\left(\mbf u^{[1]}\right) \, \mbf R_N^{[3]} + \jac \mbf F_{1, 0, 0}\left(\mbf u^{[3]}\right) \, \mbf R_N^{[1]}
        \\
        & \quad + \jac \mbf F_{1, 0, 0} \left(\mbf u^{[2]}\right) \, \mbf R_N^{[2]} + a_1 \, \jac \mbf F_{1, 1, 0} \left(\mbf u^{[1]}\right) \, \mbf R_N^{[2]} + b_1 \, \jac \mbf F_{1, 0, 1} \left(\mbf u^{[1]}\right) \, \mbf R_N^{[2]}
        \\
        & \quad + a_1 \, \jac \mbf F_{1, 1, 0}\left(\mbf u^{[2]}\right) \, \mbf R_N^{[1]} + b_1 \, \jac \mbf F_{1, 0, 1}\left(\mbf u^{[2]}\right) \, \mbf R_N^{[1]} + a_2 \, \jac \mbf F_{1, 1, 0}\left(\mbf u^{[1]}\right) \, \mbf R_N^{[1]}
        \\
        & \quad + b_2 \, \jac \mbf F_{1, 0, 1}\left(\mbf u^{[1]}\right) \, \mbf R_N^{[1]} + a_1^2 \, \jac \mbf F_{1, 2, 0}\left(\mbf u^{[1]}\right) \, \mbf R_N^{[1]} + a_1 \, b_1 \, \jac \mbf F_{1, 1, 1}\left(\mbf u^{[1]}\right) \, \mbf R_N^{[1]}
        \\
        & \quad + b_1^2 \, \jac \mbf F_{1, 0, 2}\left(\mbf u^{[1]}\right) \, \mbf R_N^{[1]} + \jac \mbf F_{2, 0, 0}\left(\mbf u^{[1]}, \mbf u^{[1]}\right) \, \mbf R_N^{[2]} + 2 \, \jac \mbf F_{2, 0, 0}\left(\mbf u^{[1]}, \mbf u^{[2]}\right) \, \mbf R_N^{[1]}
        \\
        & \quad + a_1 \, \jac \mbf F_{2, 1, 0}\left(\mbf u^{[1]}, \mbf u^{[1]}\right) \, \mbf R_N^{[1]} + b_1 \, \jac \mbf F_{2, 0, 1}\left(\mbf u^{[1]}, \mbf u^{[1]}\right) \, \mbf R_N^{[1]} + \jac \mbf F_{3, 0, 0}\left(\mbf u^{[1]}, \mbf u^{[1]}, \mbf u^{[1]}\right) \, \mbf R_N^{[1]}
        \\
        & \quad + k^2 \, \hat D \, \left(\frac{\partial^2 \mbf R_N^{[4]}}{\partial x^2} + 2 \, \frac{\partial^2 \mbf R_N^{[2]}}{\partial x \partial X}\right),
    \end{align*}
    which has a solution given by
    \begin{align}
        \mbf R_N^{[4]} &= 2 \, A_1 \, \bar B_1 \, \mbf W_0^{[4]} + 4 \, \abs{A_1}^2 \, A_1 \, \bar B_1 \, \mbf W_{0, 2}^{[4]} + i \, \bar A_1 \, B_1' \, \mbf W_{0, 3}^{[4]} + i \, A_1' \, \bar B_1 \, \mbf W_{0, 3}^{[4]} + B_1 \, e^{ix} \, \mbf W_1^{[4]} \notag
        \\
        & + A_1^2 \, \bar B_1 \, e^{ix} \, \mbf W_{1, 2}^{[4]} + 2 \, \abs{A_1}^2 \, B_1 \, e^{ix} \, \mbf W_{1, 2}^{[4]} + i \, B_1' \, e^{ix} \, \mbf W_{1, 3}^{[4]} + 2 \, A_1 \, B_1 \, e^{2ix} \, \mbf W_2^{[4]} \notag
        \\
        & + A_1^3 \, \bar B_1 \, e^{2ix} \, \mbf W_{2, 2}^{[4]} + 3 \, \abs{A_1}^2 \, A_1 \, B_1 \, e^{2ix} \, \mbf W_{2, 2}^{[4]} + i \, A_1 \, B_1' \, e^{2ix} \, \mbf W_{2, 3}^{[4]} + i \, A_1' \, B_1 \, e^{2ix} \, \mbf W_{2, 3}^{[4]} \notag
        \\
        & + 3 \, A_1^2 \, B_1 \, e^{3ix} \, \mbf W_3^{[4]} + 4 \, A_1^3 \, B_1 \, e^{4ix} \, \mbf W_4^{[4]} + 2 \, A_1 \, \bar B_2 \, \mbf W_0^{[3]} + 2 \, A_2 \, \bar B_1 \, \mbf W_0^{[3]} + 2 \, A_2 \, \bar B_2 \, \mbf W_0^{[2]} \notag
        \\
        & + 2 \, A_1 \, \bar B_3 \, \mbf W_0^{[2]} + 2 \, A_3 \, \bar B_1 \, \mbf W_0^{[2]} + B_2 \, e^{ix} \, \mbf W_1^{[3]} + A_1^2 \, \bar B_2 \, e^{ix} \, \mbf W_{1, 2}^{[3]} + 2 \, \abs{A_1}^2 \, B_2 \, e^{ix} \, \mbf W_{1, 2}^{[3]} \notag
        \\
        & + 2 \, A_1 \, A_2 \, \bar B_1 \, e^{ix} \, \mbf W_{1, 2}^{[3]} + 2 \, A_1 \, \bar A_2 \, B_1 \, e^{ix} \, \mbf W_{1, 2}^{[3]} + 2 \, \bar A_1 \, A_2 \, B_1 \, e^{ix} \, \mbf W_{1, 2}^{[3]} + i \, B_2' \, e^{ix} \, \mbf W_{1, 3}^{[3]} \notag
        \\
        & + B_3 \, e^{ix} \, \mbf W_1^{[2]} + B_4 \, e^{ix} \, \bs \phi_1^{[1]} + 2 \, A_1 \, B_2 \, e^{2ix} \, \mbf W_2^{[3]} + 2 \, A_2 \, B_1 \, e^{2ix} \, \mbf W_2^{[3]} + 2 \, A_2 \, B_2 \, e^{2ix} \, \mbf W_2^{[2]} \notag
        \\
        & + 2 \, A_1 \, B_3 \, e^{2ix} \, \mbf W_2^{[2]} + 2 \, A_3 \, B_1 \, e^{2ix} \, \mbf W_2^{[2]} + 3 \, A_1^2 \, B_2 \, e^{3ix} \, \mbf W_3^{[3]} + 6 \, A_1 \, A_2 \, B_1 \, e^{3ix} \, \mbf W_3^{[3]} + c.c., \label{fourth_remainder}
    \end{align}
    where $B_4 = B_4(X)$.

    \paragraph{Order $\mathcal O\left(\varepsilon^5\right)$:} At this order, we have
    {\allowdisplaybreaks
	\begin{align*}
        \mbf 0 &= \jac \mbf f_{0, 0}(\mbf 0) \, \mbf R_N^{[5]} + a_1 \, \jac \mbf f_{1, 0}(\mbf 0) \, \mbf R_N^{[4]} + b_1 \, \jac \mbf f_{0, 1}(\mbf 0) \, \mbf R_N^{[4]} + a_2 \, \jac \mbf f_{1, 0}(\mbf 0) \, \mbf R_N^{[3]}
        \\
        & \quad + b_2 \, \jac \mbf f_{0, 1}(\mbf 0) \, \mbf R_N^{[3]} + a_3 \, \jac \mbf f_{1, 0}(\mbf 0) \, \mbf R_N^{[2]} + b_3 \, \jac \mbf f_{0, 1}(\mbf 0) \, \mbf R_N^{[2]} + a_4 \, \jac \mbf f_{1, 0}(\mbf 0) \, \mbf R_N^{[1]}
        \\
        & \quad + b_4 \, \jac \mbf f_{0, 1}(\mbf 0) \, \mbf R_N^{[1]} + a_1^2 \, \jac \mbf f_{2, 0}(\mbf 0) \, \mbf R_N^{[3]} + a_1 \, b_1 \, \jac \mbf f_{1, 1}(\mbf 0) \, \mbf R_N^{[3]} + b_1^2 \, \jac \mbf f_{0, 2}(\mbf 0) \, \mbf R_N^{[3]}
        \\
        & \quad + 2 \, a_1 \, a_2 \, \jac \mbf f_{2, 0}(\mbf 0) \, \mbf R_N^{[2]} + a_1 \, b_2 \, \jac \mbf f_{1, 1}(\mbf 0) \, \mbf R_N^{[2]} + a_2 \, b_1 \, \jac \mbf f_{1, 1}(\mbf 0) \, \mbf R_N^{[2]} + 2 \, b_1 \, b_2 \, \jac \mbf f_{0, 2}(\mbf 0) \, \mbf R_N^{[2]}
        \\
        & \quad + 2 \, a_1 \, a_3 \, \jac \mbf f_{2, 0}(\mbf 0) \, \mbf R_N^{[1]} + a_1 \, b_3 \, \jac \mbf f_{1, 1}(\mbf 0) \, \mbf R_N^{[1]} + a_3 \, b_1 \, \jac \mbf f_{1, 1}(\mbf 0) \, \mbf R_N^{[1]} + 2 \, b_1 \, b_3 \, \jac \mbf f_{0, 2}(\mbf 0) \, \mbf R_N^{[1]}
        \\
        & \quad + a_2^2 \, \jac \mbf f_{2, 0}(\mbf 0) \, \mbf R_N^{[1]} + a_2 \, b_2 \, \jac \mbf f_{1, 1}(\mbf 0) \, \mbf R_N^{[1]} + b_2^2 \, \jac \mbf f_{0, 2}(\mbf 0) \, \mbf R_N^{[1]} + a_1^3 \, \jac \mbf f_{3, 0}(\mbf 0) \, \mbf R_N^{[2]}
        \\
        & \quad + a_1^2 \, b_1 \, \jac \mbf f_{2, 1}(\mbf 0) \, \mbf R_N^{[2]} + a_1 \, b_1^2 \, \jac \mbf f_{1, 2}(\mbf 0) \, \mbf R_N^{[2]} + b_1^3 \, \jac \mbf f_{0, 3}(\mbf 0) \, \mbf R_N^{[2]} + 3 \, a_1^2 \, a_2 \, \jac \mbf f_{3, 0}(\mbf 0) \, \mbf R_N^{[1]}
        \\
        & \quad + 2 \, a_1 \, a_2 \, b_1 \, \jac \mbf f_{2, 1}(\mbf 0) \, \mbf R_N^{[1]} + a_2 \, b_1^2 \, \mbf F_{1, 2}(\mbf 0) \, \mbf R_N^{[1]} + a_1^2 \, b_2 \, \jac \mbf f_{2, 1}(\mbf 0) \, \mbf R_N^{[1]}
        \\
        & \quad + 2 \, a_1 \,  b_1 \, b_2 \, \jac \mbf f_{1, 2}(\mbf 0) \, \mbf R_N^{[1]} + 3 \, b_1^2 \, b_2 \, \jac \mbf f_{0, 3}(\mbf 0) \, \mbf R_N^{[1]} + a_1^4 \, \jac \mbf f_{4, 0}(\mbf 0) \, \mbf R_N^{[1]} + a_1^3 \, b_1 \, \jac \mbf f_{3, 1}(\mbf 0) \, \mbf R_N^{[1]}
        \\
        & \quad + a_1^2 \, b_1^2 \, \jac \mbf f_{2, 2}(\mbf 0) \, \mbf R_N^{[1]} + a_1 \, b_1^3 \, \jac \mbf f_{1, 3}(\mbf 0) \, \mbf R_N^{[1]} + b_1^4 \, \jac \mbf f_{0, 4}(\mbf 0) \, \mbf R_N^{[1]} + \jac \mbf F_{1, 0, 0}\left(\mbf u^{[1]}\right) \, \mbf R_N^{[4]}
        \\
        & \quad + \jac \mbf F_{1, 0, 0}\left(\mbf u^{[4]}\right) \, \mbf R_N^{[1]} + \jac \mbf F_{1, 0, 0} \left(\mbf u^{[2]}\right) \, \mbf R_N^{[3]} + \jac \mbf F_{1, 0, 0}\left(\mbf u^{[3]}\right) \, \mbf R_N^{[2]} + a_1 \, \jac \mbf F_{1, 1, 0}\left(\mbf u^{[1]}\right) \, \mbf R_N^{[3]}
        \\
        & \quad + b_1 \, \jac \mbf F_{1, 0, 1}\left(\mbf u^{[1]}\right) \, \mbf R_N^{[3]} + a_1 \, \jac \mbf F_{1, 1, 0}\left(\mbf u^{[3]}\right) \, \mbf R_N^{[1]} + b_1 \, \jac \mbf F_{1, 0, 1}\left(\mbf u^{[3]}\right) \, \mbf R_N^{[1]}
        \\
        & \quad + a_1 \, \jac \mbf F_{1, 1, 0}\left(\mbf u^{[2]}\right) \, \mbf R_N^{[2]} + b_1 \, \jac \mbf F_{1, 0, 1}\left(\mbf u^{[2]}\right) \, \mbf R_N^{[2]} + a_2 \, \jac \mbf F_{1, 1, 0}\left(\mbf u^{[1]}\right) \, \mbf R_N^{[2]}
        \\
        & \quad + b_2 \, \jac \mbf F_{1, 0, 1}\left(\mbf u^{[1]}\right) \, \mbf R_N^{[2]} + a_2 \, \jac \mbf F_{1, 1, 0}\left(\mbf u^{[2]}\right) \, \mbf R_N^{[1]} + b_2 \, \jac \mbf F_{1, 0, 1}\left(\mbf u^{[2]}\right) \, \mbf R_N^{[1]}
        \\
        & \quad + a_3 \, \jac \mbf F_{1, 1, 0}\left(\mbf u^{[1]}\right) \, \mbf R_N^{[1]} + b_3 \, \jac \mbf F_{1, 0, 1}\left(\mbf u^{[1]}\right) \, \mbf R_N^{[1]} + a_1^2 \, \jac \mbf F_{1, 2, 0}\left(\mbf u^{[1]}\right) \, \mbf R_N^{[2]}
        \\
        & \quad + a_1 \, b_1 \, \jac \mbf F_{1, 1, 1}\left(\mbf u^{[1]}\right) \, \mbf R_N^{[2]} + b_1^2 \, \jac \mbf F_{1, 0, 2}\left(\mbf u^{[1]}\right) \, \mbf R_N^{[2]} + a_1^2 \, \jac \mbf F_{1, 2, 0}\left(\mbf u^{[2]}\right) \, \mbf R_N^{[1]}
        \\
        & \quad + a_1 \, b_1 \, \jac \mbf F_{1, 1, 1}\left(\mbf u^{[2]}\right) \, \mbf R_N^{[1]} + b_1^2 \, \jac \mbf F_{1, 0, 2}\left(\mbf u^{[2]}\right) \, \mbf R_N^{[1]} + 2 \, a_1 \, a_2 \, \jac \mbf F_{1, 2, 0}\left(\mbf u^{[1]}\right) \, \mbf R_N^{[1]}
        \\
        & \quad + a_1 \, b_2 \, \jac \mbf F_{1, 1, 1}\left(\mbf u^{[1]}\right) \, \mbf R_N^{[1]} + a_2 \, b_1 \, \jac \mbf F_{1, 1, 1}\left(\mbf u^{[1]}\right) \, \mbf R_N^{[1]} + 2 \, b_1 \, b_2 \, \jac \mbf F_{1, 0, 2}\left(\mbf u^{[1]}\right) \, \mbf R_N^{[1]}
        \\
        & \quad + a_1^3 \, \jac \mbf F_{1, 3, 0}\left(\mbf u^{[1]}\right)  \, \mbf R_N^{[1]} + a_1^2 \, b_1 \, \jac \mbf F_{1, 2, 1}\left(\mbf u^{[1]}\right) \, \mbf R_N^{[1]} + a_1 \, b_1^2 \, \jac \mbf F_{1, 1, 2}\left(\mbf u^{[1]}\right) \, \mbf R_N^{[1]}
        \\
        & \quad + b_1^3 \, \jac \mbf F_{1, 0, 3}\left(\mbf u^{[1]}\right)  \, \mbf R_N^{[1]} + \jac \mbf F_{2, 0, 0}\left(\mbf u^{[1]}, \mbf u^{[1]}\right) \, \mbf R_N^{[3]} + 2 \, \jac \mbf F_{2, 0, 0}\left(\mbf u^{[1]}, \mbf u^{[3]}\right) \, \mbf R_N^{[1]}
        \\
        & \quad + 2 \, \jac \mbf F_{2, 0, 0}\left(\mbf u^{[1]}, \mbf u^{[2]}\right) \, \mbf R_N^{[2]} + \jac \mbf F_{2, 0, 0}\left(\mbf u^{[2]}, \mbf u^{[2]}\right) \, \mbf R_N^{[1]} + a_1 \, \jac \mbf F_{2, 1, 0}\left(\mbf u^{[1]}, \mbf u^{[1]}\right) \, \mbf R_N^{[2]}
        \\
        & \quad + b_1 \, \jac \mbf F_{2, 0, 1}\left(\mbf u^{[1]}, \mbf u^{[1]}\right) \, \mbf R_N^{[2]} + 2 \, a_1 \, \jac \mbf F_{2, 1, 0}\left(\mbf u^{[1]}, \mbf u^{[2]}\right) \, \mbf R_N^{[1]} + 2 \, b_1 \, \jac \mbf F_{2, 0, 1}\left(\mbf u^{[1]}, \mbf u^{[2]}\right) \, \mbf R_N^{[1]}
        \\
        & \quad + a_2 \, \jac \mbf F_{2, 1, 0}\left(\mbf u^{[1]}, \mbf u^{[1]}\right) \, \mbf R_N^{[1]} + b_2 \, \jac \mbf F_{2, 0, 1}\left(\mbf u^{[1]}, \mbf u^{[1]}\right) \, \mbf R_N^{[1]} + a_1^2 \, \jac \mbf F_{2, 2, 0}\left(\mbf u^{[1]}, \mbf u^{[1]}\right) \, \mbf R_N^{[1]}
        \\
        & \quad + a_1 \, b_1 \, \jac \mbf F_{2, 1, 1}\left(\mbf u^{[1]}, \mbf u^{[1]}\right) \, \mbf R_N^{[1]} + b_1^2 \, \jac \mbf F_{2, 0, 2}\left(\mbf u^{[1]}, \mbf u^{[1]}\right) \, \mbf R_N^{[1]}
        \\
        & \quad + \jac \mbf F_{3, 0, 0}\left(\mbf u^{[1]}, \mbf u^{[1]}, \mbf u^{[1]}\right) \, \mbf R_N^{[2]} + 3 \, \jac \mbf F_{3, 0, 0}\left(\mbf u^{[1]}, \mbf u^{[1]}, \mbf u^{[2]}\right) \, \mbf R_N^{[1]}
        \\
        & \quad + a_1 \, \jac \mbf F_{3, 1, 0}\left(\mbf u^{[1]}, \mbf u^{[1]}, \mbf u^{[1]}\right) \, \mbf R_N^{[1]} + b_1 \, \jac \mbf F_{3, 0, 1}\left(\mbf u^{[1]}, \mbf u^{[1]}, \mbf u^{[1]}\right) \, \mbf R_N^{[1]}
        \\
        & \quad + \jac \mbf F_{4, 0, 0}\left(\mbf u^{[1]}, \mbf u^{[1]}, \mbf u^{[1]}, \mbf u^{[1]}\right) \, \mbf R_N^{[1]} + k^2 \, \hat D \, \left(\frac{\partial^2 \mbf R_N^{[5]}}{\partial x^2} + 2 \, \frac{\partial^2 \mbf R_N^{[3]}}{\partial x \partial X} + \frac{\partial^2 \mbf R_N^{[1]}}{\partial X^2}\right),
    \end{align*}
    }
    Again, as in Section \ref{sec:regular_asymptotics}, at this order, we need to determine a solvability condition associated with the resonant terms that are multiples of $e^{ix}$. \evs{That} part of this equation is given by
    \begin{align}
        \left(\jac \mbf f_{0, 0}(\mbf 0) - k^2 \, \hat D\right) \, \mbf R_N^{[5]} &= - \mbf Q_1^{[5]} \, B_1'' - i \, \mbf Q_2^{[5]} \, B_1' - i \, \mbf Q_3^{[5]} \, \abs{A_1}^2 \, B_1' - i \, \mbf Q_4^{[5]} \, A_1^2 \, \bar B_1' - \mbf Q_5^{[5]} \, B_1 \notag
        \\
        & \quad - 2 \, \mbf Q_6^{[5]} \, \abs{A_1}^2 \, B_1 - 3 \, \mbf Q_7^{[5]} \, \abs{A_1}^4 B_1 - i \, \mbf Q_3^{[5]} \, \bar A_1 \, A_1' \, B_1 \label{solv_cond_rem}
        \\
        & \quad - 2 \, i \, \mbf Q_4^{[5]} \, A_1 \, \bar A_1' \, B_1 - \mbf Q_6^{[5]} \, A_1^2 \, \bar B_1 - 2 \, \mbf Q_7^{[5]} \, \abs{A_1}^2 A_1^2 \, \bar B_1 \notag
        \\
        & \quad - i \, \mbf Q_3^{[5]} \, A_1 \, A_1' \, \bar B_1 - \mbox{non-homogeneous part} - \ldots, \notag
    \end{align}
    where `non-homogeneous part' represents the terms associated with $\mbf E(\delta b)$ that appear at this order and `$\ldots$' stands for terms that are multiples of $A_2, A_3, A_4, B_2, B_3. B_4$ and are not required to establish this solvability condition as they are non-resonant.

    With this, when applying the inner product with respect to $\bs \psi$ on both sides of \eqref{solv_cond_rem}, we obtain
    \begin{equation}
        \begin{aligned}
            \alpha_1 \, B_1'' + i \, \alpha_2 \, B_1' + i \, \alpha_3 \, \abs{A_1}^2 \, B_1' + i \, \alpha_3 \, A_1 \, A_1' \, \bar B_1 + i \, \alpha_3 \, \bar A_1 \, A_1' \, B_1 + i \, \alpha_4 \, A_1^2 \, \bar B_1' + 2 \, i \, \alpha_4 \, A_1 \, \bar A_1' \, B_1
            \\
            + \alpha_5 \, B_1 + \alpha_6 \, A_1^2 \, \bar B_1 + 2 \, \alpha_6 \, \abs{A_1}^2 \, B_1 + 2 \, \alpha_7 \, \abs{A_1}^2 \, A_1^2 \, \bar B_1 + 3 \, \alpha_7 \, \abs{A_1}^4 \, B_1 + \frac{\dd \delta}{\dd A_1}\left(A_1\right) = 0,
        \end{aligned} \label{rem_eq}
    \end{equation}
    where $\dfrac{\dd \delta}{\dd A_1}$ is the function one obtains when applying the inner product with the non-homogeneous part of \eqref{solv_cond_rem}, and satisfies $\dfrac{\dd \delta}{\dd A_1}\left(R_1 \, e^{ix}\right) = \dfrac{\dd \delta}{\dd R_1}\left(R_1\right) \, e^{ix}$ \Big(note that the term $\dfrac{\dd \delta}{\dd A_1}\left(A_1\right)$ is an abuse of notation used to simplify \evs{the expression}\Big). Now, if we set $B_1 = \left(R_B + i \, R_1 \, \varphi_B\right) \, e^{i \, \varphi_1}$, where $R_B \evs{= R_B(X)}$ and $\varphi_B \evs{= \varphi_B(X)}$ are real functions, then we can split the resulting equation into
    \begin{equation}
        \begin{aligned}
            8 \, \alpha_1 \left(2 \, \frac{\dd \delta}{\dd R_1} + R_1^2 \, \left(R_1 \, \left(10 \, \alpha_7 \, R_1 \, R_B - \left(\alpha_3 - 3 \, \alpha_4\right) \, \phi_B'\right) + 6 \, \alpha_6 \, R_B\right) + 2 \, \alpha_5 \, R_B\right)
            \\
            + 16 \, \alpha_1^2 \, R_B'' + R_B \left(4 \, \alpha_2^2 + \left(11 \, \alpha_3^2 - 2 \, \alpha_4 \, \alpha_3 - 13 \, \alpha_4^2\right) R_1^4 + 24 \, \alpha_2 \left(\alpha_3 - \alpha_4\right) R_1^2\right) = 0,
        \end{aligned} \label{reequationremainder}
    \end{equation}
    and
    \begin{align}
        4 \, R_1 \left(2 \, \alpha_1 \left(3 \, \alpha_3 \, R_B \, R_1' + 3 \, \alpha_4 \, R_B \, R_1' + 2 \, \alpha_1 \, \varphi_B'' + 2 \, \alpha_5 \, \varphi_B\right) + \alpha_2^2 \, \varphi_B\right) \notag
        \\
        + 16 \, \alpha_1^2 \left(\varphi_B \, R_1'' + 2 \, R_1' \, \varphi_B'\right) + \left(3 \, \alpha_3^2 - 2 \, \alpha_4 \, \alpha_3 - 5 \, \alpha_4^2 + 16 \, \alpha_1 \, \alpha_7\right) R_1^5 \, \varphi_B \label{imageqremainder}
        \\
        + 8 \, \left(\alpha_2 \left(\alpha_3 - \alpha_4\right) + 2 \, \alpha_1 \alpha_6\right) R_1^3 \, \varphi_B + 8 \, \alpha_1 \left(\alpha_3 + \alpha_4\right) R_1^2 \, R_B' = 0. \notag
    \end{align}
    Now, in order to simplify expressions, note that to make all these expansions possible, we need the coefficient $\alpha_1$ to be different from zero, which lets us find the coefficients $\alpha_5, \alpha_6$ and $\alpha_7$ from \eqref{betadef}:
    \begin{align*}
        \alpha_5 &\evs{=} - \alpha_1 \, \beta_1 - \frac{\alpha_2^2}{4 \, \alpha _1},
        \\
        \alpha_6 &= - 2 \, \alpha_1 \, \beta_3 - \frac{\alpha_2 \, \left(\alpha_3 - \alpha_4\right)}{2 \, \alpha_1},
        \\
        \alpha_7 &= - 3 \, \alpha_1 \, \beta_5 - \frac{\left(\alpha_3 + \alpha_4\right) \, \left(3 \, \alpha_3 - 5 \, \alpha_4\right)}{16 \, \alpha_1}
    \end{align*}
    With this, \eqref{imageqremainder} becomes
    \begin{align*}
        R_1^2 \, \varphi_B'' + 2 \, R_1 \, R_1' \, \varphi_B' = - \frac{\left(\alpha_3 + \alpha_4\right)}{2 \, \alpha_1} \left(R_1^3 \, R_B\right)',
    \end{align*}
    which implies
    \begin{align}
        \varphi_B' = - \frac{\alpha_3 + \alpha_4}{2 \, \alpha_1} \, R_1 \, R_B + 4 \, \alpha_1 \, \frac{L_3}{R_1^2}. \label{varphi_B}
    \end{align}
    Therefore, when replacing \eqref{varphi_B} into \eqref{reequationremainder}, we obtain
    \begin{align}
        R_B'' - \left(\beta_1 + 6 \, \beta_3 \, R_1^2 + 15 \, \beta_5 \, R_1^4\right) \, R_B = 2 \, \left(\alpha_3 - 3 \, \alpha_4\right) \, L_3 \, R_1 - \frac{1}{\alpha_1} \, \frac{\dd \delta}{\dd R_1}\left(R_1\right) = 0, \label{eqforRb}
    \end{align}
    which is a non-homogeneous second-order equation just as \eqref{eqforR2}, but with an extra non-homogeneous term. The solutions will then have the same form \evs{in addition to} an extra term. In particular, the homogeneous solution to \eqref{eqforRb} is given by
    \begin{align*}
        R_{B, h}(X) = R_1' \, \left(L_1 + L_2 \, \int_{X_0}^X \frac{\dd s}{\left(R_1'\right)^2}\right),
    \end{align*}
    whilst the particular one is given by
    \begin{align*}
        R_{B, p}(X) = R_1' \, \left(\left(\alpha_3 - 3 \, \alpha_4\right) \, L_3 \, \int_{X_0}^X \frac{R_1^2}{\left(R_1'\right)^2} \, \dd s - \frac{1}{\alpha_1} \, \int_{X_0}^X \frac{\delta\left(R_1\right)}{\left(R_1'\right)^2} \, \dd s\right).
    \end{align*}    
    Therefore, by excluding $L_1$, we have
    \begin{align*}
        R_B(X) = R_1' \, \left(\int_{X_0}^X \frac{L_2 + \left(\alpha_3 - 3 \, \alpha_4\right) \, L_3 \, R_1^2}{\left(R_1'\right)^2} \, \dd s - \frac{1}{\alpha_1} \, \int_{X_0}^X \frac{\delta\left(R_1\right)}{\left(R_1'\right)^2} \, \dd s\right),
    \end{align*}
    which implies
    \begin{align*}
        \varphi_B(X) &= - \frac{\alpha_3 + \alpha_4}{2 \, \alpha_1} \, \int_{X_0}^X R_1 \, R_1' \, \left(\int_{X_0}^\omega \frac{L_2 + \left(\alpha_3 - 3 \, \alpha_4\right) \, L_3 \, R_1^2}{\left(R_1'\right)^2} \, \dd s - \frac{1}{\alpha_1} \, \int_{X_0}^\omega \frac{\delta\left(R_1\right)}{\left(R_1'\right)^2} \, \dd s\right) \, \dd \omega
        \\
        & \quad + 4 \, \alpha_1 \, L_3 \, \int_{X_0}^X \frac{1}{R_1^2} \, \dd s,
    \end{align*}
    and
    \begin{align}
        B_1 &= \left(R_B + i \, R_1 \, \varphi_B\right) \, e^{i \, \varphi_1} \notag
        \\
        &= \begin{multlined}[t]
            L_1 \, A_1' - i \, L_4 \, \bar A_1 + \left(R_1' \, \left(\int_{X_0}^X \frac{L_2 + \left(\alpha_3 - 3 \, \alpha_4\right) \, L_3 \, R_1^2}{\left(R_1'\right)^2} \, \dd s - \frac{1}{\alpha_1} \, \int_{X_0}^X \frac{\delta\left(R_1\right)}{\left(R_1'\right)^2} \, \dd s\right)\right.
            \\
            + i \, R_1 \, \left(4 \, \alpha_1 \, L_3 \, \int_{X_0}^X \frac{1}{R_1^2} \, \dd s \right.
            \\
            - \frac{\alpha_3 + \alpha_4}{2 \, \alpha_1} \, \int_{X_0}^X R_1 \, R_1' \, \left(\int_{X_0}^\omega \frac{L_2 + \left(\alpha_3 - 3 \, \alpha_4\right) \, L_3 \, R_1^2}{\left(R_1'\right)^2} \, \dd s\right.
            \\
            \left.\left.\left. - \frac{1}{\alpha_1} \, \int_{X_0}^\omega \frac{\delta\left(R_1\right)}{\left(R_1'\right)^2} \, \dd s\right) \, \dd \omega\right)\right) \, e^{i \, \varphi_1(X)},
        \end{multlined} \label{B_1exp}
    \end{align}

    \begin{remark}
        Once again, $R_B$ and $\varphi_B$ must, technically, be real functions (see Remark \ref{remark:non-real}). Nevertheless, we make the integration centred at $X_0$ to continue our analysis at singularities.
    \end{remark}

\section{Remainder with forcing due to truncation} \label{sec:second_residual}
    In Section \ref{sec:first_residual}, we analysed the equation for the remainder whilst ignoring the forcing due to truncation. That helped us obtain a general expression for the amplitude of the residual, but the solution we obtained for it misses a term. In particular, we note that, due to the linearity of equation \eqref{remaindereq}, some constants of integration need to be determined in order to know the actual expression of \evs{the} remainder. To find such constants, we focus our attention on a Stokes' line, at which key dominant terms get \evs{triggered}, \evs{letting} us focus on the key part of the remainder to study its boundedness. \evs{Specifically,} with the forcing due to truncation, equation \eqref{remaindereq} becomes:
    \begin{equation}
        \begin{aligned}
            \mbf 0 &= \sum_{p = 0}^{M_a} \sum_{q = 1}^{M_b} a^p \, b^q \, \jac \mbf f_{p, q}(\mbf S) \, \mbf R_N + k^2 \, \hat D \, \left(\frac{\partial^2 \mbf R_N}{\partial x^2} + 2 \, \varepsilon^2 \, \frac{\partial^2 \mbf R_N}{\partial x \partial X} + \varepsilon^4 \, \frac{\partial^2 \mbf R_N}{\partial X^2}\right)
            \\
            & \qquad + \varepsilon^{N + 1} \, k^2 \, \hat D \, \left(2 \, \frac{\partial^2 \mbf u^{[N - 1]}}{\partial x \partial X} + \frac{\partial^2 \mbf u^{[N - 3]}}{\partial X^2}\right) + \varepsilon^{N + 3} \, k^2 \, \hat D \, \frac{\partial^2 \mbf u^{[N - 1]}}{\partial X^2}.
        \end{aligned} \label{truncation_analysis}
    \end{equation}
    \evs{Also, l}et us recall that our leading-order solution provided by $X_0$ is given by \eqref{forrhs+}. Therefore, the dominant terms triggered by $X_0$, present in the forcing of \eqref{truncation_analysis}, are given by
    \begin{align*}
        \varepsilon^{N + 1} \, \kappa^N \, e^{i \hat x} \, \frac{\Gamma\left(\frac{N}{2} + \gamma\right)}{\left(X_0 - X\right)^{\frac{N}{2} + \gamma}} \sum_{r = 0, 2} e^{i (r - 1) \hat x} \, \mbf C_r,
    \end{align*}
    where
    \begin{align*}
        \mbf C_r &= h_r^{[0]} \, \left(\frac{1 + 2 \, i \, \kappa^2 \, r}{\kappa^4} - \varepsilon^2 \, \left(\frac{N}{2} + \gamma\right) \, \frac{1}{\kappa^2 \, \left(X - X_0\right)}\right) \, k^2 \, \hat D \, \bs \phi_1^{[1]}.
    \end{align*}
    Furthermore, two consecutive terms have a ratio given, approximately, by
    \begin{align*}
        \frac{\varepsilon^{N + 3} \, \kappa^{N + 2} \, \dfrac{\Gamma\left(\frac{N + 2}{2} + \gamma\right)}{\left(X_0 - X\right)^{\frac{N + 2}{2} + \gamma}} \, e^{irx}}{\varepsilon^{N + 1} \, \kappa^N \, \dfrac{\Gamma\left(\frac{N}{2} + \gamma\right)}{\left(X_0 - X\right)^{\frac{N}{2} + \gamma}} \, e^{irx}} \sim \varepsilon^2 \, \kappa^2 \, \left(\frac{\frac{N}{2} + \gamma}{X_0 - X}\right).
    \end{align*}
    We want consecutive terms to become the same order as $N \to \infty$, so we need
    \begin{align*}
        \varepsilon^2 \, \abs{\kappa}^2 \, \frac{\abs{\frac{N}{2} + \gamma}}{\abs{X_0 - X}} \sim 1, \qquad \text{\evs{which implies}} \qquad \frac{N}{2} + \gamma \sim \frac{\abs{X_0 - X}}{\varepsilon^2 \, \abs{\kappa}^2} + \nu,
    \end{align*}
    where $\nu$ is a small finite number used to ensure that $N$ is an even natural number.
    
    We seek to find the Stokes' line, which is a singularity that triggers the dominant part of the remainder. With that in mind, we make the change $X - X_0 = \rho \, e^{i \theta}$, where $\rho > 0$ and $\theta \in \mathbb R$, which implies that $\dfrac{N}{2} + \gamma \sim \dfrac{\rho}{\varepsilon^2 \, \abs{\kappa}^2} + \nu$, and $X = X_0 + \rho \, e^{i \theta}$. Therefore, when using the phase shifts introduced in \eqref{first_order_approximation}, we have
    \begin{align*}
        e^{i \hat x} = e^{i \left(X - X_0 + X_0\right)/\varepsilon^2 - i\hat \chi} = e^{i \left(X - X_0\right)/\varepsilon^2} \, e^{iX_0/\varepsilon^2} \, e^{- i\hat \chi}.
    \end{align*}
    
    Furthermore, by using Stirling's approximation, we note that
    {\allowdisplaybreaks
    \begin{align*}
        e^{i\left(X - X_0\right)/\varepsilon^2} \, &\varepsilon^{N + 2 \, \gamma} \, \kappa^{N + 2 \, \gamma} \, \frac{\Gamma\left(\frac{N}{2} + \gamma\right)}{\left(X_0 - X\right)^{\frac{N}{2} + \gamma}}
        \\
        &\sim e^{i\left(X - X_0\right)/\varepsilon^2} \, \varepsilon^{N + 2 \, \gamma} \, \kappa^{N + 2 \, \gamma} \, \frac{1}{\left(X_0 - X\right)^{\frac{N}{2} + \gamma}} \, \sqrt{\frac{2 \, \pi}{\frac{N}{2} + \gamma}} \, \left(\frac{\frac{N}{2} + \gamma}{e}\right)^{\frac{N}{2} + \gamma}
        \\
        &= \sqrt{2 \, \pi} \, e^{i\left(X - X_0\right)/\varepsilon^2} \, \varepsilon^{N + 2 \, \gamma} \, \kappa^{N + 2 \, \gamma} \, \frac{1}{\left(X_0 - X\right)^{\frac{N}{2} + \gamma}} \, \left(\frac{N}{2} + \gamma\right)^{\frac{N - 1}{2} + \gamma} \, \frac{1}{e^{\frac{N}{2} + \gamma}}
        \\
        &\sim \sqrt{2 \, \pi} \, \varepsilon^{N + 2 \, \gamma} \, \kappa^{N + 2 \, \gamma} \, \frac{1}{\left(- \rho \, e^{i \, \theta}\right)^{\frac{N}{2} + \gamma}} \, \left(\frac{N}{2} + \gamma\right)^{\frac{N - 1}{2} + \gamma} \, \frac{e^{i\rho \, e^{i \theta}/\varepsilon^2}}{e^{\frac{N}{2} + \gamma}}
        \\
        &\sim \frac{\sqrt{2 \, \pi} \, \varepsilon^{2 \, \left(\frac{\rho}{\varepsilon^2 \, \abs{\kappa}^2} + \nu\right)} \, \kappa^{2 \, \left(\frac{\rho}{\varepsilon^2 \, \abs{\kappa}^2} + \nu\right)}}{\left(\rho \, e^{i \, (\theta + \pi)}\right)^{\frac{\rho}{\varepsilon^2 \, \abs{\kappa}^2} + \nu}} \, \left(\frac{\rho}{\varepsilon^2 \, \abs{\kappa}^2} + \nu\right)^{\frac{\rho}{\varepsilon^2 \, \abs{\kappa}^2} + \nu - \frac{1}{2}} \, \frac{e^{i \rho \, e^{i \theta}/\varepsilon^2}}{e^{\frac{\rho}{\varepsilon^2 \, \abs{\kappa}^2} + \nu}}
        \\
        &\sim \sqrt{2 \, \pi} \, \frac{\kappa^{2 \, \left(\frac{\rho}{\varepsilon^2 \, \abs{\kappa}^2} + \nu\right)}}{\abs{\kappa}^{2 \, \left(\frac{\rho}{\varepsilon^2 \, \abs{\kappa}^2} + \nu\right)}} \, \frac{1}{\left(e^{i (\theta + \pi)}\right)^{\frac{\rho}{\varepsilon^2 \, \abs{\kappa}^2} + \nu}} \, e^{\nu} \, \frac{\varepsilon \, \abs{\kappa}}{\rho^{\frac{1}{2}}} \, \frac{e^{i \rho \, e^{i \theta}/\varepsilon^2}}{e^{\frac{\rho}{\varepsilon^2 \, \abs{\kappa}^2} + \nu}}
        \\
        &= \sqrt{2 \, \pi} \, \varepsilon \, \abs{\kappa} \, \frac{\kappa^{2 \, \left(\frac{\rho}{\varepsilon^2 \, \abs{\kappa}^2} + \nu\right)}}{\rho^{\frac{1}{2}} \, \abs{\kappa}^{2 \, \left(\frac{\rho}{\varepsilon^2 \, \abs{\kappa}^2} + \nu\right)}} \, \frac{e^{- \frac{\rho}{\varepsilon^2 \, \abs{\kappa}^2}}}{\left(e^{i (\theta + \pi)}\right)^{\frac{\rho}{\varepsilon^2 \, \abs{\kappa}^2} + \nu}} \, e^{i \rho \, e^{i \theta}/\varepsilon^2}.
    \end{align*}
    }
    Here, observe that, \evs{when $\theta = - \dfrac{\pi}{2}$, we have $e^{i \rho \, e^{i \theta}/\varepsilon^2} = e^{\rho/\varepsilon^2}$, which implies that} this quantity is exponentially small except in a neighborhood of the Stokes' line, $\theta = - \dfrac{\pi}{2}$. Therefore, we let $\theta = - \dfrac{\pi}{2} + \varepsilon \, \hat \theta$, which implies
    \begin{align*}
        e^{i \left(X - X_0\right)/\varepsilon^2} \, \varepsilon^{N + 2 \, \gamma} \, \kappa^{N + 2 \, \gamma} \, \frac{\Gamma\left(\frac{N}{2} + \gamma\right)}{\left(X_0 - X\right)^{\frac{N}{2} + \gamma}} \sim \frac{\sqrt{2 \, \pi} \, \varepsilon \, \abs{\kappa} \, \kappa^{2 \, \left(\frac{\rho}{\varepsilon^2 \, \abs{\kappa}^2} + \nu\right)}}{\rho^{\frac{1}{2}} \, \abs{\kappa}^{2 \, \left(\frac{\rho}{\varepsilon^2 \, \abs{\kappa}^2} + \nu\right)}} \, \frac{e^{- \frac{\rho}{\varepsilon^2 \, \abs{\kappa}^2}}}{\left(e^{i \left(\frac{\pi}{2} + \varepsilon \, \hat \theta\right)}\right)^{\frac{\rho}{\varepsilon^2 \, \abs{\kappa}^2} + \nu}} \, e^{\rho \, e^{i \varepsilon \, \hat \theta}/\varepsilon^2}
        \\
        \sim \sqrt{2 \, \pi} \, \varepsilon \, \abs{\kappa} \, \frac{\kappa^{2 \, \left(\frac{\rho}{\varepsilon^2 \, \abs{\kappa}^2} + \nu\right)}}{\rho^{\frac{1}{2}} \, \abs{\kappa}^{2 \, \left(\frac{\rho}{\varepsilon^2 \, \abs{\kappa}^2} + \nu\right)}} \, \frac{e^{- \frac{\rho}{\varepsilon^2 \, \abs{\kappa}^2}}}{\left(i \, e^{i \, \varepsilon \, \hat \theta}\right)^{\frac{\rho}{\varepsilon^2 \, \abs{\kappa}^2} + \nu}} \, e^{\rho \, \left(1 + i \, \varepsilon \, \hat \theta - \frac{\varepsilon^2 \, \hat \theta^2}{2}\right)/\varepsilon^2}.
    \end{align*}
    Now, if we take $\kappa^2 = \kappa_+^2$, then
    \begin{align*}
        e^{i \left(X - X_0\right)/\varepsilon^2} \, \varepsilon^{N + 2 \, \gamma} \, \kappa^{N + 2 \, \gamma} \, \frac{\Gamma\left(\frac{N}{2} + \gamma\right)}{\left(X_0 - X\right)^{\frac{N}{2} + \gamma}} &\sim \sqrt{2 \, \pi} \, \varepsilon \, \abs{\kappa} \, \frac{1}{\rho^{\frac{1}{2}}} \, \frac{e^{- \frac{\rho}{\varepsilon^2 \, \abs{\kappa}^2}}}{\left(e^{i \, \varepsilon \, \hat \theta}\right)^{\frac{\rho}{\varepsilon^2 \, \abs{\kappa}^2} + \nu}} \, e^{\rho \, \left(1 + i \, \varepsilon \, \hat \theta - \frac{\varepsilon^2 \, \hat \theta^2}{2}\right)/\varepsilon^2}
        \\
        &\sim \frac{\sqrt{2 \, \pi} \, \varepsilon \, \abs{\kappa}}{\rho^{\frac{1}{2}}} \, e^{- \rho \, \frac{\hat \theta^2}{2}}
    \end{align*}
    Furthermore, with the scales we have defined, we obtain
    \begin{align*}
        \mbf C_r &= h_r^{[0]} \, \left(\frac{1 + 2 \, i \, \kappa^2 \, r}{\kappa^4} - \varepsilon^2 \, \left(\frac{N}{2} + \gamma\right) \, \frac{1}{\kappa^2 \, \left(X - X_0\right)}\right) \, k^2 \, \hat D \, \bs \phi_1^{[1]}
        \\
        &\sim h_r^{[0]} \, \left(\frac{1 + 2 \, i \, \kappa^2 \, r}{\kappa^4} - \varepsilon^2 \, \left(\frac{\rho}{\varepsilon^2 \, \abs{\kappa}^2} + \nu\right) \, \frac{1}{\kappa^2 \, \rho \, e^{i \, \left(- \frac{\pi}{2} + \varepsilon \, \hat \theta\right)}}\right) \, k^2 \, \hat D \, \bs \phi_1^{[1]}
        \\
        &\sim h_r^{[0]} \, \left(\frac{1 + 2 \, i \, \kappa^2 \, r}{\kappa^4} - \frac{i}{\abs{\kappa}^2 \, \kappa^2 \, e^{i \, \varepsilon \, \hat \theta}}\right) \, k^2 \, \hat D \, \bs \phi_1^{[1]}
        \\
        &\sim \frac{h_r^{[0]}}{\kappa^2} \, \left(\frac{1 + 2 \, i \, \kappa^2 \, r}{\kappa^2} - \frac{i}{\abs{\kappa}^2} \, \left(1 - i \, \varepsilon \, \hat \theta\right)\right) \, k^2 \, \hat D \, \bs \phi_1^{[1]}
        \\
        &\sim h_r^{[0]} \, \left(2 \, r - 2 + i \, \varepsilon \, \hat \theta\right) \, k^2 \, \hat D \, \bs \phi_1^{[1]},
    \end{align*}
    which implies
    \begin{align}
        \mbf C_0 \sim h_0^{[0]} \, \left(- 2 + i \, \varepsilon \, \hat \theta\right) \, k^2 \, \hat D \, \bs \phi_1^{[1]}, \qquad \text{and} \qquad \mbf C_2 \sim h_2^{[0]} \, \left(2 + i \, \varepsilon \, \hat \theta\right) \, k^2 \, \hat D \, \bs \phi_1^{[1]}. \label{change-DM}
    \end{align}
    Therefore, the second line of \eqref{truncation_analysis} is, asymptotically, given by
    \begin{multline*}
        \varepsilon^{N + 1} \, \kappa^N \, e^{ix} \, \frac{\Gamma\left(\frac{N}{2} + \gamma\right)}{\left(X_0 - X\right)^{\frac{N}{2} + \gamma}} \sum_{r = 0, 2} e^{i (r - 1) \hat x} \, \mbf C_r
        \\
        = \varepsilon^{1 - 2 \, \gamma} \, \kappa^{- 2 \, \gamma} \, e^{i X_0/\varepsilon^2} \, e^{- i \hat \chi} \, \left(e^{i \left(X - X_0\right)/\varepsilon^2} \, \varepsilon^{N + 2 \, \gamma} \, \kappa^{N + 2 \, \gamma} \, \frac{\Gamma\left(\frac{N}{2} + \gamma\right)}{\left(X_0 - X\right)^{\frac{N}{2} + \gamma}}\right) \, \sum_{r = 0, 2} e^{i (r - 1) x} \, \mbf C_r
        \\
        \sim \varepsilon^{1 - 2 \, \gamma} \, \kappa^{- 2 \, \gamma} \, e^{i X_0/\varepsilon^2} \, e^{- i \hat \chi} \, \left(\frac{\sqrt{2 \, \pi} \, \varepsilon \, \abs{\kappa}}{\rho^{\frac{1}{2}}} \, e^{- \rho \, \frac{\hat \theta^2}{2}}\right) \, \sum_{r = 0, 2} e^{i (r - 1) \hat x} \, \mbf C_r
        \\
        \sim \varepsilon^{2 - 2 \, \gamma} \, \kappa^{- 2 \, \gamma} \, e^{i X_0/\varepsilon^2} \, e^{- i \hat \chi} \, \frac{\sqrt{2 \, \pi} \, \abs{\kappa} \, k^2}{\rho^{\frac{1}{2}}} \, e^{- \rho \, \frac{\hat \theta^2}{2}}
        \\
        \times \left(\left(- 2 + i \, \hat \theta \, \varepsilon\right) \, h_0^{[0]} \, e^{- i \hat x} + \left(2 + i \, \hat \theta \, \varepsilon\right) \, h_2^{[0]} \, e^{i \hat x}\right) \, \hat D \, \bs \phi_1^{[1]}.
    \end{multline*}
    Moreover,
    \begin{align*}
        \frac{\partial}{\partial \hat \theta} = \frac{\partial \theta}{\partial \hat \theta} \, \frac{\partial X}{\partial \theta} \, \frac{\partial}{\partial X} = i \, \rho \, e^{i \, \theta} \, \varepsilon \, \frac{\partial}{\partial X} = \rho \, e^{i \, \varepsilon \, \hat \theta} \, \varepsilon \, \frac{\partial}{\partial X},
    \end{align*}
    which yields
    \begin{align*}
        X_0 - X &= - \rho \, e^{i\theta} = i \, \rho \, e^{i \, \varepsilon \, \hat \theta},
        \\
        e^{i \hat x} &= e^{\rho \, e^{i \, \varepsilon \, \hat \theta}/\varepsilon^2} \, e^{i X_0/\varepsilon^2} \, e^{- i \hat \chi}
        \\
        \frac{\partial}{\partial X} &= \frac{e^{- i \varepsilon \, \hat \theta}}{\rho \, \varepsilon} \, \frac{\partial}{\partial \hat \theta}.
    \end{align*}
    Thus, \eqref{truncation_analysis} becomes
    \begin{multline*}
        \mbf 0 = \sum_{p = 0}^{M_a} \, \sum_{q = 0}^{M_b} a^p \, b^q \, \jac \mbf f_{p, q}(\mbf S) \, \mbf R_{N, 2} + k^2 \, \hat D \, \left(\frac{\partial^2 \mbf R_{N, 2}}{\partial x^2} + 2 \, \frac{\varepsilon}{\rho} \frac{\partial^2 \mbf R_{N, 2}}{\partial x \partial \hat \theta} e^{- i \, \varepsilon \, \hat \theta} + \frac{\varepsilon^2}{\rho^2} \, \frac{\partial^2 \mbf R_{N, 2}}{\partial \hat \theta^2} e^{- 2 \, i \, \varepsilon \, \hat \theta}\right)
        \\
        + \varepsilon^{2 - 2 \, \gamma} \, \kappa^{- 2 \, \gamma} \, e^{i X_0/\varepsilon^2} \, e^{- i \hat \chi} \, \frac{\sqrt{2 \, \pi} \, \abs{\kappa} \, k^2}{\rho^{\frac{1}{2}}} \, e^{- \rho \, \frac{\hat \theta^2}{2}}
        \\
        \times \left(\left(- 2 + i \, \hat \theta \, \varepsilon\right) \, h_0^{[0]} \, e^{- ix} + \left(2 + i \, \hat \theta \, \varepsilon\right) \, h_2^{[0]} \, e^{ix}\right) \, \hat D \, \bs \phi_1^{[1]},
    \end{multline*}
    which can be expanded as
    \begin{multline}
        \mbf 0 = \sum_{p = 0}^{M_a} \, \sum_{q = 0}^{M_b} a^p \, b^q \, \jac \mbf f_{p, q}(\mbf S) \, \mbf R_{N, 2}
        \\
        + k^2 \, \hat D \, \left(\frac{\partial^2 \mbf R_{N, 2}}{\partial x^2} + 2 \, \frac{\varepsilon}{\rho} \left(1 - i \, \varepsilon \, \hat \theta\right) \frac{\partial^2 \mbf R_{N, 2}}{\partial x \partial \hat \theta} + \frac{\varepsilon^2}{\rho^2} \, \left(1 - 2 \, i \, \varepsilon \, \hat \theta\right) \, \frac{\partial^2 \mbf R_{N,2}}{\partial \hat \theta^2}\right)
        \\
        + \varepsilon^{2 - 2 \, \gamma} \, \kappa^{- 2 \, \gamma} \, e^{i X_0/\varepsilon^2} \, e^{- i \hat \chi} \, \frac{\sqrt{2 \, \pi} \, \abs{\kappa} \, k^2}{\rho^{\frac{1}{2}}} \, e^{- \rho \, \frac{\hat \theta^2}{2}}
        \\
        \times \left(\left(- 2 + i \, \hat \theta \, \varepsilon\right) \, h_0^{[0]} \, e^{- ix} + \left(2 + i \, \hat \theta \, \varepsilon\right) \, h_2^{[0]} \, e^{ix}\right) \, \hat D \, \bs \phi_1^{[1]}, \hspace{0.5cm} \label{exp_rem}
    \end{multline}
    Now, we proceed to solve \eqref{exp_rem}. To do so, we use the following expansion:
    \begin{align*}
        \mbf R_{N, 2} = \varepsilon^{- 2 \, \gamma} \, e^{i X_0/\varepsilon^2} \, \sum_{j \geq 1} \varepsilon^j \, \mbf R_{N, 2}^{[j]}.
    \end{align*}
    Therefore, when solving this equation order by order, we have
    
    \paragraph{Order $\mathcal O(\varepsilon)$:} At this order, \eqref{exp_rem} becomes
    \begin{align*}
        \left(\jac \mbf f_{0, 0}(\mbf 0) + k^2 \, \hat D \, \frac{\partial^2}{\partial x^2}\right) \mbf R_{N, 2}^{[1]} = \mbf 0,
    \end{align*}
    which has the following solution:
    \begin{align*}
        \mbf R_{N, 2}^{[1]} &= \left(C_1\left(\hat \theta\right) \, e^{ix} + C_{- 1}\left(\hat \theta\right) \, e^{- ix}\right) \, \bs \phi_1^{[1]},
    \end{align*}
    and we will omit the dependence on $\hat \theta$ from now on when there is no confusion.
    \paragraph{Order $\mathcal O\left(\varepsilon^2\right)$:} At this order, \eqref{exp_rem} becomes
    \begin{align*}
        \left(\jac \mbf f_{0, 0}(\mbf 0) + k^2 \, \hat D \, \frac{\partial^2}{\partial x^2}\right) \, \mbf R_{N, 2}^{[2]} &= - \jac \mbf F_{1, 0, 0}\left(\mbf u^{[1]}\right) \, \mbf R_{N, 2}^{[1]} - a_1 \, \jac \mbf f_{1, 0}(\mbf 0) \, \mbf R_{N, 2}^{[1]}
        \\
        & \quad - b_1 \, \jac \mbf f_{0, 1}(\mbf 0) \, \mbf R_{N, 2}^{[1]} - \frac{2}{\rho} \, k^2 \, \hat D \, \frac{\partial^2 \mbf R_{N, 2}^{[1]}}{\partial x \partial \hat \theta}
        \\
        & \quad + \kappa^{- 2 \, \gamma} \, e^{- i \hat \chi} \, \frac{\sqrt{2 \, \pi} \, \abs{\kappa}}{\rho^{\frac{1}{2}}} \, e^{- \rho \, \frac{\hat \theta^2}{2}} \, \left(h_0^{[0]} \, e^{- ix} \, \left(2 \, k^2 \, \hat D \, \bs \phi_1^{[1]}\right)\right)
        \\
        & \quad - \kappa^{- 2 \, \gamma} \, e^{- i \hat \chi} \, \frac{\sqrt{2 \, \pi} \, \abs{\kappa}}{\rho^{\frac{1}{2}}} \, e^{- \rho \, \frac{\hat \theta^2}{2}} \, \left(h_2^{[0]} \, e^{ix} \, \left(2 \, k^2 \, \hat D \, \bs \phi_1^{[1]}\right)\right).
    \end{align*}
    Therefore,
    \begin{align*}
        \mbf R_{N, 2}^{[2]} = 2 \, \left(A_1 \, C_{- 1} + \bar A_1 \, C_1\right) \, \mbf W_0^{[2]} + \left(C_1 \, e^{ix} + C_{- 1} \, e^{-ix}\right) \, \mbf W_1^{[2]} - \frac{1}{\rho} \, \left(i \, C_1' \, e^{ix} - i \, C_{- 1}' \, e^{- ix}\right) \mbf W_{1, 3}^{[3]}
        \\
        - \kappa^{- 2 \, \gamma} \, e^{- i \hat \chi} \, \frac{\sqrt{2 \, \pi} \, \abs{\kappa}}{\rho^{\frac{1}{2}}} \, e^{- \rho \, \frac{\hat \theta^2}{2}} \, h_0^{[0]} \, e^{- ix} \, \mbf W_{1, 3}^{[3]} + \kappa^{- 2 \, \gamma} \, e^{- i \hat \chi} \, \frac{\sqrt{2 \, \pi} \, \abs{\kappa}}{\rho^{\frac{1}{2}}} \, e^{- \rho \, \frac{\hat \theta^2}{2}} \, h_2^{[0]} \, e^{ix} \, \mbf W_{1, 3}^{[3]}
        \\
        + 2 \, \left(A_1 \, C_1 \, e^{2ix} + \bar A_1 \, C_{- 1} \, e^{-2ix}\right) \, \mbf W_2^{[2]}.
    \end{align*}
    \paragraph{Order $\mathcal O\left(\varepsilon^3\right)$:} At this order, \eqref{exp_rem} becomes
    {\allowdisplaybreaks
    \begin{align*}
        \left(\jac \mbf f_{0, 0}(\mbf 0) + k^2 \, \hat D \, \frac{\partial^2}{\partial x^2}\right) \, \mbf R_{N, 2}^{[3]} &= - \jac \mbf F_{1, 0, 0}\left(\mbf u^{[2]}\right) \, \mbf R_{N, 2}^{[1]} - a_1 \, \jac \mbf F_{1, 1, 0}\left(\mbf u^{[1]}\right) \, \mbf R_{N, 2}^{[1]}
        \\
        & - b_1 \, \jac \mbf F_{1, 0, 1}\left(\mbf u^{[1]}\right) \, \mbf R_{N, 2}^{[1]} - \jac \mbf F_{2, 0, 0}\left(\mbf u^{[1]}, \mbf u^{[1]}\right) \, \mbf R_{N, 2}^{[1]}
        \\
        & - \jac \mbf F_{1, 0, 0}\left(\mbf u^{[1]}\right) \, \mbf R_{N, 2}^{[2]} - a_1 \, \jac \mbf f_{1, 0}(\mbf 0) \, \mbf R_{N, 2}^{[2]} - b_1 \, \jac \mbf f_{0, 1}(\mbf 0) \, \mbf R_{N, 2}^{[2]}
        \\
        & - a_2 \, \jac \mbf f_{1, 0}(\mbf 0) \, \mbf R_{N, 2}^{[1]} - b_2 \, \jac \mbf f_{0, 1}(\mbf 0) \, \mbf R_{N, 2}^{[1]} - a_1^2 \, \jac \mbf f_{2, 0}(\mbf 0) \, \mbf R_{N, 2}^{[1]}
        \\
        & - a_1 \, b_1 \, \jac \mbf f_{1, 1}(\mbf 0) \, \mbf R_{N, 2}^{[1]} - b_1^2 \, \jac \mbf f_{0, 2}(\mbf 0) \, \mbf R_{N, 2}^{[1]} - \frac{2}{\rho} \, k^2 \, \hat D \, \frac{\partial^2 \mbf R_{N, 2}^{[2]}}{\partial x \partial \hat \theta}
        \\
        & + 2 \, i \, \frac{\hat \theta}{\rho} \, k^2 \, \hat D \, \frac{\partial^2 \mbf R_{N, 2}^{[1]}}{\partial x \partial \hat \theta} - \frac{1}{\rho^2} \, k^2 \, \hat D \, \frac{\partial^2 \mbf R_{N, 2}^{[1]}}{\partial \hat \theta^2}
        \\
        & - i \, \kappa^{- 2 \, \gamma} \, e^{- i \hat \chi} \, \frac{\sqrt{2 \, \pi} \, \abs{\kappa} \, k^2}{\rho^{\frac{1}{2}}} \, e^{- \rho \, \frac{\hat \theta^2}{2}} \, \hat \theta \, \left(h_0^{[0]} \, e^{- i x} + h_2^{[0]} \, \, e^{i x}\right) \, \hat D \, \bs \phi_1^{[1]}.
    \end{align*}
    }
    Now, as usual, we need to determine solvability conditions to ensure that this equation has a solution. In particular, we obtain two solvability conditions given by
    \begin{align}
        \left \langle \hat D \, \mbf W_{1, 3}^{[3]}, \bs \psi\right \rangle \, \frac{2 \, k^2 \, \left(C_{- 1}'' - i \, \kappa^{- 2 \, \gamma} \, e^{- i \hat \chi} \, \sqrt{2 \, \pi} \, \abs{\kappa} \, \rho^{3/2} \, e^{- \rho \, \frac{\hat \theta^2}{2}} \, \hat \theta \, h_0^{[0]}\right)}{\rho^2} &= 0, \label{solvcondC-1}
        \\
        \left \langle \hat D \, \mbf W_{1, 3}^{[3]}, \bs \psi \right \rangle \, \frac{2 \, k^2 \, \left(C_1'' - i \, \kappa^{- 2 \, \gamma} \, e^{- i \hat \chi} \, \sqrt{2 \, \pi} \, \abs{\kappa} \, \rho^{3/2} \, e^{- \rho \, \frac{\hat \theta^2}{2}} \, \hat \theta \, h_2^{[0]}\right)}{\rho^2} &= 0, \label{solvcondC1}
    \end{align}    
    where $'$ denotes the derivative with respect to $\theta$, and we have assumed that $a_1 + b_1 = 0$ to ensure the convergence of the expressions we develop below, noting that, generically, $a_1 = b_1 = 0$ (see \eqref{zero_first_order}). This yields
    \begin{align*}
        C_{- 1}'' &= i \, \kappa^{- 2 \, \gamma} \, e^{- i \hat \chi} \, \sqrt{2 \, \pi} \, \abs{\kappa} \, \rho^{3/2} \, e^{- \rho \, \frac{\hat \theta^2}{2}} \, \hat \theta \, h_0^{[0]},
        \\
        C_1'' &= i \, \kappa^{- 2 \, \gamma} \, e^{- i \hat \chi} \, \sqrt{2 \, \pi} \, \abs{\kappa} \, \rho^{3/2} \, e^{- \rho \, \frac{\hat \theta^2}{2}} \, \hat \theta \, h_2^{[0]}.
    \end{align*}
    Now, when integrating these expressions, we note that
    \begin{align*}
        C_{- 1} &= - i \, \kappa^{- 2 \, \gamma} \, e^{- i \hat \chi} \, \sqrt{2 \, \pi} \, \abs{\kappa} \, \rho^{\frac{1}{2}} \, h_0^{[0]} \, \int_{- \infty}^{\hat \theta} e^{-\frac{\nu^2 \, \rho}{2}} \dd \nu,
        \\
        C_1 &= - i \, \kappa^{- 2 \, \gamma} \, e^{- i \hat \chi} \, \sqrt{2 \, \pi} \, \abs{\kappa} \, \rho^{\frac{1}{2}} \, h_2^{[0]} \, \int_{- \infty}^{\hat \theta} e^{- \frac{\nu^2 \, \rho}{2}} \, \dd \nu,
    \end{align*}
    which implies
    \begin{align*}
        C_{- 1} &= - 2 \, i \, \kappa^{- 2 \, \gamma} \, e^{- i \hat \chi} \, \sqrt{\pi} \, \abs{\kappa} \, h_0^{[0]} \, \int_{- \infty}^{\sqrt{\frac{\rho}{2}} \, \hat \theta}  e^{-\nu ^2} \, \dd \nu,
        \\
        C_1 &= - 2 \, i \, \kappa ^{- 2 \, \gamma} \, e^{- i \hat \chi} \, \sqrt{\pi} \, \abs{\kappa} \, h_2^{[0]} \, \int_{- \infty}^{\sqrt{\frac{\rho}{2}} \, \hat \theta} e^{- \nu ^2} \, \dd \nu.
    \end{align*}
    where we have set the condition $C_{- 1}, C_1 \to 0$ as $\hat \theta \to - \infty$. Thus, \evs{when} we cross the Stokes' line as $\hat \theta$ goes from $- \infty$ to $\infty$,
    \begin{align}
        - 2 \, i \, \pi \, \varepsilon^{- 2 \, \gamma} \, \kappa^{- 2 \, \gamma} \, e^{i X_0/\varepsilon^2} \, e^{- i \hat \chi} \, \abs{\kappa} \, \left(h_2^{[0]} \, e^{ix} + h_0^{[0]} \, e^{- ix}\right) \, \bs \phi_1^{[1]} \label{key}
    \end{align}
    gets triggered by $\kappa_+^2$. Furthermore, by symmetry, we have that
    \begin{align}
        2 \, i \, \pi \, \varepsilon^{- 2 \, \bar \gamma} \, \kappa^{- 2 \, \bar \gamma} \, e^{- i \bar X_0/\varepsilon^2} \, e^{i \hat \chi} \, \abs{\kappa} \, \left(\ol{h_2^{[0]}} \, e^{- ix} + \ol{h_0^{[0]}} \, e^{ix}\right) \, \bs \phi_1^{[1]}
    \end{align}
    gets triggered by $\kappa_-^2$.
    
    Now, let us recall that the outer solution we had obtained for the remainder in Section \ref{sec:first_residual} is given by $\mbf R_N^{[1]} = B_1 \, e^{ix} \, \bs \phi_1^{[1]} + c.c.,$ where $B_1$ is given by \eqref{B_1exp}. Now, as in \cite{Chapman}, we note that $L_1$ (respectively, $L_4$) corresponds to translating the solution $\mbf u^{[1]}$ in $X$ (respectively, $x$). Therefore, these are not relevant to our analysis. The constant we need to focus on is $L_2$, which controls the solution triggered at the Stokes' line. Therefore, for simplicity, we take $L_1 = L_3 = L_4 = 0$, which yields
    \begin{align*}
        B_1 &= \begin{multlined}[t]
            \left(R_1' \, \left(\int_{X_0}^X \frac{L_2}{\left(R_1'\right)^2} \, \dd s - \frac{1}{\alpha_1} \, \int_{X_0}^X \frac{\delta\left(R_1\right)}{\left(R_1'\right)^2} \, \dd s\right)\right.
            \\
            \left. - i \, \frac{\alpha_3 + \alpha_4}{2 \, \alpha_1} \, R_1 \, \int_{X_0}^X R_1 \, R_1' \, \left(\int_{X_0}^\omega \frac{L_2}{\left(R_1'\right)^2} \, \dd s - \frac{1}{\alpha_1} \, \int_{X_0}^\omega \frac{\delta\left(R_1\right)}{\left(R_1'\right)^2} \, \dd s\right) \, \dd \omega\right) \, e^{i \, \varphi_1(X)}.
        \end{multlined}
    \end{align*}
    Furthermore, if we focus on the contribution of the term related to $L_2$ in order to match solutions, we have that
    \begin{equation}
        \begin{aligned}
            B_1 &= L_2 \, \left(R_1' \, \int_{X_0}^X \frac{\dd s}{\left(R_1'\right)^2} - i \, \frac{\alpha_3 + \alpha_4}{2 \, \alpha_1} \, R_1 \, \int_{X_0}^X \left(R_1 \, R_1' \, \int_{X_0}^\omega \frac{\dd s}{\left(R_1'\right)^2}\right) \, \dd \omega\right) \, e^{i \, \varphi_1(X)},
            \\
            &= \frac{L_2}{K_2} \, h_2^{[0]},
        \end{aligned} \label{b1}
    \end{equation}
    which has an asymptotic behavior given by,
    \begin{align}
        B_1 \sim - L_2 \, \frac{\sqrt{2} \, \beta_3}{8 \, \beta_1^2} \, \sqrt{- \frac{\beta_1}{\beta_3}} \, (1 + 2 \, \eta \,  i) \, e^{2 \, \sqrt{\beta_1} \, X} \, e^{i \, \varphi_1(X)}, \label{B_1infinity}
    \end{align}
    as $X \to \infty$.
    
    Moreover, when matching \eqref{b1} and \eqref{key} with $\kappa^2 = \kappa_+^2$, we obtain
    \begin{align}
        L_2 = L_2^+ &= - 2 \, i \, \pi \, \varepsilon^{- 2 \, \gamma} \, \kappa^{- 2 \, \gamma} \, e^{i X_0/\varepsilon^2} \, e^{- i \hat \chi} \, \abs{\kappa} \, K_2. \label{L_2+def}
    \end{align}
    On the other hand, when considering $\kappa^2 = \kappa_-^2$, we have
    \begin{align}
        L_2 = L_2^- &= 2 \, i \, \pi \, \varepsilon^{- 2 \, \bar \gamma} \, \kappa^{- 2 \, \bar \gamma} \, e^{- i \bar X_0/\varepsilon^2} \, e^{i \hat \chi} \, \abs{\kappa} \, \bar K_2. \label{L_2-def}
    \end{align}
    \begin{remark}
        We highlight a key difference between our formulation in equation \eqref{truncation_analysis} and \cite[Equation 120]{Chapman} and \cite[Equation 7.4]{Dean}. Although the approach followed here follows the same ideas as those developed in \cite{Chapman,Dean}, we have to bear in mind that the Swift-Hohenberg equation is a fourth-order partial differential equation, which implies that the definition of $X = \varepsilon^2 \, x$ produces several terms out of $\partial_{xxxx} u$. Therefore, the authors of those papers simplified the equation of the remainder by using the explicit form of their model, whilst making an abuse of notation by using terms that are not part of the expansion to simplify expressions. In our context, as a stationary reaction-diffusion equation is only a second-order partial differential equation, then that is not necessary, as the forcing due to truncation gets reduced to the last three terms in \eqref{truncation_analysis}.
    \end{remark}

    \subsection{Final prediction of the width of homoclinic snaking} \label{sub:summary}
        Merging all the information we have developed so far, we summarize the expressions that need to be used to determine the width of the homoclinic snaking close to codimension-two Turing bifurcation points \evs{to conclude with the bound given by \eqref{bound_statement}}.

        In particular, from \eqref{B_1infinity}, we have that the homogeneous part of the amplitude of the remainder, considering the contribution from $\kappa_+^2$ and $\kappa_-^2$, is given by:
        \begin{align*}
            B_{1, h} \sim \left(- L_2^+ \, \frac{\sqrt{2} \, \beta_3}{8 \, \beta_1^2} - L_2^- \, \frac{\sqrt{2} \, \beta_3}{8 \, \beta_1^2}\right) \, \sqrt{- \frac{\beta_1}{\beta_3}} \, (1 + 2 \, \eta \,  i) \, e^{2 \, \sqrt{\beta_1} \, X} \, e^{i \, \varphi_1(X)},
        \end{align*}
        as $X \to \infty$. On the other hand, the term associated with the separation from the Maxwell point introduced in equation \eqref{rem_eq} will produce a particular solution with the form
        \begin{align*}
            B_{1, p} \sim R(\varepsilon) \, \delta b \, \sqrt{- \frac{\beta_1}{\beta_3}} \, (1 + 2 \, \eta \,  i) \, e^{2 \, \sqrt{\beta_1} \, X} \, e^{i \, \varphi_1(X)},
        \end{align*}
        as $X \to \infty$, where $R(\varepsilon)$ is a rational function in $\varepsilon$ that needs to be determined for each system one wants to study (see examples in Section \ref{sec:examples}). Thus, by considering these two contributions, we have that $B_1$ has an asymptotic behaviour given by
        \begin{equation}
            \begin{aligned}
                B_1 &\sim \left(- \frac{\sqrt{2} \, \beta_3}{4 \, \beta_1^2} \, \Re\left(L_2^+\right) + p(\delta b, \varepsilon)\right) \, \sqrt{- \frac{\beta_1}{\beta_3}} \, (1 + 2 \, \eta \,  i) \, e^{2 \, \sqrt{\beta_1} \, X} \, e^{i \, \varphi_1(X)}
                \\
                &= \begin{multlined}[t]
                    \left(- \frac{\sqrt{2} \, \beta_3}{2 \, \beta_1^2} \, \frac{\pi \, \abs{K_2} \, e^{- \frac{\pi}{2} \, \left(\frac{1}{k \, \sqrt{\beta_1} \, \varepsilon^2} + \eta\right)}}{\varepsilon^6} \, \cos\left(K_2^o - \hat \chi + 2 \, \eta \, \log(\varepsilon)\right) + R(\varepsilon) \, \delta b\right)
                    \\
                    \times \sqrt{- \frac{\beta_1}{\beta_3}} \, (1 + 2 \, \eta \,  i) \, e^{2 \, \sqrt{\beta_1} \, X} \, e^{i \, \varphi_1(X)},
                \end{multlined}
            \end{aligned} \label{B_1full}
        \end{equation}
        as $X \to \infty$, where $K_2 = \abs{K_2} \, e^{i \, K_2^o}$. Therefore, in order to have a valid expansion, we need to ensure that this expression tends to zero as $X \to \infty$, so we need the coefficient of $e^{2 \, \sqrt{\beta_1} \, X}$ to equal 0. Thus, as the cosine function is bounded between $- 1$ and $1$, we have that there exists a value of $\hat \chi$ so that said coefficient equals zero if and only if
        \begin{align}
            \abs{\delta b} \leq \abs{\frac{\sqrt{2} \, \beta_3}{2 \, \beta_1^2 \, R(\varepsilon)}} \, \frac{\pi \, \abs{K_2} \, e^{- \frac{\pi}{2} \, \left(\frac{1}{k \, \sqrt{\beta_1} \, \varepsilon^2} + \eta\right)}}{\varepsilon^6}, \label{final_bound}
        \end{align}
        where the coefficients $\beta_1$ and $\beta_3$ are defined in \eqref{betadef} and $K_2$, defined in \eqref{K_2}, requires the determination of $\lambda_1$, which is a generic expression for the terms $c_r^{[0]}$ defined by the limit \eqref{c0-def}. We highlight that \eqref{final_bound} provides an exponentially decaying bound for the width of the snaking as $\varepsilon \to 0^+$.

        \evs{We highlight that although this completes the main purpose of this article by giving us the width of the homoclinic snaking, an extra step is required to do a complete analysis of the snaking. That step involves joining the two fronts that we have been studying (see sub-Section \ref{sub:solving_amplitude}) to obtain an approximation of the bifurcation curves close to the codimension-two Turing bifurcation point. We make that analysis for completeness, but we relegate it to \ref{ap:joining_fronts}.}

\section{Examples} \label{sec:examples}
    We now proceed to illustrate the theory by computing the necessary components of the beyond-all-orders theory for several common examples. Readers are invited to run the code for each example themselves, which is provided at \cite{beyond-code}, using Python and Mathematica (to understand the logic of the code, the reader is invited to read \ref{app:code}). In each case, we compare our analytical width of the snake given by \eqref{final_bound} with a numerical evaluation of fold points of homoclinic trajectories of the associated spatial-dynamics problem on the real line, using AUTO \cite{auto}.
    \subsection{Swift-Hohenberg 2-3} \label{sub:SH23}
        We start by applying our general theory to the Swift-Hohenberg equation, which was first deduced from the equations for thermal convection. That equation, with quadratic and cubic nonlinearities, was studied in \cite{Chapman}, and served as the main inspiration for this paper. We highlight that said equation has been studied from different points of view and different nonlinearities as it is one of the simplest pattern-forming equations (see, e.g.~\cite{knobloch,snakes-ladders-knobloch,Dean,peletier}, and references therein). The Swift-Hohenberg 2-3 equation can be written as
        \begin{align*}
            \partial_t u = - C \, u - 3 \, E \, u^2 - u^3 - \left(1 + \partial_{xx}\right)^2 u,
        \end{align*}
        where $u = u(x, t)$ is a real variable. Nevertheless, as explained in Section \ref{sec:introduction}, this type of equation can be written as \eqref{geneq} simply by defining an auxiliary variable $v = \partial_{xx} u$ as:        
        \begin{equation}
            \begin{aligned}
                \partial_t u &= - (1 + C) \, u - 3 \, E \, u^2 - u^3 - 2 \, v - \partial_{xx} v,
                \\
                0 &= v - \partial_{xx} u,
            \end{aligned} \label{SH23}
        \end{equation}
        and we use $a = C$, and $b = E$.
        
        To evidence the use of the code shared to study this type of problem, we state the parameter expansions as they were considered in \cite{Chapman}:
        \begin{equation}
            \begin{aligned}
                C &= C_4 \, \varepsilon^4,
                \\
                E &= E_0 + E_2 \, \varepsilon^2 + E_4 \, \varepsilon^4 + \delta E.
            \end{aligned} \label{expansionSH23}
        \end{equation}
        Now, we note that we can obtain the values of the variables shown in Table \ref{tab:valSH23} simply by running the Python code at \cite{beyond-code}. In particular, we highlight that the values obtained are consistent with the results obtained in \cite{Chapman}.
        \begin{table}
            \centering
            {\def\arraystretch{2.3}
            \begin{tabular}{|c|c|c|c|c|c|c|c|c|c|}
                \hline
                \textbf{Variable} & $C_4$ & $E_0$ & $E_2$ & $E_4$ & $\alpha_1$ & $\alpha_3$ & $\alpha_5$ & $\alpha_6$ & $\alpha_7$
                \\
                \hline
                \textbf{Value} & 1 & $\dfrac{\sqrt{114}}{38}$ & $4 \, \dfrac{\sqrt{20919}}{1083}$ & $\dfrac{63711 \, \sqrt{114}}{10069012}$ & 4 & $\dfrac{16}{19}$ & - 1 & $8 \, \dfrac{\sqrt{734}}{19}$ & $- \dfrac{8820}{361}$
                \\
                \hline
            \end{tabular}
            }
            \caption{Values of the key parameters in the asymptotic expansion of system \eqref{SH23}. \evs{The values of the parameters that are not shown, but have the same name format with a different subscript to the ones presented equal zero. These variables were obtained with the Python code in \cite{beyond-code}, with the variable \texttt{modelname=`Swift-Hohenberg'}.}}
            \label{tab:valSH23}
        \end{table}
        
        To find the width of the snaking, we need to focus on the fifth-order equation for the remainder, \eqref{rem_eq}. To study that equation, we note that the only terms that will affect the remainder are given by the linearization of \eqref{SH23} at $(u, v) = (0, 0)$, and we add $- 3 \, \delta E \, u^2$ to the equation for $u$ to take into account the separation of $E$ from the Maxwell point. In particular, we need such a term to appear in the solvability condition at order five. With this in mind, we note that $\mbf F_{2, 0, 1}\left(\mbf u^{[1]}, \mbf u^{[1]}\right)$ will not have resonant terms at that order. The first of those resonant terms will turn up due to $2 \, \mbf F_{2, 0, 1}\left(\mbf u^{[1]}, \mbf u^{[2]}\right)$ at order five, which lets us know that the first term coming from the quadratic expression needs to be part of $\mbf R_N^{[4]}$ for the resonant term to appear at order five. In particular, when using the same notation as before, we have that $\mbf R_N^{[4]}$ is given by \eqref{fourth_remainder} plus
        \begin{align*}
            \frac{\delta E}{\varepsilon^2} \, \abs{A_1}^2 \, \mbf W_*^{[4]} + \frac{\delta E}{\varepsilon^2} \, A_1^2 \, e^{2ix} \, \mbf W_{*, 2}^{[4]} + c.c.,
        \end{align*}
        where
        \begin{align*}
            \mathcal M_0 \, \mbf W_*^{[4]} &= - \mbf F_{2, 0, 1}\left(\bs \phi_1^{[1]}, \bs \phi_1^{[1]}\right),
            \\
            \mathcal M_2 \, \mbf W_{*, 2}^{[4]} &= - \mbf F_{2, 0, 1}\left(\bs \phi_1^{[1]}, \bs \phi_1^{[1]}\right),
        \end{align*}
        and the terms $\varepsilon^2$ in the denominators are put to match the orders of $\mbf R_N^{[4]}$ and $\mbf F_{2, 0, 1}\left(\mbf u^{[1]}, \mbf u^{[1]}\right)$.
        
        With this, we conclude that the equation for the amplitude of the remainder, \eqref{rem_eq}, becomes
        \begin{align*}
            \alpha_1 \, B_1'' + i \, \alpha_3 \, \abs{A_1}^2 \, B_1' + i \, \alpha_3 \, A_1 \, A_1' \, \bar B_1 + i \, \alpha_3 \, \bar A_1 \, A_1' \, B_1 + \alpha_5 \, B_1 + \alpha_6 \, A_1^2 \, \bar B_1 + 2 \, \alpha_6 \, \abs{A_1}^2 \, B_1
            \\
            + 2 \, \alpha_7 \, \abs{A_1}^2 \, A_1^2 \, \bar B_1 + 3 \, \alpha_7 \, \abs{A_1}^4 \, B_1 + 2 \, \sqrt{114} \, \frac{\delta E}{\varepsilon^2} \, \abs{A_1}^2 \, A_1 = 0,
        \end{align*}
        which implies that $\dfrac{\dd \delta}{\dd A_1}\left(A_1\right) = 2 \, \sqrt{114} \, \delta E \, \abs{A_1}^2 \, A_1$. Therefore, (see \eqref{B_1full})
        \begin{align*}
            B_1 &\sim \sqrt{- \frac{\beta_1}{\beta_3}} \, \left(- \frac{\sqrt{2} \, \beta_3 \, L_2^+}{8 \, \beta_1^2} - \frac{\sqrt{2} \, \beta_3 \, L_2^-}{8 \, \beta_1^2} + \frac{\sqrt{57} \, \delta E}{2 \, \varepsilon^2 \, \alpha_1 \, \beta_3}\right) \, (1 + 2 \, \eta \, i) \, e^X
            \\
            &= \frac{1}{2} \, \sqrt{- \frac{\beta_1}{\beta_3}} \, (1 + 2 \, \eta \, i) \, e^X
            \\
            & \qquad \times \left(- \frac{\sqrt{2} \, \beta_3}{\beta_1^2} \, \frac{\pi \, \abs{K_2} \, e^{- \frac{\pi}{2} \, \left(\eta + \frac{1}{\sqrt{\beta_1} \, \varepsilon^2}\right)}}{\varepsilon^6} \, \cos(K_2^o - \hat \chi + 2 \, \eta \, \log(\varepsilon)) + \frac{\sqrt{57} \, \delta E}{\varepsilon^2 \, \alpha_1 \, \beta_3}\right),
        \end{align*}
        as $X \to \infty$, where $\beta_1, \beta_3$ are defined in \eqref{betadef} and $K_2 = \abs{K_2} \, e^{i \, K_2^o}$ was defined in \eqref{K_2}.
        
        Thus, to ensure that our solution remains bounded, we need to force the coefficient of this expression to be zero, which implies that we need
        \begin{align}
            \abs{\delta E} \leq \frac{\sqrt{2} \, \alpha_1 \, \beta_3^2}{\sqrt{57} \, \beta_1^2} \, \frac{\pi \, \abs{K_2} \, e^{- \frac{\pi}{2} \, \left(\eta + \frac{1}{\sqrt{\beta_1} \, \varepsilon^2}\right)}}{\varepsilon^4}, \label{boundSH23}
        \end{align}
        for there to exist a value of $\hat \chi$ so that $B_1$ decays to zero as $X$ tends to infinity. Note that \eqref{boundSH23} is an exponentially small term as $\varepsilon \to 0^+$.
    
        Now, note that all the variables in \eqref{boundSH23} are known except for $K_2$. To determine that variable, we need to run the recurrence given by \eqref{DM} with the initial conditions given by
        \begin{align*}
            A_1 &= \frac{\sqrt{19} \, \sqrt[4]{98861726} \, i}{734},
            \\
            A_3 &= \begin{multlined}[t]
                \frac{\sqrt{19} \, \sqrt[4]{734}}{360647576448} \, \left(\left(- 345333190 + 2602107 \, \pi \, i\right) \, \sqrt{367}\right.
                \\
                \left. + \left(954973269 \, \pi + 236726744 \, i\right) \, \sqrt{2}\right),
            \end{multlined}
        \end{align*}
        which correspond to the leading order coefficients of $A_1$ and $A_3$, respectively, at $X_0$ (see ansatz \eqref{ansatz-inner-sol}). This is done to obtain an approximation of $c_2^{[0]}$, given by \eqref{c0-def}. In particular, when running $n$ up to 181, we obtain the graphs shown in Figure \ref{fig:c20SH23}, where panel (a) (respectively, (b)) shows the convergence of the magnitude (respectively, angle) of $c_2^{[0]}$. In particular, the dotted lines correspond to the approximation of each of these quantities for different values of $n/2$, and the continuous line is a graph of the best curve of the form
        \begin{align}
            f(n) = c_0 + \frac{c_1}{n} + \frac{c_2}{n^2} \label{best-fit}
        \end{align}
        that fits each set of points, ignoring the first few points. These curves were obtained using the procedure outlined in \cite{Dean}, and the lines appear to approximate the curves well. In particular, this lets us know that $c_2^{[0]} \to 0.0303085 \, e^{- 2.57629 \, i}$ as $n\to \infty$, which implies that (see \eqref{K_2}):
        \begin{align*}
            K_2 = - \frac{6 \, i \, \left(- 2 \, \beta_3\right)^{- 1/2}}{(3 - 2 \, \eta \, i)} \, c_0^{[0]} \approx 0.0300695 + 0.0195978 \, i,
        \end{align*}
        and
        \begin{align*}
            \abs{K_2} = 0.0358922, \qquad K_2^o = 0.577605.
        \end{align*}
        With this, we can see that the width of the snaking in terms of $\varepsilon$ is approximately given by
        \begin{align}
            \abs{\delta E} \leq \frac{2.82934 \, e^{- \frac{\pi}{\varepsilon^2}}}{\varepsilon^4}, \label{perfect-SH23}
        \end{align}
        which is close to the function obtained by numerical fitting in \cite[Equation (163)]{Chapman}. A numerical graph showing the fit of this function compared to numerics is shown in Figure \ref{fig:matchingSH23}, which shows that our approximation of the width of the homoclinic snaking is remarkably good.

        \begin{figure}
            \centering
            \begin{subfigure}[b]{0.48\textwidth}
                 \centering
                 \includegraphics[width = \textwidth]{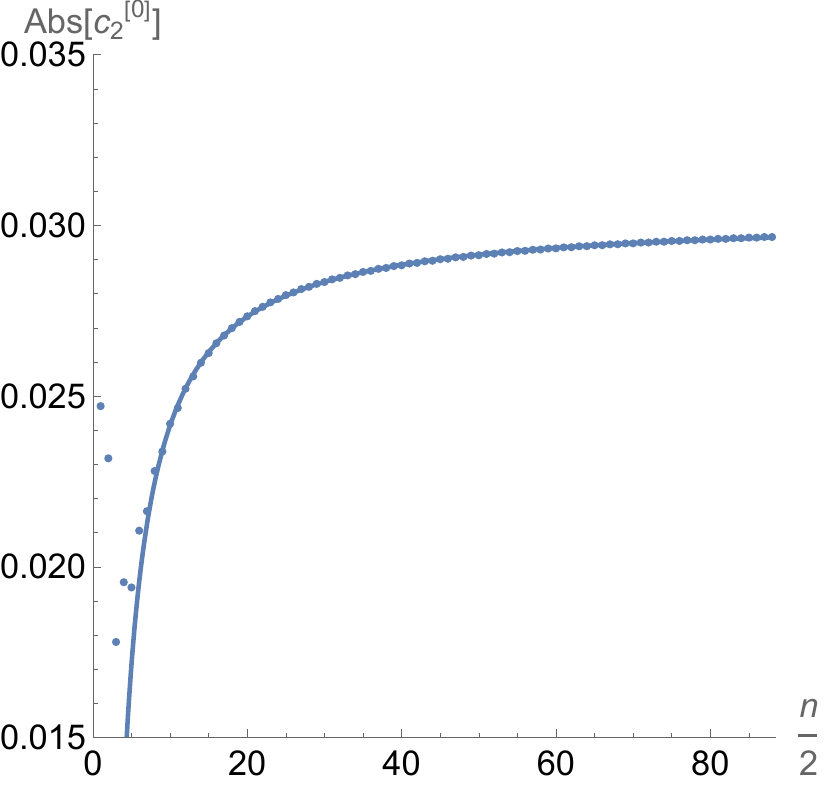}
                 \caption{}
                 \label{fig:c0normSH23}
             \end{subfigure}
             \hfill
             \begin{subfigure}[b]{0.48\textwidth}
                 \centering
                 \includegraphics[width = \textwidth]{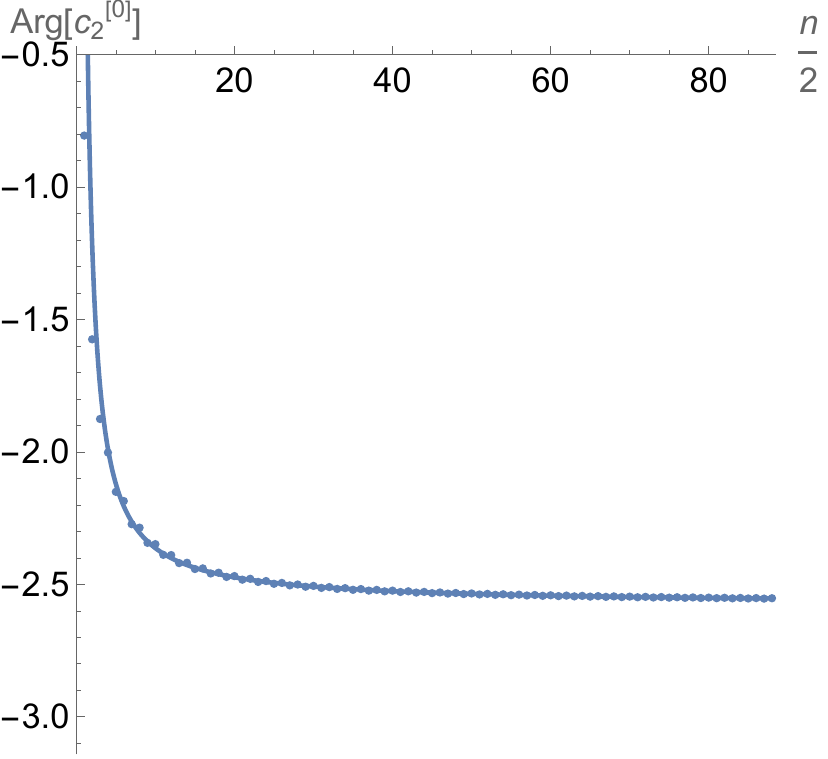}
                 \caption{}
                 \label{fig:c0angleSH23}
             \end{subfigure}
            \caption{Result of the iteration of \eqref{DM} to approximate $c_2^{[0]}$ for system \eqref{SH23}. The points correspond to the iteration of the argument in the limit \eqref{c0-def} for different values of $n$, and continuous lines are the best fit for these points with a function given by \eqref{best-fit}. (a) Behaviour of the absolute value of the argument of the limit in \eqref{c0-def}. (b) Behaviour of the angle of the argument of the limit in \eqref{c0-def} for $r = 2$. \evs{Both of these graphs were obtained with the Mathematica code within the folders \texttt{Iteration/Swift-Hohenberg 2-3} in \cite{beyond-code}.}}
            \label{fig:c20SH23}
        \end{figure}

        \begin{figure}
            \centering
            \begin{subfigure}[b]{0.49\textwidth}
                 \centering
                 \includegraphics[width = \linewidth]{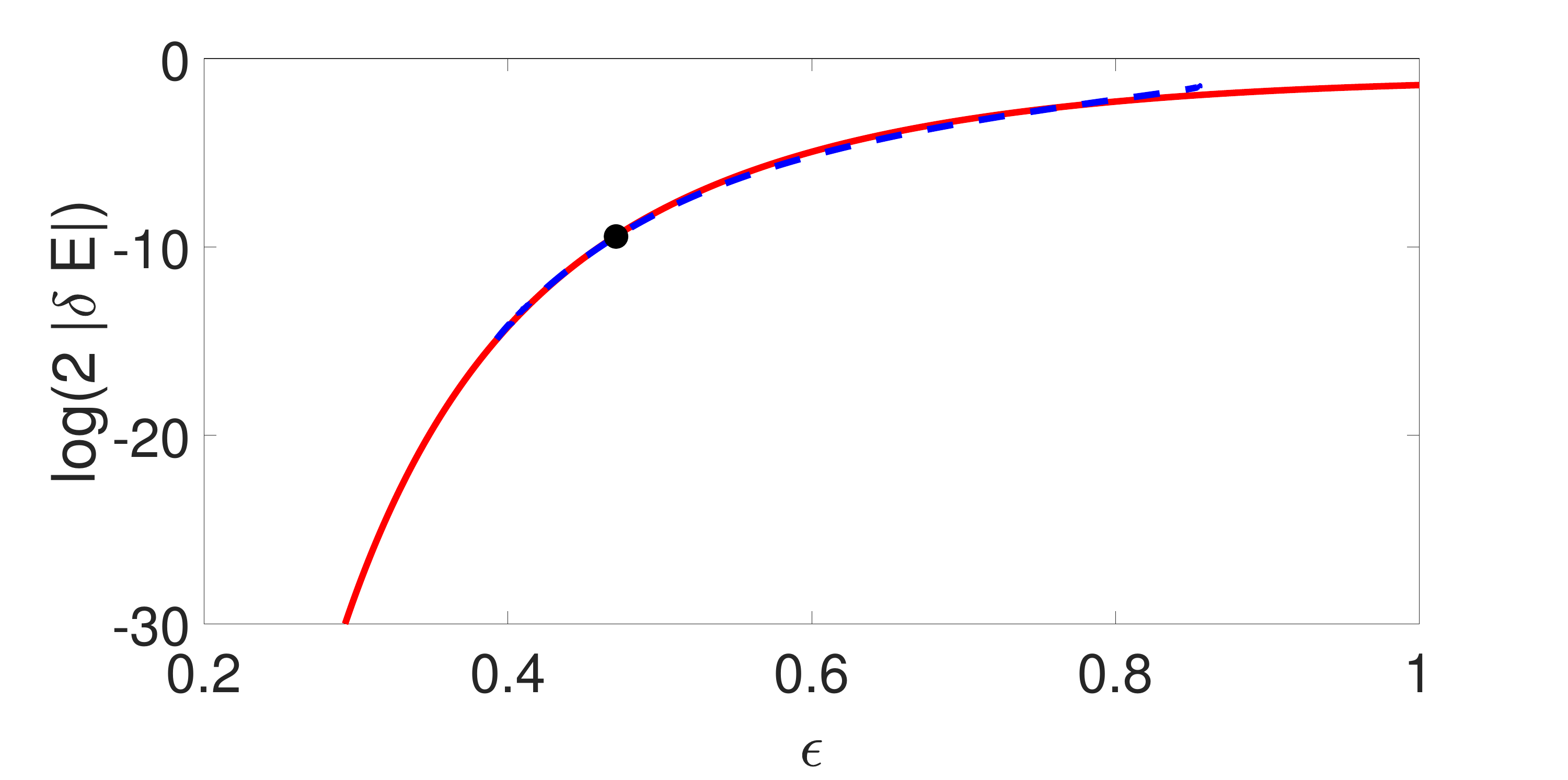}
                 \caption{}
             \end{subfigure}
             \hfill
             \begin{subfigure}[b]{0.49\textwidth}
                 \centering
                 \includegraphics[width = \textwidth]{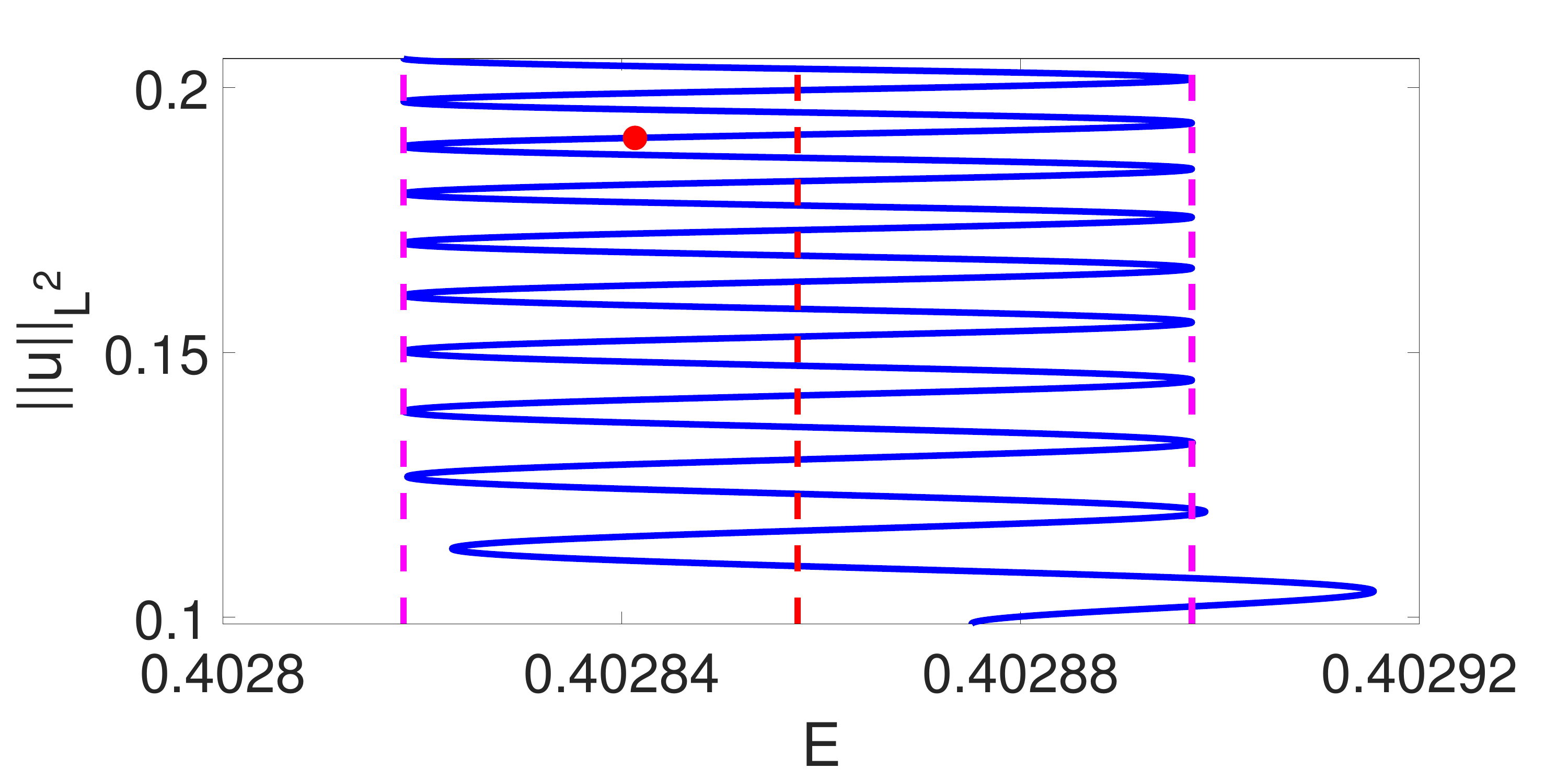}
                 \caption{}
             \end{subfigure}
             \\
             \begin{subfigure}[b]{0.7\textwidth}
                 \centering
                 \includegraphics[width = \textwidth]{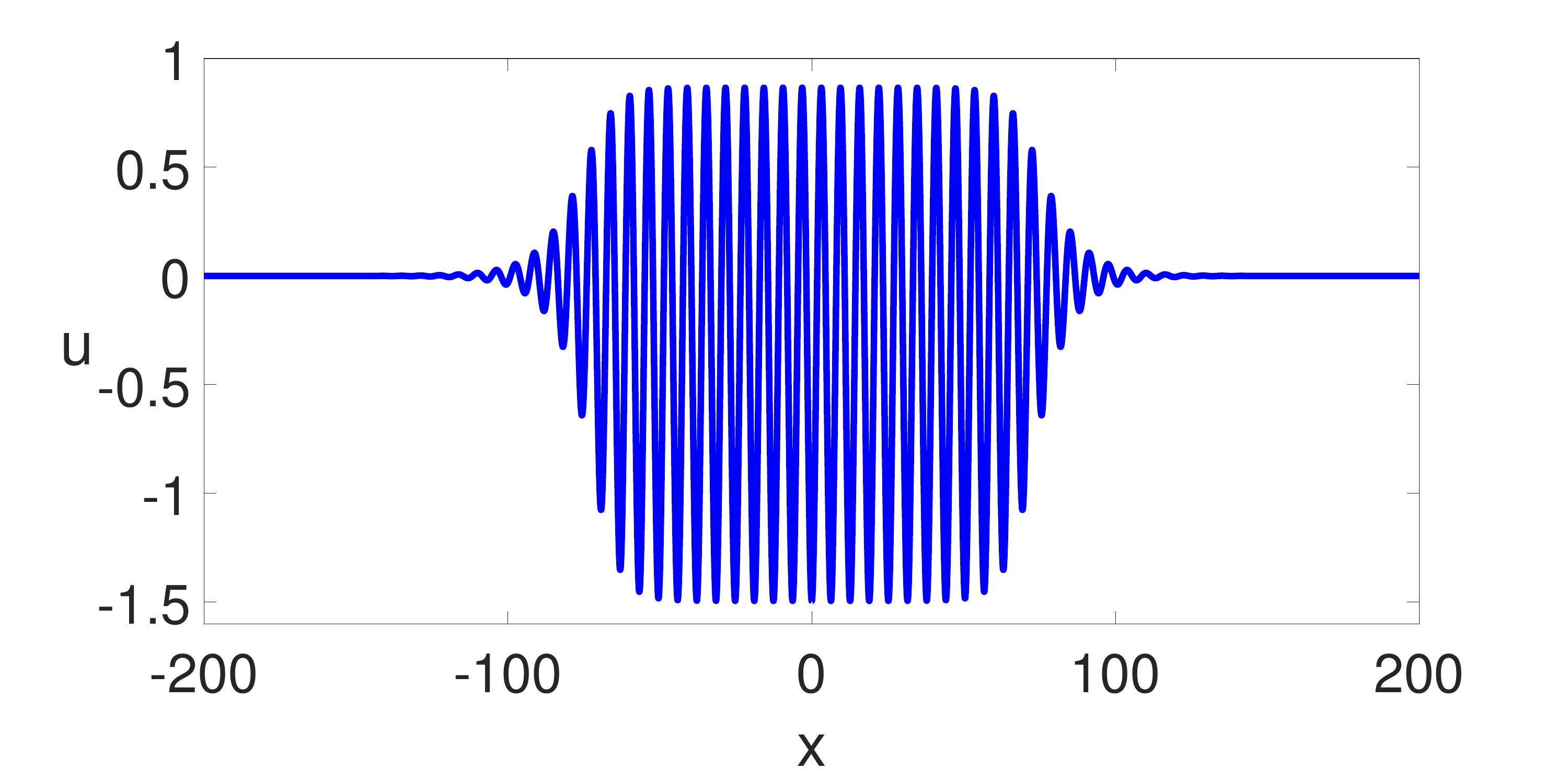}
                 \caption{}
             \end{subfigure}
            \caption{(a) Numerical comparison between the analytical formula for the width of the homoclinic snaking, \eqref{perfect-SH23} (red continuous line), and the numerical results obtained with Auto (blue dashed line) for system \eqref{SH23}. (b) Graph of one branch of the homoclinic snaking curve (blue curve), together with three vertical dashed lines, representing the folds and the Maxwell point at $\varepsilon \approx 0.471138$ (point marked in black in panel (a)). (c) Graph of the homoclinic orbit to the homogeneous steady state $\mbf P = \mbf 0$ at the red point marked in panel (b), corresponding to $E \approx 0.40284133406$.}
            \label{fig:matchingSH23}
        \end{figure}

    \subsection{Swift-Hohenberg 3-5}
        As a second example, we study the same type of equation we studied previously, but with different nonlinearities. In particular, we now apply our theory to the Swift-Hohenberg 3-5 equation, which can be written as
        \begin{equation}
            \begin{aligned}
                \partial_t u &= r \, u - u - 2 \, v + s \, u^3 - u^5 - v_{xx},
                \\
                0 &= v - \partial_{xx} u,
            \end{aligned} \label{SH35}
        \end{equation}
        and we use $a = s$ and $b = r$. This example is motivated by the fact that it has the symmetry $(u, v) \to - (u, v)$, which changes the dominant value of $\kappa$ and the dominant modes to be taken into account to study the width of the snaking \cite{Dean}. In particular, we show that the theory we have developed can also be applied in these cases, as shown in \ref{app:symmetric_case}.
        
        First, we note that a degenerate Turing bifurcation in \eqref{SH35} occurs at $(r, s) = (0, 0)$ with a wavenumber $k = 1$. Furthermore, the parameters obtained when making the expansions as developed throughout this article are shown in Table \ref{tab:valSH35}.
        \begin{table}
            \centering
            {\def\arraystretch{2.3}
            \begin{tabular}{|c|c|c|c|c|c|c|c|c|c|c|c|c|c|}
                \hline
                \textbf{Variable} & $r_4$ & $s_2$ & $\alpha_1$ & $\alpha_5$ & $\alpha_6$ & $\alpha_7$ & $\beta_1$ & $\beta_3$ & $\beta_5$
                \\
                \hline
                \textbf{Value} & $- \dfrac{27}{160}$ & $1$ & 4 & $- \dfrac{27}{160}$ & 3 & $- 10$ & $\dfrac{27}{640}$ & $- \dfrac{3}{8}$ & $\dfrac{5}{6}$
                \\
                \hline
            \end{tabular}
            }
            \caption{Values of some of the key parameters in the asymptotic expansion of system \eqref{SH35}. \evs{The values of the parameters that are not shown, but have the same name format with a different subscript to the ones presented equal zero. These variables were obtained with the Python code in \cite{beyond-code}, with the variable \texttt{modelname=`SHDM'}.}}
            \label{tab:valSH35}
        \end{table}
    
        As in the previous example, to obtain an approximation of the width of the snaking, it is important to know $c_0^{[0]}$ and $c_2^{[0]}$. Nevertheless, due to the special symmetry $(u, v) \to - (u, v)$ \eqref{SH35} has, these terms equal zero. The dominant modes in this case are $r = \pm 1, \pm 3$, so we need to compute $c_1^{[0]}$ and $c_3^{[0]}$ (see sub-Section \ref{sub:kappa_explanation}). The convergence of $c_1^{[0]}$ when using our code with initial conditions
        \begin{align*}
            A_1 = \frac{10^{\frac{3}{4}} \, \sqrt[4]{3} \, i}{10}, \qquad \text{and} \qquad A_3 = - \frac{10^{\frac{3}{4}} \, \sqrt[4]{3}}{16}
        \end{align*}
        is shown in Figure \ref{fig:c10SH35}. The convergence for $c_3^{[0]}$ looks very similar to the one for $c_1^{[0]}$ and is, therefore, omitted (the result can be checked by using the code that is contained in the GitHub repository provided in \cite{beyond-code}). We highlight that the results obtained are consistent with \cite{Dean}. In fact, we also conclude that $c_1^{[0]} \approx 0.101 \, i$, which is the same value as was found in \cite{Dean}, but with a different sign. The change in sign here is consistent with the choice of $\kappa^2 = \kappa_+^2/2$ to make our analysis (see \ref{app:symmetric_case} for an explanation of the division by 2).
        
        Next, to complete this example, we note that the calculation in this case can be carried out using the same ansatzes at each stage, but with a few minor changes (see \ref{app:symmetric_case}). Now, note that the parameter $r$  only influences the linear term in \eqref{SH35}. Therefore, when considering the equation for the remainder, we only need to add a linear term to \eqref{rem_eq}:
        \begin{equation*}
            \begin{aligned}
                \alpha_1 \, B_1'' + i \, \alpha_2 \, B_1' + i \, \alpha_3 \, \abs{A_1}^2 \, B_1' + i \, \alpha_3 \, A_1 \, A_1' \, \bar B_1 + i \, \alpha_3 \, \bar A_1 \, A_1' \, B_1 + i \, \alpha_4 \, A_1^2 \, \bar B_1' + 2 \, i \, \alpha_4 \, A_1 \, \bar A_1' \, B_1
                \\
                + \alpha_5 \, B_1 + \alpha_6 \, A_1^2 \, \bar B_1 + 2 \, \alpha_6 \, \abs{A_1}^2 \, B_1 + 2 \, \alpha_7 \, \abs{A_1}^2 \, A_1^2 \, \bar B_1 + 3 \, \alpha_7 \, \abs{A_1}^4 \, B_1 + \frac{\delta r}{\varepsilon^4} \, A_1 = 0,
            \end{aligned}
        \end{equation*}
        which implies that $\dfrac{\dd \delta}{\dd A_1}\left(A_1\right) = \dfrac{\delta r}{\varepsilon^4} \, A_1$. Therefore, when taking into account the little changes one needs to make in order to study the width of the snaking for \eqref{SH35} (see \ref{app:symmetric_case}), we have
        \begin{align*}
            B_1 &\sim \frac{1}{8 \, \beta_1} \, \sqrt{- \frac{2 \, \beta_1}{\beta_3}} \, \left(- L_2^+ \, \frac{\beta_3}{\beta_1} - L_2^- \, \frac{\beta_3}{\beta_1} - \frac{\delta r}{\alpha_1}\right) \, e^{2 \, \sqrt{\beta_1} \, X} \, e^{i \, \varphi_1}
            \\
            &= \frac{1}{8 \, \beta_1} \, \sqrt{- \frac{2 \, \beta_1}{\beta_3}} \, \left(- \frac{\beta_3}{\beta_1} \, \frac{16 \, \pi \, \abs{K_2} \, e^{- \frac{\pi}{\sqrt{\beta_1} \, \varepsilon^2}}}{\varepsilon^6} \, \cos\left(K_2^o - 2 \, \hat \chi\right) - \frac{\delta r}{\varepsilon^4 \, \alpha_1}\right) \, e^{2 \, \sqrt{\beta_1} \, X} \, e^{i \, \varphi_1},
        \end{align*}
        as $X \to \infty$. With this, similar to the previous example, we need the condition
        \begin{align*}
            \abs{\delta r} \leq - \frac{\alpha_1 \, \beta_3}{\beta_1} \, \frac{16 \, \pi \, \abs{K_2} \, e^{- \frac{\pi}{\sqrt{\beta_1} \, \varepsilon^2}}}{\varepsilon^2},
        \end{align*}
        to ensure that there is a \evs{value of $\hat \chi$} to ensure that the remainder tends to 0 as $X \to \infty$, where \evs{(see \eqref{K_2})}
        \begin{align*}
            \abs{K_2} \approx 0.0202 \, 10^{\frac{3}{4}} \, \sqrt[4]{3}.
        \end{align*}
        \begin{figure}
            \centering
            \begin{subfigure}[b]{0.48\textwidth}
                 \centering
                 \includegraphics[width = \textwidth]{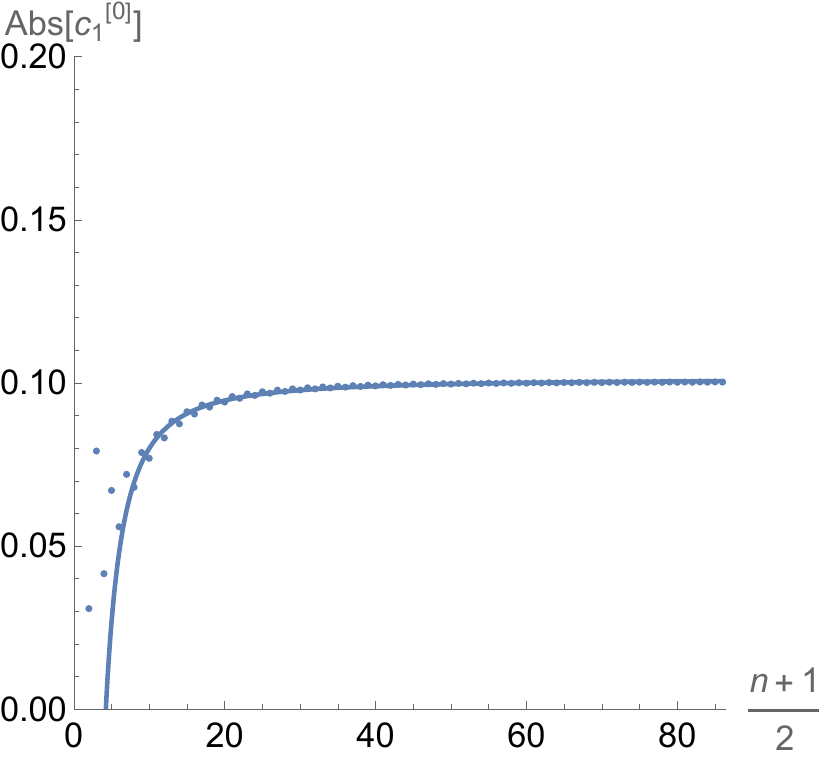}
                 \caption{}
                 \label{fig:c0normSH35}
             \end{subfigure}
             \hfill
             \begin{subfigure}[b]{0.48\textwidth}
                 \centering
                 \includegraphics[width = \textwidth]{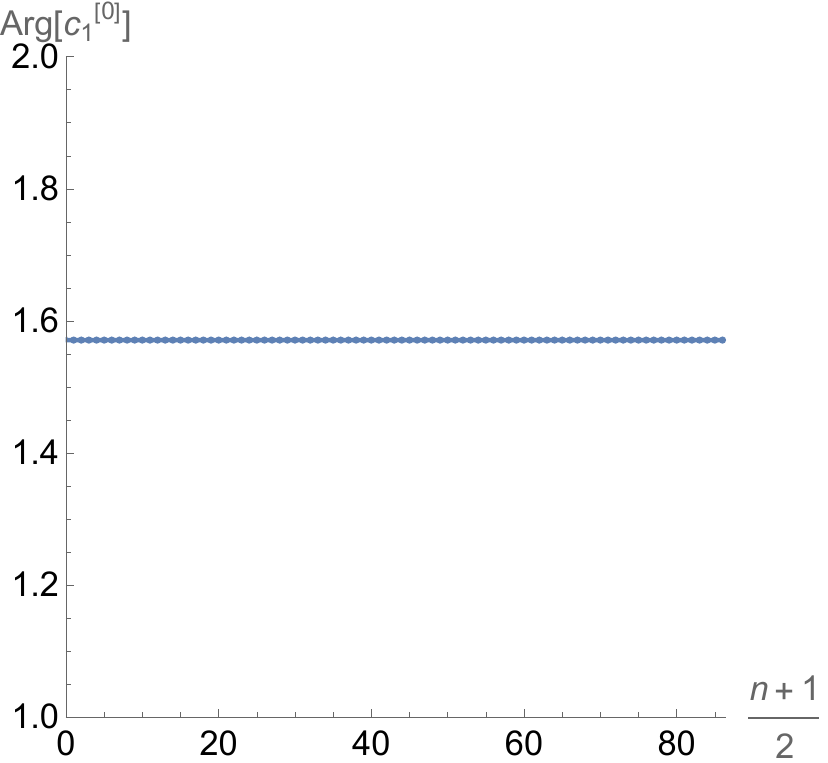}
                 \caption{}
                 \label{fig:c0angleSH35}
             \end{subfigure}
            \caption{Similar to Figure \ref{fig:c20SH23}, but to approximate $c_1^{[0]}$ for model \eqref{SH35}. \evs{Both of these graphs were obtained with the Mathematica code within the folders \texttt{Iteration/Swift-Hohenberg 3-5} in \cite{beyond-code}.}}
            \label{fig:c10SH35}
        \end{figure}

        Therefore, we conclude that
        \begin{align*}
            \abs{\delta r} \leq \frac{267.183114 \, e^{- 4.868645 \, \frac{\pi}{\varepsilon^2}}}{\varepsilon^2}.
        \end{align*}

        Finally, when making the numerical check with actual snaking curves, we obtain the graph shown in Figure \ref{fig:matching-SH35}, where we can see that the matching is, once again, excellent.
        \begin{figure}
            \centering
            \begin{subfigure}[b]{0.49\textwidth}
                 \centering
                 \includegraphics[width = \linewidth]{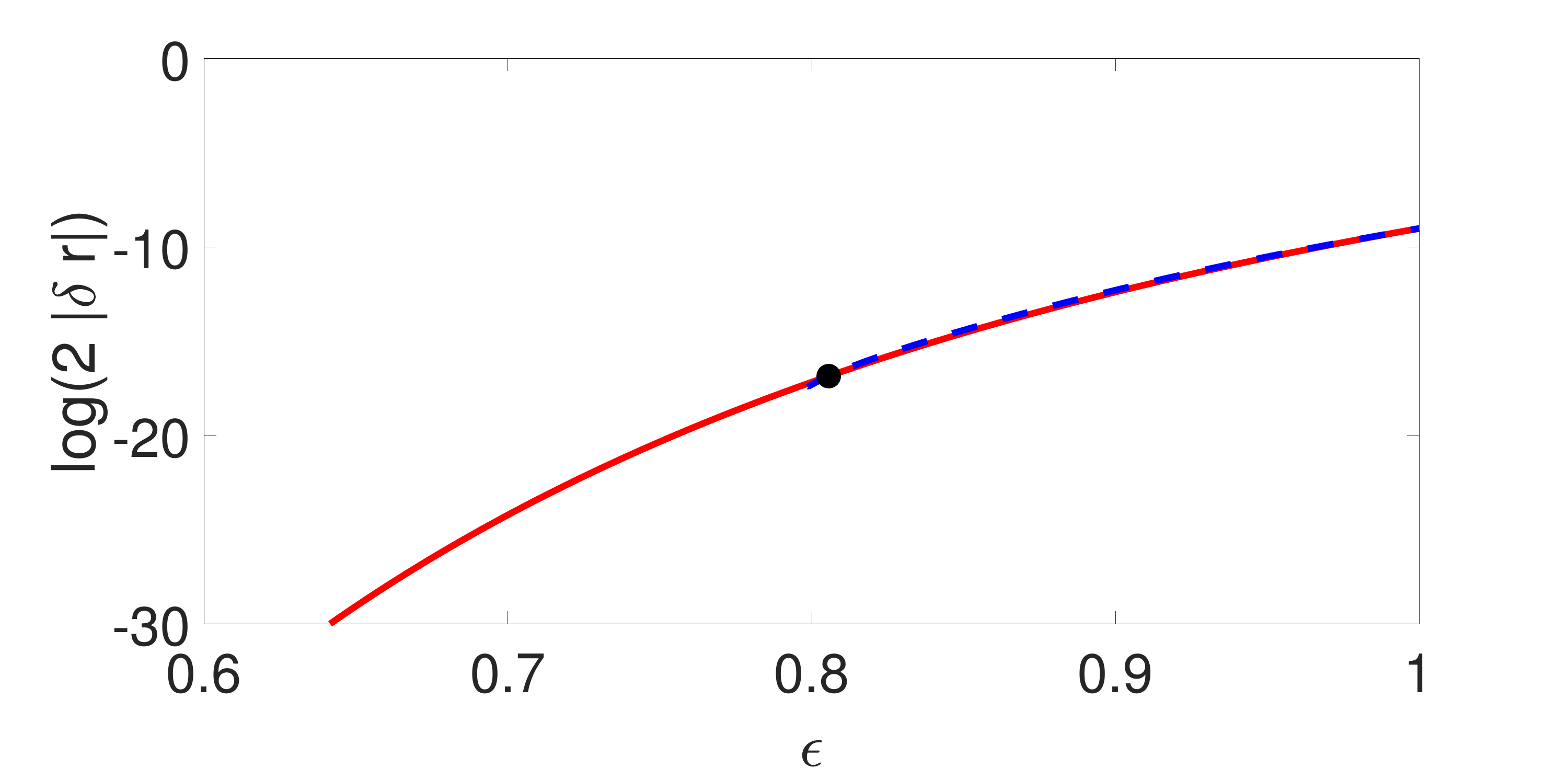}
                 \caption{}
             \end{subfigure}
             \hfill
             \begin{subfigure}[b]{0.49\textwidth}
                 \centering
                 \includegraphics[width = \textwidth]{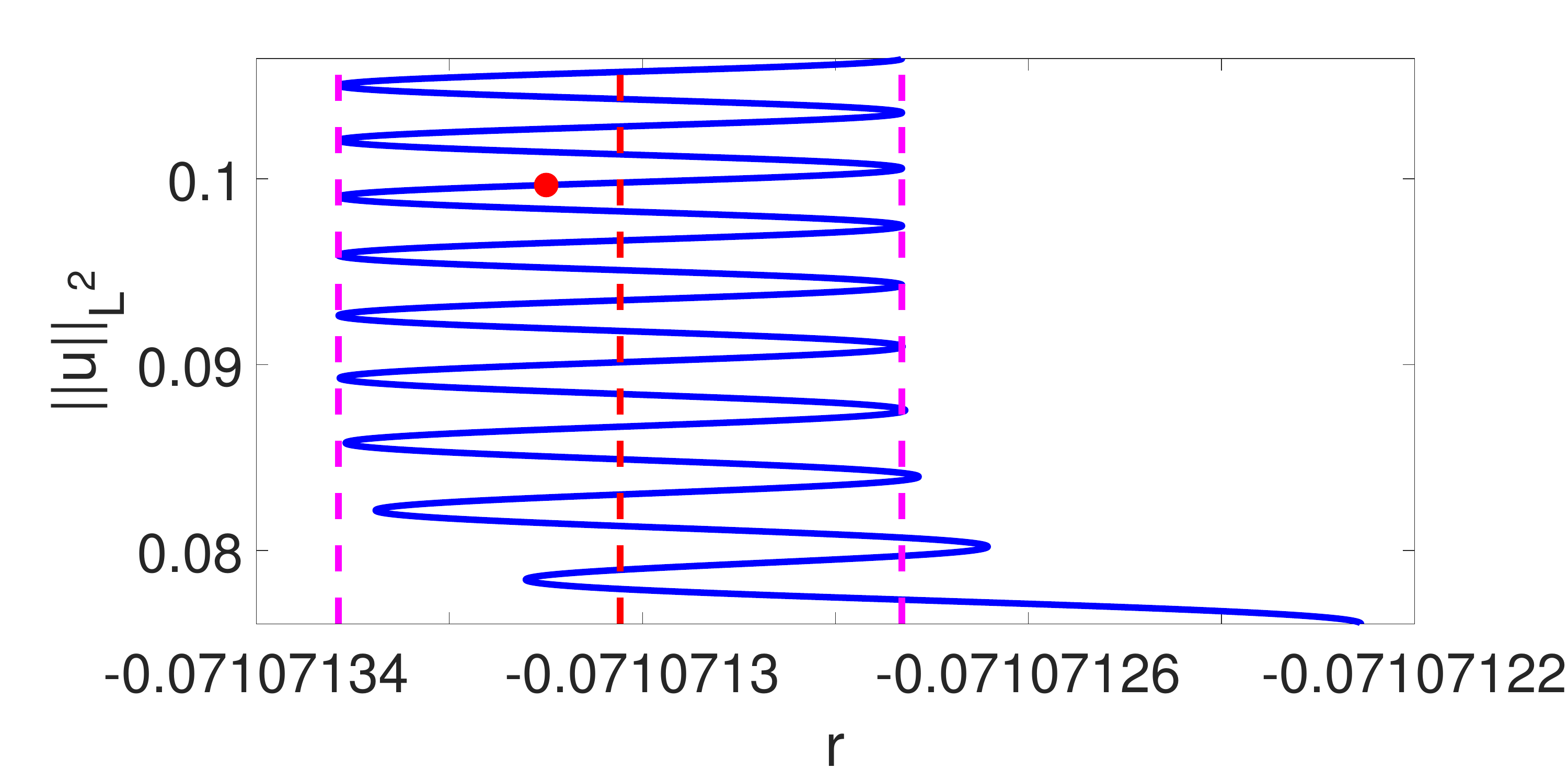}
                 \caption{}
             \end{subfigure}
             \\
             \begin{subfigure}[b]{0.7\textwidth}
                 \centering
                 \includegraphics[width = \textwidth]{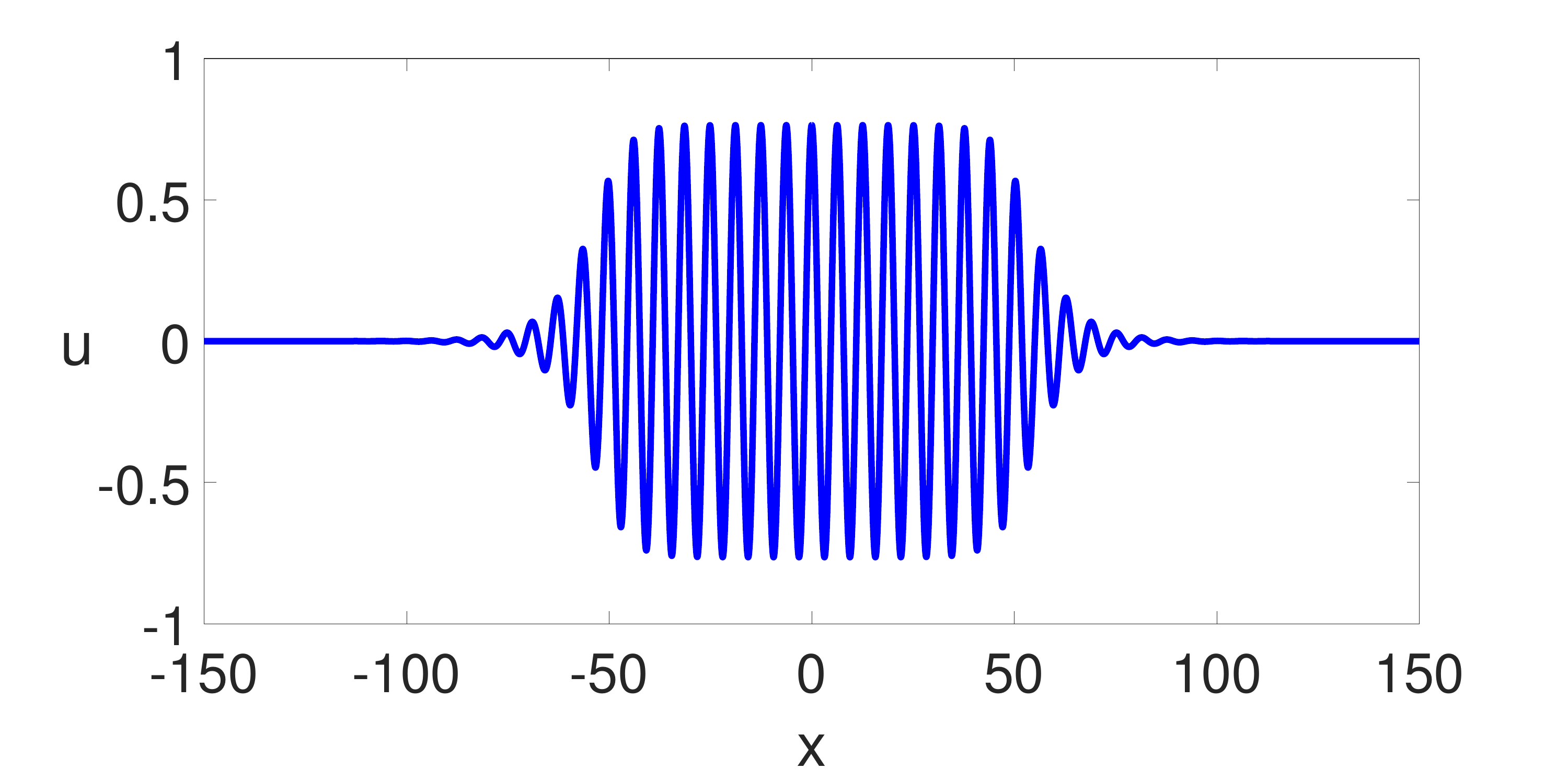}
                 \caption{}
             \end{subfigure}
            \caption{Similar to Figure \ref{fig:matchingSH23}, but for model \eqref{SH35}. The snaking in panel (b) (respectively, the solution in panel (c)) corresponds to $\varepsilon \approx 0.805579$ (respectively, $r \approx - 0.071071309973$).}
            \label{fig:matching-SH35}
        \end{figure}

    \subsection{Modified Schnakenberg system}
        To finish the study of models that have been studied in the past, we now deal with a modified version of the Schnakenberg system, which turns out to be a simple model for glycolysis that has been widely used to study pattern formation (see, e.g.~\cite{FahadWoods,HannesdeWitt,villar2023degenerate}). The modified version of the Schnakenberg system we are to study here is given by:
        \begin{equation}
            \begin{aligned}
                \partial_t u &= - u + u^2 \, v + \sigma \, \left(u - \frac{1}{v}\right)^2 + \partial_{xx} u,
                \\
                \partial_t v &= \lambda - u^2 \, v - \sigma \, \left(u - \frac{1}{v}\right)^2 + d \, \partial_{xx} v,
            \end{aligned} \label{Schnakenberg}
        \end{equation}
        and is motivated by the study carried out in \cite{HannesdeWitt}. \eqref{Schnakenberg} has a homogeneous steady state given by
        \begin{align*}
            \mbf P = (u, v) = \left(\lambda, \frac{1}{\lambda}\right).
        \end{align*}
        Furthermore, if we fix $d = 3 + \sqrt{8}$, then this steady state goes through a codimension-two Turing bifurcation at
        \begin{align*}
            (\sigma, \lambda) = \left(\frac{21}{22} \, \sqrt{2} \, \sqrt{9407 - 6651 \, \sqrt{2}} - \frac{15 \, \sqrt{2}}{22} + \frac{15}{11} \, \sqrt{9407 - 6651 \, \sqrt{2}} - \frac{31}{11}, 1\right).
        \end{align*}
        Now, let us take $a = \lambda$ and $b = \sigma$, and consider the translation $(u, v) \to (u, v) + \mbf P$. With this, system \eqref{Schnakenberg} becomes
        \begin{equation}
            \begin{aligned}
                \partial_t u &= \left(\lambda \, u \, v + u + \lambda^2 \, v\right) \, \left(\frac{u}{\lambda} + \frac{\sigma \, \left(\lambda \, u \, v + u + \lambda^2 \, v\right)}{(\lambda \, v + 1)^2} + 1\right) + \partial_{xx} u,
                \\
                \partial_t v &= - \frac{u^2 \, (\lambda \, v + 1) + 2 \, \lambda \, u \, (\lambda \, v + 1) + \lambda^3 \, v}{\lambda} - \frac{\sigma  \left(\lambda \, u \, v + u + \lambda^2 \, v\right)^2}{(\lambda \, v + 1)^2} + d \, \partial_{xx} v.
            \end{aligned}
        \end{equation}
        Now, we note that, when getting away from the Maxwell point with $\sigma$, that parameter influences infinitely many degrees of nonlinearities. Nevertheless, we only need to consider the dominant ones, which, in this case, are the quadratic and cubic ones. In particular, following the same idea as in Example \ref{sub:SH23}, we have that the equation for the amplitude of the remainder, \eqref{rem_eq}, becomes
        \begin{align*}
            \alpha_1 \, B_1'' + i \, \alpha_3 \, \abs{A_1}^2 \, B_1' + i \, \alpha_3 \, A_1 \, A_1' \, \bar B_1 + i \, \alpha_3 \, \bar A_1 \, A_1' \, B_1 + \alpha_5 \, B_1 + \alpha_6 \, A_1^2 \, \bar B_1 + 2 \, \alpha_6 \, \abs{A_1}^2 \, B_1
            \\
            + 2 \, \alpha_7 \, \abs{A_1}^2 \, A_1^2 \, \bar B_1 + 3 \, \alpha_7 \, \abs{A_1}^4 \, B_1 - \frac{2}{3} \, \sqrt{9407 - 6651 \, \sqrt{2}} \, \frac{\delta \, \sigma}{\varepsilon^2} \, \abs{A_1}^2 \, A_1 = 0.
        \end{align*}        
        Furthermore, in this case, after running the \evs{Mathematica code given in \cite{beyond-code} to carry out the recurrence}, \eqref{DM}, with initial conditions
        \begin{align*}
            A_1 &= 0.833888471879154 \, i,
            \\
            A_3 &= \evs{- 0.544581934387067 - 0.383829250137414 \, i},
        \end{align*}
        we obtain that $c_2^{[0]} \approx \evs{0.126613} \, e^{\evs{- 0.39059} \, i}$, which implies that $\abs{K_2} \approx \evs{0.210505288072226}$ \evs{(see \eqref{K_2})}. The graph showing the convergence of $c_2^{[0]}$ is shown in Figure \ref{fig:c20Schnak}.
        \begin{figure}
            \centering
            \begin{subfigure}[b]{0.48\textwidth}
                 \centering
                 \includegraphics[width = \textwidth]{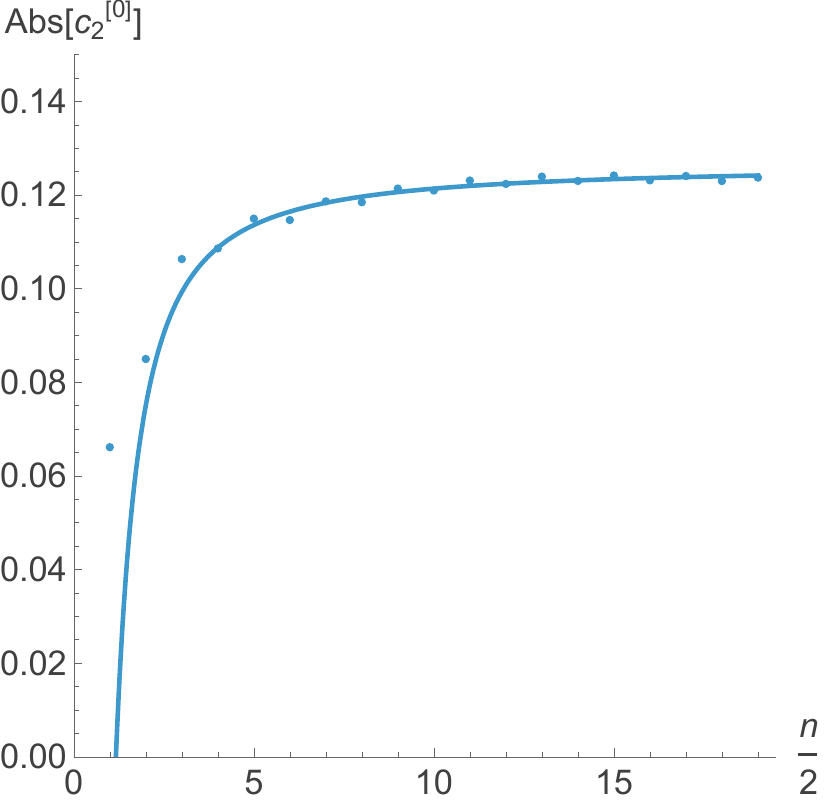}
                 \caption{}
             \end{subfigure}
             \hfill
             \begin{subfigure}[b]{0.48\textwidth}
                 \centering
                 \includegraphics[width = \textwidth]{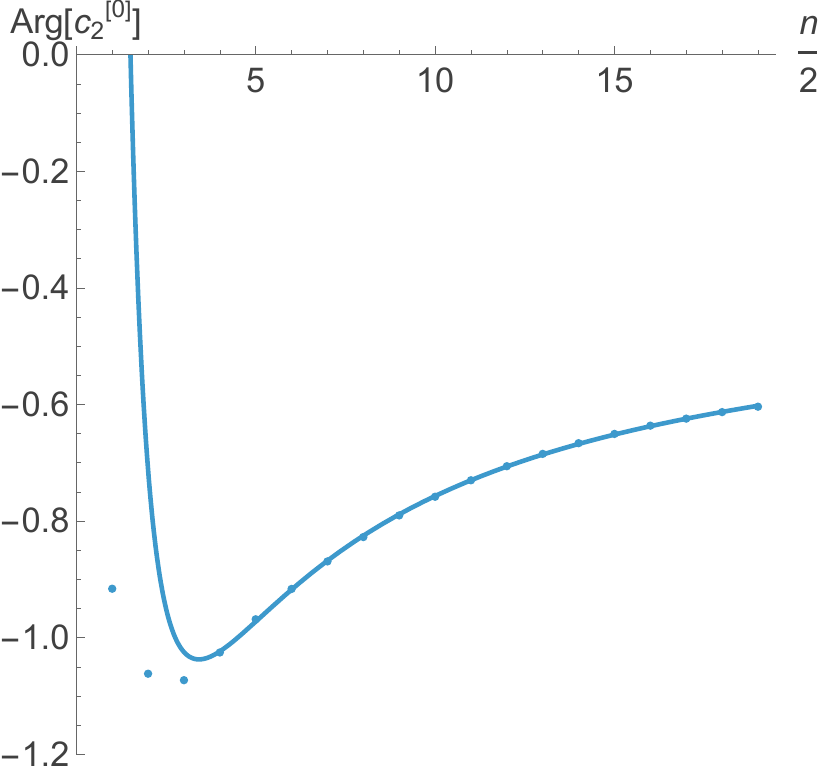}
                 \caption{}
             \end{subfigure}
            \caption{Similar to Figure \ref{fig:c20SH23}, but to approximate $c_2^{[0]}$ for model \eqref{Schnakenberg}. \evs{Both of these graphs were obtained with the Mathematica code within the folders \texttt{Iteration/Schnakenberg} in \cite{beyond-code}.}}
            \label{fig:c20Schnak}
        \end{figure}
        \begin{table}
            \centering
            {\def\arraystretch{2.8}
            \begin{tabular}{|c|c|}
                \hline
                \textbf{Variable} & \textbf{Value}
                \\
                \hline
                $\lambda_4$ & 1
                \\
                \hline
                $\sigma_2$ & \scriptsize{$- \dfrac{\sqrt{- 2595800404 - 27511240 \, \sqrt{2} \, \sqrt{9407 - 6651 \, \sqrt{2}} - 1525844 \, \sqrt{9407 - 6651 \, \sqrt{2}} + 1869381209 \, \sqrt{2}}}{242 \, \sqrt{1617 - 1139 \, \sqrt{2}}}$}
                \\
                \hline
                $\sigma_4$ & $\evs{2.42959304130293}$
                \\
                \hline
                $\alpha_1$ & $4 - 2 \, \sqrt{2}$
                \\
                \hline
                $\alpha_3$ & $- \dfrac{76294 \, \sqrt{2}}{363} - \dfrac{1140 \, \sqrt{9407 - 6651 \, \sqrt{2}}}{121} - \dfrac{444 \, \sqrt{18814 - 13302 \, \sqrt{2}}}{121} + \dfrac{113852}{363}$
                \\
                \hline
                $\alpha_4$ & $- \dfrac{626}{33} + \dfrac{2 \, \sqrt{18814 - 13302 \, \sqrt{2}}}{33} + \dfrac{10 \, \sqrt{9407 - 6651 \, \sqrt{2}}}{33} + \dfrac{421 \, \sqrt{2}}{33}$
                \\
                \hline
                $\alpha_5$ & $2 - 2 \, \sqrt{2}$
                \\
                \hline
                $\alpha_6$ & $2.83351727667069$
                \\
                \hline
                $\alpha_7$ & \footnotesize{$- \dfrac{213917678}{14641} - \dfrac{622015 \, \sqrt{18814 - 13302 \, \sqrt{2}}}{14641} + \dfrac{6617678 \, \sqrt{9407 - 6651 \, \sqrt{2}}}{43923} + \dfrac{450818570 \, \sqrt{2}}{43923}$}
                \\
                \hline
            \end{tabular}
            }
            \caption{Values of some of the key parameters in the asymptotic expansion of system \eqref{Schnakenberg}. \evs{The values of the parameters that are not shown, but have the same name format with a different subscript to the ones presented equal zero. These variables were obtained with the Python code in \cite{beyond-code}, with the variable \texttt{modelname=`Schnakenberg'}.}}
            \label{tab:valSchnak}
        \end{table}

        With this, we have that
        \begin{align*}
            B_1 &\sim \left(- \frac{\sqrt{2} \, \beta_3 \, L_2^+}{4 \, \beta_1^2} - \frac{\sqrt{2} \, \beta_3 \, L_2^-}{4 \, \beta_1^2} - \frac{\sqrt{9407 - 6651 \, \sqrt{2}}}{3 \, \sqrt{2} \, \alpha_1 \, \beta_3 \, \varepsilon ^2} \, \delta \sigma\right) \, \frac{(1 + 2 \, \eta \, i)}{2} \, \sqrt{- \frac{\beta_1}{\beta_3}} \,  e^{2 \, \sqrt{\beta_1} \, X}
            \\
            &\sim \begin{multlined}[t]
                \left(- \frac{\sqrt{2} \, \pi \, \beta_3}{\beta_1^2} \, \frac{\abs{K_2} \, e^{- \frac{\pi}{2} \, \left(\frac{1}{\sqrt{\beta_1} \, \varepsilon^2} + \eta\right)}}{\varepsilon^6} \, \cos\left(K_2^o - \hat \chi + 2 \, \eta \, \log(\varepsilon)\right)\right.
                \\
                \left. - \frac{\sqrt{9407 - 6651 \, \sqrt{2}}}{3 \, \sqrt{2} \, \alpha_1 \, \beta_3 \, \varepsilon^2} \, \delta \sigma\right) \, \frac{(1 + 2 \, \eta \, i)}{2} \, \sqrt{- \frac{\beta_1}{\beta_3}} \,  e^{2 \, \sqrt{\beta_1} \, X},
            \end{multlined}
        \end{align*}
        as $X \to \infty$. Therefore, as in the previous examples, the condition we need to ensure the remainder converges to 0 as $X \to \infty$ is given by
        \begin{align*}
            \abs{\delta \sigma} \leq \frac{6 \, \pi \, \alpha_1 \, \beta_3^2}{\beta_1^2 \, \sqrt{9407 - 6651 \, \sqrt{2}}} \, \frac{\abs{K_2} \, e^{- \frac{\pi}{2} \, \left(\frac{1}{\sqrt{\beta_1} \, \varepsilon^2} + \eta\right)}}{\varepsilon^4},
        \end{align*}
        which is, approximately, given by
        \begin{align*}
            \abs{\Delta \sigma} \leq \frac{\evs{15.866362} \, e^{- 0.594604 \, \frac{\pi}{\varepsilon^2}}}{\varepsilon^4}.
        \end{align*}
        The graph with the matching is shown in Figure \ref{fig:matching-Schnakenberg}, which depicts, once again, that the theory we have developed matches numerics superbly. We highlight that the result shown in \cite{HannesdeWitt} does not appear to be successful because of the choice of the parameters to do the matching. In fact, if one takes $b = \lambda$, then all the powers of $(\delta \lambda)^p$ will appear in the final expression for $B_1$, for every integer $p \geq 1$, and finding a bound in that case is much more intricate. In that case, it would have been better to scale $t \to \lambda \, t$ and $x \to \sqrt{\lambda} \, x$ in order to carry out this analysis. We also remark that the recurrence \eqref{DM} takes longer to run for this model in comparison to the previous two examples due to the presence of the variable $v$ in the denominator of \eqref{Schnakenberg}. Fortunately, $c_2^{[0]}$ has a tidy behaviour from the beginning (see Figure \ref{fig:c20Schnak}), so not many terms were necessary to \evs{obtain a good approximation of it}.
        \begin{figure}
            \centering
            \begin{subfigure}[b]{0.49\textwidth}
                 \centering
                 \includegraphics[width = \linewidth]{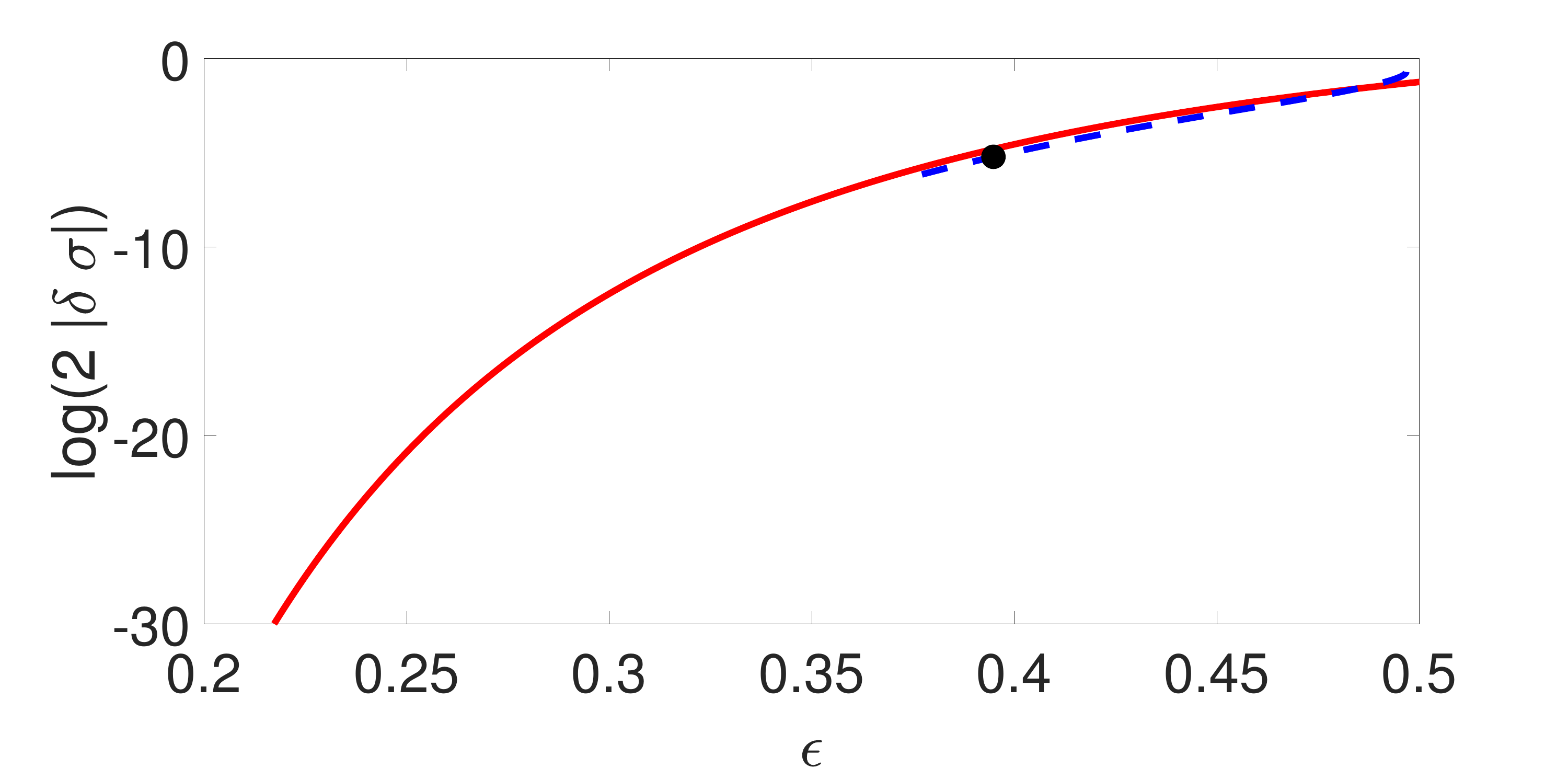}
                 \caption{}
             \end{subfigure}
             \hfill
             \begin{subfigure}[b]{0.49\textwidth}
                 \centering
                 \includegraphics[width = \linewidth]{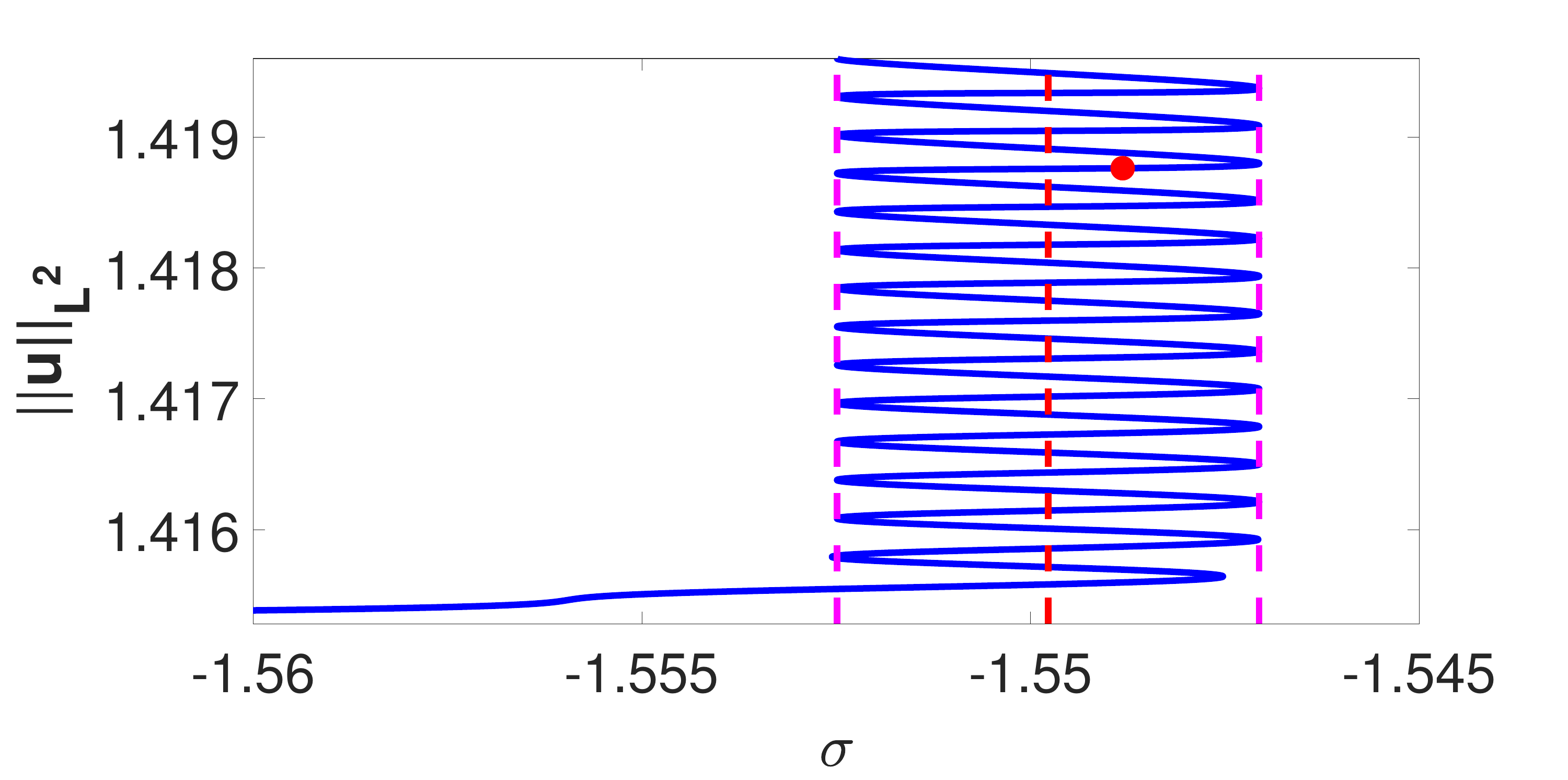}
                 \caption{}
             \end{subfigure}
             \\
             \begin{subfigure}[b]{0.49\textwidth}
                 \centering
                 \includegraphics[width = \textwidth]{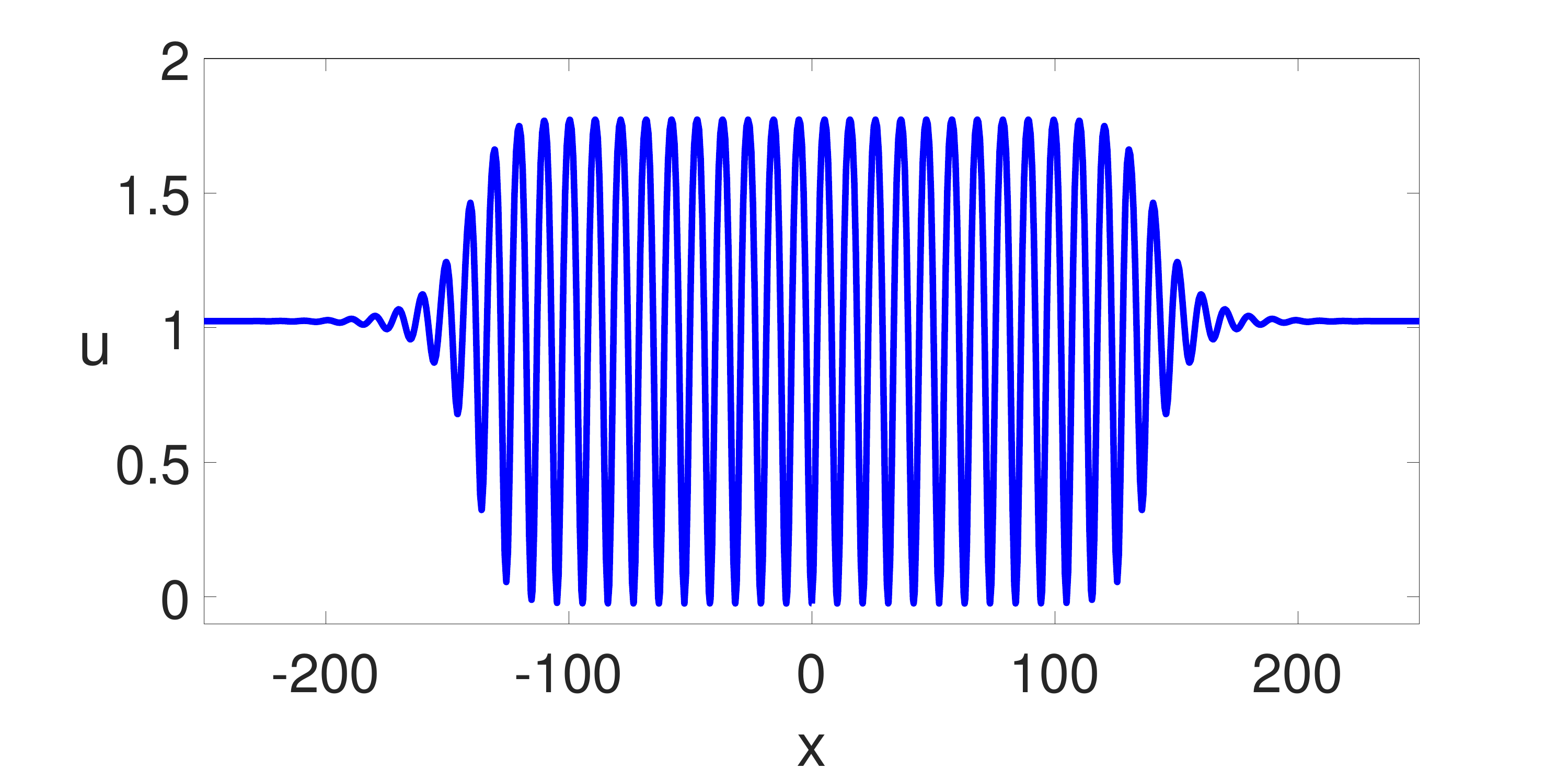}
                 \caption{}
             \end{subfigure}
             \hfill
             \begin{subfigure}[b]{0.49\textwidth}
                 \centering
                 \includegraphics[width = \textwidth]{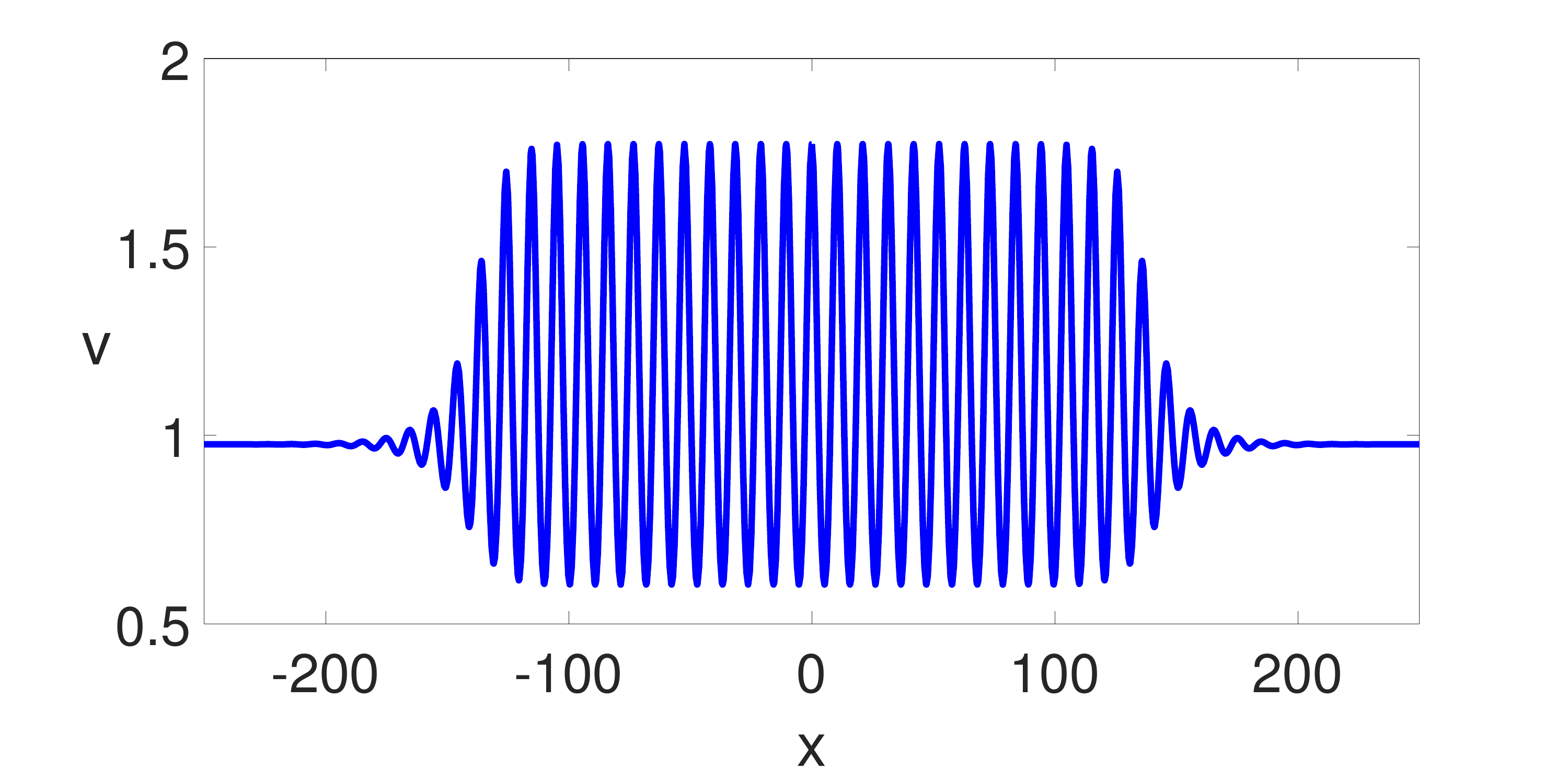}
                 \caption{}
             \end{subfigure}
            \caption{Similar to Figure \ref{fig:matchingSH23}, but for model \eqref{Schnakenberg}. The snaking in panel (b) (respectively, the solution in panels (c) and (d)) corresponds to $\varepsilon \approx 0.394849$ (respectively, $\sigma \approx - 1.5488164814$).}
            \label{fig:matching-Schnakenberg}
        \end{figure}

    \subsection{Brusselator}
        In \cite{villar2023degenerate}, the authors developed pattern formation and pattern localization for Schnakenberg-type models used to model chemical reactions or root-hair development. In particular, they proved several results for each of those models, but highlighted that the Brusselator is simple enough so that codimension-two points can be found explicitly, but complex enough to show localized structures (in fact, they proved that there is a way to scale the model so that the coordinates of its codimension-two Turing bifurcation points do not change under a change in the diffusion coefficient). Therefore, we focus on the Brusselator system, which can be written as:
        \begin{equation}
            \begin{aligned}
                \partial_t u &= a - c \, u + u^2 \, v + \delta^2 \, \partial_{xx} u,
                \\
                \partial_t v &= (c - 1) \, u - u^2 \, v + \partial_{xx} v,
            \end{aligned} \label{Brus}
        \end{equation}
        and we use $b = c$. This system has only one homogeneous steady state given by
        \begin{align*}
            \mbf P = \left(a, \frac{c - 1}{a}\right).
        \end{align*}
        which goes through a Turing bifurcation whenever $c = a^2 \, \delta^2 + 2 \, a \, \delta + 2$.
        
        The authors of \cite{villar2023degenerate} proved this model has codimension-two points at $a_\pm = \dfrac{1}{16 \, \delta} \, \left(21 \pm \sqrt{313}\right)$. In particular, if we fix $\delta = \dfrac{21 + \sqrt{313}}{48}$, then the variables of the expansion for the codimension-2 point at $a = a_+$ can be obtained \evs{with the code in \cite{beyond-code},} and are shown in Table \ref{tab:valBrus}. Now, note that when making the change
        \begin{align*}
            (u, v) &\to \left(u + a, v + \frac{c - 1}{a}\right),
        \end{align*}
        we obtain the system given by
        \begin{equation}
            \begin{aligned}
                \partial_t u &= a^2 v - \frac{u^2}{a} + 2 \, a \, u \, v + u^2 \, v - 2 \, u + c \, \left(u + \frac{u^2}{a}\right) + \delta^2 \, \partial_{xx} u,
                \\
                \partial_t v &= - a^2 \, v + \frac{u^2}{a} - 2 \, a \, u \, v - c \, u - u^2 \, v + u - c \, \left(u + \frac{u^2}{a}\right) + \partial_{xx} v.
            \end{aligned} \label{Bruschange}
        \end{equation}
        \begin{table}
            \centering
            {\def\arraystretch{2.5}
            \begin{tabular}{|c|c|}
                \hline
                \textbf{Variable} & \textbf{Value}
                \\
                \hline
                $a_2$ & 1
                \\
                \hline
                $c_2$ & $\dfrac{545 + 29 \, \sqrt{313}}{192}$
                \\
                \hline
                $c_4$ & $\dfrac{27181639265}{160777328256} + \dfrac{1139313695 \, \sqrt{313}}{53592442752}$
                \\
                \hline
                $c_6$ & \small{\evs{$\dfrac{255568350955114905604104012761}{70461515885712580123818315840} - \dfrac{74364627222521911575006255847 \, \sqrt{313}}{352307579428562900619091579200}$}}
                \\
                \hline
                $\alpha_1$ & $\dfrac{58 + 2 \, \sqrt{313}}{33}$
                \\
                \hline
                $\alpha_2$ & $- \dfrac{58 + 2 \, \sqrt{313}}{99}$
                \\
                \hline
                $\alpha_3$ & $\dfrac{8549 + 49 \, \sqrt{313}}{7128}$
                \\
                \hline
                $\alpha_4$ & $- \dfrac{316 + 20 \, \sqrt{313}}{891}$
                \\
                \hline
                $\alpha_5$ & $\dfrac{- 25180885801 + 399883939 \, \sqrt{313}}{165801619764}$
                \\
                \hline
                $\alpha_6$ & $\dfrac{55 + 3 \, \sqrt{313}}{162}$
                \\
                \hline
                $\alpha_7$ & $- \dfrac{265583 + 5555 \, \sqrt{313}}{497664}$
                \\
                \hline
            \end{tabular}
            }
            \caption{Values of some of the key parameters in the asymptotic expansion of system \eqref{Brus}. \evs{The values of the parameters that are not shown, but have the same name format with a different subscript to the ones presented equal zero. These variables were obtained with the Python code in \cite{beyond-code}, with the variable \texttt{modelname=`Brusselator'}.}}
            \label{tab:valBrus}
        \end{table}
        
        Here, we note that a change in $c$ away from the Maxwell point, which will be called $\Delta c$ to avoid confusion with $\delta$, will affect two terms, a linear and a quadratic term: $u + u^2/a$. Nevertheless, we only need to care about the dominant term, which is the linear one, which implies that the equation for the amplitude of the remainder becomes
        \begin{align*}
            \alpha_1 \, B_1'' + i \, \alpha_3 \, \abs{A_1}^2 \, B_1' + i \, \alpha_3 \, A_1 \, A_1' \, \bar B_1 + i \, \alpha_3 \, \bar A_1 \, A_1' \, B_1 + \alpha_5 \, B_1 + \alpha_6 \, A_1^2 \, \bar B_1 + 2 \, \alpha_6 \, \abs{A_1}^2 \, B_1
            \\
            + 2 \, \alpha_7 \, \abs{A_1}^2 \, A_1^2 \, \bar B_1 + 3 \, \alpha_7 \, \abs{A_1}^4 \, B_1 + \frac{1}{66} \, \left(37 - \sqrt{313}\right) \, \frac{\Delta c}{\varepsilon^4} \, A_1 = 0.
        \end{align*}
        Now, in order to determine the width of the snaking, let us recall that we need to run the recursion given by \eqref{DM} with the initial conditions given by
        \begin{align*}
            A_1 &\approx 0.109085 \, i,
            \\
            A_3 &\approx \evs{- 0.180032 - 0.130054 \, i}.
        \end{align*}
        Figure \ref{fig:c20brus} shows the results of the convergence of the argument of the limit in \eqref{c0-def} for $r = 2$, for different values of $n/2$. In particular, when fitting a curve as for the Swift-Hohenberg equation, we note that
        \begin{align*}
            c_2^{[0]} \approx \evs{0.000763389} \, e^{- \evs{0.547143} \, i}.
        \end{align*}
        \begin{figure}
            \centering
            \begin{subfigure}[b]{0.48\textwidth}
                 \centering
                 \includegraphics[width = \textwidth]{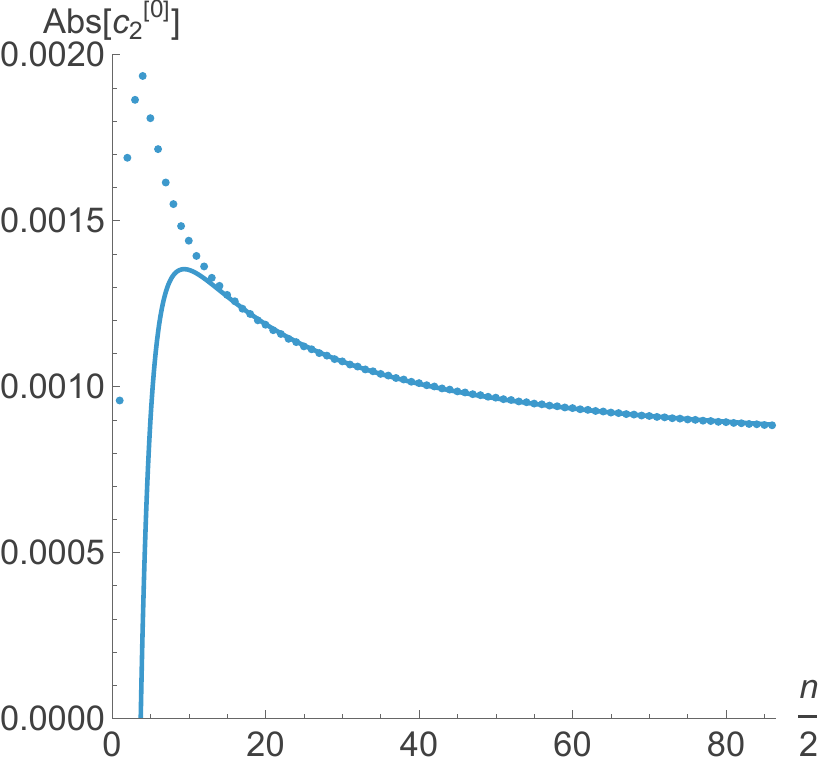}
                 \caption{}
             \end{subfigure}
             \hfill
             \begin{subfigure}[b]{0.48\textwidth}
                 \centering
                 \includegraphics[width = \textwidth]{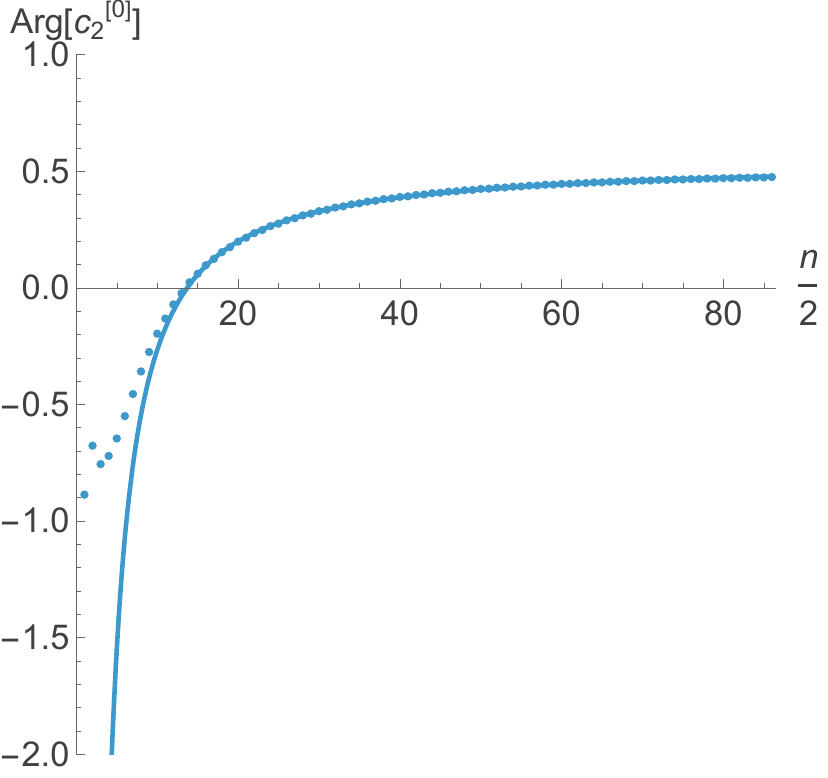}
                 \caption{}
             \end{subfigure}
            \caption{Similar to Figure \ref{fig:c20SH23}, but for model \eqref{Brus}. \evs{Both of these graphs were obtained with the Mathematica code within the folders \texttt{Iteration/Brusselator} in \cite{beyond-code}.}}
            \label{fig:c20brus}
        \end{figure}
        
        Furthermore, we have
        \begin{align*}
            B_1 &\sim \frac{1}{8 \, \beta_1} \, \sqrt{- \frac{\beta_1}{\beta_3}} \, \left(- \frac{\sqrt{2} \, \beta_3 \, L_2^+}{\beta_1} - \frac{\sqrt{2} \, \beta_3 \, L_2^-}{\beta_1} + \frac{\sqrt{313} - 37}{33 \, \sqrt{2} \, \alpha_1 \, \varepsilon^4} \, \Delta c\right) \, (1 + 2 \, \eta \, i) \, e^{2 \, \sqrt{\beta_1} \, X}
            \\
            &= \frac{1}{8 \, \beta_1} \, \sqrt{- \frac{\beta_1}{\beta_3}} \, \left(- \frac{4 \, \sqrt{2} \, \beta_3}{\beta_1} \, \frac{\pi \, \abs{K_2} \, e^{- \frac{\pi}{2} \, \left(\frac{1}{\sqrt{\beta_1} \, \varepsilon^2} + \eta\right)}}{\varepsilon^6} \, \cos\left(K_2^o - \hat \chi + 2 \, \eta \, \log (\varepsilon)\right)\right.
            \\
            &\quad \left. + \frac{\sqrt{313} - 37}{33 \, \sqrt{2} \, \alpha_1 \, \varepsilon^4} \, \Delta c\right) \, (1 + 2 \, \eta \, i) \, e^{2 \, \sqrt{\beta_1} \, X},
        \end{align*}
        as $X \to \infty$.
        
        Thus, to ensure that the remainder tends to zero as $X \to \infty$, we need
        \begin{align*}
            \abs{\Delta c} \leq \frac{264 \, \alpha_1 \, \beta_3}{\beta_1 \, \left(\sqrt{313} - 37\right)} \, \frac{\pi \, \abs{K_2} \, e^{- \frac{\pi}{2} \, \left(\frac{1}{\sqrt{\beta_1} \, \varepsilon^2} + \eta\right)}}{\varepsilon^2},
        \end{align*}
        where $\abs{K_2} \approx \evs{0.0255587}$ (see \eqref{K_2}). Therefore,
        \begin{align}
            \abs{\Delta c} \leq \frac{\evs{18.869315} \, e^{- 4.808385 \, \frac{\pi}{\varepsilon^2}}}{\varepsilon^2}. \label{wrong-brus}
        \end{align}
        The result of the matching of this function with respect to numerical computation of the snaking for model \eqref{Brus} is shown by the black line in Figure \ref{fig:bru-fit}, which shows that although the prediction we have obtained seems to follow the same behaviour as the actual width of the snaking, the coefficient in front of \eqref{wrong-brus} needs to be higher. In fact, in this case, we have that
        \begin{align*}
            \abs{\frac{h_1}{k_1}} \approx 0.006504 \ll 1,
        \end{align*}
        which implies that the mode $r = 2$ is much less dominant than $r = 0$ (see \eqref{h1-k1ratio}). Now, observe the convergence of $c_0^{[0]}$, shown in Figure \ref{fig:c00brus}. In this case, we note that we have two limits to which the curve seems to converge. This happens because the mode $r = 0$ is present for every dominant value of $\kappa^2 = \kappa_\pm^2$. That is, there is an interaction between \evs{all of} these modes.

        To summarize, if we take the average between the two limits of convergence, which turn out to be \evs{0.0760281} and \evs{0.200181}, and divide the result by \evs{8}, considering that there are four different values of $\kappa$ (that is, four different solutions, see Section \ref{sec:late_term}) \evs{and the mode $r = 0$ needs to be split into the real and imaginary parts for each of these modes (see Remark \ref{rem:factorjustification})}, we obtain
        \begin{align*}
            \abs{c_0^{[0]}} \approx \evs{0.017263}.
        \end{align*}
        Now, if we use this value for computing $K_2$, instead of $c_2^{[0]}$, taking into account that the matching does not change when considering that $h_0^{[0]}$ is defined as the complex-conjugate of $h_2^{[0]}$ when the integrals that defined them are set to be real (see \eqref{K_2}), we obtain
        \begin{align}
            \abs{\Delta c} \leq \frac{\evs{426.704943} \, e^{- 4.808385 \, \frac{\pi}{\varepsilon^2}}}{\varepsilon^2}, \label{perfect-brus}
        \end{align}
        which leads to the red line in Figure \ref{fig:bru-fit}, showing a remarkably good match, even for high values of $\varepsilon$.
        \begin{figure}
            \centering
            \begin{subfigure}[b]{0.49\textwidth}
                 \centering
                 \includegraphics[width = \textwidth]{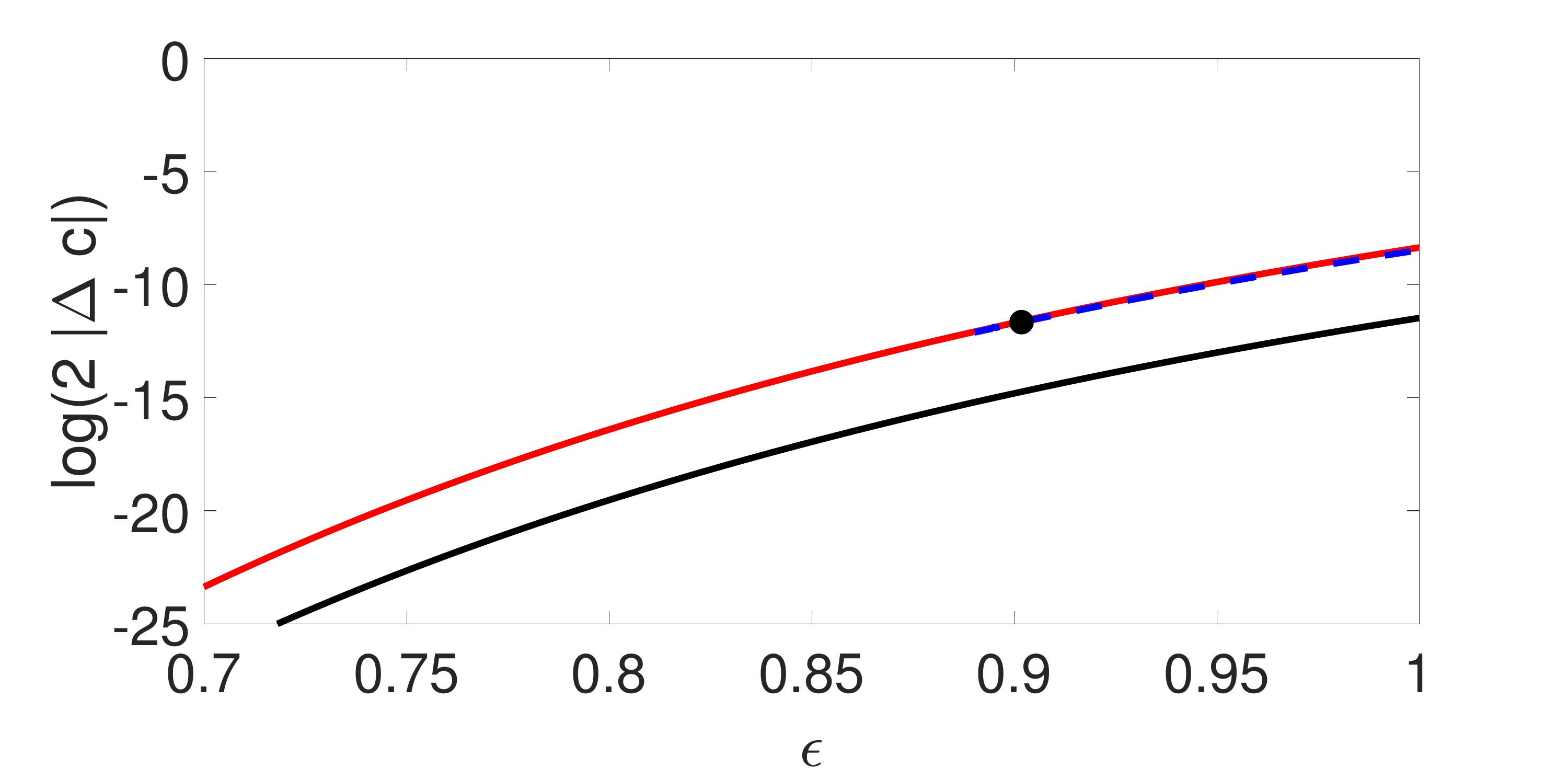}
                 \caption{}
                 \label{fig:bru-fit}
             \end{subfigure}
             \hfill
             \begin{subfigure}[b]{0.49\textwidth}
                 \centering
                 \includegraphics[width = \textwidth]{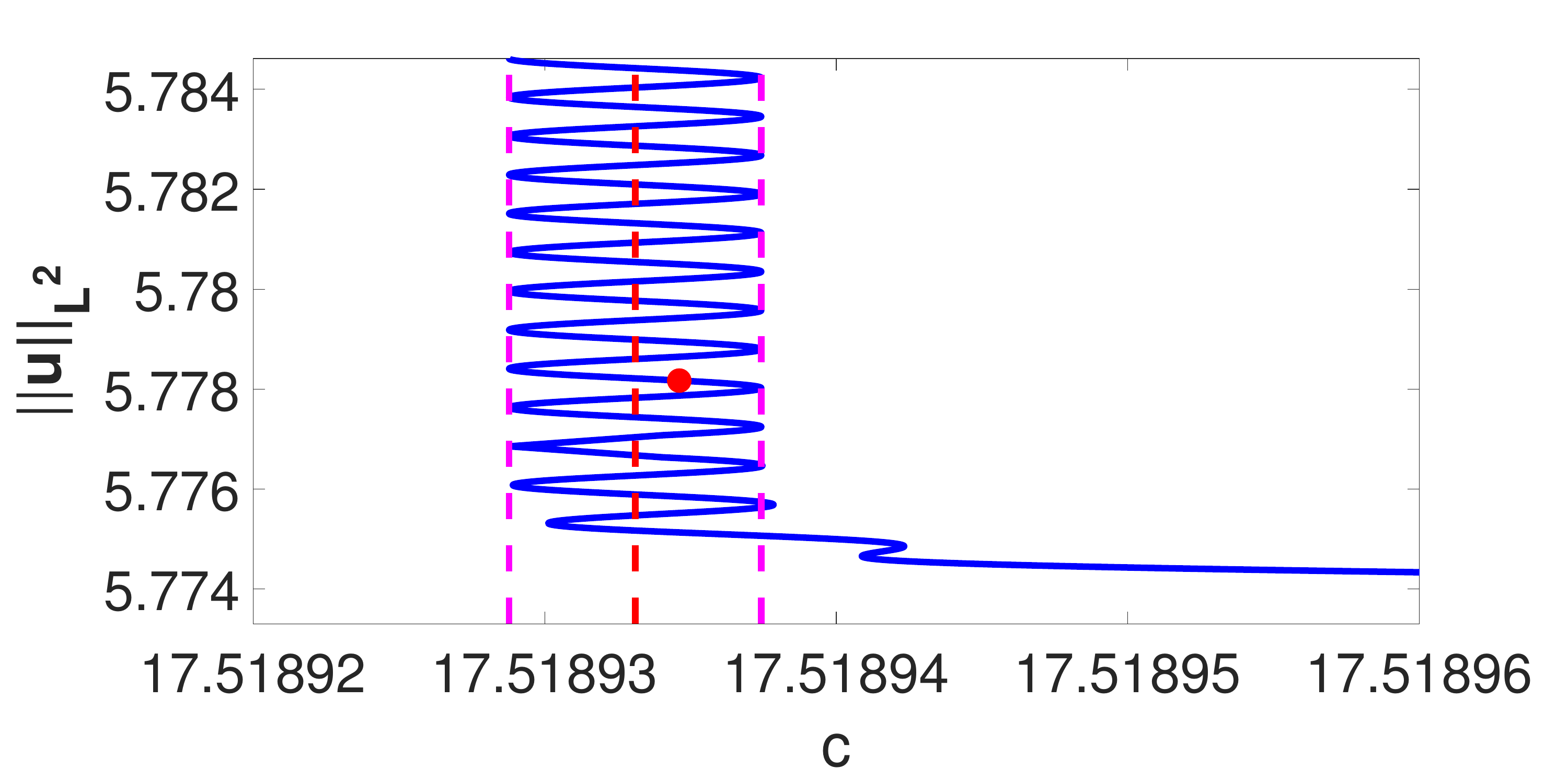}
                 \caption{}
             \end{subfigure}
             \\
             \begin{subfigure}[b]{0.48\textwidth}
                 \centering
                 \includegraphics[width = \textwidth]{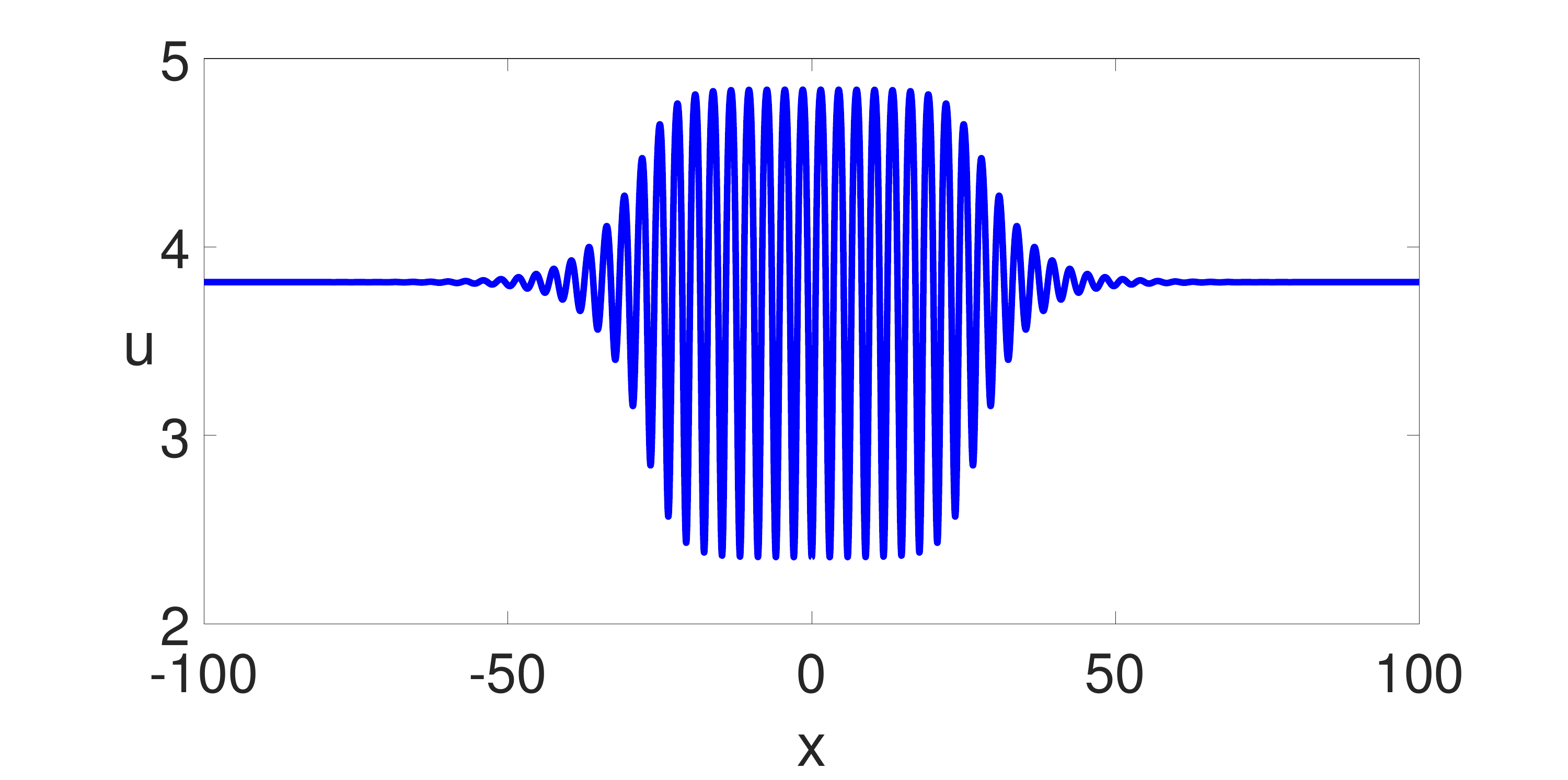}
                 \caption{}
             \end{subfigure}
             \hfill
             \begin{subfigure}[b]{0.48\textwidth}
                 \centering
                 \includegraphics[width = \textwidth]{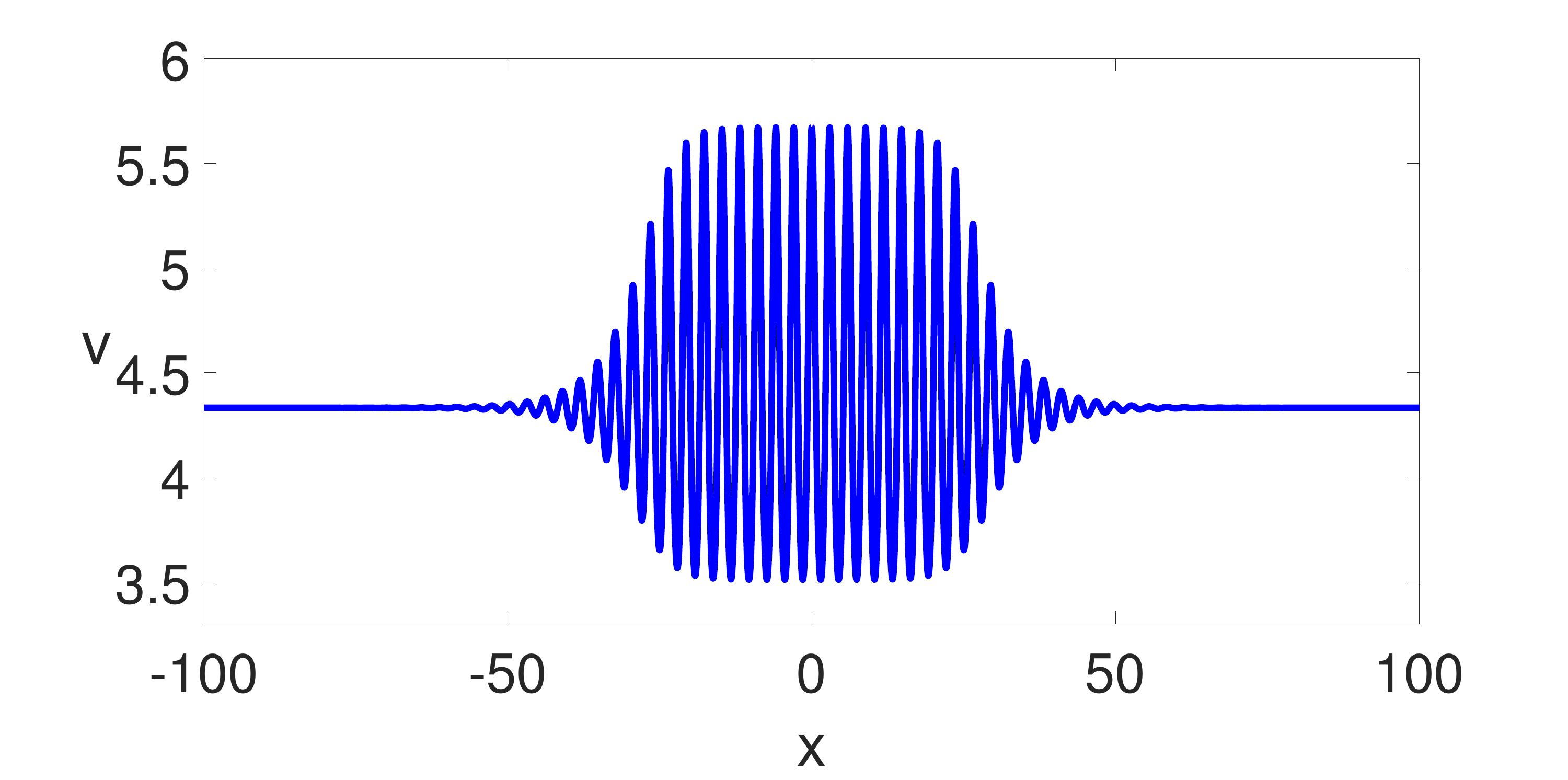}
                 \caption{}
             \end{subfigure}
            \caption{Similar to Figure \ref{fig:matchingSH23}, but for model \eqref{Brus}. Here, panel (a) shows the comparison between the width of the snaking with respect to \eqref{wrong-brus} (black line) and \eqref{perfect-brus} (red line). The snaking in panel (b) (respectively, the solution shown in panels (c) and (d)) corresponds to $\varepsilon \approx 0.901788$ (respectively, $c \approx 17.518934615$).}
            \label{fig:matching-brus}
        \end{figure}
        \begin{figure}
            \centering
            \begin{subfigure}[b]{0.48\textwidth}
                 \centering
                 \includegraphics[width = \textwidth]{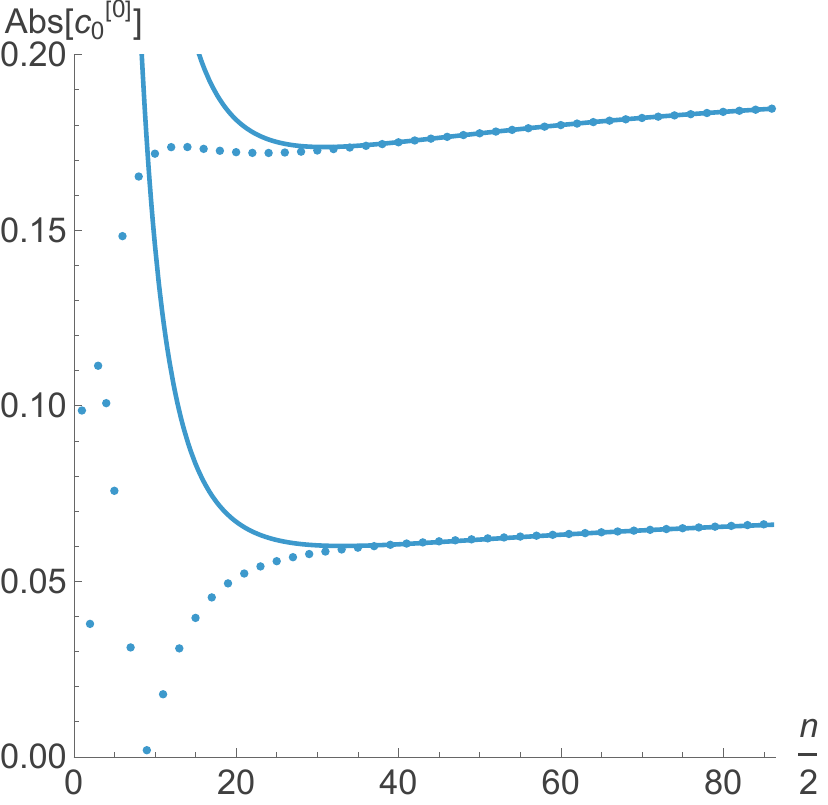}
                 \caption{}
                 \label{fig:c0normbrus}
             \end{subfigure}
             \hfill
             \begin{subfigure}[b]{0.48\textwidth}
                 \centering
                 \includegraphics[width = \textwidth]{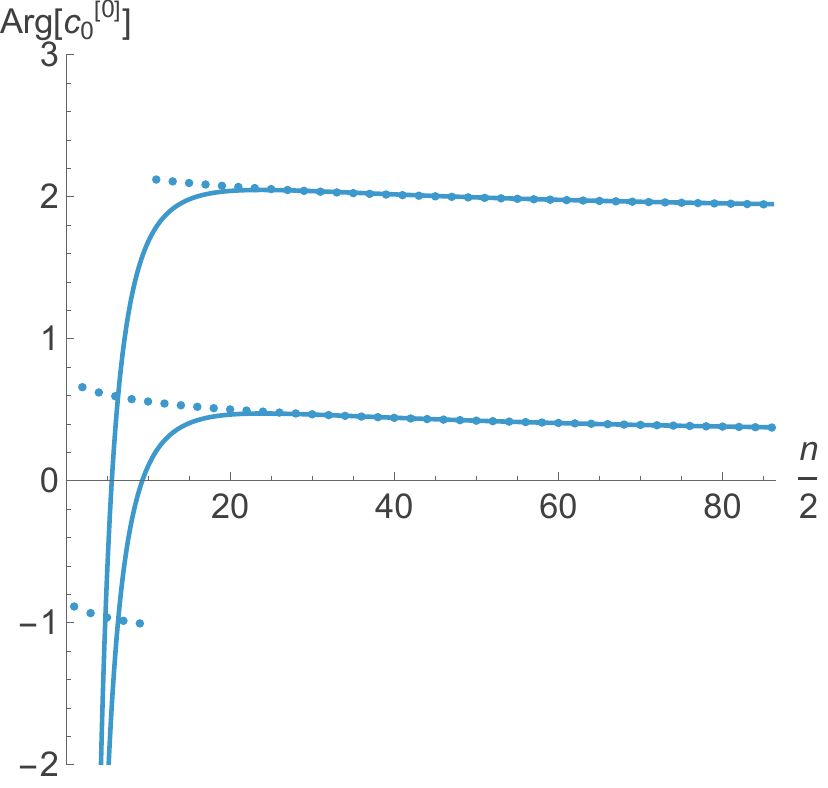}
                 \caption{}
                 \label{fig:c0anglebrus}
             \end{subfigure}
            \caption{Similar to Figure \ref{fig:c20SH23}, but to approximate $c_0^{[0]}$ for model \eqref{Brus}. \evs{Both of these graphs were obtained with the Mathematica code within the folders \texttt{Iteration/Brusselator} in \cite{beyond-code}.}}
            \label{fig:c00brus}
        \end{figure}

    \subsection{4-component Brusselator}
        Finally, we will show how our code can be used to carry out the calculations for higher-component systems. In particular, once again, motivated by the Brusselator, we now work with two Brusselator models coupled linearly, which form one 4-component reaction-diffusion system studied in \cite{AlastairBru4}. That system is given by
        \begin{equation}
            \begin{aligned}
                \partial_t u &= a - (b + 1) \, u + u^2 \, v + \alpha \, (w - u) + \delta^2 \, \partial_{xx} u,
                \\ 
                \partial_t v &= b \, u - u^2 \, v + \beta \, (z-v) + \partial_{xx} v,
                \\ 
                \partial_t w &= a - (b + 1) \, w + w^2 \, z + \alpha \, (u - w) + \delta^2 \, \partial_{xx} w,
                \\
                \partial_t z &= b \, w - w^2 \, z + \beta \, (v - z) + \partial_{xx} z,
            \end{aligned} \label{Bru4}
        \end{equation}
        and it is a model that has a homogeneous steady state given by
        \begin{align*}
            \mbf P = \left(a, \frac{b}{a}, a, \frac{b}{a}\right).
        \end{align*}
        Now, if we fix
        \begin{align*}
            \left(\alpha, \beta, \delta\right) = \left(1, \frac{1}{2}, \frac{21 + \sqrt{313}}{48}\right),
        \end{align*}
        then $\mbf P$ goes through a codimension-two Turing bifurcation at
        \begin{align*}
            (a, b) = \left(3, \frac{1}{128} \left(841 + 37 \, \sqrt{313}\right)\right).
        \end{align*}
        The variables of the expansion at this codimension-two point are shown in Table \ref{tab:valBru4}.
        \begin{table}
            \centering
            {\def\arraystretch{2.5}
            \begin{tabular}{|c|c|}
                \hline
                \textbf{Variable} & \textbf{Value}
                \\
                \hline
                $a_2$ & 1
                \\
                \hline
                $b_2$ & $\dfrac{29 \, \sqrt{313}}{192} + \dfrac{545}{192}$
                \\
                \hline
                $b_4$ & $\dfrac{27181639265}{160777328256} + \dfrac{1139313695 \, \sqrt{313}}{53592442752}$
                \\
                \hline
                $b_6$ & \small{\evs{$\dfrac{255568350955114905604104012761}{70461515885712580123818315840} - \dfrac{74364627222521911575006255847 \, \sqrt{313}}{352307579428562900619091579200}$}}
                \\
                \hline
                $\alpha_1$ & $\dfrac{4 \, \sqrt{313}}{33} + \dfrac{116}{33}$
                \\
                \hline
                $\alpha_2$ & $- \dfrac{116}{99} - \dfrac{4 \, \sqrt{313}}{99}$
                \\
                \hline
                $\alpha_3$ & $\dfrac{49 \, \sqrt{313}}{3564} + \dfrac{8549}{3564}$
                \\
                \hline
                $\alpha_4$ & $- \dfrac{40 \, \sqrt{313}}{891} - \dfrac{632}{891}$
                \\
                \hline
                $\alpha_5$ & $- \dfrac{25180885801}{82900809882} + \dfrac{399883939 \, \sqrt{313}}{82900809882}$
                \\
                \hline
                $\alpha_6$ & $\dfrac{\sqrt{313}}{27} + \dfrac{55}{81}$
                \\
                \hline
                $\alpha_7$ & $- \dfrac{265583}{248832} - \dfrac{5555 \, \sqrt{313}}{248832}$
                \\
                \hline
            \end{tabular}
            }
            \caption{Values of some of the key parameters in the asymptotic expansion of system \eqref{Bru4}. \evs{The values of the parameters that are not shown, but have the same name format with a different subscript to the ones presented equal zero. These variables were obtained with the Python code in \cite{beyond-code}, with the variable \texttt{modelname=`Bru 4'}.}}
            \label{tab:valBru4}
        \end{table}
        Now, if we make the translation $(u, v, w, z) \to (u, v, w, z) + \mbf P$, we obtain the system
        \begin{align*}
            \partial_t u &= b \, \left(u + \frac{u^2}{2 \, a}\right) + v \, (2 \, a + u)^2 - 2 \, u + w + \delta^2 \, \partial_{xx} u,
            \\
            \partial_t v &= \frac{1}{2} \, \left(- b \, \left(\frac{u \, (2 \, a + u)}{a}\right) - v \, \left(2 \, (2 \, a + u)^2 + 1\right) + z\right) + \partial_{xx} v,
            \\
            \partial_t w &= b \, \left(w + \frac{w^2}{2 \, a}\right) + z \, (2 \, a + w)^2 - 2 \, w + u + \delta^2 \, \partial_{xx} w,
            \\
            \partial_t z &= \frac{1}{2} \, \left(- b \, \left(\frac{w \, (2 \, a + w)}{a}\right) - z \, \left(2 \, (2 \, a + w)^2 + 1\right) + v\right) + \partial_{xx} z.
        \end{align*}
        Now, as in the previous examples, recall that we need to run the iteration \eqref{DM} to be able to estimate the width of the snaking in this case, with initial conditions given by
        \begin{align*}
            A_1 &\approx 0.109084591223343 \, i,
            \\
            A_3 &\approx \evs{- 0.180031849961344 - 0.130054306337421 \, i}.
        \end{align*}
        The result of that iteration is shown in Figure \ref{fig:c00bru4}. We highlight here that, as in the previous example, the most dominant mode turns out to be $r = 0$, reason why we show the convergence of $c_0^{[0]}$ instead of $c_2^{[0]}$. \evs{The corresponding limits of the recurrence in this case} turn out to be approximately \evs{0.115257} and \evs{0.247105}. Therefore, following the same idea as in the previous example, if we divide the average of these two limits by \evs{8}, we obtain that
        \begin{align*}
            \abs{c_0^{[0]}} \approx \evs{0.0226476},
        \end{align*}
        which implies $\abs{K_2} \approx \evs{0.758257}$ \evs{(see \eqref{K_2})}.
        \begin{figure}
            \centering
            \begin{subfigure}[b]{0.48\textwidth}
                 \centering
                 \includegraphics[width = \textwidth]{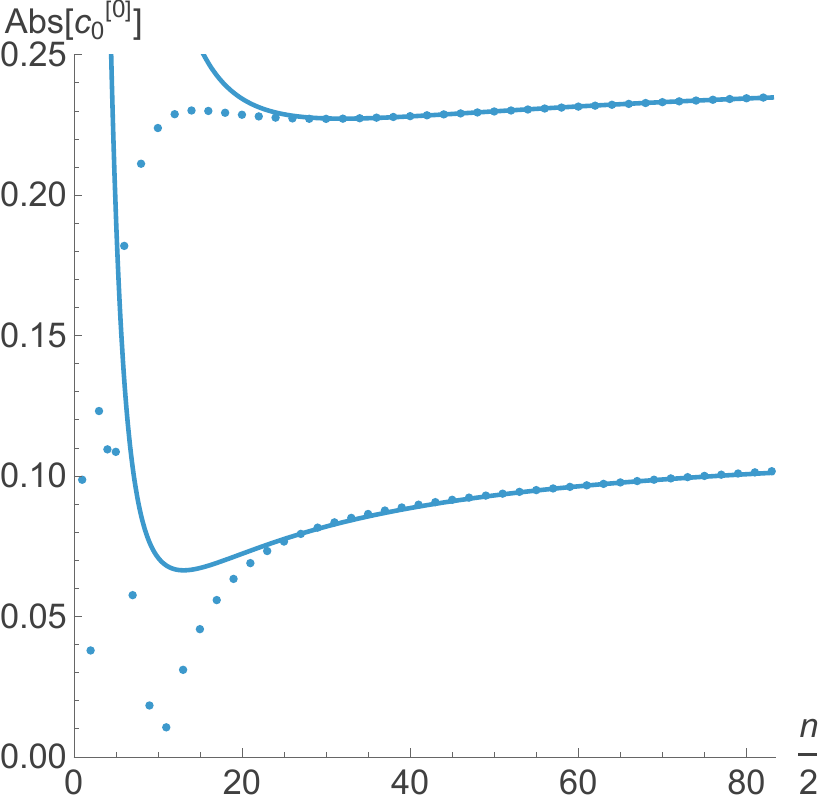}
                 \caption{}
                 \label{fig:c00-amp-bru4}
             \end{subfigure}
             \hfill
             \begin{subfigure}[b]{0.48\textwidth}
                 \centering
                 \includegraphics[width = \textwidth]{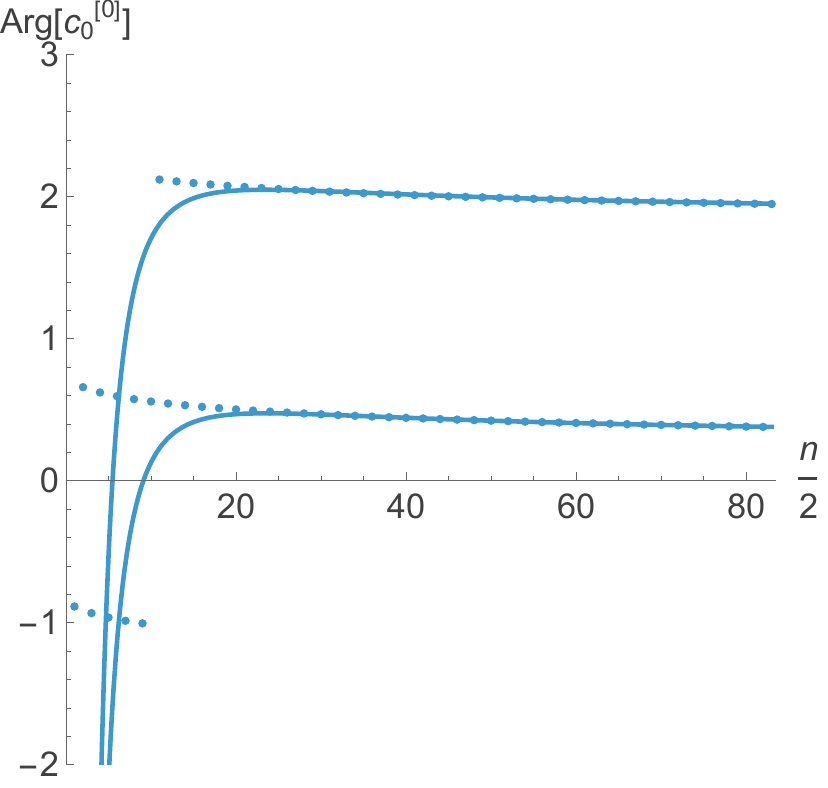}
                 \caption{}
             \end{subfigure}
            \caption{Similar to Figure \ref{fig:c00brus}, but for model \eqref{Bru4}. \evs{Both of these graphs were obtained with the Mathematica code within the folders \texttt{Iteration/Brusselator} in \cite{beyond-code}.}}
            \label{fig:c00bru4}
        \end{figure}
        
        To summarize the process, by following the same ideas as in the previous example, we have to solve the following equation for the amplitude of the remainder at order five, \eqref{rem_eq}:
        \begin{align*}
            \alpha_1 \, B_1'' + i \, \alpha_3 \, \abs{A_1}^2 \, B_1' + i \, \alpha_3 \, A_1 \, A_1' \, \bar B_1 + i \, \alpha_3 \, \bar A_1 \, A_1' \, B_1 + \alpha_5 \, B_1 + \alpha_6 \, A_1^2 \, \bar B_1 + 2 \, \alpha_6 \, \abs{A_1}^2 \, B_1
            \\
            + 2 \, \alpha_7 \, \abs{A_1}^2 \, A_1^2 \, \bar B_1 + 3 \, \alpha_7 \, \abs{A_1}^4 \, B_1 + \frac{1}{33} \, \left(37 - \sqrt{313}\right) \, \frac{\Delta b}{\varepsilon^4} \, A_1 = 0.
        \end{align*}
        With this, we have
        \begin{align*}
            B_1 &\sim \frac{1}{4 \, \beta_1} \, \sqrt{- \frac{\beta_1}{\beta_3}} \, \left(- \frac{\sqrt{2} \, \beta_3 \, L_2^+}{2 \, \beta_1} - \frac{\sqrt{2} \, \beta_3 \, L_2^-}{2 \, \beta_1} + \frac{\sqrt{313} - 37}{33 \, \sqrt{2} \, \alpha_1 \, \varepsilon^4} \, \Delta b\right) \, (1 + 2 \, \eta \, i) \, e^{2 \, \sqrt{\beta_1} \, X}
            \\
            &= \begin{multlined}[t]
                \frac{1}{4 \, \beta_1} \, \sqrt{- \frac{\beta_1}{\beta_3}} \, \left(- \frac{2 \, \sqrt{2} \, \beta_3}{\beta_1} \, \frac{\pi \, \abs{K_2} \, e^{- \frac{\pi}{2} \, \left(\frac{1}{\sqrt{\beta_1} \, \varepsilon^2} + \eta\right)}}{\varepsilon^6} \, \cos\left(K_2^o - \hat \chi + 2 \, \eta \, \log (\varepsilon)\right)\right.
                \\
                \quad \left. + \frac{\sqrt{313} - 37}{33 \, \sqrt{2} \, \alpha_1 \, \varepsilon^4} \, \Delta b\right) \, (1 + 2 \, \eta \, i) \, e^{2 \, \sqrt{\beta_1} \, X},
            \end{multlined}
        \end{align*}
        as $X \to \infty$. Therefore, as in the previous examples, we conclude that
        \begin{align*}
            \abs{\Delta b} \leq \frac{132 \, \alpha_1 \, \beta_3}{\beta_1 \, \left(\sqrt{313} - 37\right)} \, \frac{\pi \, \abs{K_2} \, e^{- \frac{\pi}{2} \, \left(\frac{1}{\sqrt{\beta_1} \, \varepsilon^2} + \eta\right)}}{\varepsilon^2},
        \end{align*}
        which can be written, approximately, as
        \begin{align}
            \abs{\Delta b} \leq \frac{\evs{559.799849} \, e^{- 4.808385 \, \frac{\pi}{\varepsilon^2}}}{\varepsilon^2}, \label{perfect-bru4}
        \end{align}
        leading to the red line in Figure \eqref{fig:fit-bru4}, which shows, once again, a remarkably good match, letting us conclude that the theory we have developed produces a great result when the expansions and constants are computed accurately.
        \begin{figure}
            \centering
            \begin{subfigure}[b]{0.49\textwidth}
                \centering
                \includegraphics[width = \textwidth]{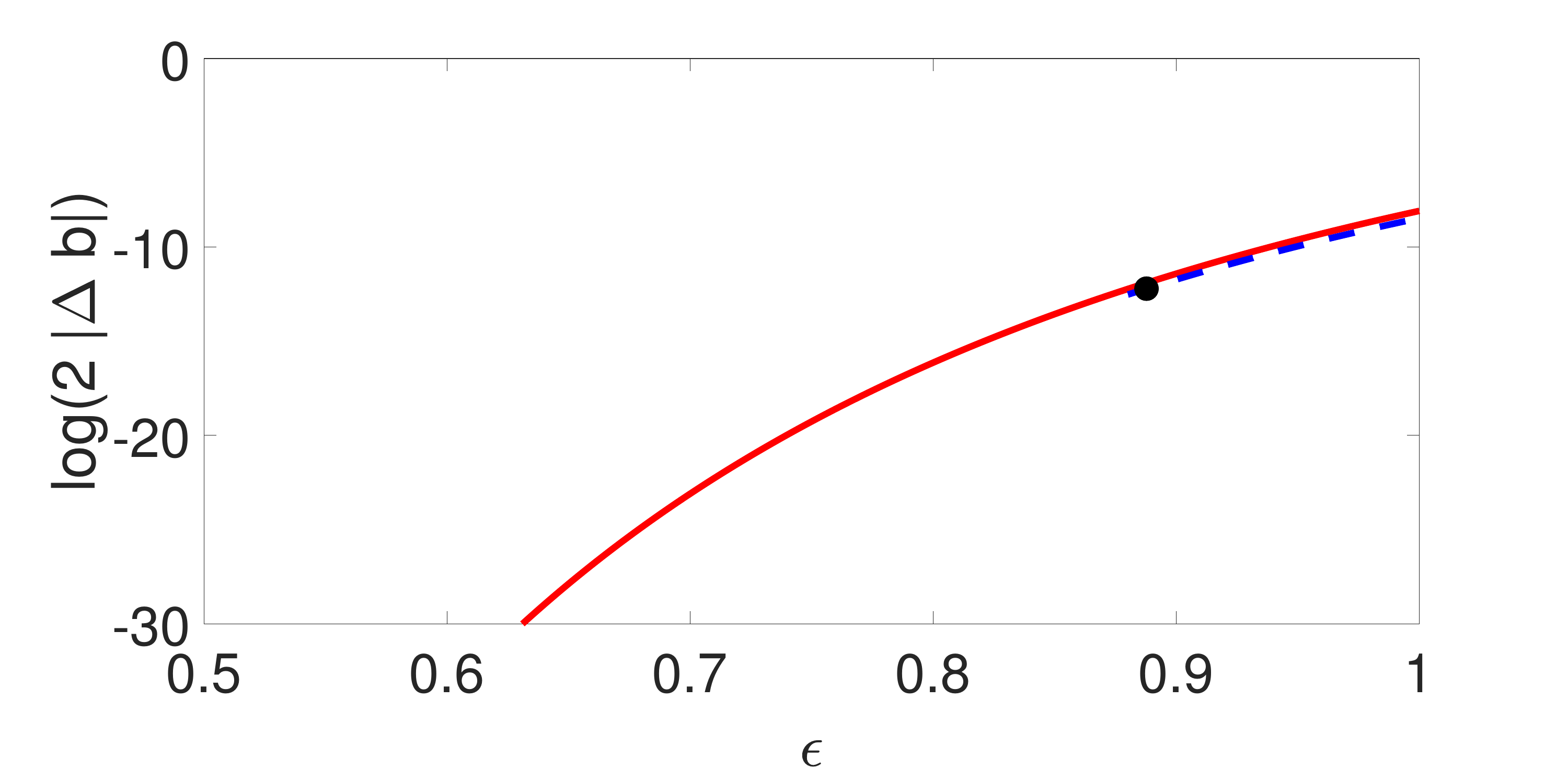}
                \caption{}
                \label{fig:fit-bru4}
            \end{subfigure}
            \hfill
            \begin{subfigure}[b]{0.49\textwidth}
                \centering
                \includegraphics[width = \textwidth]{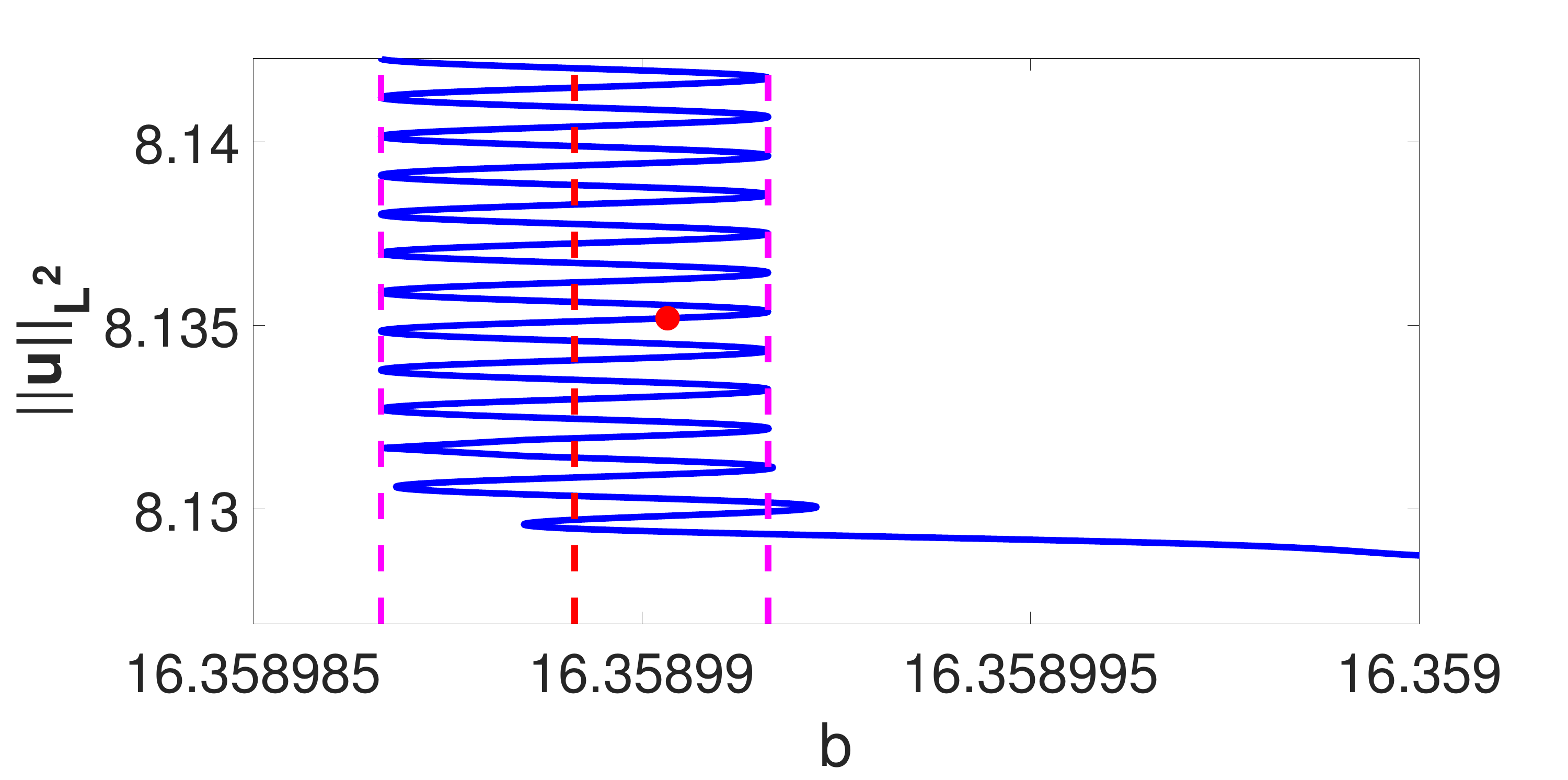}
                \caption{}
            \end{subfigure}
            \\
            \begin{subfigure}[b]{0.49\textwidth}
                \centering
                \includegraphics[width = \textwidth]{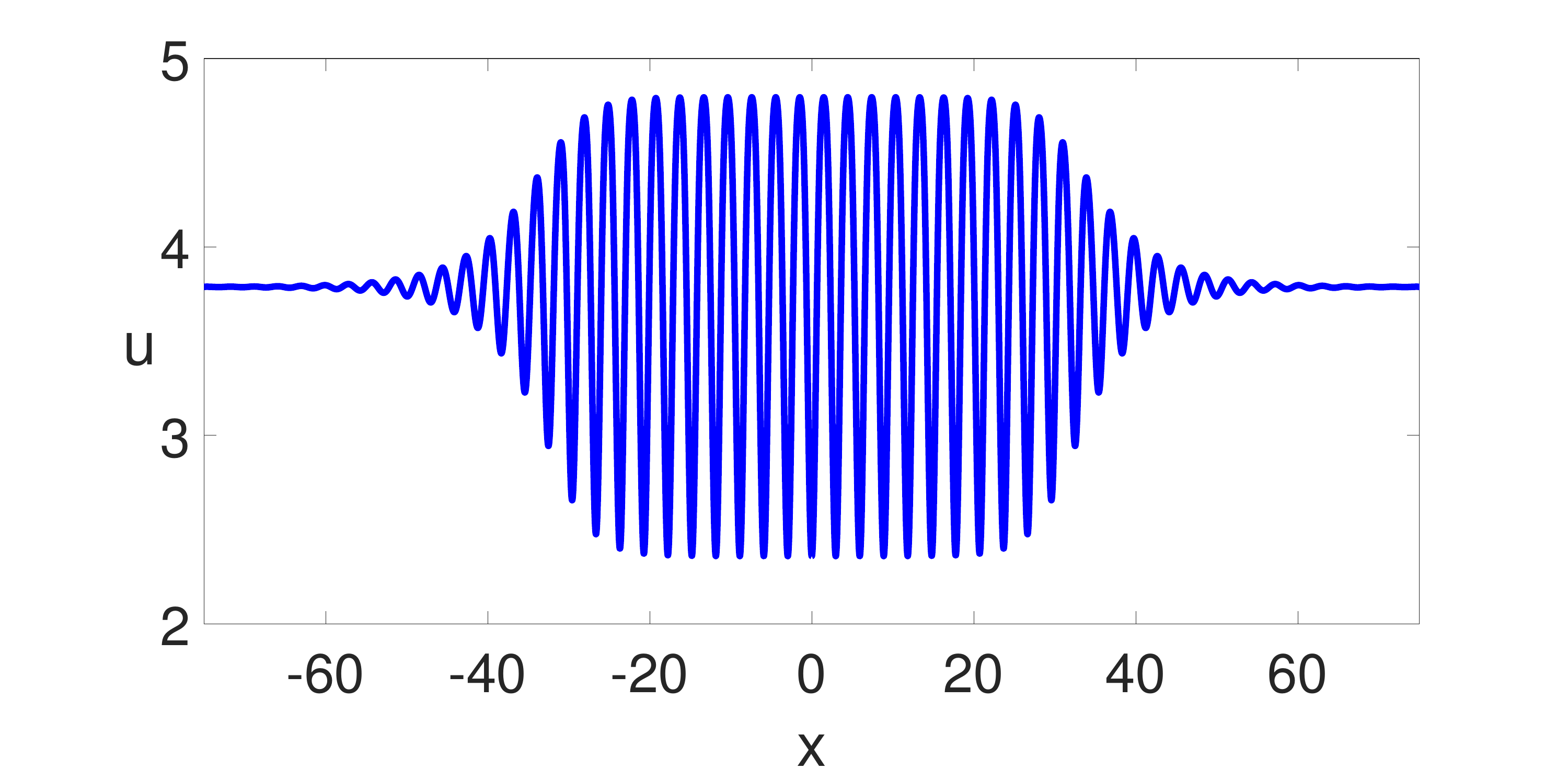}
                \caption{}
            \end{subfigure}
            \hfill
            \begin{subfigure}[b]{0.49\textwidth}
                \centering
                \includegraphics[width = \textwidth]{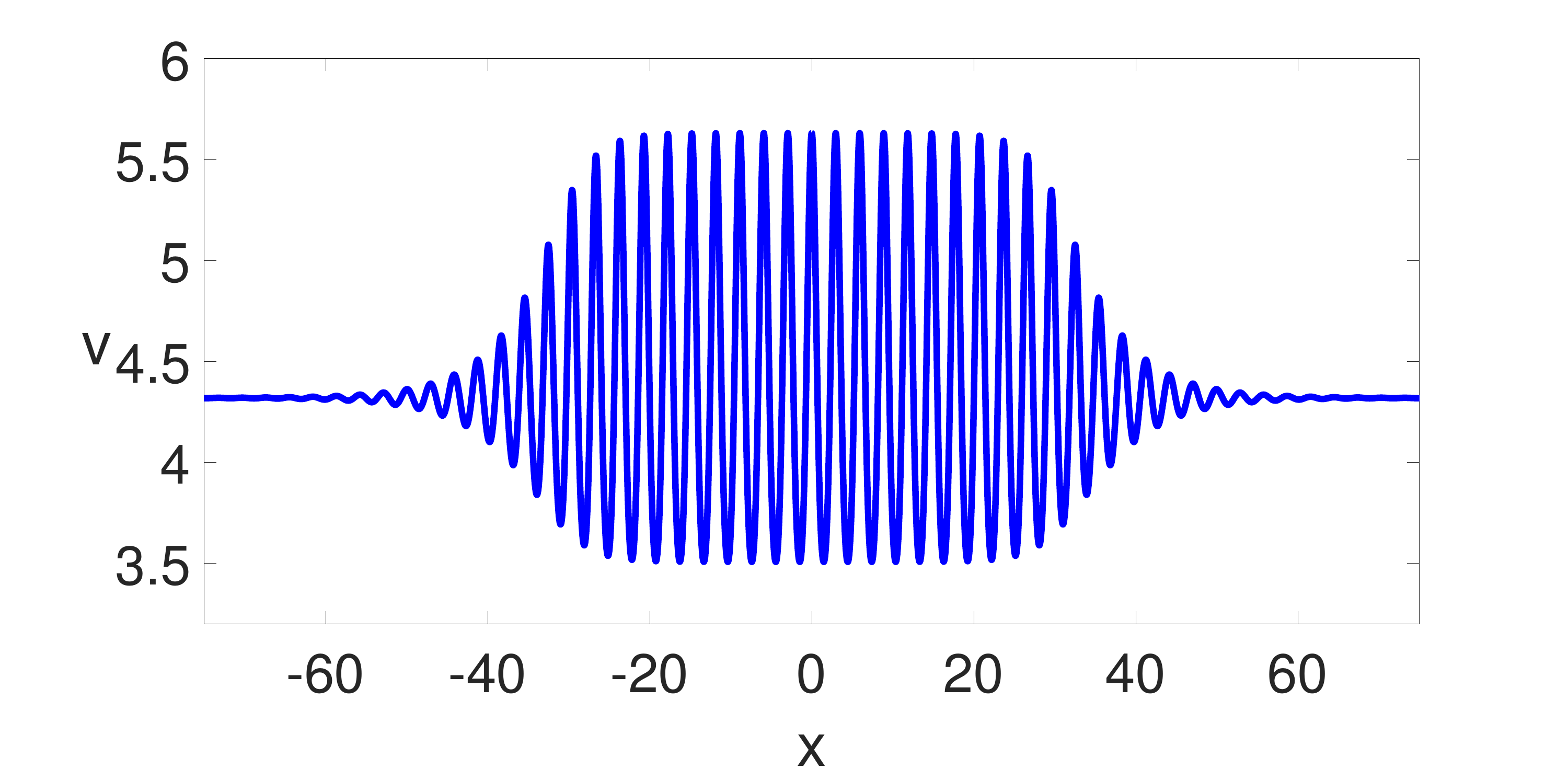}
                \caption{}
            \end{subfigure}
            \\
            \begin{subfigure}[b]{0.49\textwidth}
                \centering
                \includegraphics[width = \textwidth]{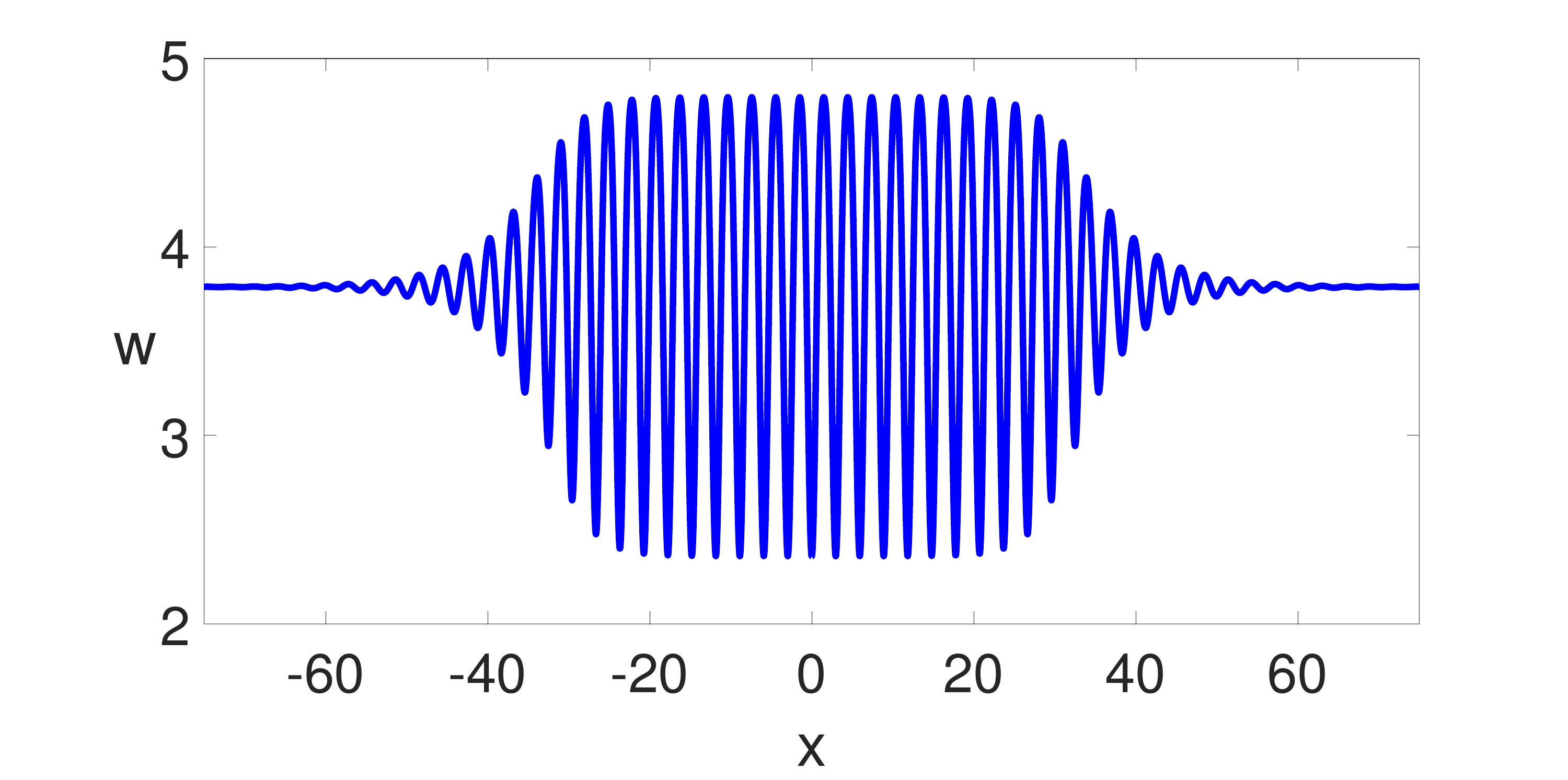}
                \caption{}
            \end{subfigure}
            \hfill
            \begin{subfigure}[b]{0.49\textwidth}
                \centering
                \includegraphics[width = \textwidth]{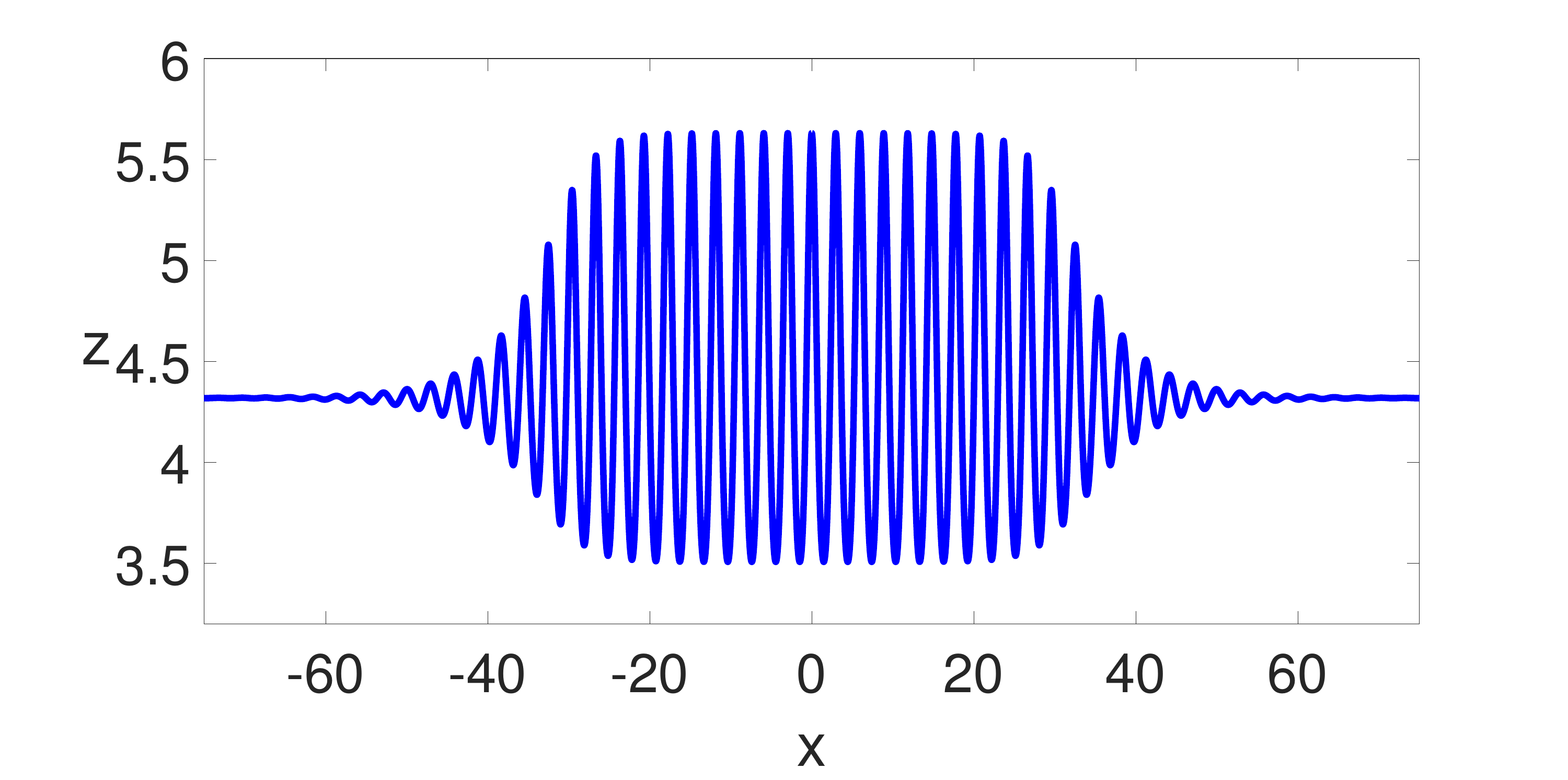}
                \caption{}
            \end{subfigure}             
            \caption{Similar to Figure \ref{fig:matchingSH23} but for model \eqref{Bru4}. The snaking in panel (b) (respectively, the solution in panels (c), (d), (e), and (f)) corresponds to $\varepsilon \approx 0.88773$ (respectively, $b \approx 16.35899033$).}
            \label{fig:matching-bru4}
        \end{figure}

    \begin{remark}
        We close this section by highlighting that it is hard to work with small values of $\varepsilon$ in numerical computations. We omit the details, but using suitable tolerances in AUTO, we had to take extremely long domains to find reliable estimates for the width of the homoclinic snaking when it becomes smaller than $\mathcal O(10^{- 7})$. Given the exponentially small asymptotic estimates of the theory, in practice, this means that the smallest values of $\varepsilon$ for which we can demonstrate the agreement of the numerics with the theory is only about 0.4. Nevertheless, we note that the homoclinic snaking occurs, generically, at a distance of order $\varepsilon^2$ or $\varepsilon^4$ from the codimension-two Turing bifurcation line, which allowed us to demonstrate good agreement between our analysis and numerics by computing how the width of the snake scales with $\varepsilon$.
    \end{remark}

\section{Discussion} \label{sec:discussion}
    In this paper, we have achieved what we set out to do. We have successfully generalized the theory of exponential asymptotics used to study homoclinic snaking close to a codimension-two Turing bifurcation in the Swift-Hohenberg equation \cite{Chapman,Dean} to arbitrary systems of reaction-diffusion equations undergoing the same instability. Accompanying codes to do all the calculations automatically in Python and Mathematica for any reaction-diffusion systems have also been provided \cite{beyond-code}, and the codes were used to do the calculations for each of the examples shown in this paper. Further study will include the use of these codes to study homoclinic snaking in models where it has been found, but not much has been said about the region in the parameter space where such a phenomenon can be found. Although this calculation is extensive, the purpose of this article is to provide tools for everyone to be able to understand it easily and use it as needed. 

    We highlight that some details in this paper have not been treated, as it is already quite lengthy. For example, we have not explained how to use information on the phase of the exponentially small estimate in order to study the fine details of the homoclinic snaking. Furthermore, with additional symmetries or conserved quantities, there can be additional branches. Moreover, it is well known that examples such as the Swift-Hohenberg equation, which have variational structures, give rise to `ladder' asymmetric stationary localised patterns connecting the two interleaving snaking branches.

    Also, we have not discussed the temporal stability of patterns. Under certain conditions on the matrix $\mathbb M$ and forms of $\mbf f$, much can be said at the general level, without needing to resort to model-specific calculations. Details are left for future work. 

    A more challenging open question is to consider analogous structures that arise near Turing bifurcations corresponding to the dispersion curve having a non-zero imaginary part. Similar codimension-two bifurcation points occur for such finite wavenumber Hopf bifurcations (also known as wave bifurcations), but are more complex due to the presence of both standing and travelling waves. Nevertheless, in earlier work \cite{villar-sepulveda-wave}, the present author has computed the regular asymptotic expansion to produce the amplitude equations for wave bifurcations up to order five. The same approach as that presented here is expected to be generalisable, but a general calculation is expected to be yet more cumbersome.

    Another generalisation under consideration is to study so-called slanted snaking, in the presence of a zero wavenumber mode in addition to the Turing instability \cite{Knobloch2}. Particular classes of systems in higher spatial dimensions would also be interesting to study, although complete generalisations there seem a long way off.

    Furthermore, as the theory developed here is similar to the WKB theory since they are both based on asymptotics \cite{WKB-theory}, and the latter has been used to study Turing bifurcations and localized patterns in heterogeneous reaction-diffusion equations (see e.g.~\cite{our-WKB}), a generalization of this approach to study localized solutions in heterogeneous systems is also under consideration.

    One last thing to highlight is that the process carried out in this article is a formal derivation of the region in the parameter space where homoclinic snaking can be found close to codimension-two Turing bifurcation points. It would be nice to develop a more rigorous approach to deal with this kind of problem, as often such approaches follow from more formal asymptotic procedures, like the ones used here.

\section*{Acknowledgements}
    I want to thank Alan R. Champneys for his incredible support and insistence on this project. I also want to acknowledge Andrew Dean, Jon Chapman, Gregory Kozyreff, Philippe Trinh, and Hannes de Witt for their insightful input at different stages of this project. Moreover, I would like to thank Bastián Jiménez-Sepúlveda for his invaluable support in optimizing the codes developed to make the intensive computations efficient. This project was funded by ANID, Beca Chile Doctorado en el extranjero, number 72210071.

\newpage

\bibliographystyle{plain}
\bibliography{refs}

@article{Knobloch2,
title={Localized
structures and front propagation in systems with a conservation law},
author={Knobloch, E.},
journal={IMA J Appl Math},
volume={81},
pages={457-487},
year={2016} }

@misc{auto,
	author = {Doedel, E.J. and {Oldeman {\em et al.}}, B. },
	note={Latest version at
	{\tt https://github.com/auto-07p}},
	title = {AUTO-07P: Continuation And Bifurcation Software For Ordinary Differential Equations},
	year = {2020}
}

@article{Knobloch1,
author={Knobloch, E.},
title={Spatial Localization in Dissipative Systems},
journal={Annual Review of Condensed Matter Physics},
volume={6},
pages={325-359},
year={2015} }

@book{ioossbook,
author    = {Haragus, M. and Iooss, G.}, 
title    = {Local Bifurcations, Center Manifolds, and Normal Forms in Infinite-Dimensional Dynamical Systems},
publisher = {Springer},
year      = {2011}
}

@article{snakebeck,
author = {Beck, M. and Knobloch, J. and Lloyd, D.J.B. and Sandstede, B. and 
Wagenknecht, T.},
title = {Snakes, Ladders, and Isolas of Localized Patterns},
journal = {SIAM Journal on Mathematical Analysis},
volume = {41},
pages = {936-972},
year = {2009} }

@article{Kozyreff, 
title={Asymptotics of Large Bound States of Localized Structures},
author={Kozyreff, G. and Chapman, S.J.},
journal={Physical Review Letters}, 
volume={97}, 
pages={044502}, 
year={2006} }

@article{Woods,
author = "Woods, P.D. and Champneys, A.R.",
title = "Heteroclinic tangles and homoclinic snaking in the unfolding of a degenerate reversible {H}amiltonian--{H}opf bifurcation",
journal = "Physica D: Nonlinear Phenomena",
volume = "129",
pages = "147-170",
year = "1999" }

@article{Woods_Intro,
 author={Champneys, A.R.},
 title={Editorial to {Homoclinic snaking at 21: In memory of Patrick Woods}},
 journal={IMA J Appl Math},
 year={2021},
 volume={86},
 pages={845-855}
 }

@article{krause_review,
  title={Modern perspectives on near-equilibrium analysis of {T}uring systems},
  author={Krause, A.L. and Gaffney, E.A. and Maini, P.K. and Klika, V.},
  journal={Philosophical Transactions of the Royal Society A},
  volume={379},
  number={2213},
  pages={20200268},
  year={2021},
  publisher={The Royal Society}
}

@article{dawes,
title = "After 1952: The later development of {Alan Turing's} ideas on the mathematics of pattern formation",
journal = "Historia Mathematica",
volume = "43",
number = "1",
pages = "49--64",
year = "2016",
issn = "0315-0860",
author = "J.H.P. Dawes"}

@article{FahadWoods,
  title={Localized patterns and semi-strong interaction, a unifying framework for reaction--diffusion systems},
  author={{Al Saadi}, F. and Champneys, A. and Verschueren, N.},
  journal={IMA Journal of Applied Mathematics},
  volume={86},
  number={5},
  pages={1031--1065},
  year={2021},
  publisher={Oxford University Press}
}

@article{Brusselator,
  title={Turing instability of Brusselator in the reaction-diffusion network},
  author={Ji, Yansu and Shen, Jianwei},
  journal={Complexity},
  volume={2020},
  year={2020},
  publisher={Hindawi}
}

@article{Chapman,
    title = {Exponential asymptotics of localised patterns and snaking bifurcation diagrams},
    journal = {Physica D: Nonlinear Phenomena},
    volume = {238},
    number = {3},
    pages = {319-354},
    year = {2009},
    issn = {0167-2789},
    doi = {https://doi.org/10.1016/j.physd.2008.10.005},
    url = {https://www.sciencedirect.com/science/article/pii/S0167278908003679},
    author = {Chapman, S.J. and Kozyreff, G.}
}

@inproceedings{knobloch,
  title={Normal form for spatial dynamics in the {S}wift-{H}ohenberg equation},
  author={Burke, J. and Knobloch, E.},
  booktitle={Conference Publications},
  volume={2007},
  pages={170},
  year={2007},
  organization={American Institute of Mathematical Sciences}
}

@article{HannesdeWitt,
  title={Beyond all order asymptotics for homoclinic snaking in a {Schnakenberg} system},
  author={{de Witt}, H.},
  journal={Nonlinearity},
  volume={32},
  number={7},
  pages={2667},
  year={2019},
  publisher={IOP Publishing}
}

@article{Dean,
  title={Exponential asymptotics of homoclinic snaking},
  author={Dean, A. and Matthews, P.C. and Cox, S.M. and King, J.R.},
  journal={Nonlinearity},
  volume={24},
  number={12},
  pages={3323},
  year={2011},
  publisher={IOP Publishing}
}

@article{turing,
	author = {Turing, A.M.},
	title = {The Chemical Basis of Morphogenesis},
	volume = {237},
	pages = {37-72},
	year = {1952},
	journal = {Philosophical Transactions of the Royal Society of London B: Biological Sciences}
}

@article{villar2023degenerate,
  title={Degenerate {T}uring bifurcation and the birth of localized patterns in activator-inhibitor systems},
  author={Villar-Sep{\'u}lveda, E. and Champneys, A.},
  journal={SIAM Journal on Applied Dynamical Systems},
  volume={22},
  number={3},
  pages={1673--1709},
  year={2023},
  publisher={SIAM}
}

@article{bramburger-review,
  title={Localized multi-dimensional patterns},
  author={Bramburger, J.J and Hill, D.J. and Lloyd, D.J.B.},
  journal={arXiv preprint arXiv:2404.14987},
  year={2024}
}

@article{burke2007homoclinic,
  title={Homoclinic snaking: structure and stability},
  author={Burke, J. and Knobloch, E.},
  journal={Chaos: An Interdisciplinary Journal of Nonlinear Science},
  volume={17},
  number={3},
  year={2007},
  publisher={AIP Publishing}
}

@article{TOMS,
  title={Computation of {Turing} Bifurcation Normal Form for n-Component Reaction-Diffusion Systems},
  author={Villar-Sep{\'u}lveda, Edgardo and Champneys, Alan},
  journal={ACM Transactions on Mathematical Software},
  volume={49},
  number={4},
  pages={1--24},
  year={2023},
  publisher={ACM New York, NY}
}

@article{villar-sepulveda-wave,
  title={Amplitude equations for wave bifurcations in reaction--diffusion systems},
  author={Villar-Sep{\'u}lveda, Edgardo and Champneys, Alan},
  journal={Nonlinearity},
  volume={37},
  number={8},
  pages={085012},
  year={2024},
  publisher={IOP Publishing}
}

@article{AlastairBru4,
  title={Spatiotemporal chaos and quasipatterns in coupled reaction--diffusion systems},
  author={Castelino, Jennifer K and Ratliff, Daniel J and Rucklidge, Alastair M and Subramanian, Priya and Topaz, Chad M},
  journal={Physica D: Nonlinear Phenomena},
  volume={409},
  pages={132475},
  year={2020},
  publisher={Elsevier}
}

@misc{beyond-code,
author = {Villar-Sepúlveda, E.},
title = {{Exponential asymptotics for homoclinic snaking}},
note={\url{https://github.com/edgardeitor/Exponential-asymptotics}
} }

@article{peletier,
  title={Pattern selection of solutions of the {Swift--Hohenberg} equation},
  author={Peletier, Lambertus A and Rottsch{\"a}fer, Vivi},
  journal={Physica D: Nonlinear Phenomena},
  volume={194},
  number={1-2},
  pages={95--126},
  year={2004},
  publisher={Elsevier}
}

@article{snakes-ladders-knobloch,
  title={Snakes and ladders: localized states in the {S}wift--{H}ohenberg equation},
  author={Burke, John and Knobloch, Edgar},
  journal={Physics Letters A},
  volume={360},
  number={6},
  pages={681--688},
  year={2007},
  publisher={Elsevier}
}

@misc{our-WKB,
    title={{Pattern Localisation in the Swift-Hohenberg Equation via Slowly Varying Spatial Heterogeneity}}, 
    author={Andrew L. Krause and Václav Klika and Edgardo Villar-Sepúlveda and Alan R. Champneys and Eamonn A. Gaffney},
    year={2025},
    eprint={2409.13043},
    archivePrefix={arXiv},
    primaryClass={nlin.PS},
    url={https://arxiv.org/abs/2409.13043}, 
}

@article{WKB-theory,
  title={New {S}tokes’ line in {WKB} theory},
  author={Berk, HL and Nevins, William McCay and Roberts, KV},
  journal={Journal of Mathematical Physics},
  volume={23},
  number={6},
  pages={988--1002},
  year={1982},
  publisher={American Institute of Physics}
}

\newpage 

\appendix
    \section{Expansion up to order 5} \label{ap:order 5}
        In this section, we continue the asymptotic expansion of system \eqref{geneq} up until order 5 to obtain the amplitude equation at that order.
        
        \paragraph{Order $\mathcal O\left(\varepsilon^4\right)$.} At this order, equation \eqref{eqtoexpand} becomes
    	\begin{align*}
    	    \mbf 0 &= \jac \mbf f_{0, 0}(\mbf 0) \, \mbf u^{[4]} + a_1 \, \jac \mbf f_{1, 0}(\mbf 0) \, \mbf u^{[3]} + b_1 \, \jac \mbf f_{0, 1}(\mbf 0) \, \mbf u^{[3]} + a_2 \, \jac \mbf f_{1, 0}(\mbf 0) \, \mbf u^{[2]} + b_2 \, \jac \mbf f_{0, 1}(\mbf 0) \, \mbf u^{[2]}
            \\
            & \quad + a_3 \, \jac \mbf f_{1, 0}(\mbf 0) \, \mbf u^{[1]} + b_3 \, \jac \mbf f_{0, 1}(\mbf 0) \, \mbf u^{[1]} + a_1^2 \, \jac \mbf f_{2, 0}(\mbf 0) \, \mbf u^{[2]} + a_1 \, b_1 \, \jac \mbf f_{1, 1}(\mbf 0) \, \mbf u^{[2]}
            \\
            & \quad + b_1^2 \, \jac \mbf f_{0, 2}(\mbf 0) \, \mbf u^{[2]} + 2 \, a_1 \, a_2 \, \jac \mbf f_{2, 0}(\mbf 0) \, \mbf u^{[1]} + \left(a_1 \, b_2 + a_2 \, b_1 \right) \, \jac \mbf f_{1, 1}(\mbf 0) \, \mbf u^{[1]}
            \\
            & \quad + 2 \, b_1 \, b_2 \, \jac \mbf f_{0, 2}(\mbf 0) \, \mbf u^{[1]} + a_1^3 \, \jac \mbf f_{3, 0}(\mbf 0) \, \mbf u^{[1]} + a_1^2 \, b_1 \, \jac \mbf f_{2, 1}(\mbf 0) \, \mbf u^{[1]} + a_1 \, b_1^2 \, \jac \mbf f_{1, 2}(\mbf 0) \, \mbf u^{[1]}
            \\
            &\quad + b_1^3 \, \jac \mbf f_{0, 3}(\mbf 0) \, \mbf u^{[1]} + 2 \, \mbf F_{2, 0, 0}\left(\mbf u^{[1]}, \mbf u^{[3]}\right) + \mbf F_{2, 0, 0} \left(\mbf u^{[2]}, \mbf u^{[2]}\right) + 2 \, a_1 \, \mbf F_{2, 1, 0} \left(\mbf u^{[1]}, \mbf u^{[2]}\right)
            \\
            &\quad + 2 \, b_1 \, \mbf F_{2, 0, 1} \left(\mbf u^{[1]}, \mbf u^{[2]}\right) + a_2 \, \mbf F_{2, 1, 0} \left(\mbf u^{[1]}, \mbf u^{[1]}\right) + b_2 \, \mbf F_{2, 0, 1} \left(\mbf u^{[1]}, \mbf u^{[1]}\right)
            \\
            & \quad + a_1^2 \, \mbf F_{2, 2, 0} \left(\mbf u^{[1]}, \mbf u^{[1]}\right) + a_1 \, b_1 \, \mbf F_{2, 1, 1} \left(\mbf u^{[1]}, \mbf u^{[1]}\right) + b_1^2 \, \mbf F_{2, 0, 2} \left(\mbf u^{[1]}, \mbf u^{[1]}\right)
            \\
            & \quad + 3 \, \mbf F_{3, 0, 0} \left(\mbf u^{[1]}, \mbf u^{[1]}, \mbf u^{[2]}\right) + a_1 \, \mbf F_{3, 1, 0}\left(\mbf u^{[1]}, \mbf u^{[1]}, \mbf u^{[1]}\right) + b_1 \, \mbf F_{3, 0, 1}\left(\mbf u^{[1]}, \mbf u^{[1]}, \mbf u^{[1]}\right)
            \\
            & \quad + \mbf F_{4, 0, 0}\left(\mbf u^{[1]}, \mbf u^{[1]}, \mbf u^{[1]}, \mbf u^{[1]}\right) + k^2 \, \hat D \, \left(\frac{\partial^2 \mbf u^{[4]}}{\partial x^2} + 2 \, \frac{\partial^2 \mbf u^{[2]}}{\partial x \partial X}\right).
    	\end{align*}
    	Therefore,
        \begin{align*}
            \mbf u^{[4]} &= \abs{A_1}^2 \, \mbf W_0^{[4]} + \abs{A_1}^4 \, \mbf W_{0, 2}^{[4]} + i \, \bar A_1 \, A_1' \, \mbf W_{0, 3}^{[4]} + A_1 \, e^{ix} \, \mbf W_1^{[4]} + \abs{A_1}^2 \, A_1 \, e^{ix} \, \mbf W_{1, 2}^{[4]}
            \\
            & \quad + i \, A_1' \, e^{ix} \, \mbf W_{1, 3}^{[4]} + A_1^2 \, e^{2ix} \, \mbf W_2^{[4]} + \abs{A_1}^2 \, A_1^2 \, e^{2ix} \, \mbf W_{2, 2}^{[4]} + i \, A_1 \, A_1' \, e^{2ix} \, \mbf W_{2, 3}^{[4]} + A_1^3 \, e^{3ix} \, \mbf W_3^{[4]}
            \\
            & \quad + A_1^4 \, e^{4ix} \, \mbf W_4^{[4]} + 2 \, A_1 \, \bar A_2 \, \mbf W_0^{[3]} + 2 \, A_1 \, \bar A_3 \, \mbf W_0^{[2]} + \abs{A_2}^2 \, \mbf W_0^{[2]} + A_2 \, e^{ix} \, \mbf W_1^{[3]}
            \\
            & \quad + A_1^2 \, \bar A_2 \, e^{ix} \, \mbf W_{1, 2}^{[3]} + 2 \, \abs{A_1}^2 \, A_2 \, e^{ix} \, \mbf W_{1, 2}^{[3]} + i \, A_2' \, e^{ix} \, \mbf W_{1, 3}^{[3]} + A_3 \, e^{ix} \, \mbf W_1^{[2]} + A_4 \, e^{ix} \, \bs \phi_1^{[1]}
            \\
            & \quad + 2 \, A_1 \, A_2 \, e^{2ix} \, \mbf W_2^{[3]} + 2 \, A_1 \, A_3 \, e^{2ix} \, \mbf W_2^{[2]} + A_2^2 \, e^{2ix} \, \mbf W_2^{[2]} + 3 \, A_1^2 \, A_2 \, e^{3ix} \, \mbf W_3^{[3]} + c.c.
        \end{align*}
    	where $A_4 = A_4(X)$,
    	\begin{align*}
    	    \mathcal M_0 \, \mbf W_0^{[4]} &= - a_1 \, \jac \mbf f_{1, 0}(\mbf 0) \, \mbf W_0^{[3]} - b_1 \, \jac \mbf f_{0, 1}(\mbf 0) \, \mbf W_0^{[3]} - a_2 \, \jac \mbf f_{1, 0}(\mbf 0) \, \mbf W_0^{[2]} - b_2 \, \jac \mbf f_{0, 1}(\mbf 0) \, \mbf W_0^{[2]}
            \\
            & \quad - a_1^2 \, \jac \mbf f_{2, 0}(\mbf 0) \, \mbf W_0^{[2]} - a_1 \, b_1 \, \jac \mbf f_{1, 1}(\mbf 0) \, \mbf W_0^{[2]} - b_1^2 \, \jac \mbf f_{0, 2}(\mbf 0) \, \mbf W_0^{[2]}
            \\
            & \quad - 2 \, \mbf F_{2, 0, 0}\left(\bs \phi_1^{[1]}, \mbf W_1^{[3]}\right) - \mbf F_{2, 0, 0} \left(\mbf W_1^{[2]}, \mbf W_1^{[2]}\right) - 2 \, a_1 \, \mbf F_{2, 1, 0} \left(\bs \phi_1^{[1]}, \mbf W_1^{[2]}\right)
    	    \\
            & \quad - 2 \, b_1 \, \mbf F_{2, 0, 1} \left(\bs \phi_1^{[1]}, \mbf W_1^{[2]}\right) - a_2 \, \mbf F_{2, 1, 0} \left(\bs \phi_1^{[1]}, \bs \phi_1^{[1]}\right) - b_2 \, \mbf F_{2, 0, 1} \left(\bs \phi_1^{[1]}, \bs \phi_1^{[1]}\right)
            \\
            & \quad - a_1^2 \, \mbf F_{2, 2, 0} \left(\bs \phi_1^{[1]},\bs \phi_1^{[1]}\right) - a_1 \, b_1 \, \mbf F_{2, 1, 1} \left(\bs \phi_1^{[1]}, \bs \phi_1^{[1]}\right) - b_1^2 \, \mbf F_{2, 0, 2} \left(\bs \phi_1^{[1]}, \bs \phi_1^{[1]}\right),
    	\end{align*}
        \begin{align*}
    	    \mathcal M_0 \, \mbf W_{0, 2}^{[4]} &= - 2 \, \mbf F_{2, 0, 0}\left(\bs \phi_1^{[1]}, \mbf W_{1,2}^{[3]}\right) - 2 \, \mbf F_{2, 0, 0} \left(\mbf W_0^{[2]}, \mbf W_0^{[2]}\right) - \mbf F_{2, 0, 0} \left(\mbf W_2^{[2]}, \mbf W_2^{[2]}\right)
    	    \\
    	    & \quad - 3 \, \mbf F_{3, 0, 0} \left(\bs \phi_1^{[1]}, \bs \phi_1^{[1]}, 2 \, \mbf W_0^{[2]} + \mbf W_2^{[2]}\right) - 3 \, \mbf F_{4, 0, 0}\left(\bs \phi_1^{[1]}, \bs \phi_1^{[1]}, \bs \phi_1^{[1]}, \bs \phi_1^{[1]}\right),
    	\end{align*}
    	\begin{align*}
    	    \mathcal M_0 \, \mbf W_{0, 3}^{[4]} = - 2 \, \mbf F_{2, 0, 0}\left(\bs \phi_1^{[1]}, \mbf W_{1,3}^{[3]}\right),
    	\end{align*}
    	\begin{align}
    	    \mathcal M_1 \, \mbf W_1^{[4]} &= - a_1 \, \jac \mbf f_{1, 0}(\mbf 0) \, \mbf W_1^{[3]} - b_1 \, \jac \mbf f_{0, 1}(\mbf 0) \, \mbf W_1^{[3]} - a_2 \, \jac \mbf f_{1, 0}(\mbf 0) \, \mbf W_1^{[2]} \notag
            \\
            & \quad - b_2 \, \jac \mbf f_{0, 1}(\mbf 0) \, \mbf W_1^{[2]} - a_3 \, \jac \mbf f_{1, 0}(\mbf 0) \, \bs \phi_1^{[1]} - b_3 \, \jac \mbf f_{0, 1}(\mbf 0) \, \bs \phi_1^{[1]} - a_1^2 \, \jac \mbf f_{2, 0}(\mbf 0) \, \mbf W_1^{[2]} \notag
    	    \\
    	    & \quad - a_1 \, b_1 \, \jac \mbf f_{1, 1}(\mbf 0) \, \mbf W_1^{[2]} - b_1^2 \, \jac \mbf f_{0, 2}(\mbf 0) \, \mbf W_1^{[2]} - 2 \, a_1 \, a_2 \, \jac \mbf f_{2, 0}(\mbf 0) \, \bs \phi_1^{[1]} \label{solvcond14}
            \\
            & \quad - \left(a_1 \, b_2 + a_2 \, b_1\right) \, \jac \mbf f_{1, 1}(\mbf 0) \, \bs \phi_1^{[1]} - 2 \, b_1 \, b_2 \, \jac \mbf f_{0, 2}(\mbf 0) \, \bs \phi_1^{[1]} - a_1^3 \, \jac \mbf f_{3, 0}(\mbf 0) \, \bs \phi_1^{[1]} \notag
            \\
            & \quad - a_1^2 \, b_1 \, \jac \mbf f_{2, 1}(\mbf 0) \, \bs \phi_1^{[1]} - a_1 \, b_1^2 \, \jac \mbf f_{1, 2}(\mbf 0) \, \bs \phi_1^{[1]} - b_1^3 \, \jac \mbf f_{0, 3}(\mbf 0) \, \bs \phi_1^{[1]}, \notag
    	\end{align}
        \begin{equation}
            \begin{aligned}
        	    \mathcal M_1 \, \mbf W_{1, 2}^{[4]} &= - a_1 \, \jac \mbf f_{1, 0}(\mbf 0) \, \mbf W_{1, 2}^{[3]} - b_1 \, \jac \mbf f_{0, 1}(\mbf 0) \, \mbf W_{1, 2}^{[3]} - 2 \, \mbf F_{2, 0, 0}\left(\bs \phi_1^{[1]}, 2 \, \mbf W_0^{[3]} + \mbf W_2^{[3]}\right)
        	    \\
        	    &\quad - 2 \, \mbf F_{2, 0, 0} \left(\mbf W_1^{[2]}, 2 \, \mbf W_0^{[2]} + \mbf W_2^{[2]}\right) - 2 \, a_1 \, \mbf F_{2, 1, 0} \left(\bs \phi_1^{[1]}, 2 \, \mbf W_0^{[2]} + \mbf W_2^{[2]}\right)
        	    \\
        	    &\quad - 2 \, b_1 \, \mbf F_{2, 0, 1} \left(\bs \phi_1^{[1]}, 2 \, \mbf W_0^{[2]} + \mbf W_2^{[2]}\right) - 9 \, \mbf F_{3, 0, 0} \left(\bs \phi_1^{[1]}, \bs \phi_1^{[1]}, \mbf W_1^{[2]}\right)
                \\
                &\quad- 3 \, a_1 \, \mbf F_{3, 1, 0}\left(\bs \phi_1^{[1]}, \bs \phi_1^{[1]}, \bs \phi_1^{[1]}\right) - 3 \, b_1 \, \mbf F_{3, 0, 1}\left(\bs \phi_1^{[1]}, \bs \phi_1^{[1]}, \bs \phi_1^{[1]}\right),
            \end{aligned} \label{solvcond24}
        \end{equation}
    	\begin{align}
    	    \mathcal M_1 \, \mbf W_{1, 3}^{[4]} = - a_1 \, \jac \mbf f_{1, 0}(\mbf 0) \, \mbf W_{1, 3}^{[3]} - b_1 \, \jac \mbf f_{0, 1}(\mbf 0) \, \mbf W_{1, 3}^{[3]} - 2 \, k^2 \, \hat D \, \mbf W_1^{[2]}, \label{solvcond34}
    	\end{align}
    	\begin{align*}
    	    \mathcal M_2 \, \mbf W_2^{[4]} &= - a_1 \, \jac \mbf f_{1, 0}(\mbf 0) \, \mbf W_2^{[3]} - b_1 \, \jac \mbf f_{0, 1}(\mbf 0) \, \mbf W_2^{[3]} - a_2 \, \jac \mbf f_{1, 0}(\mbf 0) \, \mbf W_2^{[2]} - b_2 \, \jac \mbf f_{0, 1}(\mbf 0) \, \mbf W_2^{[2]}
            \\
            & \quad - a_1^2 \, \jac \mbf f_{2, 0}(\mbf 0) \, \mbf W_2^{[2]} - a_1 \, b_1 \, \jac \mbf f_{1, 1}(\mbf 0) \, \mbf W_2^{[2]} - b_1^2 \, \jac \mbf f_{0, 2}(\mbf 0) \, \mbf W_2^{[2]}
            \\
            & \quad - 2 \, \mbf F_{2, 0, 0}\left(\bs \phi_1^{[1]}, \mbf W_1^{[3]}\right) - \mbf F_{2, 0, 0} \left(\mbf W_1^{[2]}, \mbf W_1^{[2]}\right) - 2 \, a_1 \, \mbf F_{2, 1, 0} \left(\bs \phi_1^{[1]}, \mbf W_1^{[2]}\right)
            \\
            & \quad - 2 \, b_1 \, \mbf F_{2, 0, 1} \left(\bs \phi_1^{[1]}, \mbf W_1^{[2]}\right) - a_2 \, \mbf F_{2, 1, 0} \left(\bs \phi_1^{[1]}, \bs \phi_1^{[1]}\right) - b_2 \, \mbf F_{2, 0, 1} \left(\bs \phi_1^{[1]}, \bs \phi_1^{[1]}\right)
            \\
            & \quad - a_1^2 \, \mbf F_{2, 2, 0} \left(\bs \phi_1^{[1]}, \bs \phi_1^{[1]}\right) - a_1 \, b_1 \, \mbf F_{2, 1, 1} \left(\bs \phi_1^{[1]}, \bs \phi_1^{[1]}\right) - b_1^2 \, \mbf F_{2, 0, 2} \left(\bs \phi_1^{[1]}, \bs \phi_1^{[1]}\right),
    	\end{align*}
    	\begin{align*}
    	    \mathcal M_2 \, \mbf W_{2, 2}^{[4]} &= - 2 \, \mbf F_{2, 0, 0}\left(\bs \phi_1^{[1]}, \mbf W_{1,2}^{[3]} + \mbf W_3^{[3]}\right) - 4 \, \mbf F_{2, 0, 0} \left(\mbf W_0^{[2]}, \mbf W_2^{[2]}\right)
    	    \\
    	    & \quad - 6 \, \mbf F_{3, 0, 0} \left(\bs \phi_1^{[1]}, \bs \phi_1^{[1]}, \mbf W_0^{[2]} + \mbf W_2^{[2]}\right) - 4 \, \mbf F_{4, 0, 0}\left(\bs \phi_1^{[1]}, \bs \phi_1^{[1]}, \bs \phi_1^{[1]}, \bs \phi_1^{[1]}\right),
    	\end{align*}
    	\begin{align*}
    	    \mathcal M_2 \, \mbf W_{2, 3}^{[4]} = - 2 \, \mbf F_{2, 0, 0}\left(\bs \phi_1^{[1]}, \mbf W_{1, 3}^{[3]}\right) - 8 \, k^2 \, \hat D \, \mbf W_2^{[2]},
    	\end{align*}
        \begin{align*}
            \mathcal M_3 \, \mbf W_3^{[4]} &= - a_1 \, \jac \mbf f_{1, 0}(\mbf 0) \, \mbf W_3^{[3]} - b_1 \, \jac \mbf f_{0, 1}(\mbf 0) \, \mbf W_3^{[3]} - 2 \, \mbf F_{2, 0, 0}\left(\bs \phi_1^{[1]}, \mbf W_2^{[3]}\right)
            \\
            & \quad - 2 \, \mbf F_{2,0,0} \left(\mbf W_1^{[2]}, \mbf W_2^{[2]}\right) - 2 \, a_1 \, \mbf F_{2, 1, 0}\left(\bs \phi_1^{[1]}, \mbf W_2^{[2]}\right) - 2 \, b_1 \, \mbf F_{2, 0, 1}\left(\bs \phi_1^{[1]}, \mbf W_2^{[2]}\right)
            \\
            & \quad - 3 \, \mbf F_{3, 0, 0}\left(\bs \phi_1^{[1]}, \bs \phi_1^{[1]}, \mbf W_1^{[2]}\right) - a_1 \, \mbf F_{3, 1, 0}\left(\bs \phi_1^{[1]}, \bs \phi_1^{[1]}, \bs \phi_1^{[1]}\right) - b_1 \, \mbf F_{3, 0, 1}\left(\bs \phi_1^{[1]}, \bs \phi_1^{[1]}, \bs \phi_1^{[1]}\right),
    	\end{align*}
    	and
    	\begin{align*}
    	    \mathcal M_4 \, \mbf W_4^{[4]} &= - 2 \, \mbf F_{2, 0, 0}\left(\bs \phi_1^{[1]}, \mbf W_3^{[3]}\right) - \mbf F_{2, 0, 0} \left(\mbf W_2^{[2]}, \mbf W_2^{[2]}\right)
    	    \\
    	    & \quad - 3 \, \mbf F_{3, 0, 0} \left(\bs \phi_1^{[1]}, \bs \phi_1^{[1]}, \mbf W_2^{[2]}\right) - \mbf F_{4,0,0}\left(\bs \phi_1^{[1]},\bs \phi_1^{[1]},\bs \phi_1^{[1]},\bs \phi_1^{[1]}\right).
    	\end{align*}
    	Here, we need to ensure that equations \eqref{solvcond14}, \eqref{solvcond24} and \eqref{solvcond34} have a solution so we require
        \begin{align*}
    	    \left \langle a_1 \, \jac \mbf f_{1, 0}(\mbf 0) \, \mbf W_1^{[3]} + b_1 \, \jac \mbf f_{0, 1}(\mbf 0) \, \mbf W_1^{[3]} + a_2 \, \jac \mbf f_{1, 0}(\mbf 0) \, \mbf W_1^{[2]} + b_2 \, \jac \mbf f_{0, 1}(\mbf 0) \, \mbf W_1^{[2]} \right. \notag
            \\
            + a_3 \, \jac \mbf f_{1, 0}(\mbf 0) \, \bs \phi_1^{[1]} + b_3 \, \jac \mbf f_{0, 1}(\mbf 0) \, \bs \phi_1^{[1]} + a_1^2 \, \jac \mbf f_{2, 0}(\mbf 0) \, \mbf W_1^{[2]} + a_1 \, b_1 \, \jac \mbf f_{1, 1}(\mbf 0) \, \mbf W_1^{[2]} \notag
    	    \\
    	    + b_1^2 \, \jac \mbf f_{0, 2}(\mbf 0) \, \mbf W_1^{[2]} + 2 \, a_1 \, a_2 \, \jac \mbf f_{2, 0}(\mbf 0) \, \bs \phi_1^{[1]} + \left(a_1 \, b_2 + a_2 \, b_1\right) \, \jac \mbf f_{1, 1}(\mbf 0) \, \bs \phi_1^{[1]}
            \\
            + 2 \, b_1 \, b_2 \, \jac \mbf f_{0, 2}(\mbf 0) \, \bs \phi_1^{[1]} + a_1^3 \, \jac \mbf f_{3, 0}(\mbf 0) \, \bs \phi_1^{[1]} + a_1^2 \, b_1 \, \jac \mbf f_{2, 1}(\mbf 0) \, \bs \phi_1^{[1]} + a_1 \, b_1^2 \, \jac \mbf f_{1, 2}(\mbf 0) \, \bs \phi_1^{[1]} \notag
            \\
            \left. + b_1^3 \, \jac \mbf f_{0, 3}(\mbf 0) \, \bs \phi_1^{[1]}, \bs \psi\right \rangle = 0, \notag
    	\end{align*}
    	\begin{align}
    	    \left \langle a_1 \, \jac \mbf f_{1, 0}(\mbf 0) \, \mbf W_{1, 2}^{[3]} + b_1 \, \jac \mbf f_{0, 1}(\mbf 0) \, \mbf W_{1, 2}^{[3]} + 2 \, \mbf F_{2, 0, 0}\left(\bs \phi_1^{[1]}, 2 \, \mbf W_0^{[3]} + \mbf W_2^{[3]}\right) \right. \notag
    	    \\
    	    + 2 \, \mbf F_{2, 0, 0} \left(\mbf W_1^{[2]}, 2 \, \mbf W_0^{[2]} + \mbf W_2^{[2]}\right) + 2 \, a_1 \, \mbf F_{2, 1, 0} \left(\bs \phi_1^{[1]}, 2 \, \mbf W_0^{[2]} + \mbf W_2^{[2]}\right) \label{a_1,b_1fourthorder}
    	    \\
    	    + 2 \, b_1 \, \mbf F_{2, 0, 1} \left(\bs \phi_1^{[1]}, 2 \, \mbf W_0^{[2]} + \mbf W_2^{[2]}\right) + 9 \, \mbf F_{3, 0, 0} \left(\bs \phi_1^{[1]}, \bs \phi_1^{[1]}, \mbf W_1^{[2]}\right) \notag
            \\
            \left. + 3 \, a_1 \, \mbf F_{3, 1, 0}\left(\bs \phi_1^{[1]}, \bs \phi_1^{[1]}, \bs \phi_1^{[1]}\right) + 3 \, b_1 \, \mbf F_{3, 0, 1}\left(\bs \phi_1^{[1]}, \bs \phi_1^{[1]}, \bs \phi_1^{[1]}\right),\bs \psi\right \rangle = 0, \notag
    	\end{align}
    	and
    	\begin{align}
    	    \left \langle a_1 \, \jac \mbf f_{1, 0}(\mbf 0) \, \mbf W_{1, 3}^{[3]} + b_1 \, \jac \mbf f_{0, 1}(\mbf 0) \, \mbf W_{1, 3}^{[3]} + 2 \, k^2 \, \hat D \, \mbf W_1^{[2]}, \bs \psi \right \rangle = 0. \label{a_1,b_1,2fourthorder}
    	\end{align}
        \begin{remark}
            Note equations \eqref{a_1,b_1secondorder}, \eqref{a_1,b_1fourthorder} and \eqref{a_1,b_1,2fourthorder} can be written as a linear system of equations,
            \begin{align}
                \upsilon_1 \, a_1 + \upsilon_2 \, b_1 &=0, \notag
                \\
                \upsilon_3 \, a_1 + \upsilon_4 \, b_1 &=0, \label{zero_first_order}
                \\
                \upsilon_5 \, a_1 + \upsilon_6 \, b_1 &=0, \notag
            \end{align}
            where $\upsilon_i \in \mathbb R$ for all $i = 1, \ldots, 6$. This implies that, if the determinant of any $2 \times 2$ submatrix of the matrix of coefficients of this system is different from zero, then $a_1 = b_1 = 0$. That is, generically, these variables need to be zero. Nevertheless, we carry on with a general asymptotic analysis, as there might be systems in which one can carry out the analysis with these variables being different from zero.
        \end{remark}

    	\paragraph{Order $\mathcal O\left(\varepsilon^5\right)$.} At this order, \eqref{eqtoexpand} becomes
        {\allowdisplaybreaks
    	\begin{align}
    	    \mbf 0 &= \jac \mbf f_{0, 0}(\mbf 0) \, \mbf u^{[5]} + a_1 \, \jac \mbf f_{1, 0}(\mbf 0) \, \mbf u^{[4]} + b_1 \, \jac \mbf f_{0, 1}(\mbf 0) \, \mbf u^{[4]} + a_2 \, \jac \mbf f_{1, 0}(\mbf 0) \, \mbf u^{[3]} + b_2 \, \jac \mbf f_{0, 1}(\mbf 0) \, \mbf u^{[3]} \notag
            \\
            & \quad + a_3 \, \jac \mbf f_{1, 0}(\mbf 0) \, \mbf u^{[2]} + b_3 \, \jac \mbf f_{0, 1}(\mbf 0) \, \mbf u^{[2]} + a_4 \, \jac \mbf f_{1, 0}(\mbf 0) \, \mbf u^{[1]} + b_4 \, \jac \mbf f_{0, 1}(\mbf 0) \, \mbf u^{[1]} \notag
            \\
            & \quad + a_1^2 \, \jac \mbf f_{2, 0}(\mbf 0) \, \mbf u^{[3]} + a_1 \, b_1 \, \jac \mbf f_{1, 1}(\mbf 0) \, \mbf u^{[3]} + b_1^2 \, \jac \mbf f_{0, 2}(\mbf 0) \, \mbf u^{[3]} + 2 \, a_1 \, a_2 \, \jac \mbf f_{2, 0}(\mbf 0) \, \mbf u^{[2]} \notag
            \\
            & \quad + \left(a_1 \, b_2 + a_2 \, b_1\right) \, \jac \mbf f_{1, 1}(\mbf 0) \, \mbf u^{[2]} + 2 \, b_1 \, b_2 \, \jac \mbf f_{0, 2}(\mbf 0) \, \mbf u^{[2]} + 2 \, a_1 \, a_3 \, \jac \mbf f_{2, 0}(\mbf 0) \, \mbf u^{[1]} \notag
            \\
            & \quad + \left(a_1 \, b_3 + a_3 \, b_1\right) \, \jac \mbf f_{1, 1}(\mbf 0) \, \mbf u^{[1]} + 2 \, b_1 \, b_3 \, \jac \mbf f_{0, 2}(\mbf 0) \, \mbf u^{[1]} + a_2^2 \, \jac \mbf f_{2, 0}(\mbf 0) \, \mbf u^{[1]} \notag
            \\
            & \quad + a_2 \, b_2 \, \jac \mbf f_{1, 1}(\mbf 0) \, \mbf u^{[1]} + b_2^2 \, \jac \mbf f_{0, 2}(\mbf 0) \, \mbf u^{[1]} + a_1^3 \, \jac \mbf f_{3, 0}(\mbf 0) \, \mbf u^{[2]} + a_1^2 \, b_1 \, \jac \mbf f_{2, 1}(\mbf 0) \, \mbf u^{[2]} \notag
            \\
            & \quad + a_1 \, b_1^2 \, \jac \mbf f_{1, 2}(\mbf 0) \, \mbf u^{[2]} + b_1^3 \, \jac \mbf f_{0, 3}(\mbf 0) \, \mbf u^{[2]} + 3 \, a_1^2 \, a_2 \, \jac \mbf f_{3, 0}(\mbf 0) \, \mbf u^{[1]} + 2 \, a_1 \, a_2 \, b_1 \, \jac \mbf f_{2, 1}(\mbf 0) \, \mbf u^{[1]} \notag
            \\
            & \quad + a_2 \, b_1^2 \, \jac \mbf f_{1, 2}(\mbf 0) \, \mbf u^{[1]} + a_1^2 \, b_2 \, \jac \mbf f_{2, 1}(\mbf 0) \, \mbf u^{[1]} + 2 \, a_1 \, b_1 \, b_2 \, \jac \mbf f_{1, 2}(\mbf 0) \, \mbf u^{[1]} + 3 \, b_1^2 \, b_2 \, \jac \mbf f_{0, 3}(\mbf 0) \, \mbf u^{[1]} \notag
            \\
            & \quad + a_1^4 \, \jac \mbf f_{4, 0}(\mbf 0) \, \mbf u^{[1]} + a_1^3 \, b_1 \, \jac \mbf f_{3, 1}(\mbf 0) \, \mbf u^{[1]} + a_1^2 \, b_1^2 \, \jac \mbf f_{2, 2}(\mbf 0) \, \mbf u^{[1]} + a_1 \, b_1^3 \, \jac \mbf f_{1, 3}(\mbf 0) \, \mbf u^{[1]} \notag
            \\
            & \quad + b_1^4 \, \jac \mbf f_{0, 4}(\mbf 0) \, \mbf u^{[1]} + 2 \, \mbf F_{2, 0, 0}\left(\mbf u^{[1]}, \mbf u^{[4]}\right) + 2 \, \mbf F_{2, 0, 0}\left(\mbf u^{[2]}, \mbf u^{[3]}\right) + 2 \, a_1 \, \mbf F_{2, 1, 0}\left(\mbf u^{[1]}, \mbf u^{[3]}\right) \notag
            \\
            & \quad + 2 \, b_1 \, \mbf F_{2, 0, 1}\left(\mbf u^{[1]}, \mbf u^{[3]}\right) + a_1 \, \mbf F_{2, 1, 0}\left(\mbf u^{[2]}, \mbf u^{[2]}\right) + b_1 \, \mbf F_{2, 0, 1}\left(\mbf u^{[2]}, \mbf u^{[2]}\right) + 2 \, a_2 \, \mbf F_{2, 1, 0}\left(\mbf u^{[1]}, \mbf u^{[2]}\right) \notag
            \\
            & \quad + 2 \, b_2 \, \mbf F_{2, 0, 1}\left(\mbf u^{[1]}, \mbf u^{[2]}\right) + a_3 \, \mbf F_{2, 1, 0}\left(\mbf u^{[1]}, \mbf u^{[1]}\right) + b_3 \, \mbf F_{2, 0, 1}\left(\mbf u^{[1]}, \mbf u^{[1]}\right) + 2 \, a_1^2 \, \mbf F_{2, 2, 0}\left(\mbf u^{[1]}, \mbf u^{[2]}\right) \notag
            \\
            & \quad + 2 \, a_1 \, b_1 \, \mbf F_{2, 1, 1}\left(\mbf u^{[1]}, \mbf u^{[2]}\right) + 2 \, b_1^2 \, \mbf F_{2, 0, 2}\left(\mbf u^{[1]}, \mbf u^{[2]}\right) + 2 \, a_1 \, a_2 \, \mbf F_{2, 2, 0}\left(\mbf u^{[1]}, \mbf u^{[1]}\right) \notag
    	    \\
            & \quad + \left(a_1 \, b_2 + a_2 \, b_1\right) \, \mbf F_{2, 1, 1}\left(\mbf u^{[1]}, \mbf u^{[1]}\right) + 2 \, b_1 \, b_2 \, \mbf F_{2, 0, 2}\left(\mbf u^{[1]}, \mbf u^{[1]}\right) + a_1^3 \, \mbf F_{2, 3, 0}\left(\mbf u^{[1]}, \mbf u^{[1]}\right) \notag
            \\
            & \quad + a_1^2 \, b_1 \, \mbf F_{2, 2, 1}\left(\mbf u^{[1]}, \mbf u^{[1]}\right) + a_1 \, b_1^2 \, \mbf F_{2, 1, 2}\left(\mbf u^{[1]}, \mbf u^{[1]}\right) + b_1^3 \, \mbf F_{2, 0, 3}\left(\mbf u^{[1]}, \mbf u^{[1]}\right) \notag
            \\
            & \quad + 3 \, \mbf F_{3, 0, 0}\left(\mbf u^{[1]}, \mbf u^{[1]}, \mbf u^{[3]}\right) + 3 \, \mbf F_{3, 0, 0}\left(\mbf u^{[1]}, \mbf u^{[2]}, \mbf u^{[2]}\right) + 3 \, a_1 \, \mbf F_{3, 1, 0}\left(\mbf u^{[1]}, \mbf u^{[1]}, \mbf u^{[2]}\right) \notag
            \\
            & \quad + 3 \, b_1 \, \mbf F_{3, 0, 1}\left(\mbf u^{[1]}, \mbf u^{[1]}, \mbf u^{[2]}\right) + a_2 \, \mbf F_{3, 1, 0}\left(\mbf u^{[1]}, \mbf u^{[1]}, \mbf u^{[1]}\right) + b_2 \, \mbf F_{3, 0, 1}\left(\mbf u^{[1]}, \mbf u^{[1]}, \mbf u^{[1]}\right) \notag
            \\
            & \quad + a_1^2 \, \mbf F_{3, 2, 0}\left(\mbf u^{[1]}, \mbf u^{[1]}, \mbf u^{[1]}\right) + a_1 \, b_1 \, \mbf F_{3, 1, 1}\left(\mbf u^{[1]}, \mbf u^{[1]}, \mbf u^{[1]}\right) + b_1^2 \, \mbf F_{3, 0, 2}\left(\mbf u^{[1]}, \mbf u^{[1]}, \mbf u^{[1]}\right) \notag
            \\
            & \quad + 4 \, \mbf F_{4, 0, 0} \left(\mbf u^{[1]}, \mbf u^{[1]}, \mbf u^{[1]}, \mbf u^{[2]}\right) + a_1 \, \mbf F_{4, 1, 0}\left(\mbf u^{[1]}, \mbf u^{[1]}, \mbf u^{[1]}, \mbf u^{[1]}\right) + b_1 \, \mbf F_{4, 0, 1}\left(\mbf u^{[1]}, \mbf u^{[1]}, \mbf u^{[1]}, \mbf u^{[1]}\right) \notag
            \\
            & \quad + \mbf F_{5, 0, 0}\left(\mbf u^{[1]}, \mbf u^{[1]}, \mbf u^{[1]}, \mbf u^{[1]}, \mbf u^{[1]}\right) + k^2 \, \hat D \, \left(\frac{\partial^2 \mbf u^{[5]}}{\partial x^2} + 2 \, \frac{\partial^2 \mbf u^{[3]}}{\partial x \partial X} + \frac{\partial^2 \mbf u^{[1]}}{\partial X^2}\right). \label{fifth-order-eq}
    	\end{align}
        }
        At this point, we are not interested in the full solution to this equation (find the complete development up to order seven in \ref{app:7expansion}). For now, we only need to ensure it exists. To do this, we have to obtain a solvability condition from the terms that are multiples of $e^{ix}$ on the right-hand side of \eqref{fifth-order-eq}, which are given by
    	\begin{multline*}
    	    - \mbf Q_1^{[5]} \, A_1'' - i \, \mbf Q_2^{[5]} \, A_1' - i \, \mbf Q_3^{[5]} \, \abs{A_1}^2 \, A_1' - i \, \mbf Q_4^{[5]} \, A_1^2 \, \bar A_1'
            \\
            - \mbf Q_5^{[5]} \, A_1 - \mbf Q_6^{[5]} \, \abs{A_1}^2 \, A_1 - \mbf Q_7^{[5]} \, \abs{A_1}^4 \, A_1 + \ldots,
    	\end{multline*}
        where ``$\ldots$'' are terms that depend on $A_2, A_3, A_4$ and do not influence the solvability condition, and
        {\allowdisplaybreaks
        \begin{align*}
            \mbf Q_1^{[5]} &= - 2 \, k^2 \, \hat D \, \mbf W_{1, 3}^{[3]} + k^2 \, \hat D \, \bs \phi_1^{[1]},
            \\
            \mbf Q_2^{[5]} &= a_1 \, \jac \mbf f_{1, 0}(\mbf 0) \, \mbf W_{1, 3}^{[4]} + b_1 \, \jac \mbf f_{0, 1}(\mbf 0) \, \mbf W_{1, 3}^{[4]} + a_2 \, \jac \mbf f_{1, 0}(\mbf 0) \, \mbf W_{1, 3}^{[3]} + b_2 \, \jac \mbf f_{0, 1}(\mbf 0) \, \mbf W_{1, 3}^{[3]}
            \\
            & \quad + a_1^2 \, \jac \mbf f_{2, 0}(\mbf 0) \, \mbf W_{1, 3}^{[3]} + a_1 \, b_1 \, \jac \mbf f_{1, 1}(\mbf 0) \, \mbf W_{1, 3}^{[3]} + b_1^2 \, \jac \mbf f_{0, 2}(\mbf 0) \, \mbf W_{1, 3}^{[3]} + 2 \, k^2 \, \hat D \, \mbf W_1^{[3]},
            \\
            \mbf Q_3^{[5]} &= 2 \, \mbf F_{2, 0, 0}\left(\bs \phi_1^{[1]}, \mbf W_{0, 3}^{[4]} + \mbf W_{2, 3}^{[4]}\right) + 4 \, \mbf F_{2, 0, 0}\left(\mbf W_0^{[2]}, \mbf W_{1, 3}^{[3]}\right) + 6 \, \mbf F_{3, 0, 0}\left(\bs \phi_1^{[1]}, \bs \phi_1^{[1]}, \mbf W_{1, 3}^{[3]}\right)
            \\
            & \quad + 4 \, k^2 \, \hat D \, \mbf W_{1, 2}^{[3]},
            \\
            \mbf Q_4^{[5]} &= - 2 \, \mbf F_{2, 0, 0}\left(\bs \phi_1^{[1]}, \mbf W_{0, 3}^{[4]}\right) - 2 \, \mbf F_{2, 0, 0}\left(\mbf W_2^{[2]}, \mbf W_{1, 3}^{[3]}\right) - 3 \, \mbf F_{3, 0, 0}\left(\bs \phi_1^{[1]}, \bs \phi_1^{[1]}, \mbf W_{1, 3}^{[3]}\right)
            \\
            & \quad + 2 \, k^2 \, \hat D \, \mbf W_{1, 2}^{[3]},
            \\
            \mbf Q_5^{[5]} &= a_1 \, \jac \mbf f_{1, 0}(\mbf 0) \, \mbf W_1^{[4]} + b_1 \, \jac \mbf f_{0, 1}(\mbf 0) \, \mbf W_1^{[4]} + a_2 \, \jac \mbf f_{1, 0}(\mbf 0) \, \mbf W_1^{[3]} + b_2 \, \jac \mbf f_{0, 1}(\mbf 0) \, \mbf W_1^{[3]}
            \\
            & \quad + a_3 \, \jac \mbf f_{1, 0}(\mbf 0) \, \mbf W_1^{[2]} + b_3 \, \jac \mbf f_{0, 1}(\mbf 0) \, \mbf W_1^{[2]} + a_4 \, \jac \mbf f_{1, 0}(\mbf 0) \, \bs \phi_1^{[1]} + b_4 \, \jac \mbf f_{0, 1}(\mbf 0) \, \bs \phi_1^{[1]}
            \\
            & \quad + a_1^2 \, \jac \mbf f_{2, 0}(\mbf 0) \, \mbf W_1^{[3]} + a_1 \, b_1 \, \jac \mbf f_{1, 1}(\mbf 0) \, \mbf W_1^{[3]} + b_1^2 \, \jac \mbf f_{0, 2}(\mbf 0) \, \mbf W_1^{[3]} + 2 \, a_1 \, a_2 \, \jac \mbf f_{2, 0}(\mbf 0) \, \mbf W_1^{[2]}
    	    \\
            & \quad + \left(a_1 \, b_2 + a_2 \, b_1\right) \, \jac \mbf f_{1, 1}(\mbf 0) \, \mbf W_1^{[2]} + 2 \, b_1 \, b_2 \, \jac \mbf f_{0, 2}(\mbf 0) \, \mbf W_1^{[2]} + 2 \, a_1 \, a_3 \, \jac \mbf f_{2, 0}(\mbf 0) \, \bs \phi_1^{[1]}
            \\
    	    & \quad + \left(a_1 \, b_3 + a_3 \, b_1\right) \, \jac \mbf f_{1, 1}(\mbf 0) \, \bs \phi_1^{[1]} + 2 \, b_1 \, b_3 \, \jac \mbf f_{0, 2}(\mbf 0) \, \bs \phi_1^{[1]} + a_2^2 \, \jac \mbf f_{2, 0}(\mbf 0) \, \bs \phi_1^{[1]}
            \\
            & \quad + a_2 \, b_2 \, \jac \mbf f_{1, 1}(\mbf 0) \, \bs \phi_1^{[1]} + b_2^2 \, \jac \mbf f_{0, 2}(\mbf 0) \, \bs \phi_1^{[1]} + a_1^3 \, \jac \mbf f_{3, 0}(\mbf 0) \, \mbf W_1^{[2]} + a_1^2 \, b_1 \, \jac \mbf f_{2, 1}(\mbf 0) \, \mbf W_1^{[2]}
            \\
            & \quad + a_1 \, b_1^2 \, \jac \mbf f_{1, 2}(\mbf 0) \, \mbf W_1^{[2]} + b_1^3 \, \jac \mbf f_{0, 3}(\mbf 0) \, \mbf W_1^{[2]} + 3 \, a_1^2 \, a_2 \, \jac \mbf f_{3, 0}(\mbf 0) \, \bs \phi_1^{[1]}
            \\
            & \quad + 2 \, a_1 \, a_2 \, b_1 \, \jac \mbf f_{2, 1}(\mbf 0) \, \bs \phi_1^{[1]} + a_2 \, b_1^2 \, \jac \mbf f_{1, 2}(\mbf 0) \, \bs \phi_1^{[1]} + a_1^2 \, b_2 \, \jac \mbf f_{2, 1}(\mbf 0) \, \bs \phi_1^{[1]}
    	    \\
    	    & \quad + 2 \, a_1 \, b_1 \, b_2 \, \jac \mbf f_{1, 2}(\mbf 0) \, \bs \phi_1^{[1]} + 3 \, b_1^2 \, b_2 \, \jac \mbf f_{0, 3}(\mbf 0) \, \bs \phi_1^{[1]} + a_1^4 \, \jac \mbf f_{4, 0}(\mbf 0) \, \bs \phi_1^{[1]} + a_1^3 \, b_1 \, \jac \mbf f_{3, 1}(\mbf 0) \, \bs \phi_1^{[1]}
            \\
            & \quad + a_1^2 \, b_1^2 \, \jac \mbf f_{2, 2}(\mbf 0) \, \bs \phi_1^{[1]} + a_1 \, b_1^3 \, \jac \mbf f_{1, 3}(\mbf 0) \, \bs \phi_1^{[1]} + b_1^4 \, \jac \mbf f_{0, 4}(\mbf 0) \, \bs \phi_1^{[1]},
            \\
            \mbf Q_6^{[5]} &= a_1 \, \jac \mbf f_{1, 0}(\mbf 0) \, \mbf W_{1, 2}^{[4]} + b_1 \, \jac \mbf f_{0, 1}(\mbf 0) \, \mbf W_{1, 2}^{[4]} + a_2 \, \jac \mbf f_{1, 0}(\mbf 0) \, \mbf W_{1, 2}^{[3]} + b_2 \, \jac \mbf f_{0, 1}(\mbf 0) \, \mbf W_{1, 2}^{[3]}
            \\
            & \quad + a_1^2 \, \jac \mbf f_{2, 0}(\mbf 0) \, \mbf W_{1, 2}^{[3]} + a_1 \, b_1 \, \jac \mbf f_{1, 1}(\mbf 0) \, \mbf W_{1, 2}^{[3]} + b_1^2 \, \jac \mbf f_{0, 2}(\mbf 0) \, \mbf W_{1, 2}^{[3]}
            \\
            & \quad + 2 \, \mbf F_{2, 0, 0}\left(\bs \phi_1^{[1]}, 2 \, \mbf W_0^{[4]} + \mbf W_2^{[4]}\right) + 2 \, \mbf F_{2, 0, 0}\left(\mbf W_1^{[2]}, 2 \, \mbf W_0^{[3]} + \mbf W_2^{[3]}\right)
    	    \\
    	    & \quad + 2 \, \mbf F_{2, 0, 0}\left(\mbf W_1^{[3]}, 2 \, \mbf W_0^{[2]} + \mbf W_2^{[2]}\right) + 2 \, a_1 \, \mbf F_{2, 1, 0}\left(\bs \phi_1^{[1]}, 2 \, \mbf W_0^{[3]} + \mbf W_2^{[3]}\right)
            \\
            & \quad + 2 \, b_1 \, \mbf F_{2, 0, 1}\left(\bs \phi_1^{[1]}, 2 \, \mbf W_0^{[3]} + \mbf W_2^{[3]}\right) + 2 \, a_1 \, \mbf F_{2, 1, 0}\left(\mbf W_1^{[2]}, 2 \, \mbf W_0^{[2]} + \mbf W_2^{[2]}\right)
            \\
            & \quad + 2 \, b_1 \, \mbf F_{2, 0, 1}\left(\mbf W_1^{[2]}, 2 \, \mbf W_0^{[2]} + \mbf W_2^{[2]}\right) + 2 \, a_2 \, \mbf F_{2, 1, 0}\left(\bs \phi_1^{[1]}, 2 \, \mbf W_0^{[2]} + \mbf W_2^{[2]}\right)
            \\
            & \quad + 2 \, b_2 \, \mbf F_{2, 0, 1}\left(\bs \phi_1^{[1]}, 2 \, \mbf W_0^{[2]} + \mbf W_2^{[2]}\right) + 2 \, a_1^2 \, \mbf F_{2, 2, 0}\left(\bs \phi_1^{[1]}, 2 \, \mbf W_0^{[2]} + \mbf W_2^{[2]}\right)
            \\
            & \quad + 2 \, a_1 \, b_1 \, \mbf F_{2, 1, 1}\left(\bs \phi_1^{[1]}, 2 \, \mbf W_0^{[2]} + \mbf W_2^{[2]}\right) + 2 \, b_1^2 \, \mbf F_{2, 0, 2}\left(\bs \phi_1^{[1]}, 2 \, \mbf W_0^{[2]} + \mbf W_2^{[2]}\right)
            \\
            & \quad + 9 \, \mbf F_{3, 0, 0}\left(\bs \phi_1^{[1]}, \bs \phi_1^{[1]}, \mbf W_1^{[3]}\right) + 9 \, \mbf F_{3, 0, 0}\left(\bs \phi_1^{[1]}, \mbf W_1^{[2]}, \mbf W_1^{[2]}\right)
            \\
            & \quad + 9 \, a_1 \, \mbf F_{3, 1, 0}\left(\bs \phi_1^{[1]}, \bs \phi_1^{[1]}, \mbf W_1^{[2]}\right) + 9 \, b_1 \, \mbf F_{3, 0, 1}\left(\bs \phi_1^{[1]}, \bs \phi_1^{[1]}, \mbf W_1^{[2]}\right)
            \\
            & \quad + 3 \, a_2 \, \mbf F_{3, 1, 0}\left(\bs \phi_1^{[1]}, \bs \phi_1^{[1]}, \bs \phi_1^{[1]}\right) + 3 \, b_2 \, \mbf F_{3, 0, 1}\left(\bs \phi_1^{[1]}, \bs \phi_1^{[1]}, \bs \phi_1^{[1]}\right) + 3 \, a_1^2 \, \mbf F_{3, 2, 0}\left(\bs \phi_1^{[1]}, \bs \phi_1^{[1]}, \bs \phi_1^{[1]}\right)
            \\
            & \quad + 3 \, a_1 \, b_1 \, \mbf F_{3, 1, 1}\left(\bs \phi_1^{[1]}, \bs \phi_1^{[1]}, \bs \phi_1^{[1]}\right) + 3 \, b_1^2 \, \mbf F_{3, 0, 2}\left(\bs \phi_1^{[1]}, \bs \phi_1^{[1]}, \bs \phi_1^{[1]}\right),
            \\
            \mbf Q_7^{[5]} &= 2 \, \mbf F_{2, 0, 0}\left(\bs \phi_1^{[1]}, 2 \, \mbf W_{0, 2}^{[4]} + \mbf W_{2, 2}^{[4]}\right) + 2 \, \mbf F_{2, 0, 0}\left(\mbf W_2^{[2]}, \mbf W_3^{[3]}\right) + 2 \, \mbf F_{2, 0, 0}\left(\mbf W_{1, 2}^{[3]}, 2 \, \mbf W_0^{[2]} + \mbf W_2^{[2]}\right)
    	    \\
    	    & \quad + 3 \, \mbf F_{3, 0, 0}\left(\bs \phi_1^{[1]}, \bs \phi_1^{[1]}, 3 \, \mbf W_{1, 2}^{[3]} + \mbf W_3^{[3]}\right) + 12 \, \mbf F_{3, 0, 0}\left(\bs \phi_1^{[1]}, \mbf W_0^{[2]}, \mbf W_0^{[2]}\right)
    	    \\
    	    & \quad + 6 \, \mbf F_{3, 0, 0}\left(\bs \phi_1^{[1]}, \mbf W_2^{[2]}, 2 \, \mbf W_0^{[2]} + \mbf W_2^{[2]}\right) + 8 \, \mbf F_{4, 0, 0} \left(\bs \phi_1^{[1]}, \bs \phi_1^{[1]}, \bs \phi_1^{[1]}, 3 \, \mbf W_0^{[2]} + 2 \, \mbf W_2^{[2]}\right)
    	    \\
    	    & \quad + 10 \, \mbf F_{5, 0, 0}\left(\bs \phi_1^{[1]}, \bs \phi_1^{[1]}, \bs \phi_1^{[1]}, \bs \phi_1^{[1]}, \bs \phi_1^{[1]}\right).
        \end{align*}
        }
        Again, we need to ensure that \eqref{fifth-order-eq} has a solution and, as before, we can apply the inner product with $\bs \psi$. Finally, we arrive at the amplitude equation  \eqref{firstamplitudeeq}, where $\alpha_i = \left \langle \mbf Q_i^{[5]}, \, \bs \psi \right \rangle $ for $i = 1, \ldots, 7$, and we require $\alpha_1 \neq 0$.

    \section{The case in which $c_0^{[0]} = c_2^{[0]} = 0$} \label{app:symmetric_case}
        Generically, we have that the constants $c_0^{[0]}$ or $c_2^{[0]}$ let us estimate the width of the homoclinic snaking successfully, close to codimension-two Turing bifurcation points (see Section \ref{sec:second_residual}). Nevertheless, there are cases in which both of these constants equal zero, which does not let us give a proper estimation for this width (see e.g.~\cite{Dean}). This happens, generically, when the system \eqref{geneq} has the symmetry $\mbf u \to - \mbf u$, a case in which the even terms of the expansion vanish. To overcome this issue, the same general ansatzes we used in Section \ref{sec:late_term} are still valid, but some key considerations need to be made. In particular,
        \begin{enumerate}
            \item The dominant values of $\kappa^2$ become $\pm i/2$, and the dominant modes for which these values are attained are $r = \pm 1, \pm 3$ (see sub-Section \ref{sub:kappa_explanation}).
            \item The ansatz for the values of $\gamma_r$, \eqref{gammavals} becomes
            \begin{align*}
                \gamma_0 = \gamma_{\pm 2} = \gamma - \frac{1}{2}, \qquad \gamma_{\pm 1} = \gamma, \qquad \gamma_{\pm r} = \gamma - \frac{r - 3}{2}, \quad \text{for } r \geq 3.
            \end{align*}
            \item After making the same expansion as in \eqref{Vnr} to obtain the inner solution with these new values of $\gamma_r$ (see Section \ref{sec:late_term} for the ideas that need to be followed to study this case), we conclude that
            \begin{align*}
                \evs{\gamma} \in \{- 1 - 2 \, \eta \, i, - 2 \, \eta \, i, 2 - 2 \, \eta \, i, 3 - 2 \, \eta \, i\},
            \end{align*}
            from where we take, once again, the one with the highest real part, $\gamma = 3 - 2 \, \eta \, i$.
            \item As the dominant multiples of the critical wavenumber, $k$, are now odd, we need to assume that $n$ is odd. Furthermore, we need to consider the changes of the modes in the forcing due to truncation of the equation for the remainder (see Section \ref{sec:second_residual}). In particular, the forcing due to truncation provided by $\kappa^2 = \kappa_+^2/2$ is now given by
            \begin{align*}
                \varepsilon^{N + 1} \, \kappa^N \, e^{2ix} \, \frac{\Gamma\left(\frac{N}{2} + \gamma\right)}{\left(X_0 - X\right)^{\frac{N}{2} + \gamma}} \sum_{r = 1, 3} e^{i (r - 2) \hat x} \, \mbf C_r,
            \end{align*}
            where
            \begin{align*}
                \mbf C_r
                &= 2 \, h_r^{[0]} \, \left(- 2 \, \left(2 - r\right) + 2 \, i \, \varepsilon \, \hat \theta\right) \, k^2 \, \hat D \, \bs \phi_1^{[1]},
            \end{align*}
            which implies
            \begin{align*}
                \mbf C_1 = h_1^{[0]} \, \left(- 4 + 4 \, i \, \varepsilon \, \hat \theta\right) \, k^2 \, \hat D \, \bs \phi_1^{[1]}, \qquad \mbf C_3 = h_3^{[0]} \, \left(4 + 4 \, i \, \varepsilon \, \hat \theta\right) \, k^2 \, \hat D \, \bs \phi_1^{[1]},
            \end{align*}
            and $h_1^{[0]}$, $h_3^{[0]}$ are the analogs of $h_0^{[0]}$, $h_2^{[0]}$, respectively, for this case. We highlight that the equations that determine $h_1^{[0]}$, $h_3^{[0]}$ coincide with the ones obtained for $h_0^{[0]}$, $h_2^{[0]}$ in Section \ref{sec:late_term2}.
            \item Following the same steps as in Section \ref{sec:second_residual}, we note that
            \begin{align*}
                e^{2 i \hat x} = e^{2 i \left(X - X_0 + X_0\right)/\varepsilon^2 - 2i\hat \chi} = e^{2 i \left(X - X_0\right)/\varepsilon^2} \, e^{2iX_0/\varepsilon^2} \, e^{- 2i\hat \chi},
            \end{align*}
            which yields
            \begin{align*}
                e^{2i\left(X - X_0\right)/\varepsilon^2} \, \varepsilon^{N + 2 \, \gamma} \, \kappa^{N + 2 \, \gamma} \, \frac{\Gamma\left(\frac{N}{2} + \gamma\right)}{\left(X_0 - X\right)^{\frac{N}{2} + \gamma}} \sim \sqrt{2 \, \pi} \, \varepsilon \, \abs{\kappa} \, e^{- \rho \, \hat \theta^2}.
            \end{align*}
            Therefore,
            \begin{multline*}
                \mbf R_{N, 2}^{[2]} = 2 \, \left(A_1 \, C_{- 1} + \bar A_1 \, C_1\right) \, \mbf W_0^{[2]} + \left(C_1 \, e^{ix} + C_{- 1} \, e^{-ix}\right) \, \mbf W_1^{[2]}
                \\
                - \frac{1}{\rho} \, \left(i \, C_1' \, e^{ix} - i \, C_{- 1}' \, e^{- ix}\right) \mbf W_{1, 3}^{[3]}
                \\
                - 2 \, \kappa^{- 2 \, \gamma} \, e^{- i \hat \chi} \, \frac{\sqrt{2 \, \pi} \, \abs{\kappa}}{\rho^{\frac{1}{2}}} \, e^{- \rho \, \hat \theta^2} \, h_0^{[0]} \, e^{- ix} \, \mbf W_{1, 3}^{[3]} + 2 \, \kappa^{- 2 \, \gamma} \, e^{- i \hat \chi} \, \frac{\sqrt{2 \, \pi} \, \abs{\kappa}}{\rho^{\frac{1}{2}}} \, e^{- \rho \, \hat \theta^2} \, h_2^{[0]} \, e^{ix} \, \mbf W_{1, 3}^{[3]}
                \\
                + 2 \, \left(A_1 \, C_1 \, e^{2ix} + \bar A_1 \, C_{- 1} \, e^{-2ix}\right) \, \mbf W_2^{[2]},
            \end{multline*}
            which implies that the analog of the solvability conditions at order $\mathcal O\left(\varepsilon^3\right)$, \eqref{solvcondC-1} and \eqref{solvcondC1}, yield
            \begin{align*}
                C_{- 1}'' &= 2 \, \sqrt{2 \, \pi} \, i \, \hat \theta \, h_0^{[0]} \, \rho^{3/2} \, \abs{\kappa} \, \kappa^{- 2 \, \gamma} \, e^{- \hat \theta^2 \, \rho - i \hat \chi }
                \\
                C_1'' &= 2 \, \sqrt{2 \, \pi} \, i \, \hat \theta \, h_2^{[0]} \, \rho^{3/2} \, \abs{\kappa} \, \kappa^{- 2 \, \gamma} \, e^{- \hat \theta^2 \, \rho - i \, \hat \chi},
            \end{align*}
            which yield
            \begin{align*}
                C_{- 1} &= - i \, \pi \, \kappa^{- 2 \, \gamma} \, h_0^{[0]} \, e^{- i \hat \chi},
                \\
                C_1 &= - i \, \pi \, \kappa^{- 2 \, \gamma} \, h_0^{[0]} \, e^{- i \hat \chi}.
            \end{align*}
            Therefore, in this case, as we cross the Stokes' line, the following function gets triggered by $\kappa^2 = \kappa_+^2/2$:
            \begin{align*}
                - i \, \pi \, \varepsilon^{- 2 \, \gamma} \, \kappa^{- 2 \, \gamma} \, e^{2 i X_0/\varepsilon^2} \, e^{- 2 i \hat \chi} \, \left(h_2^{[0]} \, e^{ix} + h_0^{[0]} \, e^{- ix}\right) \, \bs \phi_1^{[1]},
            \end{align*}
            which implies
            \begin{align}
                L_2 = L_2^+ &= - i \, \pi \, \varepsilon^{- 2 \, \gamma} \, \kappa^{- 2 \, \gamma} \, e^{2 i X_0/\varepsilon^2} \, e^{- 2 i \hat \chi} \, K_2,
            \end{align}
            with the corresponding development also for $\kappa_-^2/2$.
        \end{enumerate}
        In summary, the only change one needs to make in this case to apply the same theory in order to estimate the width of the snaking is to scale \eqref{L_2+def} by a factor of $\dfrac{1}{2}$, take the factor $\abs{\kappa}$ out of it, and multiply the power of the exponential terms by 2.

    \section{Joining fronts} \label{ap:joining_fronts}
        Let us recall that we have analysed two solutions at the same time (see sub-Section \ref{sub:solving_amplitude}). They correspond to fronts that join two steady states in both directions, up and down (see \eqref{R_1-up_front} and \eqref{R_1-down_front}). Now, we need to join those solutions to construct a homoclinic orbit, which turns out to be a localized solution of system \eqref{general_equation} (see Figure \ref{fig:fronts_to_join}), and obtain an approximation of the \textit{homoclinic snaking} bifurcation curve.
    
        Note that the solution we obtained for the remainder in Sections \ref{sec:first_residual} and \ref{sec:second_residual} is exponentially growing on $X$ and becomes order 1 when $X \sim \mathcal O\left(1/\varepsilon^2\right)$. Therefore, if we let $2 \, L$ be the width of the support of this solution on $x$ in \eqref{general_equation}, then the localized solution can be constructed by joining an up-front in the range $0 \leq X \leq L/\varepsilon^2$ with a down-front in the range $L/\varepsilon^2 \leq X \leq 2 \, L/\varepsilon^2$. With this, we now proceed to perform the matching.
        
        The up-front solution is given by
        \begin{align}
            \mbf u = \sum_{r = 1}^N \varepsilon^r \, \mbf u^{[r]} + \mbf R_N, \label{solution}
        \end{align}
        where
        \begin{align*}
            \mbf u^{[1]} &= A_1 \, e^{i\left(x - \hat \chi\right)} \, \bs \phi_1^{[1]} + c.c.,
            \\
            \mbf u^{[3]} &= A_3 \, e^{i\left(x - \hat \chi\right)} \, \bs \phi_1^{[1]} + c.c.,
        \end{align*}
        and a first-order approximation of the remainder has also been determined. Still, it depends on the function $\delta\left(R_1\right)$, which is related to the separation of one parameter from the Maxwell point (see equation \eqref{rem_eq}), which was studied in the examples developed in Section \ref{sec:examples}.
        
        Now, for $1 \ll X \ll 1/\varepsilon^2$, we have
        \begin{align}
            A_1 \sim \left(\Delta_{1, 1} + \Delta_{1, 2} \, e^{- 2 \, \sqrt{\beta_1} \, X}\right) \, e^{\Phi \, i},
        \end{align}
        where
        \begin{align*}
            \Delta_{1, 1} = \sqrt{- \frac{2 \, \beta_1}{\beta_3}}, \qquad \Delta_{1, 2} = \sqrt{- \frac{2 \, \beta_1}{\beta_3}} \, \frac{- 1 + 2 \, \eta \, i}{2},
            \\
            \Phi = 2 \, (\eta - \xi) \, \sqrt{\beta_1} \, X + \zeta - \eta \, \log\left(2 \, \sqrt{\beta_1}\right),
        \end{align*}
        and $\zeta = \zeta(\Xi)$ is a linear function on $\Xi = \varepsilon^2 \, X$, obtained in \ref{app:7expansion}.
        
        Furthermore, from \ref{app:7expansion}, we have that $A_3$ has an asymptotic behaviour, as $X \to \infty$, given by
        \begin{align*}
            A_3 &\sim \left(\Delta_{3, 1} + \Delta_{3, 2} \, X \, e^{- 2 \, \sqrt{\beta_1} \, X}\right) \, e^{\Phi \, i},
        \end{align*}
        where
        \begin{align*}
            \Delta_{3, 1} &= - \frac{2^{- 5/2}}{\alpha_1^2 \, \sqrt{\beta_1} \, \left(- \beta_3\right)^{7/2}} \, \left(\alpha_2 \, \beta_3^3 \, \left(\alpha_{1, 4} + 2 \, \alpha_1 \, \zeta_\Xi\right) + \alpha_2^2 \, \beta_3^3\right.
            \\
            & \quad \left. + 2 \, \alpha_1 \, \left(\beta_3^3 \, \alpha_{8, 4} - 8 \, \beta_1^3 \, \beta_{7, 3} + 4 \, \beta_1^2 \, \beta_3 \, \beta_{5, 3} - 2 \, \beta_1 \, \beta_3^2 \, \beta_{3, 3} + 4 \, i \, \alpha_1 \, \beta_1 \, \beta_3^3 \, \omega_4\right)\right),
            \\
            \Delta_{3, 2} &= \frac{2^{- 3/2}}{\alpha_1^2 \, \left(- \beta_3\right)^{7/2}} \, \left(1 - 2 \, \eta \, i\right) \, \left(\alpha_2 \, \beta_3^3 \, \left(\alpha_2 + \alpha_{1, 4} + 2 \, \alpha_1 \, \zeta_\Xi\right)\right.
            \\
            & \quad + \alpha_1 \, \left(- \beta_1^2 \, \beta_3^2 \, \left(\alpha_{6, 4} + \alpha_{7, 4}\right) + 2 \, \beta_3^3 \, \alpha_{8, 4} - 4 \, \beta_1^3 \, \beta_{7, 3}\right).
        \end{align*}
        With this, we note that the up-front solution is, asymptotically, given by
        \begin{align}
            \mbf u &\sim \varepsilon \, \mbf u^{[1]} + \varepsilon^3 \, \mbf u^{[3]} \notag
            \\
            &= \varepsilon \, \left(\left(\Delta_{1, 1} + \Delta_{1, 2} \, e^{- 2 \, \sqrt{\beta_1} \, X}\right) + \varepsilon^2 \, \left(\Delta_{3, 1} + \Delta_{3, 2} \, X \, e^{- 2 \, \sqrt{\beta_1} \, X}\right)\right) \, e^{\Phi \, i} \, e^{i \left(x - \hat \chi\right)} \, \bs \phi_1^{[1]} + c.c., \label{up-front}
        \end{align}
        and, by the symmetric properties of the amplitude equation \eqref{firstamplitudeeq} (see sub-Section, \ref{sub:solving_amplitude}), we have that the down-front is given by \eqref{up-front} when we change $X \to L/\varepsilon^2 - X$, and $x \to L/\varepsilon^4 - X$. However, the corresponding two solutions may have different phase shifts. Therefore, for $1 \ll L/\varepsilon^2 - X \ll 1/\varepsilon^2$, an asymptotic approximation of the down-front is given by
        \begin{align}
            \mbf u &\sim \varepsilon \, \left(\left(\Delta_{1, 1} + \Delta_{1, 2} \, e^{- 2 \, \sqrt{\beta_1} \, \left(L/\varepsilon^2 - X\right)}\right) + \varepsilon^2 \, \left(\Delta_{3, 1} + \Delta_{3, 2} \, \left(L/\varepsilon^2 - X\right) \, e^{- 2 \, \sqrt{\beta_1} \, \left(L/\varepsilon^2 - X\right)}\right)\right) \notag
            \\
            & \quad \times \left(2 \, \sqrt{\beta_1}\right)^{- \eta \, i} \, e^{2 \, i \, (\eta - \xi) \, \sqrt{\beta_1} \, \left(L/\varepsilon^2 - X\right)} \, e^{i \, \zeta} \, e^{i \left(L/\varepsilon^4 - x - \check \chi\right)} \, \bs \phi_1^{[1]} + c.c., \label{down-front}
        \end{align}
        where $\check \chi \in \mathbb R$ is the phase shift for the down-front.
        
        Now, we highlight that, for our expansion to be valid, we need to ensure that these match for $X = \mathcal O\left(\varepsilon^2\right)$ and $L/\varepsilon^2 - X = \mathcal O\left(1/\varepsilon^2\right)$. Nevertheless, the expressions of these fronts will not be uniformly valid in these limits due to the term $X \, e^{- 2 \, \sqrt{\beta_1} \, X}$ present in $A_3$ (unless $\Delta_{3, 2} = 0$). Therefore, in the general case, we must construct an extra outer expansion valid inside the extended periodic domain (see the periodic pattern when the fronts merge in Figure \ref{fig:fronts_to_join}) to match it with these two fronts.
    
        \subsection{Outer expansion inside the extended periodic domain} \label{sub:fast_oscillations}
            To obtain a suitable outer approximation, we take the usual expansion, \eqref{solution}, only focusing on the limit as $X \to \infty$ (see the periodic pattern when the fronts merge in Figure \ref{fig:fronts_to_join}). In particular, we note that
            \begin{align*}
                A_1 = \Delta_{1, 1} \, e^{\Phi \, i}
            \end{align*}
            solves the amplitude equation at fifth order, \eqref{firstamplitudeeq}. With this, we can solve the equation for the remainder at order five, \eqref{rem_eq}, by setting $B_1(X, \Xi) = \left(\Lambda_1(X, \Xi) + i \, \Omega_1(X, \Xi)\right) \, e^{i \, \Phi(X,\Xi)}$, which transforms that equation into
            \begin{align}
                \alpha_1 \, \Lambda_{1_{X \! X}} - \gamma_{1, 2, -} \, \Omega_{1_X} + \gamma_{1, 1} \, \Lambda_1 &= 0, \label{real}
                \\
                \alpha_1 \, \Omega_{1_{X \! X}} + \gamma_{1, 2, +} \, \Lambda_{1_X} &= 0, \label{imaginary}
            \end{align}
            where
            \begin{align*}
                \gamma_{1, 1} &= \frac{4 \, \left(\alpha_4 \, \left(\alpha_3 + \alpha_4\right) + 4 \, \alpha_1 \, \alpha_7\right) \, \beta_1^2}{3 \, \alpha_1 \, \beta_3^2},
                \\
                \gamma_{1, 2, \pm} &= 2 \, \alpha_1 \, \Phi_X + \left(\alpha_3 \pm \alpha_4\right) \, \Delta_1^2 + \alpha_2.
            \end{align*}
            Now, if we integrate \eqref{imaginary} once, we obtain
            \begin{align*}
                \Omega_{1_X} = - \dfrac{\gamma_{1, 2, +}}{\alpha_1} \, \Lambda_1,
            \end{align*}
            which transforms \eqref{real} into
            \begin{align*}
                \alpha_1 \, \Lambda_{1_{X \! X}} - 4 \, \alpha_1 \, \beta_1 \, \Lambda_1 = 0,
            \end{align*}
            Therefore,
            \begin{align*}
                \Lambda_1 = \Psi(\Xi) \, e^{2 \, \sqrt{\beta_1} \, X} + \Pi(\Xi) \, e^{- 2 \, \sqrt{\beta_1} \, X},
            \end{align*}
            which implies
            \begin{align*}
                \Omega_1 = 2 \, \eta \, \left(\Psi(\Xi) \, e^{2 \, \sqrt{\beta_1} \, X} - \Pi(\Xi) \, e^{- 2 \, \sqrt{\beta_1} \, X}\right).
            \end{align*}
            With this, we can obtain the leading-order solution of the amplitude equation at order seven, \eqref{messy-A_3} (see \ref{app:7expansion} for an explanation regarding notation and the new variables present in this expansion):
            \begin{multline*}
                \alpha_1 \, A_{3_{X \! X}} + i \, \alpha_2 \, A_{3_X} + i \, \alpha_3 \, \abs{A_1}^2 \, A_{3_X} + i \, \alpha_4 \, A_1^2 \, \bar A_{3_X} + \alpha_5 \, A_3 + 2 \, \alpha_6 \, \abs{A_1}^2 \, A_3 + 3 \, \alpha_7 \, \abs{A_1}^4 \, A_3
                \\
                + i \, \alpha_3 \, \bar A_1 \, A_{1_X} \, A_3 + 2 \, i \, \alpha_4 \, A_1 \, \bar A_{1_X} \, A_3 + \alpha_6 \, A_1^2 \, \bar A_3 + 2 \, \alpha_7 \, \abs{A_1}^2 \, A_1^2 \, \bar A_3 + i \, \alpha_3 \, A_1 \, A_{1_X} \, \bar A_3
                \\
                + i \, \alpha_{1, 4} \, A_{1_X} + i \, \alpha_{2, 4} \, \abs{A_1}^2 \, A_{1_X} + i \, \alpha_{3, 4} \, \abs{A_1}^4 \, A_{1_X} + i \, \alpha_{4, 4} \, A_1^2 \, \bar A_{1_X} + i \, \alpha_{5, 4} \, \abs{A_1}^2 \, A_1^2 \, \bar A_{1_X}
                \\
                + \alpha_{6, 4} \, \bar A_1 \, \left(A_{1_X}\right)^2 + \alpha_{7, 4} \, A_1 \, \abs{A_{1_X}}^2 + \alpha_{8, 4} \, A_1 + \alpha_{9, 4} \, \abs{A_1}^2 \, A_1 + \alpha_{10, 4} \, \abs{A_1}^4 \, A_1 + \alpha_{11, 4} \, \abs{A_1}^6 \, A_1
                \\
                + 2 \, \alpha_1 \, A_{1_{X \! \Xi}} + i \, \alpha_2 \, A_{1_\Xi} + i \, \alpha_3 \, \abs{A_1}^2 \, A_{1_\Xi} + i \, \alpha_4 \, A_1^2 \, \bar A_{1_\Xi} = 0,
            \end{multline*}
            In particular, note that $A_3 = \Delta_{3, 1} \, e^{\Phi \, i}$ solves the seventh-order equation for $A_3$.
            
            Now, we note that the full homogeneous equation for the remainder at order seven is given by \eqref{rem_7_eq} (see \ref{app:rem_order7}). However, for simplicity, as stated in \ref{app:7expansion}, we can assume that $A_2 = B_2 = 0$ as that is, generically, the case (the right-hand side of the differential equations to determine $A_2$ and $B_2$ can be forced to equal zero). This transforms the equation for the remainder at order seven into:
            {\allowdisplaybreaks
            \begin{multline*}
                \alpha_1 \, B_{3_{X \! X}} + i \, \alpha_2 \, B_{3_X} + i \, \alpha_3 \, \abs{A_1}^2 \, B_{3_X} + i \, \alpha_3 \, \bar A_1 \, A_{1_X} \, B_3 + i \, \alpha_3 \, A_1 \, A_{1_X} \, \bar B_3 + i \, \alpha_4 \, A_1^2 \, \bar B_{3_X}
                \\
                + 2 \, i \, \alpha_4 \, A_1 \, \bar A_{1_X} \, B_3 + \alpha_5 \, B_3 + \alpha_6 \, A_1^2 \, \bar B_3 + 2 \, \alpha_6 \, \abs{A_1}^2 \, B_3 + 2 \, \alpha_7 \, \abs{A_1}^2 \, A_1^2 \, \bar B_3 + 3 \, \alpha_7 \, \abs{A_1}^4 \, B_3
                \\
                + i \, \alpha_{1, 3} \, B_{1_{X \! X \! X}} + \alpha_{2, 3} \, B_{1_{X \! X}} + \alpha_{3, 3} \, \abs{A_1}^2 \, B_{1_{X \! X}} + \alpha_{3, 3} \, \bar A_1 \, A_{1_{X \! X}} \, B_1 + \alpha_{3, 3} \, A_1 \, A_{1_{X \! X}} \, \bar B_1
                \\
                + \alpha_{4, 3} \, A_1^2 \, \bar B_{1_{X \! X}} + 2 \, \alpha_{4, 3} \, A_1 \, \bar A_{1_{X \! X}} \, B_1 + i \, \alpha_{5, 3} \, B_{1_X} + i \, \alpha_{6, 3} \, \abs{A_1}^2 \, B_{1_X} + i \, \alpha_{6, 3} \, \bar A_1 \, A_{1_X} \, B_1
                \\
                + i \, \alpha_{6, 3} \, A_1 \, A_{1_X} \, \bar B_1 + i \, \alpha_{7, 3} \, \abs{A_1}^4 \, B_{1_X} + 2 \, i \, \alpha_{7, 3} \, \abs{A_1}^2 \, \bar A_1 \, A_{1_X} \, B_1 + 2 \, i \, \alpha_{7, 3} \, \abs{A_1}^2 \, A_1 \, A_{1_X} \, \bar B_1
                \\
                + i \, \alpha_{8, 3} \, A_1^2 \, \bar B_{1_X} + 2 \, i \, \alpha_{8, 3} \, A_1 \, \bar A_{1_X} \, B_1 + i \, \alpha_{9, 3} \, \abs{A_1}^2 \, A_1^2 \, \bar B_{1_X} + i \, \alpha_{9, 3} \, A_1^3 \, \bar A_{1_X} \, \bar B_1
                \\
                + 3 \, i \, \alpha_{9, 3} \, \abs{A_1}^2 \, A_1 \, \bar A_{1_X} \, B_1 + \alpha_{10, 3} \, A_{1_X}^2 \, \bar B_1 + 2 \, \alpha_{10, 3} \, \bar A_1 \, A_{1_X} \, B_{1_X} + \alpha_{11, 3} \, A_1 \, \bar A_{1_X} \, B_{1_X}
                \\
                + \alpha_{11, 3} \, A_1 \, A_{1_X} \, \bar B_{1_X} + \alpha_{11, 3} \, \abs{A_{1_X}}^2 \, B_1 + \alpha_{12, 3} \, B_1 + \alpha_{13, 3} \, A_1^2 \, \bar B_1 + 2 \, \alpha_{13, 3} \, \abs{A_1}^2 \, B_1
                \\
                + 2 \, \alpha_{14, 3} \, \abs{A_1}^2 \, A_1^2 \, \bar B_1 + 3 \, \alpha_{14, 3} \, \abs{A_1}^4 \, B_1 + 3 \, \alpha_{15, 3} \, \abs{A_1}^4 \, A_1^2 \, \bar B_1 + 4 \, \alpha_{15, 3} \, \abs{A_1}^6 \, B_1
                \\
                + 2 \, \alpha_1 \, B_{1_{X \! \Xi}} + i \, \alpha_2 \, B_{1_\Xi} + i \, \alpha_3 \, \abs{A_1}^2 \, B_{1_\Xi} + i \, \alpha_3 \, \bar A_1 \, A_{1_\Xi} \, B_1 + i \, \alpha_3 \, A_1 \, A_{1_\Xi} \, \bar B_1 + i \, \alpha_4 \, A_1^2 \, \bar B_{1_\Xi}
                \\
                + 2 \, i \, \alpha_4 \, A_1 \, \bar A_{1_\Xi} \, B_1 + i \, \alpha_3 \, \bar A_1 \, A_3 \, B_{1_X} + i \, \alpha_3 \, A_1 \, \bar A_3 \, B_{1_X} + i \, \alpha_3 \, \bar A_1 \, A_{3_X} \, B_1 + i \, \alpha_3 \, A_1 \, A_{3_X} \, \bar B_1
                \\
                + i \, \alpha_3 \, A_{1_X} \, \bar A_3 \, B_1 + i \, \alpha_3 \, A_{1_X} \, A_3 \, \bar B_1 + 2 \, i \, \alpha_4 \, A_1 \, A_3 \, \bar B_{1_X} + 2 \, i \, \alpha_4 \, A_1 \, \bar A_{3_X} \, B_1
                \\
                + 2 \, i \, \alpha_4 \, \bar A_{1_X} \, A_3 \, B_1 + 2 \, \alpha_6 \, \bar A_1 \, A_3 \, B_1 + 2 \, \alpha_6 \, A_1 \, \bar A_3 \, B_1 + 2 \, \alpha_6 \, A_1 \, A_3 \, \bar B_1 + 2 \, \alpha_7 \, A_1^3 \, \bar A_3 \, \bar B_1
                \\
                + 6 \, \alpha_7 \, \abs{A_1}^2 \, A_1 \, \bar A_3 \, B_1 + 6 \, \alpha_7 \, \abs{A_1}^2 \, \bar A_1 \, A_3 \, B_1 + 6 \, \alpha_7 \, \abs{A_1}^2 \, A_1 \, A_3 \, \bar B_1 = 0.
            \end{multline*}
            }
            Now, if we set $B_3 = \left(\Lambda_3(X, \Xi) + i \, \Omega_3(X, \Xi)\right) \, e^{i \, \Phi}$, we obtain two equations given by
            \begin{align}
                \alpha_1 \, \Lambda_{3_{X \! X}} + \gamma_{1, 3} \, \Omega_{3_X} - \gamma_{1, 4} \, \Lambda_3 = \left(\gamma_{1, 5} \, \Psi + \gamma_{1, 6} \, \Psi_\Xi\right) \, e^{2 \, \sqrt{\beta_1} \, X} + \left(\gamma_{1, 7} \, \Pi - \gamma_{1, 6} \, \Pi_\Xi\right) \, e^{- 2 \, \sqrt{\beta_1} \, X}, \label{u3eq}
                \\
                \alpha_1 \, \Omega_{3_{X \! X}} + \gamma_{1, 8} \, \Lambda_{3_X} = \left(\gamma_{1, 9} \, \Psi + \gamma_{1, 8} \, \Psi_\Xi\right) \, e^{2 \, \sqrt{\beta_1} \, X} + \left(\gamma_{1, 10} \, \Pi + \gamma_{1, 8} \, \Pi_\Xi\right) \, e^{- 2 \, \sqrt{\beta_1} \, X}, \label{v3eq}
            \end{align}
            where
            {\allowdisplaybreaks
            \begin{align*}
                \gamma_{1, 3} &= \frac{\left(\alpha_3 - 3 \, \alpha_4\right) \, \beta_1}{\beta_3},
                \\
                \gamma_{1, 4} &= 4 \, \alpha_1 \, \beta_1 - \frac{\gamma_{1, 3} \, \gamma_{1, 8}}{\alpha_1},
                \\
                \gamma_{1, 5} &= \frac{1}{24 \, \alpha_1^3 \, \beta_3^5} \, \left(2 \, \alpha_1 \, \left(- 2 \, \alpha_3 \, \beta_1 \, \left(\alpha_1 \, \beta_3^2 \, \left(28 \, \beta_1^2 \, \left(\alpha_{3, 4} + \alpha_{5, 4}\right) - 15 \, \beta_3 \, \beta_1 \, \left(\alpha_{2, 4} + \alpha_{4, 4}\right)\right.\right.\right.\right.
                \\
                & \quad \left. + 9 \, \beta_3^2 \, \alpha_{1, 4} + 6 \, \alpha_1 \, \beta_3^2 \, \left(3 \, \zeta_\Xi - 2 \, \sqrt{\beta_1} \, \omega_4\right)\right) + \alpha_4 \, \left(4 \, \beta_1^2 \, \beta_3 \, \left(7 \, \beta_3 \, \alpha_{6, 4} + 15 \, \beta_{5, 3}\right)\right.
                \\
                & \quad \left.\left. + 15 \, \beta_3^3 \, \alpha_{8, 4} - 120 \, \beta_1^3 \, \beta_{7, 3} - 30 \, \beta_1 \, \beta_3^2 \, \beta_{3, 3}\right) + 15 \, \alpha_4^2 \, \beta_1^{3/2} \, \beta_3^2 \, \omega_4\right)
                \\
                & \quad - 3 \, \left(2 \, \alpha_1 \, \alpha_4 \, \beta_1 \, \beta_3^2 \, \left(- 28 \, \beta_1^2 \, \left(\alpha_{3, 4} + \alpha_{5, 4}\right) + 15 \, \beta_1 \, \beta_3 \, \left(\alpha_{2, 4} + \alpha_{4, 4}\right) - 9 \, \beta_3^2 \, \alpha_{1, 4}\right.\right.
                \\
                & \quad \left. + 6 \, \alpha_1 \, \beta_3^2 \, \left(2 \, \sqrt{\beta_1} \, \omega_4 - 3 \, \zeta_\Xi\right)\right) + \alpha_4^2 \, \beta_1 \, \left(4 \, \beta_1^2 \, \beta_3 \, \left(7 \, \beta_3 \, \alpha_{6, 4} + 15 \, \beta_{5, 3}\right) + 15 \, \beta_3^3 \, \alpha_{8, 4}\right.
                \\
                & \quad \left. - 120 \, \beta_1^3 \, \beta_{7, 3} - 30 \, \beta_1 \, \beta_3^2 \, \beta_{3, 3}\right) + 16 \, \alpha_1^2 \, \beta_3^2 \, \left(- 2 \, \beta_3^3 \, \alpha_{8, 4} + 4 \, \beta_1^3 \, \beta_{7, 3} - 4 \, \beta_1^2 \, \beta_3 \, \beta_{5, 3}\right.
                \\
                & \quad \left.\left. + 3 \, \beta_1 \, \beta_3^2 \, \beta_{3, 3}\right) + 6 \, \alpha_4^3 \, \beta_1^{5/2} \, \beta_3^2 \,  \omega_4\right) + \alpha_3^2 \, \beta_1 \, \left(\beta_3 \, \left(\beta_3 \, \left(28 \, \beta_1^2 \, \alpha_{6, 4} + 15 \, \beta_3 \, \alpha_{8, 4}\right.\right.\right.
                \\
                & \quad \left.\left.\left.\left. - 30 \, \beta_1 \, \beta_{3, 3} - 6 \, \alpha_4 \, \beta_1^{3/2} \, \omega_4\right) + 60 \, \beta_1^2 \, \beta_{5, 3}\right) - 120 \, \beta_1^3 \, \beta_{7, 3}\right) + 6 \, \alpha_3^3 \, \beta_1^{5/2} \, \beta_3^2 \, \omega_4\right)
                \\
                & \quad + 3 \, \alpha_2 \, \beta_3^3 \, \left(10 \, \left(\alpha_3 - 3 \, \alpha_4\right) \, \alpha_1 \, \beta_1 \, \left(\left(\alpha_3 + \alpha_4\right) \, \zeta_\Xi - 2 \, \beta_1 \, \alpha_{6, 4}\right)\right.
                \\
                & \quad \left. + 8 \, \alpha_1^2 \, \beta_3 \, \left(4 \, \beta_3 \, \alpha_{1, 4} + 3 \, \left(\alpha_4 - \alpha_3\right) \, \beta_1\right) + 5 \, \left(\alpha_3 + \alpha_4\right) \, \left(\alpha_3 - 3 \, \alpha_4\right) \, \beta_1 \, \alpha_{1, 4}\right.
                \\
                & \quad \left.\left. + 72 \, \alpha_1^3 \, \beta_3^2 \, \zeta_\Xi\right) + 3 \, \alpha_2^2 \, \beta_3^3 \, \left(36 \, \alpha_1^2 \, \beta_3^2 + 5 \, \left(\alpha_3 + \alpha_4\right) \, \left(\alpha_3 - 3 \, \alpha_4\right) \, \beta_1\right)\right),
                \\
                \gamma_{1, 6} &= - \frac{\sqrt{\beta_1} \, \left(\left(\alpha_3^2 - 2 \, \alpha_3 \, \alpha_4 - 3 \, \alpha_4^2\right) \, \beta_1 + 8 \, \alpha_1^2 \, \beta_3^2\right)}{2 \, \alpha_1 \, \beta_3^2},
                \\
                \gamma_{1, 7} &= \gamma_{1, 5} - \frac{\left(\alpha_3 - 3 \, \alpha_4\right) \, \beta_1^{3/2} \, \left(\left(\alpha_3 + \alpha_4\right)^2 \, \beta_1 + 4 \, \alpha_1^2 \, \beta_3^2\right)}{\alpha_1^2 \, \beta_3^3} \, \omega_4,
                \\
                \gamma_{1, 8} &= - \frac{\left(\alpha_3 + \alpha_4\right) \, \beta_1}{\beta_3},
                \\
                \gamma_{1, 9} &= \frac{1}{48 \, \alpha_1^4 \, \sqrt{\beta_1} \, \beta_3^6} \, \left(2 \, \alpha_1 \, \left(12 \, \alpha_1^2 \, \left(\alpha_3 + \alpha_4\right) \, \beta_3^2 \, \left(4 \, \beta_1 \, \left(\beta_3 \, \left(\beta_3 \, \left(2 \, \beta_1^2 \, \alpha_{6, 4} + \beta_3 \, \alpha_{8, 4} - 2 \, \beta_1 \, \beta_{3, 3}\right) \right.\right.\right.\right.\right.
                \\
                & \quad \left.\left.\left. + 4 \, \beta_1^2 \, \beta_{5, 3}\right) - 8 \, \beta_1^3 \, \beta_{7, 3}\right) + \alpha_3 \, \beta_1^2 \, \beta_3^2 \, \left(2 \, \sqrt{\beta_1} \, \omega_4 - \zeta_\Xi\right) + \alpha_4 \, \beta_1^2 \, \beta_3^2 \, \left(2 \, \sqrt{\beta_1} \, \omega_4 + 3 \, \zeta_\Xi\right)\right)
                \\
                & \quad - 48 \, \alpha_1^3 \, \beta_1 \, \beta_3^4 \, \left(4 \, \beta_1^2 \, \left(\alpha_{3, 4} + \alpha_{5, 4}\right) - 2 \, \beta_1 \, \beta_3 \, \left(\alpha_{2, 4} + \alpha_{4, 4}\right) + \beta_3^2 \, \alpha_{1, 4}\right)
                \\
                & \quad - 2 \, \alpha_1 \, \left(\alpha_3 + \alpha_4\right) \, \left(\alpha_3 - 3 \, \alpha_4\right) \, \beta_1^2 \, \beta_3^2 \, \left(4 \, \beta_1^2 \, \left(\alpha_{3, 4} + \alpha_{5, 4}\right) - 3 \, \beta_1 \, \beta_3 \, \left(\alpha_{2, 4} + \alpha_{4, 4}\right)\right.
                \\
                & \quad \left. + 3 \, \beta_3^2 \, \alpha_{1, 4}\right) + \left(\alpha_3 + \alpha_4\right)^2 \, \left(\alpha_3 - 3 \, \alpha_4\right) \, \beta_1^2 \, \left(4 \, \beta_1^2 \, \beta_3 \, \left(\beta_3 \, \alpha_{6, 4} + 3 \, \beta_{5, 3}\right) + 3 \, \beta_3^3 \, \alpha_{8, 4}\right.
                \\
                & \quad \left.\left. - 24 \, \beta_1^3 \, \beta_{7, 3} - 6 \, \beta_1 \, \beta_3^2 \, \beta_{3, 3}\right) + 96 \, \alpha_1^4 \, \beta_1 \, \beta_3^6 \, \left(\sqrt{\beta_1} \, \omega_4 - \zeta_\Xi\right)\right)
                \\
                & \quad + 3 \, \alpha_2 \, \beta_1 \, \beta_3^3 \, \left(8 \, \alpha_1^3 \, \beta_3^2 \, \left(5 \, \left(\alpha_3 + \alpha_4\right) \, \zeta_\Xi - 8 \, \beta_1 \, \alpha_{6, 4}\right)\right.
                \\
                & \quad + 2 \, \left(\alpha_3 + \alpha_4\right) \, \left(\alpha_3 - 3 \, \alpha_4\right) \, \alpha_1 \, \beta_1 \, \left(\left(\alpha_3 + \alpha_4\right) \, \zeta_\Xi - 2 \, \beta_1 \, \alpha_{6, 4}\right)
                \\
                & \quad \left. - 8 \, \left(\alpha_3 + \alpha_4\right) \, \alpha_1^2 \, \beta_3 \, \left(\left(\alpha_3 - \alpha_4\right) \, \beta_1 - 2 \, \beta_3 \, \alpha_{1, 4}\right) + \left(\alpha_3 + \alpha_4\right)^2 \, \left(\alpha_3 - 3 \, \alpha_4\right) \, \beta_1 \, \alpha_{1, 4}\right)
                \\
                & \quad \left. + 3 \, \alpha_2^2 \, \left(\alpha_3 + \alpha_4\right) \, \beta_1 \, \beta_3^3 \, \left(20 \, \alpha_1^2 \, \beta_3^2 + \left(\alpha_3 + \alpha_4\right) \, \left(\alpha_3 - 3 \, \alpha_4\right) \, \beta_1\right)\right),
                \\
                \gamma_{1, 10} &= - \gamma_{1, 9} + \frac{2 \, \beta_1 \, \left(\left(\alpha_3 + \alpha_4\right)^2 \, \beta_1 + 4 \, \alpha_1^2 \, \beta_3^2\right)}{\alpha_1 \, \beta_3^2} \, \omega_4.
            \end{align*}
            }
            Next, note that we can integrate \eqref{v3eq} with respect to $X$, to obtain
            \begin{align*}
                \Omega_{3_X} = \frac{1}{\alpha_1} \, \left(\frac{1}{2 \, \sqrt{\beta_1}} \, \left(\left(\gamma_{1, 9} \, \Psi + \gamma_{1, 8} \, \Psi_\Xi\right) \, e^{2 \, \sqrt{\beta_1} \, X} - \left(\gamma_{1, 10} \, \Pi + \gamma_{1, 8} \, \Pi_\Xi\right) \, e^{- 2 \, \sqrt{\beta_1} \, X}\right) - \gamma_{1, 8} \, \Lambda_3\right),
            \end{align*}
            which can be replaced into \eqref{u3eq}, yielding
            \begin{align}
                \alpha_1 \, \Lambda_{3_{X \! X}} - 4 \, \alpha_1 \, \beta_1 \, \Lambda_3 = - 4 \, \alpha_1 \, \sqrt{\beta_1} \, \left(\left(\gamma_{1, 11} \, \Psi + \Psi_\Xi\right) \, e^{2 \sqrt{\beta _1} X} + \left(\gamma_{1, 11} \, \Pi - \Pi_\Xi\right) \, e^{- 2 \, \sqrt{\beta_1} \, X}\right), \label{secularseventhremainder}
            \end{align}
            where
            {\allowdisplaybreaks
            \begin{align*}
                \gamma_{1, 11} &= \frac{1}{384 \, \alpha_1^6 \, \sqrt{\beta_1} \, \beta_3^7} \, \left(3 \, \alpha_2^2 \, \left(- 144 \, \alpha_1^4 \, \beta_3^4 - 28 \, \alpha_1^2 \, \left(\alpha_3 - 3 \, \alpha_4\right) \, \left(\alpha_3 + \alpha_4\right) \, \beta_1 \, \beta_3^2\right.\right.
                \\
                & \quad \left. + \left(\alpha_3 + \alpha_4\right)^2 \, \left(\alpha_3 - 3 \, \alpha_4\right)^2 \, \beta_1^2\right) \, \beta_3^3 - 3 \, \alpha_2 \, \left(288 \, \alpha_1^5 \, \beta_3^4 \, \zeta_\Xi\right.
                \\
                & \quad \left. + 16 \, \alpha_1^4 \, \beta_3^3 \, \left(\left(\alpha_3 - 15 \, \alpha_4\right) \, \beta_1 + 8 \, \beta_3 \, \alpha_{1, 4}\right)\right.
                \\
                & \quad + 8 \, \alpha_1^3 \, \left(\alpha_3 - 3 \, \alpha_4\right) \, \beta_1 \, \beta_3^2 \left(12 \, \beta_1 \, \alpha_{6, 4} + 7 \, \left(\alpha_3 + \alpha_4\right) \, \zeta_\Xi\right)
                \\
                & \quad + 8 \, \alpha_1^2 \, \left(\alpha_3 + \alpha_4\right) \, \left(\alpha_3 - 3 \, \alpha_4\right) \, \beta_1 \, \beta_3 \, \left(\left(\alpha_3 - \alpha_4\right) \, \beta_1 + 4 \, \beta_3 \, \alpha_{1, 4}\right)
                \\
                & \quad - 2 \, \alpha_1 \, \left(\alpha_3 - 3 \, \alpha_4\right)^2 \, \left(\alpha_3 + \alpha_4\right) \, \beta_1^2 \, \left(\left(\alpha_3 + \alpha_4\right) \, \zeta_\Xi - 2 \, \beta_1 \, \alpha_{6, 4}\right)
                \\
                & \quad \left. - \left(\alpha_3 + \alpha_4\right)^2 \, \left(\alpha_3 - 3 \, \alpha_4\right)^2 \, \beta_1^2 \, \alpha_{1, 4}\right) \, \beta_3^3
                \\
                & \quad + 2 \, \alpha_1 \, \left(- 48 \, \alpha_1^3 \, \left(\alpha_3 - 3 \, \alpha_4\right) \, \beta_1 \, \left(4 \, \beta_1^2 \, \left(\alpha_{3, 4} + \alpha_{5, 4}\right)\right.\right.
                \\
                & \quad \left. - 3 \, \beta_1 \, \beta_3 \, \left(\alpha_{2, 4} + \alpha_{4, 4}\right) + 3 \, \beta_3^2 \, \alpha_{1, 4}\right) \, \beta_3^4 + 96 \, \alpha_1^4 \, \left(8 \, \beta_1^3 \, \beta_{7, 3} - 8 \, \beta_1^2 \, \beta_3 \, \beta_{5, 3}\right.
                \\
                & \quad \left. + \beta_3^2 \, \left(6 \, \beta_{3, 3} + \left(9 \, \alpha_4 - 3 \, \alpha_3\right) \, \zeta_\Xi\right) \, \beta_1 - 4 \, \beta_3^3 \, \alpha_{8, 4}\right) \, \beta_3^4
                \\
                & \quad - 2 \, \alpha_1 \, \left(\alpha_3 + \alpha_4\right) \, \left(\alpha_3 - 3 \, \alpha_4\right)^2 \, \beta_1^2 \, \beta_3^2 \, \left(4 \, \left(\alpha_{3, 4} + \alpha_{5, 4}\right) \, \beta_1^2\right.
                \\
                & \quad \left. - 3 \, \beta_3 \, \left(\alpha_{2, 4} + \alpha_{4, 4}\right) \, \beta_1 + 3 \, \beta_3^2 \, \alpha_{1, 4}\right)
                \\
                & \quad - 4 \, \alpha_1^2 \, \left(\alpha_3 - 3 \, \alpha_4\right) \, \left(\alpha_3 + \alpha_4\right) \, \beta_1 \, \beta_3^2 \, \left(60 \, \beta_1^3 \, \beta_{7, 3}\right.
                \\
                & \quad - \beta_1^2 \, \beta_3 \, \left(\beta_3 \, \left(31 \, \alpha_{6, 4} + 7 \, \alpha_{7, 4}\right) + 16 \, \beta_{5, 3}\right)
                \\
                & \quad \left. + \beta_3^2 \, \left(3 \, \left(\alpha_3 - 3 \, \alpha_4\right) \, \zeta_\Xi - 6 \, \beta_{3, 3}\right) \, \beta_1 + 24 \, \beta_3^3 \, \alpha_{8, 4}\right)
                \\
                & \quad - \left(\alpha_3 - 3 \, \alpha_4\right)^2 \, \left(\alpha_3 + \alpha_4\right)^2 \, \beta_1^2 \, \left(24 \, \beta_1^3 \, \beta_{7, 3} - 4 \, \beta_3 \, \left(\beta_3 \, \alpha_{6, 4} + 3 \, \beta_{5, 3}\right) \, \beta_1^2\right.
                \\
                & \quad \left.\left.\left. + 6 \, \beta_1 \, \beta_3^2 \, \beta_{3, 3} - 3 \, \beta_3^3 \, \alpha_{8, 4}\right)\right)\right).
            \end{align*}
            }
            Now, as the term on the right-hand side of \eqref{secularseventhremainder} is secular, we need to set
            \begin{align*}
                \Psi = \Psi_0 \, e^{- \gamma_{1, 11} \, \Xi}, \qquad \text{and} \qquad \Pi = \Pi_0 \, e^{\gamma_{1, 11} \, \Xi},
            \end{align*}
            in order to ensure the solution is bounded.
    
            In summary, we can conclude that the far-field solution inside the localized pattern is given by
            \begin{align}
                \mbf u^* &\sim \left(\varepsilon \, \left(\Delta_{1, 1} + \Lambda_1 + i \, \Omega_1\right) + \varepsilon^3 \, \Delta_{3, 1}\right) \, e^{\Phi \, i} \, e^{i \left(x + \tilde \chi\right)} \, \bs \phi_1^{[1]} + c.c. \notag
                \\
                &= \varepsilon \, \left(\Delta_{1, 1} + (1 + 2 \, \eta \, i) \, \Psi_0 \, e^{2 \, \sqrt{\beta_1} \, X} \, e^{- \gamma_{1, 11} \, \Xi}\right. \notag
                \\
                &\qquad \left. + (1 - 2 \, \eta \, i) \, \Pi_0 \, e^{- 2 \, \sqrt{\beta_1} \, X} \, e^{\gamma_{1, 11} \, \Xi} + \varepsilon^2 \, \Delta_{3, 1}\right) \, e^{\Phi \, i} \, e^{i \left(x - \tilde \chi\right)} \, \bs \phi_1^{[1]} \notag
                \\
                & + c.c., \label{fast_oscillations}
            \end{align}
            where $\tilde \chi$ is the phase shift for the fast oscillations.
            
            The determination of the constants $\Psi_0$ and $\Pi_0$ can be carried out by matching the fronts and the far-field solutions inside the localized patterns. Specifically, they can be found by using \cite{Chapman}
            \begin{align}
                \lim_{X \to \infty} \mbf u\left(x, X, 0^-\right) \sim \lim_{\Xi \to 0^+} \mbf u^*(x, X, \Xi), \label{limit1}
            \end{align}
            for matching \eqref{up-front} to \eqref{fast_oscillations}, and
            \begin{align}
                \lim_{X - L/\varepsilon^2 \to - \infty} \mbf u\left(x, X, L^+\right) \sim \lim_{\Xi \to L^-} \mbf u^*(x, X, \Xi). \label{limit2}
            \end{align}
            for matching \eqref{down-front} to \eqref{fast_oscillations}.
    
            However, these limits depend on the function $\delta = \delta\left(R_1\right)$, which is related to the separation of one parameter from the Maxwell point (see equation \eqref{rem_eq}), and depends on the specific form of the vector field posed in \eqref{geneq}.
    
            \subsection{Summary of \ref{ap:joining_fronts}}
                The calculations we carried out up until Section \ref{sec:second_residual} let us obtain an approximation of the width of the homoclinic snaking close to codimension-two Turing bifurcation points. Nevertheless, from the beginning, we stated that we have been studying two solutions at the same time, an up-front and a down-front (see sub-Section \ref{sub:solving_amplitude}). \ref{ap:joining_fronts} gives the final conditions that need to be met to match these fronts. We have to bear in mind that these conditions depend on $\delta b$, the separation of parameter $b$ from the Maxwell point and, therefore, after determining bounds for that parameter by using \eqref{final_bound}, one can use conditions \eqref{limit1} and \eqref{limit2} to obtain approximations of the homoclinic snaking as a function of $\varepsilon$.
    
                We highlight that, in the end, although the limits \eqref{limit1} and \eqref{limit2} seem complicated to evaluate, what one really needs to do to match the fronts is to equate the corresponding coefficients of the constant and exponential terms in $X$ in \eqref{up-front}, \eqref{down-front}, and \eqref{fast_oscillations}. It is worth noting that, in the limit \eqref{limit2}, one needs to match the corresponding exponentials by making use of the complex conjugates of \eqref{up-front}, \eqref{down-front}, and \eqref{fast_oscillations}, for the equations to be well-posed (see e.g.~\cite{Chapman}).
        
    \section{Expansion up to order 7} \label{app:7expansion}
        To iterate \eqref{DM}, we need to obtain the dominant mode of $A_3$ and use it as part of the initial condition. Therefore, we need to study the amplitude equation at order 7. To do this, we define an extra small variable for the expansion: $\Xi = \varepsilon^2 \, X = \varepsilon^4 \, x$.
        
        Therefore,
        \begin{align*}
            \frac{\partial^2 \mbf u}{\partial x^2} = \frac{\partial^2 \mbf u}{\partial x^2} + 2 \, \varepsilon^2 \, \frac{\partial^2 \mbf u}{\partial x \partial X} + \varepsilon^4 \, \left(\frac{\partial^2 \mbf u}{\partial X^2} + 2 \, \frac{\partial^2 \mbf u}{\partial x \partial \Xi}\right) + 2 \, \varepsilon^6 \, \frac{\partial^2 \mbf u}{\partial X \partial \Xi} + \varepsilon^8 \, \frac{\partial^2 \mbf u}{\partial \Xi^2}.
        \end{align*}
        Thus, the equation at order five gets an extra term given by
        \begin{align*}
            2 \, k^2 \, \hat D \, \frac{\partial^2 \mbf u^{[1]}}{\partial x \partial \Xi},
        \end{align*}
        which does not affect the Ginzburg-Landau equation, \eqref{firstamplitudeeq}, but it does imply that $A_1 = A_1(X, \Xi)$, which yields $\varphi_1 = \varphi_1(X, \Xi) = \varphi_1(X) + \zeta(\Xi)$, where $\zeta$ is a real function still to be determined.
        
        Here, we use the notation $A_X = \partial_X A$ or $A_X = \dfrac{\dd A}{\dd X}$ to denote partial or ordinary derivatives, disregarding the change of notation depending on whether functions depend on one or more variables. With this, we have
        \begin{align*}
            \mbf u^{[5]} &= \abs{A_1}^2 \, \mbf W_0^{[5]} + \abs{A_1}^4 \, \mbf W_{0, 2}^{[5]} + i \, \bar A_1 \, A_{1_X} \, \mbf W_{0, 3}^{[5]} + A_{1_{X \! X}} \, e^{ix} \, \mbf W_1^{[5]} + i \, A_{1_X} \, e^{ix} \, \mbf W_{1, 2}^{[5]}
            \\
            & \quad + i \, \abs{A_1}^2 \, A_{1_X} \, e^{ix} \, \mbf W_{1, 3}^{[5]} + i \, A_1^2 \, \bar A_{1_X} \, e^{ix} \, \mbf W_{1, 4}^{[5]} + A_1 \, e^{ix} \, \mbf W_{1, 5}^{[5]} + \abs{A_1}^2 \, A_1 \, e^{ix} \, \mbf W_{1, 6}^{[5]}
            \\
            & \quad + \abs{A_1}^4 \, A_1 \, e^{ix} \, \mbf W_{1, 7}^{[5]} + i \, A_{1_\Xi} \, \mbf W_{1, 3}^{[3]} \, e^{ix} + A_1^2 \, e^{2ix} \, \mbf W_2^{[5]} + \abs{A_1}^2 \, A_1^2 \, e^{2ix} \, \mbf W_{2, 2}^{[5]}
            \\
            & \quad + i \, A_1 \, A_{1_X} \, e^{2ix} \, \mbf W_{2, 3}^{[5]} + A_1^3 \, e^{3ix} \, \mbf W_3^{[5]} + \abs{A_1}^2 \, A_1^3 \, e^{3ix} \, \mbf W_{3, 2}^{[5]} + i \, A_1^2 \, A_{1_X} \, e^{3ix} \, \mbf W_{3, 3}^{[5]}
            \\
            & \quad + 2 \, A_1 \, \bar A_2 \, \mbf W_0^{[4]} + 4 \, \abs{A_1}^2 \, A_1 \, \bar A_2 \, \mbf W_{0, 2}^{[4]} + i \, A_{1_X} \, \bar A_2 \, \mbf W_{0, 3}^{[4]} + i \, \bar A_1 \, A_{2_X} \, \mbf W_{0, 3}^{[4]}
            \\
            & \quad + 2 \, A_1 \, \bar A_3 \, \mbf W_0^{[3]} + \abs{A_2}^2 \, \mbf W_0^{[3]} + 2 \, A_2 \, \bar A_3 \, \mbf W_0^{[2]} + 2 \, A_1 \, \bar A_4 \, \mbf W_0^{[2]}
            \\
            & \quad + A_2 \, e^{ix} \, \mbf W_1^{[4]} + A_1^2 \, \bar A_2 \, e^{ix} \, \mbf W_{1, 2}^{[4]} + 2 \, \abs{A_1}^2 \, A_2 \, e^{ix} \, \mbf W_{1, 2}^{[4]} + i \, A_{2_X} \, e^{ix} \, \mbf W_{1, 3}^{[4]}
            \\
            & \quad + A_3 \, e^{ix} \, \mbf W_1^{[3]} + 2 \, A_1 \, \abs{A_2}^2 \, e^{ix} \, \mbf W_{1, 2}^{[3]} + \bar A_1 \, A_2^2 \, e^{ix} \, \mbf W_{1, 2}^{[3]} + A_1^2 \, \bar A_3 \, e^{ix} \, \mbf W_{1, 2}^{[3]}
            \\
            & \quad + 2 \, \abs{A_1}^2 \, A_3 \, e^{ix} \, \mbf W_{1, 2}^{[3]} + i \, A_{3_X} \, e^{ix} \, \mbf W_{1, 3}^{[3]} + A_4 \, e^{ix} \, \mbf W_1^{[2]} + A_5 \, e^{ix} \, \bs \phi_1^{[1]} + 2 \, A_1 \, A_2 \, e^{2ix} \, \mbf W_2^{[4]}
            \\
            & \quad + 3 \, \abs{A_1}^2 \, A_1 \, A_2 \, e^{2ix} \, \mbf W_{2, 2}^{[4]} + A_1^3 \, \bar A_2 \, e^{2ix} \, \mbf W_{2, 2}^{[4]} + i \, A_1 \, A_{2_X} \, e^{2ix} \, \mbf W_{2, 3}^{[4]} + i \, A_{1_X} \, A_2 \, e^{2ix} \, \mbf W_{2, 3}^{[4]}
            \\
            & \quad + 2 \, A_1 \, A_3 \, e^{2ix} \, \mbf W_2^{[3]} + A_2^2 \, e^{2ix} \, \mbf W_2^{[3]} + 2 \, A_1 \, A_4 \, e^{2ix} \, \mbf W_2^{[2]} + 2 \, A_2 \, A_3 \, e^{2ix} \, \mbf W_2^{[2]}
            \\
            & \quad + 3 \, A_1^2 \, A_2 \, e^{3ix} \, \mbf W_3^{[4]} + 3 \, A_1^2 \, A_3 \, e^{3ix} \, \mbf W_3^{[3]} + 3 \, A_1 \, A_2^2 \, e^{3ix} \, \mbf W_3^{[3]} + \ldots + c.c.,
        \end{align*}
        where `$\ldots$' represents terms that are multiples of $e^{4ix}$ or $e^{5ix}$, and
        \begin{align*}
            \mathcal M_0 \, \mbf W_0^{[5]} &= - a_1 \, \jac \mbf f_{1, 0}(\mbf 0) \, \mbf W_0^{[4]} - b_1 \, \jac \mbf f_{0, 1}(\mbf 0) \, \mbf W_0^{[4]} - a_2 \, \jac \mbf f_{1, 0}(\mbf 0) \, \mbf W_0^{[3]} - b_2 \, \jac \mbf f_{0, 1}(\mbf 0) \, \mbf W_0^{[3]}
            \\
            & \quad - a_3 \, \jac \mbf f_{1, 0}(\mbf 0) \, \mbf W_0^{[2]} - b_3 \, \jac \mbf f_{0, 1}(\mbf 0) \, \mbf W_0^{[2]} - a_1^2 \, \jac \mbf f_{2, 0}(\mbf 0) \, \mbf W_0^{[3]} - a_1 \, b_1 \, \jac \mbf f_{1, 1}(\mbf 0) \, \mbf W_0^{[3]}
            \\
            & \quad - b_1^2 \,  \jac \mbf f_{0, 2}(\mbf 0) \, \mbf W_0^{[3]} - 2 \, a_1 \, a_2 \, \jac \mbf f_{2, 0}(\mbf 0) \, \mbf W_0^{[2]} - \left(a_1 \, b_2 + a_2 \, b_1\right) \, \jac \mbf f_{1, 1}(\mbf 0) \, \mbf W_0^{[2]}
            \\
            & \quad - 2 \, b_1 \, b_2 \, \jac \mbf f_{0, 2}(\mbf 0) \, \mbf W_0^{[2]} - a_1^3 \, \jac \mbf f_{3, 0}(\mbf 0) \, \mbf W_0^{[2]} - a_1^2 \, b_1 \, \jac \mbf f_{2, 1}(\mbf 0) \, \mbf W_0^{[2]}
            \\
            & \quad - a_1 \, b_1^2 \, \jac \mbf f_{1, 2}(\mbf 0) \, \mbf W_0^{[2]} - b_1^3 \, \jac \mbf f_{0, 3}(\mbf 0) \, \mbf W_0^{[2]} - 2 \, \mbf F_{2, 0, 0}\left(\bs \phi_1^{[1]}, \mbf W_1^{[4]}\right)
            \\
            & \quad - 2 \, \mbf F_{2, 0, 0}\left(\mbf W_1^{[2]}, \mbf W_1^{[3]}\right) - 2 \, a_1 \, \mbf F_{2, 1, 0}\left(\bs \phi_1^{[1]}, \mbf W_1^{[3]}\right) - 2 \, b_1 \, \mbf F_{2, 0, 1}\left(\bs \phi_1^{[1]}, \mbf W_1^{[3]}\right)
            \\
            & \quad - a_1 \, \mbf F_{2, 1, 0}\left(\mbf W_1^{[2]}, \mbf W_1^{[2]}\right) - b_1 \, \mbf F_{2, 0, 1}\left(\mbf W_1^{[2]}, \mbf W_1^{[2]}\right) - 2 \, a_2 \, \mbf F_{2, 1, 0}\left(\bs \phi_1^{[1]}, \mbf W_1^{[2]}\right)
            \\
            & \quad - 2 \, b_2 \, \mbf F_{2, 0, 1}\left(\bs \phi_1^{[1]}, \mbf W_1^{[2]}\right) - a_3 \, \mbf F_{2, 1, 0}\left(\bs \phi_1^{[1]}, \bs \phi_1^{[1]}\right) - b_3 \, \mbf F_{2, 0, 1}\left(\bs \phi_1^{[1]}, \bs \phi_1^{[1]}\right)
            \\
            & \quad - 2 \, a_1^2 \, \mbf F_{2, 2, 0}\left(\bs \phi_1^{[1]}, \mbf W_1^{[2]}\right) - 2 \, a_1 \, b_1 \, \mbf F_{2, 1, 1}\left(\bs \phi_1^{[1]}, \mbf W_1^{[2]}\right) - 2 \, b_1^2 \, \mbf F_{2, 0, 2}\left(\bs \phi_1^{[1]}, \mbf W_1^{[2]}\right)
            \\
            & \quad - 2 \, a_1 \, a_2 \, \mbf F_{2, 2, 0}\left(\bs \phi_1^{[1]}, \bs \phi_1^{[1]}\right) - \left(a_1 \, b_2 + a_2 \, b_1\right) \, \mbf F_{2, 1, 1}\left(\bs \phi_1^{[1]}, \bs \phi_1^{[1]}\right)
            \\
            & \quad - 2 \, b_1 \, b_2 \, \mbf F_{2, 0, 2}\left(\bs \phi_1^{[1]}, \bs \phi_1^{[1]}\right) - a_1^3 \, \mbf F_{2, 3, 0}\left(\bs \phi_1^{[1]}, \bs \phi_1^{[1]}\right) - a_1^2 \, b_1 \, \mbf F_{2, 2, 1}\left(\bs \phi_1^{[1]}, \bs \phi_1^{[1]}\right)
            \\
            & \quad - a_1 \, b_1^2 \, \mbf F_{2, 1, 2}\left(\bs \phi_1^{[1]}, \bs \phi_1^{[1]}\right) - b_1^3 \, \mbf F_{2, 0, 3}\left(\bs \phi_1^{[1]}, \bs \phi_1^{[1]}\right),
        \end{align*}
        \begin{align*}
            \mathcal M_0 \, \mbf W_{0, 2}^{[5]} &= - a_1 \, \jac \mbf f_{1, 0}(\mbf 0) \, \mbf W_{0, 2}^{[4]} - b_1 \, \jac \mbf f_{0, 1}(\mbf 0) \, \mbf W_{0, 2}^{[4]} - 2 \, \mbf F_{2, 0, 0}\left(\bs \phi_1^{[1]}, \mbf W_{1, 2}^{[4]}\right)
            \\
            & \quad - 4 \, \mbf F_{2, 0, 0}\left(\mbf W_0^{[2]}, \mbf W_0^{[3]}\right) - 2 \, \mbf F_{2, 0, 0}\left(\mbf W_1^{[2]}, \mbf W_{1, 2}^{[3]}\right) - 2 \, \mbf F_{2, 0, 0}\left(\mbf W_2^{[2]}, \mbf W_2^{[3]}\right)
            \\
            & \quad - 2 \, a_1 \, \mbf F_{2, 1, 0}\left(\bs \phi_1^{[1]}, \mbf W_{1, 2}^{[3]}\right) - 2 \, b_1 \, \mbf F_{2, 0, 1}\left(\bs \phi_1^{[1]}, \mbf W_{1, 2}^{[3]}\right) - 2 \, a_1 \, \mbf F_{2, 1, 0}\left(\mbf W_0^{[2]}, \mbf W_0^{[2]}\right)
            \\
            & \quad - 2 \, b_1 \, \mbf F_{2, 0, 1}\left(\mbf W_0^{[2]}, \mbf W_0^{[2]}\right) - a_1 \, \mbf F_{2, 1, 0}\left(\mbf W_2^{[2]}, \mbf W_2^{[2]}\right) - b_1 \, \mbf F_{2, 0, 1}\left(\mbf W_2^{[2]}, \mbf W_2^{[2]}\right)
            \\
            & \quad - 3 \, \mbf F_{3, 0, 0}\left(\bs \phi_1^{[1]}, \bs \phi_1^{[1]}, 2 \, \mbf W_0^{[3]} + \mbf W_2^{[3]}\right) - 6 \, \mbf F_{3, 0, 0}\left(\bs \phi_1^{[1]}, \mbf W_1^{[2]}, 2 \, \mbf W_0^{[2]} + \mbf W_2^{[2]}\right)
            \\
            & \quad - 3 \, a_1 \, \mbf F_{3, 1, 0}\left(\bs \phi_1^{[1]}, \bs \phi_1^{[1]}, 2 \, \mbf W_0^{[2]} + \mbf W_2^{[2]}\right) - 3 \, b_1 \, \mbf F_{3, 0, 1}\left(\bs \phi_1^{[1]}, \bs \phi_1^{[1]}, 2 \, \mbf W_0^{[2]} + \mbf W_2^{[2]}\right)
            \\
            & \quad - 12 \, \mbf F_{4, 0, 0}\left(\bs \phi_1^{[1]}, \bs \phi_1^{[1]}, \bs \phi_1^{[1]}, \mbf W_1^{[2]}\right) - 3 \, a_1 \, \mbf F_{4, 1, 0}\left(\bs \phi_1^{[1]}, \bs \phi_1^{[1]}, \bs \phi_1^{[1]}, \bs \phi_1^{[1]}\right)
            \\
            & \quad - 3 \, b_1 \, \mbf F_{4, 0, 1}\left(\bs \phi_1^{[1]}, \bs \phi_1^{[1]}, \bs \phi_1^{[1]}, \bs \phi_1^{[1]}\right),
        \end{align*}
        and
        \begin{align*}
            \mathcal M_0 \, \mbf W_{0, 3}^{[5]} &= - a_1 \, \jac \mbf f_{1, 0}(\mbf 0) \, \mbf W_{0, 3}^{[4]} - b_1 \, \jac \mbf f_{0, 1}(\mbf 0) \, \mbf W_{0, 3}^{[4]} - 2 \, \mbf F_{2, 0, 0}\left(\bs \phi_1^{[1]}, \mbf W_{1, 3}^{[4]}\right)
            \\
            & \quad - 2 \, \mbf F_{2, 0, 0}\left(\mbf W_1^{[2]}, \mbf W_{1, 3}^{[3]}\right) - 2 \, a_1 \, \mbf F_{2, 1, 0}\left(\bs \phi_1^{[1]}, \mbf W_{1, 3}^{[3]}\right) - 2 \, b_1 \, \mbf F_{2, 0, 1}\left(\bs \phi_1^{[1]}, \mbf W_{1, 3}^{[3]}\right).
        \end{align*}
        Moreover, to find the resonant solution to \eqref{fifth-order-eq}, we note that the corresponding equation will be given by
        \begin{multline}
            \mathcal M_1 \, \mbf u_1^{[5]} = - \mbf Q_1^{[5]} \, A_1'' - i \, \mbf Q_2^{[5]} \, A_1' - i \, \mbf Q_3^{[5]} \, \abs{A_1}^2 \, A_1' - i \, \mbf Q_4^{[5]} \, A_1^2 \, \bar A_1'
            \\
            - \mbf Q_5^{[5]} \, A_1 - \mbf Q_6^{[5]} \, \abs{A_1}^2 \, A_1 - \mbf Q_7^{[5]} \, \abs{A_1}^4 \, A_1. \label{resonant-order-five}
        \end{multline}
        Unfortunately, we cannot simply solve the equation for each term separately, as we cannot ensure that $\mbf Q_p^{[5]}\in \mbox{im}\left(\mathcal M_1\right)$ for every $p = 1, \ldots, 7$. Nevertheless, from \eqref{firstamplitudeeq}, we know that
        \begin{align*}
            \alpha_1 \, A_1'' + i \, \alpha_2 \, A_1' + i \, \alpha_3 \, \abs{A_1}^2 \, A_1' + i \, \alpha_4 \, A_1^2 \, \bar A_1' + \alpha_5 \, A_1 + \alpha_6 \, \abs{A_1}^2 \, A_1 + \alpha_7 \, \abs{A_1}^4 \, A_1 &= 0,
        \end{align*}
        which implies
        \begin{multline*}
            \left(\alpha_1 \, A_1'' + i \, \alpha_2 \, A_1' + i \, \alpha_3 \, \abs{A_1}^2 \, A_1' + i \, \alpha_4 \, A_1^2 \, \bar A_1'\right.
            \\
            \left. + \alpha_5 \, A_1 + \alpha_6 \, \abs{A_1}^2 \, A_1 + \alpha_7 \, \abs{A_1}^4 \, A_1\right) \, \frac{\bs \psi}{\langle \bs \psi, \bs \psi \rangle} = \mbf 0.
        \end{multline*}
        Therefore, we see that \eqref{resonant-order-five} is equivalent to
        \begin{align*}
            \mathcal M_1 \mbf u_1^{[5]} &= \left(- \mbf Q_1^{[5]} + \alpha_1 \, \frac{\bs \psi}{\langle \bs \psi, \bs \psi\rangle}\right) \, A_1'' + i \, \left(- \mbf Q_2^{[5]} + \alpha_2 \, \frac{\bs \psi}{\langle \bs \psi, \bs \psi\rangle}\right) \, A_1'
            \\
            & \quad + i \, \left(- \mbf Q_3^{[5]} + \alpha_3 \, \frac{\bs \psi}{\langle \bs \psi, \bs \psi\rangle}\right) \, \abs{A_1}^2 \, A_1' + i \, \left(- \mbf Q_4^{[5]} + \alpha_4 \, \frac{\bs \psi}{\langle \bs \psi, \bs \psi\rangle}\right) \, A_1^2 \, \bar A_1'
            \\
            & \quad + \left(- \mbf Q_5^{[5]} + \alpha_5 \, \frac{\bs \psi}{\langle \bs \psi, \bs \psi\rangle}\right) \, A_1 + \left(- \mbf Q_6^{[5]} + \alpha_6 \, \frac{\bs \psi}{\langle \bs \psi, \bs \psi\rangle}\right) \, \abs{A_1}^2 \, A_1
            \\
            & \quad + \left(- \mbf Q_7^{[5]} + \alpha_7 \, \frac{\bs \psi}{\langle \bs \psi, \bs \psi\rangle}\right) \, \abs{A_1}^4 \, A_1,
        \end{align*}
        which fulfills that the coefficient of each term related to $A_1$ is a vector that belongs to $\mbox{im}\left(\mathcal M_1\right)$. Therefore,
        \begin{align*}
            \mathcal M_1 \, \mbf W_1^{[5]} = - \mbf Q_1^{[5]} + \frac{\alpha_1}{\left \langle \bs \psi, \bs \psi\right \rangle} \, \bs \psi,
        \end{align*}
        \begin{align*}
            \mathcal M_1 \, \mbf W_{1, p}^{[5]} = - \mbf Q_p^{[5]} + \frac{\alpha_p}{\left \langle \bs \psi, \bs \psi\right \rangle} \, \bs \psi \quad \text{for } p = 2, \ldots, 7.
        \end{align*}
        Furthermore,
        \begin{align*}
            \mathcal M_2 \, \mbf W_2^{[5]} &= - a_1 \, \jac \mbf f_{1, 0}(\mbf 0) \, \mbf W_2^{[4]} - b_1 \, \jac \mbf f_{0, 1}(\mbf 0) \, \mbf W_2^{[4]} - a_2 \, \jac \mbf f_{1, 0}(\mbf 0) \, \mbf W_2^{[3]} - b_2 \, \jac \mbf f_{0, 1}(\mbf 0) \, \mbf W_2^{[3]}
            \\
            & \quad - a_3 \, \jac \mbf f_{1, 0}(\mbf 0) \, \mbf W_2^{[2]} - b_3 \, \jac \mbf f_{0, 1}(\mbf 0) \, \mbf W_2^{[2]} - a_1^2 \, \jac \mbf f_{2, 0}(\mbf 0) \, \mbf W_2^{[3]} - a_1 \, b_1 \, \jac \mbf f_{1, 1}(\mbf 0) \, \mbf W_2^{[3]}
            \\
            & \quad - b_1^2 \, \jac \mbf f_{0, 2}(\mbf 0) \, \mbf W_2^{[3]} - 2 \, a_1 \, a_2 \, \jac \mbf f_{2, 0}(\mbf 0) \, \mbf W_2^{[2]} - \left(a_1 \, b_2 + a_2 \, b_1\right) \, \jac \mbf f_{1, 1}(\mbf 0) \, \mbf W_2^{[2]}
            \\
            & \quad - 2 \, b_1 \, b_2 \, \jac \mbf f_{0, 2}(\mbf 0) \, \mbf W_2^{[2]} - a_1^3 \, \jac \mbf f_{3, 0}(\mbf 0) \, \mbf W_2^{[2]} - a_1^2 \, b_1 \, \jac \mbf f_{2, 1}(\mbf 0) \, \mbf W_2^{[2]}
            \\
            & \quad - a_1 \, b_1^2 \, \jac \mbf f_{1, 2}(\mbf 0) \, \mbf W_2^{[2]} - b_1^3 \, \jac \mbf f_{0, 3}(\mbf 0) \, \mbf W_2^{[2]} - 2 \, \mbf F_{2, 0, 0}\left(\bs \phi_1^{[1]}, \mbf W_1^{[4]}\right)
            \\
            & \quad - 2 \, \mbf F_{2, 0, 0}\left(\mbf W_1^{[2]}, \mbf W_1^{[3]}\right) - 2 \, a_1 \, \mbf F_{2, 1, 0}\left(\bs \phi_1^{[1]}, \mbf W_1^{[3]}\right) - 2 \, b_1 \, \mbf F_{2, 0, 1}\left(\bs \phi_1^{[1]}, \mbf W_1^{[3]}\right)
            \\
            & \quad - a_1 \, \mbf F_{2, 1, 0}\left(\mbf W_1^{[2]}, \mbf W_1^{[2]}\right) - b_1 \, \mbf F_{2, 0, 1}\left(\mbf W_1^{[2]}, \mbf W_1^{[2]}\right) - 2 \, a_2 \, \mbf F_{2, 1, 0}\left(\bs \phi_1^{[1]}, \mbf W_1^{[2]}\right)
            \\
            & \quad - 2 \, b_2 \, \mbf F_{2, 0, 1}\left(\bs \phi_1^{[1]}, \mbf W_1^{[2]}\right) - a_3 \, \mbf F_{2, 1, 0}\left(\bs \phi_1^{[1]}, \bs \phi_1^{[1]}\right) - b_3 \, \mbf F_{2, 0, 1}\left(\bs \phi_1^{[1]}, \bs \phi_1^{[1]}\right)
            \\
            & \quad - 2 \, a_1^2 \, \mbf F_{2, 2, 0}\left(\bs \phi_1^{[1]}, \mbf W_1^{[2]}\right) - 2 \, a_1 \, b_1 \, \mbf F_{2, 1, 1}\left(\bs \phi_1^{[1]}, \mbf W_1^{[2]}\right) - 2 \, b_1^2 \, \mbf F_{2, 0, 2}\left(\bs \phi_1^{[1]}, \mbf W_1^{[2]}\right)
            \\
            & \quad - 2 \, a_1 \, a_2 \, \mbf F_{2, 2, 0}\left(\bs \phi_1^{[1]}, \bs \phi_1^{[1]}\right) - \left(a_1 \, b_2 + a_2 \, b_1\right) \, \mbf F_{2, 1, 1}\left(\bs \phi_1^{[1]}, \bs \phi_1^{[1]}\right)
            \\
            & \quad - 2 \, b_1 \, b_2 \, \mbf F_{2, 0, 2}\left(\bs \phi_1^{[1]}, \bs \phi_1^{[1]}\right) - a_1^3 \, \mbf F_{2, 3, 0}\left(\bs \phi_1^{[1]}, \bs \phi_1^{[1]}\right) - a_1^2 \, b_1 \, \mbf F_{2, 2, 1}\left(\bs \phi_1^{[1]}, \bs \phi_1^{[1]}\right)
            \\
            & \quad - a_1 \, b_1^2 \, \mbf F_{2, 1, 2}\left(\bs \phi_1^{[1]}, \bs \phi_1^{[1]}\right) - b_1^3 \, \mbf F_{2, 0, 3}\left(\bs \phi_1^{[1]}, \bs \phi_1^{[1]}\right),
        \end{align*}
        \begin{align*}
            \mathcal M_2 \, \mbf W_{2, 2}^{[5]} &= - a_1 \, \jac \mbf f_{1, 0}(\mbf 0) \, \mbf W_{2, 2}^{[4]} - b_1 \, \jac \mbf f_{0, 1}(\mbf 0) \, \mbf W_{2, 2}^{[4]} - 2 \, \mbf F_{2, 0, 0}\left(\bs \phi_1^{[1]}, \mbf W_{1, 2}^{[4]} + \mbf W_3^{[4]}\right)
            \\
            & \quad - 4 \, \mbf F_{2, 0, 0}\left(\mbf W_0^{[2]}, \mbf W_2^{[3]}\right) - 2 \, \mbf F_{2, 0, 0}\left(\mbf W_1^{[2]}, \mbf W_{1, 2}^{[3]} + \mbf W_3^{[3]}\right) - 4 \, \mbf F_{2, 0, 0}\left(\mbf W_2^{[2]}, \mbf W_0^{[3]}\right)
            \\
            & \quad - 2 \, a_1 \, \mbf F_{2, 1, 0}\left(\bs \phi_1^{[1]}, \mbf W_{1, 2}^{[3]} + \mbf W_3^{[3]}\right) - 2 \, b_1 \, \mbf F_{2, 0, 1}\left(\bs \phi_1^{[1]}, \mbf W_{1, 2}^{[3]} + \mbf W_3^{[3]}\right)
            \\
            & \quad - 4 \, a_1 \, \mbf F_{2, 1, 0}\left(\mbf W_0^{[2]}, \mbf W_2^{[2]}\right) - 4 \, b_1 \, \mbf F_{2, 0, 1}\left(\mbf W_0^{[2]}, \mbf W_2^{[2]}\right)
            \\
            & \quad - 6 \, \mbf F_{3, 0, 0}\left(\bs \phi_1^{[1]}, \bs \phi_1^{[1]}, \mbf W_0^{[3]} + \mbf W_2^{[3]}\right) - 12 \, \mbf F_{3, 0, 0}\left(\bs \phi_1^{[1]}, \mbf W_1^{[2]}, \mbf W_0^{[2]} + \mbf W_2^{[2]}\right)
            \\
            & \quad - 6 \, a_1 \, \mbf F_{3, 1, 0}\left(\bs \phi_1^{[1]}, \bs \phi_1^{[1]}, \mbf W_0^{[2]} + \mbf W_2^{[2]}\right) - 6 \, b_1 \, \mbf F_{3, 0, 1}\left(\bs \phi_1^{[1]}, \bs \phi_1^{[1]}, \mbf W_0^{[2]} + \mbf W_2^{[2]}\right)
            \\
            & \quad - 16 \, \mbf F_{4, 0, 0}\left(\bs \phi_1^{[1]}, \bs \phi_1^{[1]}, \bs \phi_1^{[1]}, \mbf W_1^{[2]}\right) - 4 \, a_1 \, \mbf F_{4, 1, 0}\left(\bs \phi_1^{[1]}, \bs \phi_1^{[1]}, \bs \phi_1^{[1]}, \bs \phi_1^{[1]}\right)
            \\
            & \quad - 4 \, b_1 \, \mbf F_{4, 0, 1}\left(\bs \phi_1^{[1]}, \bs \phi_1^{[1]}, \bs \phi_1^{[1]}, \bs \phi_1^{[1]}\right),
        \end{align*}
        \begin{align*}
            \mathcal M_2 \, \mbf W_{2, 3}^{[5]} &= - a_1 \, \jac \mbf f_{1, 0}(\mbf 0) \, \mbf W_{2, 3}^{[4]} - b_1 \, \jac \mbf f_{0, 1}(\mbf 0) \, \mbf W_{2, 3}^{[4]} - 2 \, \mbf F_{2, 0, 0}\left(\bs \phi_1^{[1]}, \mbf W_{1, 3}^{[4]}\right)
            \\
            & \quad - 2 \, \mbf F_{2, 0, 0}\left(\mbf W_1^{[2]}, \mbf W_{1, 3}^{[3]}\right) - 2 \, a_1 \, \mbf F_{2, 1, 0}\left(\bs \phi_1^{[1]}, \mbf W_{1, 3}^{[3]}\right) - 2 \, b_1 \, \mbf F_{2, 0, 1}\left(\bs \phi_1^{[1]}, \mbf W_{1, 3}^{[3]}\right)
            \\
            & \quad - 8 \, k^2 \, \hat D \, \mbf W_2^{[3]},
        \end{align*}
        \begin{align*}
            \mathcal M_3 \, \mbf W_3^{[5]} &= - a_1 \, \jac \mbf f_{1, 0}(\mbf 0) \, \mbf W_3^{[4]} - b_1 \, \jac \mbf f_{0, 1}(\mbf 0) \, \mbf W_3^{[4]} - a_2 \, \jac \mbf f_{1, 0}(\mbf 0) \, \mbf W_3^{[3]} - b_2 \, \jac \mbf f_{0, 1}(\mbf 0) \, \mbf W_3^{[3]}
            \\
            & \quad - a_1^2 \, \jac \mbf f_{2, 0}(\mbf 0) \, \mbf W_3^{[3]} - a_1 \, b_1 \, \jac \mbf f_{1, 1}(\mbf 0) \, \mbf W_3^{[3]} - b_1^2 \, \jac \mbf f_{0, 2}(\mbf 0) \, \mbf W_3^{[3]}
            \\
            & \quad - 2 \, \mbf F_{2, 0, 0}\left(\bs \phi_1^{[1]}, \mbf W_2^{[4]}\right) - 2 \, \mbf F_{2, 0, 0}\left(\mbf W_1^{[2]}, \mbf W_2^{[3]}\right) - 2 \, \mbf F_{2, 0, 0}\left(\mbf W_2^{[2]}, \mbf W_1^{[3]}\right)
            \\
            & \quad - 2 \, a_1 \, \mbf F_{2, 1, 0}\left(\bs \phi_1^{[1]}, \mbf W_2^{[3]}\right) - 2 \, b_1 \, \mbf F_{2, 0, 1}\left(\bs \phi_1^{[1]}, \mbf W_2^{[3]}\right) - 2 \, a_1 \, \mbf F_{2, 1, 0}\left(\mbf W_1^{[2]}, \mbf W_2^{[2]}\right)
            \\
            & \quad - 2 \, b_1 \, \mbf F_{2, 0, 1}\left(\mbf W_1^{[2]}, \mbf W_2^{[2]}\right) - 2 \, a_2 \, \mbf F_{2, 1, 0}\left(\bs \phi_1^{[1]}, \mbf W_2^{[2]}\right) - 2 \, b_2 \, \mbf F_{2, 0, 1}\left(\bs \phi_1^{[1]}, \mbf W_2^{[2]}\right)
            \\
            & \quad - 2 \, a_1^2 \, \mbf F_{2, 2, 0}\left(\bs \phi_1^{[1]}, \mbf W_2^{[2]}\right) - 2 \, a_1 \, b_1 \, \mbf F_{2, 1, 1}\left(\bs \phi_1^{[1]}, \mbf W_2^{[2]}\right) - 2 \, b_1^2 \, \mbf F_{2, 0, 2}\left(\bs \phi_1^{[1]}, \mbf W_2^{[2]}\right)
            \\
            & \quad - 3 \, \mbf F_{3, 0, 0}\left(\bs \phi_1^{[1]}, \bs \phi_1^{[1]}, \mbf W_1^{[3]}\right) - 3 \, \mbf F_{3, 0, 0}\left(\bs \phi_1^{[1]}, \mbf W_1^{[2]}, \mbf W_1^{[2]}\right)
            \\
            & \quad - 3 \, a_1 \, \mbf F_{3, 1, 0}\left(\bs \phi_1^{[1]}, \bs \phi_1^{[1]}, \mbf W_1^{[2]}\right) - 3 \, b_1 \, \mbf F_{3, 0, 1}\left(\bs \phi_1^{[1]}, \bs \phi_1^{[1]}, \mbf W_1^{[2]}\right)
            \\
            & \quad - a_2 \, \mbf F_{3, 1, 0}\left(\bs \phi_1^{[1]}, \bs \phi_1^{[1]}, \bs \phi_1^{[1]}\right) - b_2 \, \mbf F_{3, 0, 1}\left(\bs \phi_1^{[1]}, \bs \phi_1^{[1]}, \bs \phi_1^{[1]}\right) - a_1^2 \, \mbf F_{3, 2, 0}\left(\bs \phi_1^{[1]}, \bs \phi_1^{[1]}, \bs \phi_1^{[1]}\right)
            \\
            & \quad - a_1 \, b_1 \, \mbf F_{3, 1, 1}\left(\bs \phi_1^{[1]}, \bs \phi_1^{[1]}, \bs \phi_1^{[1]}\right) - b_1^2 \, \mbf F_{3, 0, 2}\left(\bs \phi_1^{[1]}, \bs \phi_1^{[1]}, \bs \phi_1^{[1]}\right),
        \end{align*}
        \begin{align*}
            \mathcal M_3 \, \mbf W_{3, 2}^{[5]} &= - 2 \, \mbf F_{2, 0, 0}\left(\bs \phi_1^{[1]}, \mbf W_{2, 2}^{[4]} + \mbf W_4^{[4]}\right) - 4 \, \mbf F_{2, 0, 0}\left(\mbf W_0^{[2]}, \mbf W_3^{[3]}\right) - 2 \, \mbf F_{2, 0, 0}\left(\mbf W_2^{[2]}, \mbf W_{1, 2}^{[3]}\right)
            \\
            & \quad - 3 \, \mbf F_{3, 0, 0}\left(\bs \phi_1^{[1]}, \bs \phi_1^{[1]}, \mbf W_{1, 2}^{[3]} + 2 \, \mbf W_3^{[3]}\right) - 3 \, \mbf F_{3, 0, 0}\left(\bs \phi_1^{[1]}, \mbf W_2^{[2]}, 4 \, \mbf W_0^{[2]} + \mbf W_2^{[2]}\right)
            \\
            & \quad - 4 \, \mbf F_{4, 0, 0}\left(\bs \phi_1^{[1]}, \bs \phi_1^{[1]}, \bs \phi_1^{[1]}, 2 \, \mbf W_0^{[2]} + 3 \, \mbf W_2^{[2]}\right) - 5 \, \mbf F_{5, 0, 0}\left(\bs \phi_1^{[1]}, \bs \phi_1^{[1]}, \bs \phi_1^{[1]}, \bs \phi_1^{[1]}, \bs \phi_1^{[1]}\right),
        \end{align*}
        and
        \begin{align*}
            \mathcal M_3 \, \mbf W_{3, 3}^{[5]} &= - 2 \, \mbf F_{2, 0, 0}\left(\bs \phi_1^{[1]}, \mbf W_{2, 3}^{[4]}\right) - 2 \, \mbf F_{2, 0, 0}\left(\mbf W_2^{[2]}, \mbf W_{1, 3}^{[3]}\right) - 3 \, \mbf F_{3, 0, 0}\left(\bs \phi_1^{[1]}, \bs \phi_1^{[1]}, \mbf W_{1, 3}^{[3]}\right)
            \\
            & \quad - 18 \, k^2 \, \hat D \, \mbf W_3^{[3]}.
        \end{align*}
        
        \paragraph{Order $\mathcal O\left(\varepsilon^6\right)$.} At this order, \eqref{eqtoexpand} becomes
        {
        \allowdisplaybreaks
        \begin{align*}
            \mbf 0 &= \jac \mbf f_{0, 0}(\mbf 0) \, \mbf u^{[6]} + a_1 \, \jac \mbf f_{1, 0}(\mbf 0) \, \mbf u^{[5]} + b_1 \, \jac \mbf f_{0, 1}(\mbf 0) \, \mbf u^{[5]} + a_2 \, \jac \mbf f_{1, 0}(\mbf 0) \, \mbf u^{[4]} + b_2 \, \jac \mbf f_{0, 1}(\mbf 0) \, \mbf u^{[4]}
            \\
            & \quad + a_3 \, \jac \mbf f_{1, 0}(\mbf 0) \, \mbf u^{[3]} + b_3 \, \jac \mbf f_{0, 1}(\mbf 0) \, \mbf u^{[3]} + a_4 \, \jac \mbf f_{1, 0}(\mbf 0) \, \mbf u^{[2]} + b_4 \, \jac \mbf f_{0, 1}(\mbf 0) \, \mbf u^{[2]}
            \\
            & \quad + a_5 \, \jac \mbf f_{1, 0}(\mbf 0) \, \mbf u^{[1]} + b_5 \, \jac \mbf f_{0, 1}(\mbf 0) \, \mbf u^{[1]} + a_1^2 \, \jac \mbf f_{2, 0}(\mbf 0) \, \mbf u^{[4]} + a_1 \, b_1 \, \jac \mbf f_{1, 1}(\mbf 0) \, \mbf u^{[4]}
            \\
            & \quad + b_1^2 \, \jac \mbf f_{0, 2}(\mbf 0) \, \mbf u^{[4]} + 2 \, a_1 \, a_2 \, \jac \mbf f_{2, 0}(\mbf 0) \, \mbf u^{[3]} + \left(a_1 \, b_2 + a_2 \, b_1\right) \, \jac \mbf f_{1, 1}(\mbf 0) \, \mbf u^{[3]}
            \\
            & \quad + 2 \, b_1 \, b_2 \, \jac \mbf f_{0, 2}(\mbf 0) \, \mbf u^{[3]} + 2 \, a_1 \, a_3 \, \jac \mbf f_{2, 0}(\mbf 0) \, \mbf u^{[2]} + \left(a_1 \, b_3 + a_3 \, b_1\right) \, \jac \mbf f_{1, 1}(\mbf 0) \, \mbf u^{[2]}
            \\
            & \quad + 2 \, b_1 \, b_3 \, \jac \mbf f_{0, 2}(\mbf 0) \, \mbf u^{[2]} + 2 \, a_1 \, a_4 \, \jac \mbf f_{2, 0}(\mbf 0) \, \mbf u^{[1]} + \left(a_1 \, b_4 + a_4 \, b_1\right) \, \jac \mbf f_{1, 1}(\mbf 0) \, \mbf u^{[1]}
            \\
            & \quad + 2 \, b_1 \, b_4 \, \jac \mbf f_{0, 2}(\mbf 0) \, \mbf u^{[1]} + a_2^2 \, \jac \mbf f_{2, 0}(\mbf 0) \, \mbf u^{[2]} + a_2 \, b_2 \, \jac \mbf f_{1, 1}(\mbf 0) \, \mbf u^{[2]} + b_2^2 \, \jac \mbf f_{0, 2}(\mbf 0) \, \mbf u^{[2]}
            \\
            & \quad + 2 \, a_2 \, a_3 \, \jac \mbf f_{2, 0}(\mbf 0) \, \mbf u^{[1]} + \left(a_2 \, b_3 + a_3 \, b_2\right) \, \jac \mbf f_{1, 1}(\mbf 0) \, \mbf u^{[1]} + 2 \, b_2 \, b_3 \, \jac \mbf f_{0, 2}(\mbf 0) \, \mbf u^{[1]}
            \\
            & \quad + a_1^3 \, \jac \mbf f_{3, 0}(\mbf 0) \, \mbf u^{[3]} + a_1^2 \, b_1 \, \jac \mbf f_{2, 1}(\mbf 0) \, \mbf u^{[3]} + a_1 \, b_1^2 \, \jac \mbf f_{1, 2}(\mbf 0) \, \mbf u^{[3]} + b_1^3 \, \jac \mbf f_{0, 3}(\mbf 0) \, \mbf u^{[3]}
            \\
            & \quad + 3 \, a_1^2 \, a_2 \, \jac \mbf f_{3, 0}(\mbf 0) \, \mbf u^{[2]} + 2 \, a_1 \, a_2 \, b_1 \, \jac \mbf f_{2, 1}(\mbf 0) \, \mbf u^{[2]} + a_2 \, b_1^2 \, \jac \mbf f_{1, 2}(\mbf 0) \, \mbf u^{[2]}
            \\
            & \quad + a_1^2 \, b_2 \, \jac \mbf f_{2, 1}(\mbf 0) \, \mbf u^{[2]} + 2 \, a_1 \, b_1 \, b_2 \, \jac \mbf f_{1, 2}(\mbf 0) \, \mbf u^{[2]} + 3 \, b_1^2 \, b_2 \, \jac \mbf f_{0, 3}(\mbf 0) \, \mbf u^{[2]}
            \\
            & \quad + 3 \, a_1^2 \, a_3 \, \jac \mbf f_{3, 0}(\mbf 0) \, \mbf u^{[1]} + 2 \, a_1 \, a_3 \, b_1 \, \jac \mbf f_{2, 1}(\mbf 0) \, \mbf u^{[1]} + a_3 \, b_1^2 \, \jac \mbf f_{1, 2}(\mbf 0) \, \mbf u^{[1]} + a_1^2 \, b_3 \, \jac \mbf f_{\evs{2, 1}}(\mbf 0) \, \mbf u^{[1]}
            \\
            & \quad + 2 \, a_1 \, b_1 \, b_3 \, \jac \mbf f_{1, 2}(\mbf 0) \, \mbf u^{[1]} + 3 \, b_1^2 \, b_3 \, \jac \mbf f_{0, 3}(\mbf 0) \, \mbf u^{[1]} + 3 \, a_1 \, a_2^2 \, \jac \mbf f_{\evs{3, 0}}(\mbf 0) \, \mbf u^{[1]}
            \\
            & \quad + 2 \, a_1 \, a_2 \, b_2 \, \jac \mbf f_{2, 1}(\mbf 0) \, \mbf u^{[1]} + a_1 \, b_2^2 \, \jac \mbf f_{1, 2}(\mbf 0) \, \mbf u^{[1]} + a_2^2 \, b_1 \, \jac \mbf f_{2, 1}(\mbf 0) \, \mbf u^{[1]}
            \\
            & \quad + 2 \, a_2 \, b_1 \, b_2 \, \jac \mbf f_{1, 2}(\mbf 0) \, \mbf u^{[1]} + 3 \, b_1 \, b_2^2 \, \jac \mbf f(\mbf 0) \, \mbf u^{[1]} + a_1^4 \, \jac \mbf f_{4, 0}(\mbf 0) \, \mbf u^{[2]} + a_1^3 \, b_1 \, \jac \mbf f_{3, 1}(\mbf 0) \, \mbf u^{[2]}
            \\
            & \quad + a_1^2 \, b_1^2 \, \jac \mbf f_{2, 2}(\mbf 0) \, \mbf u^{[2]} + a_1 \, b_1^3 \, \jac \mbf f_{1, 3}(\mbf 0) \, \mbf u^{[2]} + b_1^4 \, \jac \mbf f_{0, 4}(\mbf 0) \, \mbf u^{[2]} + 4 \, a_1^3 \, a_2 \, \jac \mbf f_{4, 0}(\mbf 0) \, \mbf u^{[1]}
            \\
            & \quad + 3 \, a_1^2 \, a_2 \, b_1 \, \jac \mbf f_{3, 1}(\mbf 0) \, \mbf u^{[1]} + 2 \, a_1 \, a_2 \, b_1^2 \, \jac \mbf f_{2, 2}(\mbf 0) \, \mbf u^{[1]} + a_2 \, b_1^3 \, \jac \mbf f_{1, 3}(\mbf 0) \, \mbf u^{[1]}
            \\
            & \quad + a_1^3 \, b_2 \, \jac \mbf f_{3, 1}(\mbf 0) \, \mbf u^{[1]} + 2 \, a_1^2 \, b_1 \, b_2 \, \jac \mbf f_{2, 2}(\mbf 0) \, \mbf u^{[1]} + 3 \, a_1 \, b_1^2 \, b_2 \, \jac \mbf f_{1, 3}(\mbf 0) \, \mbf u^{[1]}
            \\
            & \quad + 4 \, b_1^3 \, b_2 \, \jac \mbf f_{0, 4}(\mbf 0) \, \mbf u^{[1]} + a_1^5 \, \jac \mbf f_{5, 0}(\mbf 0) \, \mbf u^{[1]} + a_1^4 \, b_1 \, \jac \mbf f_{4, 1}(\mbf 0) \, \mbf u^{[1]} + a_1^3 \, b_1^2 \, \jac \mbf f_{3, 2}(\mbf 0) \, \mbf u^{[1]}
            \\
            & \quad + a_1^2 \, b_1^3 \, \jac \mbf f_{2, 3}(\mbf 0) \, \mbf u^{[1]} + a_1 \, b_1^4 \, \jac \mbf f_{1, 4} \, \mbf u^{[1]} + b_1^5 \, \jac \mbf f_{0, 5}(\mbf 0) \, \mbf u^{[1]} + 2 \, \mbf F_{2, 0, 0}\left(\mbf u^{[1]}, \mbf u^{[5]}\right)
            \\
            & \quad + 2 \, \mbf F_{2, 0, 0}\left(\mbf u^{[2]}, \mbf u^{[4]}\right) + \mbf F_{2, 0, 0}\left(\mbf u^{[3]}, \mbf u^{[3]}\right) + 2 \, a_1 \, \mbf F_{2, 1, 0}\left(\mbf u^{[1]}, \mbf u^{[4]}\right) + 2 \, b_1 \,  \mbf F_{2, 0, 1}\left(\mbf u^{[1]}, \mbf u^{[4]}\right)
            \\
            & \quad + 2 \, a_1 \, \mbf F_{2, 1, 0}\left(\mbf u^{[2]}, \mbf u^{[3]}\right) + 2 \, b_1 \, \mbf F_{2, 0, 1}\left(\mbf u^{[2]}, \mbf u^{[3]}\right) + 2 \, a_2 \, \mbf F_{2, 1, 0}\left(\mbf u^{[1]}, \mbf u^{[3]}\right)
            \\
            & \quad + 2 \, b_2 \, \mbf F_{2, 0, 1}\left(\mbf u^{[1]}, \mbf u^{[3]}\right) + a_2 \, \mbf F_{2, 1, 0}\left(\mbf u^{[2]}, \mbf u^{[2]}\right) + b_2 \, \mbf F_{2, 0, 1}\left(\mbf u^{[2]}, \mbf u^{[2]}\right) + 2 \, a_3 \, \mbf F_{2, 1, 0}\left(\mbf u^{[1]}, \mbf u^{[2]}\right)
            \\
            & \quad + 2 \, b_3 \, \mbf F_{2, 0, 1}\left(\mbf u^{[1]}, \mbf u^{[2]}\right) + a_4 \, \mbf F_{2, 1, 0}\left(\mbf u^{[1]}, \mbf u^{[1]}\right) + b_4 \, \mbf F_{2, 0, 1}\left(\mbf u^{[1]}, \mbf u^{[1]}\right) + 2 \, a_1^2 \, \mbf F_{2, 2, 0}\left(\mbf u^{[1]}, \mbf u^{[3]}\right)
            \\
            & \quad + 2 \, a_1 \, b_1 \, \mbf F_{2, 1, 1}\left(\mbf u^{[1]}, \mbf u^{[3]}\right) + 2 \, b_1^2 \, \mbf F_{2, 0, 2}\left(\mbf u^{[1]}, \mbf u^{[3]}\right) + a_1^2 \, \mbf F_{2, 2, 0}\left(\mbf u^{[2]}, \mbf u^{[2]}\right)
            \\
            & \quad + a_1 \, b_1 \, \mbf F_{2, 1, 1}\left(\mbf u^{[2]}, \mbf u^{[2]}\right) + b_1^2 \, \mbf F_{2, 0, 2}\left(\mbf u^{[2]}, \mbf u^{[2]}\right) + 4 \, a_1 \, a_2 \, \mbf F_{2, 2, 0}\left(\mbf u^{[1]}, \mbf u^{[2]}\right)
            \\
            & \quad + 2 \, \left(a_1 \, b_2 + a_2 \, b_1\right) \, \mbf F_{2, 1, 1}\left(\mbf u^{[1]}, \mbf u^{[2]}\right) + 4 \, b_1 \, b_2 \, \mbf F_{2, 0, 2}\left(\mbf u^{[1]}, \mbf u^{[2]}\right) + 2 \, a_1 \, a_3 \, \mbf F_{2, 2, 0}\left(\mbf u^{[1]}, \mbf u^{[1]}\right)
            \\
            & \quad + \left(a_1 \, b_3 + a_3 \, b_1\right) \, \mbf F_{2, 1, 1}\left(\mbf u^{[1]}, \mbf u^{[1]}\right) + 2 \, b_1 \, b_3 \, \mbf F_{2, 0, 2}\left(\mbf u^{[1]}, \mbf u^{[1]}\right) + a_2^2 \, \mbf F_{2, 2, 0}\left(\mbf u^{[1]}, \mbf u^{[1]}\right)
            \\
            & \quad + a_2 \, b_2 \, \mbf F_{2, 1, 1}\left(\mbf u^{[1]}, \mbf u^{[1]}\right) + b_2^2 \, \mbf F_{2, 0, 2}\left(\mbf u^{[1]}, \mbf u^{[1]}\right) + 2 \, a_1^3 \, \mbf F_{2, 3, 0}\left(\mbf u^{[1]}, \mbf u^{[2]}\right)
            \\
            & \quad + 2 \, a_1^2 \, b_1 \, \mbf F_{2, 0, 1}\left(\mbf u^{[1]}, \mbf u^{[2]}\right) + 2 \, a_1 \, b_1^2 \, \mbf F_{2, 1, 2}\left(\mbf u^{[1]}, \mbf u^{[2]}\right) + 2 \, b_1^3 \, \mbf F_{2, 0, 3}\left(\mbf u^{[1]}, \mbf u^{[2]}\right)
            \\
            & \quad + 3 \, a_1^2 \, a_2 \, \mbf F_{2, 3, 0}\left(\mbf u^{[1]}, \mbf u^{[1]}\right) + 2 \, a_1 \, a_2 \, b_1 \, \mbf F_{2, 2, 1}\left(\mbf u^{[1]}, \mbf u^{[1]}\right) + a_2 \, b_1^2 \, \mbf F_{2, 1, 2}\left(\mbf u^{[1]}, \mbf u^{[1]}\right)
            \\
            & \quad + a_1^2 \, b_2 \, \mbf F_{2, 2, 1}\left(\mbf u^{[1]}, \mbf u^{[1]}\right) + 2 \, a_1 \, b_1 \, b_2 \, \mbf F_{2, 1, 2}\left(\mbf u^{[1]}, \mbf u^{[1]}\right) + 3 \, b_1^2 \, b_2 \, \mbf F_{2, 0, 3}\left(\mbf u^{[1]}, \mbf u^{[1]}\right)
            \\
            & \quad + a_1^4 \, \mbf F_{2, 4, 0}\left(\mbf u^{[1]}, \mbf u^{[1]}\right) + a_1^3 \, b_1 \, \mbf F_{2, 3, 1}\left(\mbf u^{[1]}, \mbf u^{[1]}\right) + a_1^2 \, b_1^2 \, \mbf F_{2, 2, 2}\left(\mbf u^{[1]}, \mbf u^{[1]}\right)
            \\
            & \quad + a_1 \, b_1^3 \, \mbf F_{2, 1, 3}\left(\mbf u^{[1]}, \mbf u^{[1]}\right) + b_1^4 \, \mbf F_{2, 0, 4}\left(\mbf u^{[1]}, \mbf u^{[1]}\right) + 3 \, \mbf F_{3, 0, 0}\left(\mbf u^{[1]}, \mbf u^{[1]}, \mbf u^{[4]}\right)
            \\
            & \quad + 6 \, \mbf F_{3, 0, 0}\left(\mbf u^{[1]}, \mbf u^{[2]}, \mbf u^{[3]}\right) + \mbf F_{3, 0, 0}\left(\mbf u^{[2]}, \mbf u^{[2]}, \mbf u^{[2]}\right) + 3 \, a_1 \, \mbf F_{3, 1, 0}\left(\mbf u^{[1]}, \mbf u^{[1]}, \mbf u^{[3]}\right)
            \\
            & \quad + 3 \, b_1 \, \mbf F_{3, 0, 1}\left(\mbf u^{[1]}, \mbf u^{[1]}, \mbf u^{[3]}\right) + 3 \, a_1 \, \mbf F_{3, 1, 0}\left(\mbf u^{[1]}, \mbf u^{[2]}, \mbf u^{[2]}\right) + 3 \, b_1 \, \mbf F_{3, 0, 1}\left(\mbf u^{[1]}, \mbf u^{[2]}, \mbf u^{[2]}\right)
            \\
            & \quad + 3 \, a_2 \, \mbf F_{3, 1, 0}\left(\mbf u^{[1]}, \mbf u^{[1]}, \mbf u^{[2]}\right) + 3 \, b_2 \, \mbf F_{3, 0, 1}\left(\mbf u^{[1]}, \mbf u^{[1]}, \mbf u^{[2]}\right) + a_3 \, \mbf F_{3, 1, 0}\left(\mbf u^{[1]}, \mbf u^{[1]}, \mbf u^{[1]}\right)
            \\
            & \quad + b_3 \, \mbf F_{3, 0, 1}\left(\mbf u^{[1]}, \mbf u^{[1]}, \mbf u^{[1]}\right) + 3 \, a_1^2 \, \mbf F_{3, 2, 0}\left(\mbf u^{[1]}, \mbf u^{[1]}, \mbf u^{[2]}\right) + 3 \, a_1 \, b_1 \, \mbf F_{3, 1, 1}\left(\mbf u^{[1]}, \mbf u^{[1]}, \mbf u^{[2]}\right)
            \\
            & \quad + 3 \, b_1^2 \, \mbf F_{3, 0, 2}\left(\mbf u^{[1]}, \mbf u^{[1]}, \mbf u^{[2]}\right) + 2 \, a_1 \, a_2 \, \mbf F_{3, 2, 0}\left(\mbf u^{[1]}, \mbf u^{[1]}, \mbf u^{[1]}\right)
            \\
            & \quad + \left(a_1 \, b_2 + a_2 \, b_1\right) \, \mbf F_{3, 1, 1}\left(\mbf u^{[1]}, \mbf u^{[1]}, \mbf u^{[1]}\right) + 2 \, b_1 \, b_2 \, \mbf F_{3, 0, 2}\left(\mbf u^{[1]}, \mbf u^{[1]}, \mbf u^{[1]}\right)
            \\
            & \quad + a_1^3 \, \mbf F_{3, 3, 0}\left(\mbf u^{[1]}, \mbf u^{[1]}, \mbf u^{[1]}\right) + a_1^2 \, b_1 \, \mbf F_{3, 2, 1}\left(\mbf u^{[1]}, \mbf u^{[1]}, \mbf u^{[1]}\right) + a_1 \, b_1^2 \, \mbf F_{3, 1, 2}\left(\mbf u^{[1]}, \mbf u^{[1]}, \mbf u^{[1]}\right)
            \\
            & \quad + b_1^3 \, \mbf F_{3, 0, 3}\left(\mbf u^{[1]}, \mbf u^{[1]}, \mbf u^{[1]}\right) + 4 \, \mbf F_{4, 0, 0}\left(\mbf u^{[1]}, \mbf u^{[1]}, \mbf u^{[1]}, \mbf u^{[3]}\right) + 6 \, \mbf F_{4, 0, 0}\left(\mbf u^{[1]}, \mbf u^{[1]}, \mbf u^{[2]}, \mbf u^{[2]}\right)
            \\
            & \quad + 4 \, a_1 \, \mbf F_{4, 1, 0}\left(\mbf u^{[1]}, \mbf u^{[1]}, \mbf u^{[1]}, \mbf u^{[2]}\right) + 4 \, b_1 \, \mbf F_{4, 0, 1}\left(\mbf u^{[1]}, \mbf u^{[1]}, \mbf u^{[1]}, \mbf u^{[2]}\right)
            \\
            & \quad + a_2 \, \mbf F_{4, 1, 0}\left(\mbf u^{[1]}, \mbf u^{[1]}, \mbf u^{[1]}, \mbf u^{[1]}\right) + b_2 \, \mbf F_{4, 0, 2}\left(\mbf u^{[1]}, \mbf u^{[1]}, \mbf u^{[1]}, \mbf u^{[1]}\right) + a_1^2 \, \mbf F_{4, 2, 0}\left(\mbf u^{[1]}, \mbf u^{[1]}, \mbf u^{[1]}, \mbf u^{[1]}\right)
            \\
            & \quad + a_1 \, b_1 \, \mbf F_{4, 1, 1}\left(\mbf u^{[1]}, \mbf u^{[1]}, \mbf u^{[1]}, \mbf u^{[1]}\right) + b_1^2 \, \mbf F_{4, 0, 2}\left(\mbf u^{[1]}, \mbf u^{[1]}, \mbf u^{[1]}, \mbf u^{[1]}\right)
            \\
            & \quad + 5 \, \mbf F_{5, 0, 0}\left(\mbf u^{[1]}, \mbf u^{[1]}, \mbf u^{[1]}, \mbf u^{[1]}, \mbf u^{[2]}\right) + a_1 \, \mbf F_{5, 1, 0}\left(\mbf u^{[1]}, \mbf u^{[1]}, \mbf u^{[1]}, \mbf u^{[1]}, \mbf u^{[1]}\right)
            \\
            & \quad + b_1 \, \mbf F_{5, 0, 1}\left(\mbf u^{[1]}, \mbf u^{[1]}, \mbf u^{[1]}, \mbf u^{[1]}, \mbf u^{[1]}\right) + \mbf F_{6, 0, 0}\left(\mbf u^{[1]}, \mbf u^{[1]}, \mbf u^{[1]}, \mbf u^{[1]}, \mbf u^{[1]}, \mbf u^{[1]}\right)
            \\
            & \quad + k^2 \, \hat D \, \left(\frac{\partial^2 \mbf u^{[6]}}{\partial x^2} + 2 \, \frac{\partial^2 \mbf u^{[4]}}{\partial x \partial X} + 2 \, \frac{\partial^2 \mbf u^{[2]}}{\partial x \partial \Xi} + \frac{\partial^2 \mbf u^{[2]}}{\partial X^2}\right).
        \end{align*}
        }
        From this, we note that
        {\allowdisplaybreaks
        \begin{align*}
            \mbf u^{[6]} &= \abs{A_1}^2 \, \mbf W_0^{[6]} + \abs{A_1}^4 \, \mbf W_{0, 2}^{[6]} + \abs{A_1}^6 \, \mbf W_{0, 3}^{[6]} + \abs{A_{1_X}}^2 \, \mbf W_{0, 4}^{[6]} + i \, \bar A_1 \, A_{1_X} \, \mbf W_{0, 5}^{[6]}
            \\
            & \quad + i \, \abs{A_1}^2 \, \bar A_1 \, A_{1_X} \, \mbf W_{0, 6}^{[6]} + \bar A_1 \, A_{1_{X \! X}} \, \mbf W_{0, 7}^{[6]} + i \, \bar A_1 \, A_{1_\Xi} \, \mbf W_{0, 3}^{[4]} + A_{1_{X \! X}} \, e^{ix} \, \mbf W_1^{[6]}
            \\
            & \quad + i \, A_{1_X} \, e^{ix} \, \mbf W_{1, 2}^{[6]} + i \, \abs{A_1}^2 \, A_{1_X} \, e^{ix} \, \mbf W_{1, 3}^{[6]} + i \, A_1^2 \, \bar A_{1_X} \, e^{ix} \, \mbf W_{1, 4}^{[6]} + A_1 \, e^{ix} \, \mbf W_{1, 5}^{[6]}
            \\
            & \quad + \abs{A_1}^2 \, A_1 \, e^{ix} \, \mbf W_{1, 6}^{[6]} + \abs{A_1}^4 \, A_1 \, e^{ix} \, \mbf W_{1, 7}^{[6]} + i \, A_{1_\Xi} \, e^{ix} \, \mbf W_{1, 3}^{[4]} + A_1^2 \, e^{2ix} \, \mbf W_2^{[6]}
            \\
            & \quad + \abs{A_1}^2 \, A_1^2 \, e^{2ix} \, \mbf W_{2, 2}^{[6]} + \abs{A_1}^4 \, A_1^2 \, e^{2ix} \, \mbf W_{2, 3}^{[6]} + \left(A_{1_X}\right)^2 \, e^{2ix} \, \mbf W_{2, 4}^{[6]}
            \\
            & \quad + i \, A_1 \, A_{1_X} \, e^{2ix} \, \mbf W_{2, 5}^{[6]} + i \, \abs{A_1}^2 \, A_1 \, A_{1_X} \, e^{2ix} \, \mbf W_{2, 6}^{[6]} + i \, A_1^3 \, \bar A_{1_X} \, e^{2ix} \, \mbf W_{2, 7}^{[6]}
            \\
            & \quad + A_1 \, A_{1_{X \! X}} \, e^{2ix} \, \mbf W_{2, 8}^{[6]} + i \, A_1 \, A_{1_\Xi} \, e^{2ix} \, \mbf W_{2, 3}^{[4]} + 2 \, A_1 \, \bar A_2 \, \mbf W_0^{[5]} + 4 \, \abs{A_1}^2 \, A_1 \, \bar A_2 \, \mbf W_{0, 2}^{[5]}
            \\
            & \quad + i \, A_{1_X} \, \bar A_2 \, \mbf W_{0, 3}^{[5]} + i \, \bar A_1 \, A_{2_X} \, \mbf W_{0, 3}^{[5]} + \abs{A_2}^2 \, \mbf W_0^{[4]} + 2 \, A_1 \, \bar A_3 \, \mbf W_0^{[4]} + 4 \, \abs{A_1}^2 \, A_1 \, \bar A_3 \, \mbf W_{0, 2}^{[4]}
            \\
            & \quad + 4 \, \abs{A_1}^2 \, \abs{A_2}^2 \, \mbf W_{0, 2}^{[4]} + 2 \, A_1^2 \, \bar A_2^2 \, \mbf W_{0, 2}^{[4]} + i \, A_{1_X} \, \bar A_3 \, \mbf W_{0, 3}^{[4]} + i \, \bar A_1 \, A_{3_X} \, \mbf W_{0, 3}^{[4]}
            \\
            & \quad + i \, \bar A_2 \, A_{2_X} \, \mbf W_{0, 3}^{[4]} + 2 \, A_1 \, \bar A_4 \, \mbf W_0^{[3]} + 2 \, A_2 \, \bar A_3 \, \mbf W_0^{[3]} + 2 \, A_1 \, \bar A_5 \, \mbf W_0^{[2]} + 2 \, A_2 \, \bar A_4 \, \mbf W_0^{[2]}
            \\
            & \quad + \abs{A_3}^2 \, \mbf W_0^{[2]} + A_{2_{X \! X}} \, e^{ix} \, \mbf W_1^{[5]} + i \, A_{2_X} \, e^{ix} \, \mbf W_{1, 2}^{[5]} + i \, \abs{A_1}^2 \, A_{2_X} \, e^{ix} \, \mbf W_{1, 3}^{[5]}
            \\
            & \quad + i \, A_1 \, A_{1_X} \, \bar A_2 \, e^{ix} \, \mbf W_{1, 3}^{[5]} + i \, \bar A_1 \, A_{1_X} \, A_2 \, e^{ix} \, \mbf W_{1, 3}^{[5]} + i \, A_1^2 \, \bar A_{2_X} \, e^{ix} \, \mbf W_{1, 4}^{[5]}
            \\
            & \quad + 2 \, i \, A_1 \, \bar A_{1_X} \, A_2 \, e^{ix} \, \mbf W_{1, 4}^{[5]} + A_2 \, e^{ix} \, \mbf W_{1, 5}^{[5]} + 2 \, \abs{A_1}^2 \, A_2 \, e^{ix} \, \mbf W_{1, 6}^{[5]} + A_1^2 \, \bar A_2 \, e^{ix} \, \mbf W_{1, 6}^{[5]}
            \\
            & \quad + 3 \, \abs{A_1}^4 \, A_2 \, e^{ix} \, \mbf W_{1, 7}^{[5]} + 2 \, \abs{A_1}^2 \, A_1^2 \, \bar A_2 \, e^{ix} \, \mbf W_{1, 7}^{[5]} + i \, A_{2_\Xi} \, e^{ix} \, \mbf W_{1, 3}^{[3]} + A_3 \, e^{ix} \, \mbf W_1^{[4]}
            \\
            & \quad + 2 \, \abs{A_1}^2 \, A_3 \, e^{ix} \, \mbf W_{1, 2}^{[4]} + \bar A_1 \, A_2^2 \, e^{ix} \, \mbf W_{1, 2}^{[4]} + 2 \, \abs{A_2}^2 \, A_1 \, e^{ix} \, \mbf W_{1, 2}^{[4]} + A_1^2 \, \bar A_3 \, e^{ix} \, \mbf W_{1, 2}^{[4]}
            \\
            & \quad + i \, A_{3_X} \, e^{ix} \, \mbf W_{1, 3}^{[4]} + A_4 \, e^{ix} \, \mbf W_1^{[3]} + 2 \, \abs{A_1}^2 \, A_4 \, e^{ix} \, \mbf W_{1, 2}^{[3]} + A_1^2 \, \bar A_4 \, e^{ix} \, \mbf W_{1, 2}^{[3]}
            \\
            & \quad + 2 \, A_1 \, A_2 \, \bar A_3 \, e^{ix} \, \mbf W_{1, 2}^{[3]} + 2 \, A_1 \, \bar A_2 \, A_3 \, e^{ix} \, \mbf W_{1, 2}^{[3]} + 2 \, \bar A_1 \, A_2 \, A_3 \, e^{ix} \, \mbf W_{1, 2}^{[3]} + \abs{A_2}^2 \, A_2 \, e^{ix} \, \mbf W_{1, 2}^{[3]}
            \\
            & \quad + i \, A_{4_X} \, e^{ix} \, \mbf W_{1, 3}^{[3]} + A_5 \, e^{ix} \, \mbf W_1^{[2]} + A_6 \, e^{ix} \, \bs \phi_1^{[1]} + 2 \, A_1 \, A_2 \, e^{2ix} \, \mbf W_2^{[5]} + 3 \, \abs{A_1}^2 \, A_1 \, A_2 \, e^{2ix} \, \mbf W_{2, 2}^{[5]}
            \\
            & \quad + A_1^3 \, \bar A_2 \, e^{2ix} \, \mbf W_{2, 2}^{[5]} + i \, A_1 \, A_{2_X} \, e^{2ix} \, \mbf W_{2, 3}^{[5]} + i \, A_{1_X} \, A_2 \, e^{2ix} \, \mbf W_{2, 3}^{[5]} + A_2^2 \, e^{2ix} \, \mbf W_2^{[4]}
            \\
            & \quad + 2 \, A_1 \, A_3 \, e^{2ix} \, \mbf W_2^{[4]} + 3 \, \abs{A_1}^2 \, A_1 \, A_3 \, e^{2ix} \, \mbf W_{2, 2}^{[4]} + A_1^3 \, \bar A_3 \, e^{2ix} \, \mbf W_{2, 2}^{[4]} + 3 \, A_1^2 \, \abs{A_2}^2 \, e^{2ix} \, \mbf W_{2, 2}^{[4]}
            \\
            & \quad + 3 \, \abs{A_1}^2 \, A_2^2 \, e^{2ix} \, \mbf W_{2, 2}^{[4]} + i \, A_{1_X} \, A_3 \, e^{2ix} \, \mbf W_{2, 3}^{[4]} + i \, A_1 \, A_{3_X} \, e^{2ix} \, \mbf W_{2, 3}^{[4]} + i \, A_2 \, A_{2_X} \, e^{2ix} \, \mbf W_{2, 3}^{[4]}
            \\
            & \quad + 2 \, A_1 \, A_4 \, e^{2ix} \, \mbf W_2^{[3]} + 2 \, A_2 \, A_3 \, e^{2ix} \, \mbf W_2^{[3]} + 2 \, A_1 \, A_5 \, e^{2ix} \, \mbf W_2^{[2]} + 2 \, A_2 \, A_4 \, e^{2ix} \, \mbf W_2^{[2]}
            \\
            & \quad + A_3^2 \, e^{2ix} \, \mbf W_2^{[2]} + \ldots + c.c.
        \end{align*}
        }
        where `\ldots' represents terms that are multiples of $e^{3ix}$, $e^{4ix}$, $e^{5ix}$ or $e^{6ix}$, and
        {\allowdisplaybreaks
        \begin{align*}
            \mathcal M_0 \, \mbf W_0^{[6]} &= - a_1 \, \jac \mbf f_{1, 0}(\mbf 0) \, \mbf W_0^{[5]} - b_1 \, \jac \mbf f_{0, 1}(\mbf 0) \, \mbf W_0^{[5]} - a_2 \, \jac \mbf f_{1, 0}(\mbf 0) \, \mbf W_0^{[4]} - b_2 \, \jac \mbf f_{0, 1}(\mbf 0) \, \mbf W_0^{[4]}
            \\
            & \quad - a_3 \, \jac \mbf f_{1, 0}(\mbf 0) \, \mbf W_0^{[3]} - b_3 \, \jac \mbf f_{0, 1}(\mbf 0) \, \mbf W_0^{[3]} - a_4 \, \jac \mbf f_{1, 0}(\mbf 0) \, \mbf W_0^{[2]} - b_4 \, \jac \mbf f_{0, 1}(\mbf 0) \, \mbf W_0^{[2]}
            \\
            & \quad - a_1^2 \, \jac \mbf f_{2, 0}(\mbf 0) \, \mbf W_0^{[4]} - a_1 \, b_1 \, \jac \mbf f_{1, 1}(\mbf 0) \, \mbf W_0^{[4]} - b_1^2 \, \jac \mbf f_{0, 2}(\mbf 0) \, \mbf W_0^{[4]}
            \\
            & \quad - 2 \, a_1 \, a_2 \, \jac \mbf f_{2, 0}(\mbf 0) \, \mbf W_0^{[3]} - \left(a_1 \, b_2 + a_2 \, b_1\right) \, \jac \mbf f_{1, 1}(\mbf 0) \, \mbf W_0^{[3]} - 2 \, b_1 \, b_2 \, \jac \mbf f_{0, 2}(\mbf 0) \, \mbf W_0^{[3]}
            \\
            & \quad - 2 \, a_1 \, a_3 \, \jac \mbf f_{2, 0}(\mbf 0) \, \mbf W_0^{[2]} - \left(a_1 \, b_3 + a_3 \, b_1\right) \, \jac \mbf f_{1, 1}(\mbf 0) \, \mbf W_0^{[2]} - 2 \, b_1 \, b_3 \, \jac \mbf f_{0, 2}(\mbf 0) \, \mbf W_0^{[2]}
            \\
            & \quad - a_2^2 \, \jac \mbf f_{2, 0}(\mbf 0) \, \mbf W_0^{[2]} - a_2 \, b_2 \, \jac \mbf f_{1, 1}(\mbf 0) \, \mbf W_0^{[2]} - b_2^2 \, \jac \mbf f_{0, 2}(\mbf 0) \, \mbf W_0^{[2]} - a_1^3 \, \jac \mbf f_{3, 0}(\mbf 0) \, \mbf W_0^{[3]}
            \\
            & \quad - a_1^2 \, b_1 \, \jac \mbf f_{2, 1}(\mbf 0) \, \mbf W_0^{[3]} - a_1 \, b_1^2 \, \jac \mbf f_{1, 2}(\mbf 0) \, \mbf W_0^{[3]} - b_1^3 \, \jac \mbf f_{0, 3}(\mbf 0) \, \mbf W_0^{[3]}
            \\
            & \quad - 3 \, a_1^2 \, a_2 \, \jac \mbf f_{3, 0}(\mbf 0) \, \mbf W_0^{[2]} - 2 \, a_1 \, a_2 \, b_1 \, \jac \mbf f_{2, 1}(\mbf 0) \, \mbf W_0^{[2]} - a_2 \, b_1^2 \, \jac \mbf f_{1, 2}(\mbf 0) \, \mbf W_0^{[2]}
            \\
            & \quad - a_1^2 \, b_2 \, \jac \mbf f_{2, 1}(\mbf 0) \, \mbf W_0^{[2]} - 2 \, a_1 \, b_1 \, b_2 \, \jac \mbf f_{1, 2}(\mbf 0) \, \mbf W_0^{[2]} - 3 \, b_1^2 \, b_2 \, \jac \mbf f_{0, 3}(\mbf 0) \, \mbf W_0^{[2]}
            \\
            & \quad - a_1^4 \, \jac \mbf f_{4, 0}(\mbf 0) \, \mbf W_0^{[2]} - a_1^3 \, b_1 \, \jac \mbf f_{3, 1}(\mbf 0) \, \mbf W_0^{[2]} - a_1^2 \, b_1^2 \, \jac \mbf f_{2, 2}(\mbf 0) \, \mbf W_0^{[2]}
            \\
            & \quad - a_1 \, b_1^3 \, \jac \mbf f_{1, 3}(\mbf 0) \, \mbf W_0^{[2]} - b_1^4 \, \jac \mbf f_{0, 4}(\mbf 0) \, \mbf W_0^{[2]} - 2 \, \mbf F_{2, 0, 0}\left(\bs \phi_1^{[1]}, \mbf W_{1, 5}^{[5]}\right)
            \\
            & \quad - 2 \, \mbf F_{2, 0, 0}\left(\mbf W_1^{[2]}, \mbf W_1^{[4]}\right) - \mbf F_{2, 0, 0}\left(\mbf W_1^{[3]}, \mbf W_1^{[3]}\right) - 2 \, a_1 \, \mbf F_{2, 1, 0}\left(\bs \phi_1^{[1]}, \mbf W_1^{[4]}\right)
            \\
            & \quad - 2 \, b_1 \, \mbf F_{2, 0, 1}\left(\bs \phi_1^{[1]}, \mbf W_1^{[4]}\right) - 2 \, a_1 \, \mbf F_{2, 1, 0}\left(\mbf W_1^{[2]}, \mbf W_1^{[3]}\right) - 2 \, b_1 \, \mbf F_{2, 0, 1}\left(\mbf W_1^{[2]}, \mbf W_1^{[3]}\right)
            \\
            & \quad - 2 \, a_2 \, \mbf F_{2, 1, 0}\left(\bs \phi_1^{[1]}, \mbf W_1^{[3]}\right) - 2 \, b_2 \, \mbf F_{2, 0, 1}\left(\bs \phi_1^{[1]}, \mbf W_1^{[3]}\right) - a_2 \, \mbf F_{2, 1, 0}\left(\mbf W_1^{[2]}, \mbf W_1^{[2]}\right)
            \\
            & \quad - b_2 \, \mbf F_{2, 0, 1}\left(\mbf W_1^{[2]}, \mbf W_1^{[2]}\right) - 2 \, a_3 \, \mbf F_{2, 1, 0}\left(\bs \phi_1^{[1]}, \mbf W_1^{[2]}\right) - 2 \, b_3 \, \mbf F_{2, 0, 1}\left(\bs \phi_1^{[1]}, \mbf W_1^{[2]}\right)
            \\
            & \quad - a_4 \, \mbf F_{2, 1, 0}\left(\bs \phi_1^{[1]}, \bs \phi_1^{[1]}\right) - b_4 \, \mbf F_{2, 0, 1}\left(\bs \phi_1^{[1]}, \bs \phi_1^{[1]}\right) - 2 \, a_1^2 \, \mbf F_{2, 2, 0}\left(\bs \phi_1^{[1]}, \mbf W_1^{[3]}\right)
            \\
            & \quad - 2 \, a_1 \, b_1 \, \mbf F_{2, 1, 1}\left(\bs \phi_1^{[1]}, \mbf W_1^{[3]}\right) - 2 \, b_1^2 \, \mbf F_{2, 0, 2}\left(\bs \phi_1^{[1]}, \mbf W_1^{[3]}\right) - a_1^2 \, \mbf F_{2, 2, 0}\left(\mbf W_1^{[2]}, \mbf W_1^{[2]}\right)
            \\
            & \quad - a_1 \, b_1 \, \mbf F_{2, 1, 1}\left(\mbf W_1^{[2]}, \mbf W_1^{[2]}\right) - b_1^2 \, \mbf F_{2, 0, 2}\left(\mbf W_1^{[2]}, \mbf W_1^{[2]}\right) - 4 \, a_1 \, a_2 \, \mbf F_{2, 2, 0}\left(\bs \phi_1^{[1]}, \mbf W_1^{[2]}\right)
            \\
            & \quad - 2 \, \left(a_1 \, b_2 + a_2 \, b_1\right) \, \mbf F_{2, 1, 1}\left(\bs \phi_1^{[1]}, \mbf W_1^{[2]}\right) - 4 \, b_1 \, b_2 \, \mbf F_{2, 0, 2}\left(\bs \phi_1^{[1]}, \mbf W_1^{[2]}\right)
            \\
            & \quad - 2 \, a_1 \, a_3 \, \mbf F_{2, 2, 0}\left(\bs \phi_1^{[1]}, \bs \phi_1^{[1]}\right) - \left(a_1 \, b_3 + a_3 \, b_1\right) \, \mbf F_{2, 1, 1}\left(\bs \phi_1^{[1]}, \bs \phi_1^{[1]}\right)
            \\
            & \quad - 2 \, b_1 \, b_3 \, \mbf F_{2, 0, 2}\left(\bs \phi_1^{[1]}, \bs \phi_1^{[1]}\right) - a_2^2 \, \mbf F_{2, 2, 0}\left(\bs \phi_1^{[1]}, \bs \phi_1^{[1]}\right) - a_2 \, b_2 \, \mbf F_{2, 1, 1}\left(\bs \phi_1^{[1]}, \bs \phi_1^{[1]}\right)
            \\
            & \quad - b_2^2 \, \mbf F_{2, 0, 2}\left(\bs \phi_1^{[1]}, \bs \phi_1^{[1]}\right) - 2 \, a_1^3 \, \mbf F_{2, 3, 0}\left(\bs \phi_1^{[1]}, \mbf W_1^{[2]}\right) - 2 \, a_1^2 \, b_1 \, \mbf F_{2, 2, 1}\left(\bs \phi_1^{[1]}, \mbf W_1^{[2]}\right)
            \\
            & \quad - 2 \, a_1 \, b_1^2 \, \mbf F_{2, 1, 2}\left(\bs \phi_1^{[1]}, \mbf W_1^{[2]}\right) - 2 \, b_1^3 \, \mbf F_{2, 0, 3}\left(\bs \phi_1^{[1]}, \mbf W_1^{[2]}\right) - 3 \, a_1^2 \, a_2 \, \mbf F_{2, 3, 0}\left(\bs \phi_1^{[1]}, \bs \phi_1^{[1]}\right)
            \\
            & \quad - 2 \, a_1 \, a_2 \, b_1 \, \mbf F_{2, 2, 1}\left(\bs \phi_1^{[1]}, \bs \phi_1^{[1]}\right) - a_2 \, b_1^2 \, \mbf F_{2, 1, 2}\left(\bs \phi_1^{[1]}, \bs \phi_1^{[1]}\right) - a_1^2 \, b_2 \, \mbf F_{2, 2, 1}\left(\bs \phi_1^{[1]}, \bs \phi_1^{[1]}\right)
            \\
            & \quad - 2 \, a_1 \, b_1 \, b_2 \, \mbf F_{2, 1, 2}\left(\bs \phi_1^{[1]}, \bs \phi_1^{[1]}\right) - 3 \, b_1^2 \, b_2 \, \mbf F_{2, 0, 3}\left(\bs \phi_1^{[1]}, \bs \phi_1^{[1]}\right) - a_1^4 \, \mbf F_{2, 4, 0}\left(\bs \phi_1^{[1]}, \bs \phi_1^{[1]}\right)
            \\
            & \quad - a_1^3 \, b_1 \, \mbf F_{2, 3, 1}\left(\bs \phi_1^{[1]}, \bs \phi_1^{[1]}\right) - a_1^2 \, b_1^2 \, \mbf F_{2, 2, 2}\left(\bs \phi_1^{[1]}, \bs \phi_1^{[1]}\right) - a_1 \, b_1^3 \, \mbf F_{2, 1, 3}\left(\bs \phi_1^{[1]}, \bs \phi_1^{[1]}\right)
            \\
            & \quad - b_1^4 \, \mbf F_{2, 0, 4}\left(\bs \phi_1^{[1]}, \bs \phi_1^{[1]}\right),
        \end{align*}
        }
        {\allowdisplaybreaks
        \begin{align*}
            \mathcal M_0 \, \mbf W_{0, 2}^{[6]} &= - a_1 \, \jac \mbf f_{1, 0}(\mbf 0) \, \mbf W_{0, 2}^{[5]} - b_1 \, \jac \mbf f_{0, 1}(\mbf 0) \, \mbf W_{0, 2}^{[5]} - a_2 \, \jac \mbf f_{1, 0}(\mbf 0) \, \mbf W_{0, 2}^{[4]} - b_2 \, \jac \mbf f_{0, 1}(\mbf 0) \, \mbf W_{0, 2}^{[4]}
            \\
            & \quad - a_1^2 \, \jac \mbf f_{2, 0}(\mbf 0) \, \mbf W_{0, 2}^{[4]} - a_1 \, b_1 \, \jac \mbf f_{1, 1}(\mbf 0) \, \mbf W_{0, 2}^{[4]} - b_1^2 \, \jac \mbf f_{0, 2}(\mbf 0) \, \mbf W_{0, 2}^{[4]}
            \\
            & \quad - 2 \, \mbf F_{2, 0, 0}\left(\bs \phi_1^{[1]}, \mbf W_{1, 6}^{[5]}\right) - 4 \, \mbf F_{2, 0, 0}\left(\mbf W_0^{[2]}, \mbf W_0^{[4]}\right) - 2 \, \mbf F_{2, 0, 0}\left(\mbf W_1^{[2]}, \mbf W_{1, 2}^{[4]}\right)
            \\
            & \quad - 2 \, \mbf F_{2, 0, 0}\left(\mbf W_2^{[2]}, \mbf W_2^{[4]}\right) - 2 \, \mbf F_{2, 0, 0}\left(\mbf W_0^{[3]}, \mbf W_0^{[3]}\right) - 2 \, \mbf F_{2, 0, 0}\left(\mbf W_1^{[3]}, \mbf W_{1, 2}^{[3]}\right)
            \\
            & \quad - \mbf F_{2, 0, 0}\left(\mbf W_2^{[3]}, \mbf W_2^{[3]}\right) - 2 \, a_1 \, \mbf F_{2, 1, 0}\left(\bs \phi_1^{[1]}, \mbf W_{1, 2}^{[4]}\right) - 2 \, b_1 \, \mbf F_{2, 0, 1}\left(\bs \phi_1^{[1]}, \mbf W_{1, 2}^{[4]}\right)
            \\
            & \quad - 4 \, a_1 \, \mbf F_{2, 1, 0}\left(\mbf W_0^{[2]}, \mbf W_0^{[3]}\right) - 4 \, b_1 \, \mbf F_{2, 0, 1}\left(\mbf W_0^{[2]}, \mbf W_0^{[3]}\right) - 2 \, a_1 \, \mbf F_{2, 1, 0}\left(\mbf W_1^{[2]}, \mbf W_{1, 2}^{[3]}\right)
            \\
            & \quad - 2 \, b_1 \, \mbf F_{2, 0, 1}\left(\mbf W_1^{[2]}, \mbf W_{1, 2}^{[3]}\right) - 2 \, a_1 \, \mbf F_{2, 1, 0}\left(\mbf W_2^{[2]}, \mbf W_2^{[3]}\right) - 2 \, b_1 \, \mbf F_{2, 0, 1}\left(\mbf W_2^{[2]}, \mbf W_2^{[3]}\right)
            \\
            & \quad - 2 \, a_2 \, \mbf F_{2, 1, 0}\left(\bs \phi_1^{[1]}, \mbf W_{1, 2}^{[3]}\right) - 2 \, b_2 \, \mbf F_{2, 0, 1}\left(\bs \phi_1^{[1]}, \mbf W_{1, 2}^{[3]}\right) - 2 \, a_2 \, \mbf F_{2, 1, 0}\left(\mbf W_0^{[2]}, \mbf W_0^{[2]}\right)
            \\
            & \quad - 2 \, b_2 \, \mbf F_{2, 0, 1}\left(\mbf W_0^{[2]}, \mbf W_0^{[2]}\right) - a_2 \, \mbf F_{2, 1, 0}\left(\mbf W_2^{[2]}, \mbf W_2^{[2]}\right) - b_2 \, \mbf F_{2, 0, 1}\left(\mbf W_2^{[2]}, \mbf W_2^{[2]}\right)
            \\
            & \quad - 2 \, a_1^2 \, \mbf F_{2, 2, 0}\left(\bs \phi_1^{[1]}, \mbf W_{1, 2}^{[3]}\right) - 2 \, a_1 \, b_1 \, \mbf F_{2, 1, 1}\left(\bs \phi_1^{[1]}, \mbf W_{1, 2}^{[3]}\right) - 2 \, b_1^2 \, \mbf F_{2, 0, 2}\left(\bs \phi_1^{[1]}, \mbf W_{1, 2}^{[3]}\right)
            \\
            & \quad - 2 \, a_1^2 \, \mbf F_{2, 2, 0}\left(\mbf W_0^{[2]}, \mbf W_0^{[2]}\right) - 2 \, a_1 \, b_1 \, \mbf F_{2, 1, 1}\left(\mbf W_0^{[2]}, \mbf W_0^{[2]}\right) - 2 \, b_1^2 \, \mbf F_{2, 0, 2}\left(\mbf W_0^{[2]}, \mbf W_0^{[2]}\right)
            \\
            & \quad - a_1^2 \, \mbf F_{2, 2, 0}\left(\mbf W_2^{[2]}, \mbf W_2^{[2]}\right) - a_1 \, b_1 \, \mbf F_{2, 1, 1}\left(\mbf W_2^{[2]}, \mbf W_2^{[2]}\right) - b_1^2 \, \mbf F_{2, 0, 2}\left(\mbf W_2^{[2]}, \mbf W_2^{[2]}\right)
            \\
            & \quad - 3 \, \mbf F_{3, 0, 0}\left(\bs \phi_1^{[1]}, \bs \phi_1^{[1]}, 2 \, \mbf W_0^{[4]} + \mbf W_2^{[4]}\right) - 6 \, \mbf F_{3, 0, 0}\left(\bs \phi_1^{[1]}, \mbf W_1^{[2]}, 2 \, \mbf W_0^{[3]} + \mbf W_2^{[3]}\right)
            \\
            & \quad - 6 \, \mbf F_{3, 0, 0}\left(\bs \phi_1^{[1]}, \mbf W_1^{[3]}, 2 \, \mbf W_0^{[2]} + \mbf W_2^{[2]}\right) - 3 \, \mbf F_{3, 0, 0}\left(\mbf W_1^{[2]}, \mbf W_1^{[2]}, 2 \, \mbf W_0^{[2]} + \mbf W_2^{[2]}\right)
            \\
            & \quad - 3 \, a_1 \, \mbf F_{3, 1, 0}\left(\bs \phi_1^{[1]}, \bs \phi_1^{[1]}, 2 \, \mbf W_0^{[3]} + \mbf W_2^{[3]}\right) - 3 \, b_1 \, \mbf F_{3, 0, 1}\left(\bs \phi_1^{[1]}, \bs \phi_1^{[1]}, 2 \, \mbf W_0^{[3]} + \mbf W_2^{[3]}\right)
            \\
            & \quad - 6 \, a_1 \, \mbf F_{3, 1, 0}\left(\bs \phi_1^{[1]}, \mbf W_1^{[2]}, 2 \, \mbf W_0^{[2]} + \mbf W_2^{[2]}\right) - 6 \, b_1 \, \mbf F_{3, 0, 1}\left(\bs \phi_1^{[1]}, \mbf W_1^{[2]}, 2 \, \mbf W_0^{[2]} + \mbf W_2^{[2]}\right)
            \\
            & \quad - 3 \, a_2 \, \mbf F_{3, 1, 0}\left(\bs \phi_1^{[1]}, \bs \phi_1^{[1]}, 2 \, \mbf W_0^{[2]} + \mbf W_2^{[2]}\right) - 3 \, b_2 \, \mbf F_{3, 0, 1}\left(\bs \phi_1^{[1]}, \bs \phi_1^{[1]}, 2 \, \mbf W_0^{[2]} + \mbf W_2^{[2]}\right)
            \\
            & \quad - 3 \, a_1^2 \, \mbf F_{3, 2, 0}\left(\bs \phi_1^{[1]}, \bs \phi_1^{[1]}, 2 \, \mbf W_0^{[2]} + \mbf W_2^{[2]}\right) - 3 \, a_1 \, b_1 \, \mbf F_{3, 1, 1}\left(\bs \phi_1^{[1]}, \bs \phi_1^{[1]}, 2 \, \mbf W_0^{[2]} + \mbf W_2^{[2]}\right)
            \\
            & \quad - 3 \, b_1^2 \, \mbf F_{3, 0, 2}\left(\bs \phi_1^{[1]}, \bs \phi_1^{[1]}, 2 \, \mbf W_0^{[2]} + \mbf W_2^{[2]}\right) - 12 \, \mbf F_{4, 0, 0}\left(\bs \phi_1^{[1]}, \bs \phi_1^{[1]}, \bs \phi_1^{[1]}, \mbf W_1^{[3]}\right)
            \\
            & \quad - 18 \, \mbf F_{4, 0, 0}\left(\bs \phi_1^{[1]}, \bs \phi_1^{[1]}, \mbf W_1^{[2]}, \mbf W_1^{[2]}\right) - 12 \, a_1 \, \mbf F_{4, 1, 0}\left(\bs \phi_1^{[1]}, \bs \phi_1^{[1]}, \bs \phi_1^{[1]}, \mbf W_1^{[2]}\right)
            \\
            & \quad - 12 \, b_1 \, \mbf F_{4, 0, 1}\left(\bs \phi_1^{[1]}, \bs \phi_1^{[1]}, \bs \phi_1^{[1]}, \mbf W_1^{[2]}\right) - 3 \, a_2 \, \mbf F_{4, 1, 0}\left(\bs \phi_1^{[1]}, \bs \phi_1^{[1]}, \bs \phi_1^{[1]}, \bs \phi_1^{[1]}\right)
            \\
            & \quad - 3 \, b_2 \, \mbf F_{4, 0, 1}\left(\bs \phi_1^{[1]}, \bs \phi_1^{[1]}, \bs \phi_1^{[1]}, \bs \phi_1^{[1]}\right) - 3 \, a_1^2 \, \mbf F_{4, 2, 0}\left(\bs \phi_1^{[1]}, \bs \phi_1^{[1]}, \bs \phi_1^{[1]}, \bs \phi_1^{[1]}\right)
            \\
            & \quad - 3 \, a_1 \, b_1 \, \mbf F_{4, 1, 1}\left(\bs \phi_1^{[1]}, \bs \phi_1^{[1]}, \bs \phi_1^{[1]}, \bs \phi_1^{[1]}\right) - 3 \, b_1^2 \, \mbf F_{4, 0, 2}\left(\bs \phi_1^{[1]}, \bs \phi_1^{[1]}, \bs \phi_1^{[1]}, \bs \phi_1^{[1]}\right),
        \end{align*}
        }
        \begin{align*}
            \mathcal M_0 \, \mbf W_{0, 3}^{[6]} &= - 2 \, \mbf F_{2, 0, 0}\left(\bs \phi_1^{[1]}, \mbf W_{1, 7}^{[5]}\right) - 4 \, \mbf F_{2, 0, 0}\left(\mbf W_0^{[2]}, \mbf W_{0, 2}^{[4]}\right) - 2 \, \mbf F_{2, 0, 0}\left(\mbf W_2^{[2]}, \mbf W_{2, 2}^{[4]}\right)
            \\
            & \quad - \mbf F_{2, 0, 0}\left(\mbf W_{1, 2}^{[3]}, \mbf W_{1, 2}^{[3]}\right) - \mbf F_{2, 0, 0}\left(\mbf W_3^{[3]}, \mbf W_3^{[3]}\right) - 3 \, \mbf F_{3, 0, 0}\left(\bs \phi_1^{[1]}, \bs \phi_1^{[1]}, 2 \, \mbf W_{0, 2}^{[4]} + \mbf W_{2, 2}^{[4]}\right)
            \\
            & \quad - 6 \, \mbf F_{3, 0, 0}\left(\bs \phi_1^{[1]}, \mbf W_2^{[2]}, \mbf W_3^{[3]}\right) - 6 \, \mbf F_{3, 0, 0}\left(\bs \phi_1^{[1]}, \mbf W_{1, 2}^{[3]}, 2 \, \mbf W_0^{[2]} + \mbf W_2^{[2]}\right)
            \\
            & \quad - 4 \, \mbf F_{3, 0, 0}\left(\mbf W_0^{[2]}, \mbf W_0^{[2]}, \mbf W_0^{[2]}\right) - 6 \, \mbf F_{3, 0, 0}\left(\mbf W_0^{[2]}, \mbf W_2^{[2]}, \mbf W_2^{[2]}\right)
            \\
            & \quad - 4 \, \mbf F_{4, 0, 0}\left(\bs \phi_1^{[1]}, \bs \phi_1^{[1]}, \bs \phi_1^{[1]}, 3 \, \mbf W_{1, 2}^{[3]} + \mbf W_3^{[3]}\right) - 24 \, \mbf F_{4, 0, 0}\left(\bs \phi_1^{[1]}, \bs \phi_1^{[1]}, \mbf W_0^{[2]}, \mbf W_0^{[2]}\right)
            \\
            & \quad - 12 \, \mbf F_{4, 0, 0}\left(\bs \phi_1^{[1]}, \bs \phi_1^{[1]}, \mbf W_2^{[2]}, 2 \, \mbf W_0^{[2]} + \mbf W_2^{[2]}\right)
            \\
            & \quad - 10 \, \mbf F_{5, 0, 0}\left(\bs \phi_1^{[1]}, \bs \phi_1^{[1]}, \bs \phi_1^{[1]}, \bs \phi_1^{[1]}, 3 \, \mbf W_0^{[2]} + 2 \, \mbf W_2^{[2]}\right)
            \\
            & \quad - 10 \, \mbf F_{6, 0, 0}\left(\bs \phi_1^{[1]}, \bs \phi_1^{[1]}, \bs \phi_1^{[1]}, \bs \phi_1^{[1]}, \bs \phi_1^{[1]}, \bs \phi_1^{[1]}\right),
        \end{align*}
        \begin{align*}
            \mathcal M_0 \, \mbf W_{0, 4}^{[6]} = - \mbf F_{2, 0, 0}\left(\mbf W_{1, 3}^{[3]}, \mbf W_{1, 3}^{[3]}\right) - 2 \, k^2 \, \hat D \, \mbf W_0^{[2]},
        \end{align*}
        \begin{align*}
            \mathcal M_0 \, \mbf W_{0, 5}^{[6]} &= - a_1 \, \jac \mbf f_{1, 0}(\mbf 0) \, \mbf W_{0, 3}^{[5]} - b_1 \, \jac \mbf f_{0, 1}(\mbf 0) \, \mbf W_{0, 3}^{[5]} - a_2 \, \jac \mbf f_{1, 0}(\mbf 0) \, \mbf W_{0, 3}^{[4]} - b_2 \, \jac \mbf f_{0, 1}(\mbf 0) \, \mbf W_{0, 3}^{[4]}
            \\
            & \quad - a_1^2 \, \jac \mbf f_{2, 0}(\mbf 0) \, \mbf W_{0, 3}^{[4]} - a_1 \, b_1 \, \jac \mbf f_{1, 1}(\mbf 0) \, \mbf W_{0, 3}^{[4]} - b_1^2 \, \jac \mbf f_{0, 2}(\mbf 0) \, \mbf W_{0, 3}^{[4]}
            \\
            & \quad - 2 \, \mbf F_{2, 0, 0}\left(\bs \phi_1^{[1]}, \mbf W_{1, 2}^{[5]}\right) - 2 \, \mbf F_{2, 0, 0}\left(\mbf W_1^{[2]}, \mbf W_{1, 3}^{[4]}\right) - 2 \, \mbf F_{2, 0, 0}\left(\mbf W_1^{[3]}, \mbf W_{1, 3}^{[3]}\right)
            \\
            & \quad - 2 \, a_1 \, \mbf F_{2, 1, 0}\left(\bs \phi_1^{[1]}, \mbf W_{1, 3}^{[4]}\right) - 2 \, b_1 \, \mbf F_{2, 0, 1}\left(\bs \phi_1^{[1]}, \mbf W_{1, 3}^{[4]}\right) - 2 \, a_1 \, \mbf F_{2, 1, 0}\left(\mbf W_1^{[2]}, \mbf W_{1, 3}^{[3]}\right)
            \\
            & \quad - 2 \, b_1 \, \mbf F_{2, 0, 1}\left(\mbf W_1^{[2]}, \mbf W_{1, 3}^{[3]}\right) - 2 \, a_2 \, \mbf F_{2, 1, 0}\left(\bs \phi_1^{[1]}, \mbf W_{1, 3}^{[3]}\right) - 2 \, b_2 \, \mbf F_{2, 0, 1}\left(\bs \phi_1^{[1]}, \mbf W_{1, 3}^{[3]}\right)
            \\
            & \quad - 2 \, a_1^2 \, \mbf F_{2, 2, 0}\left(\bs \phi_1^{[1]}, \mbf W_{1, 3}^{[3]}\right) - 2 \, a_1 \, b_1 \, \mbf F_{2, 1, 1}\left(\bs \phi_1^{[1]}, \mbf W_{1, 3}^{[3]}\right) - 2 \, b_1^2 \, \mbf F_{2, 0, 2}\left(\bs \phi_1^{[1]}, \mbf W_{1, 3}^{[3]}\right),
        \end{align*}
        \begin{align*}
            \mathcal M_0 \, \mbf W_{0, 6}^{[6]} &= - 2 \, \mbf F_{2, 0, 0}\left(\bs \phi_1^{[1]}, \mbf W_{1, 3}^{[5]} - \mbf W_{1, 4}^{[5]}\right) - 4 \, \mbf F_{2, 0, 0}\left(\mbf W_0^{[2]}, \mbf W_{0, 3}^{[4]}\right) - 2 \, \mbf F_{2, 0, 0}\left(\mbf W_2^{[2]}, \mbf W_{2, 3}^{[4]}\right)
            \\
            & \quad - 2 \, \mbf F_{2, 0, 0}\left(\mbf W_{1, 2}^{[3]}, \mbf W_{1, 3}^{[3]}\right) - 3 \, \mbf F_{3, 0, 0}\left(\bs \phi_1^{[1]}, \bs \phi_1^{[1]}, 2 \, \mbf W_{0, 3}^{[4]} + \mbf W_{2, 3}^{[4]}\right)
            \\
            & \quad - 6 \, \mbf F_{3, 0, 0}\left(\bs \phi_1^{[1]}, \mbf W_{1, 3}^{[3]}, 2 \, \mbf W_0^{[2]} + \mbf W_2^{[2]}\right) - 12 \, \mbf F_{4, 0, 0}\left(\bs \phi_1^{[1]}, \bs \phi_1^{[1]}, \bs \phi_1^{[1]}, \mbf W_{1, 3}^{[3]}\right),
        \end{align*}
        and
        \begin{align*}
            \evs{\mathcal M_0} \, \mbf W_{0, 7}^{[6]} = - 2 \, \mbf F_{2, 0, 0}\left(\bs \phi_1^{[1]}, \mbf W_1^{[5]}\right) - 2 \, k^2 \, \hat D \, \mbf W_0^{[2]}.
        \end{align*}
        Furthermore, the determination of the vectors that are a multiple of $e^{ix}$ leads to a solvability condition given by
        \begin{multline}
            \alpha_1 \, A_{2_{X \! X}} + i \, \alpha_2 \, A_{2_X} + i \, \alpha_3 \, \abs{A_1}^2 \, A_{2_X} + i \, \alpha_3 \, A_1 \, A_{1_X} \, \bar A_2 + i \, \alpha_3 \, \bar A_1 \, A_{1_X} \, A_2 + i \, \alpha_4 \, A_1^2 \, \bar A_{2_X}
            \\
            + 2 \, i \, \alpha_4 \, A_1 \, \bar A_{1_X} \, A_2 + \alpha_5 \, A_2 + \alpha_6 \, A_1^2 \, \bar A_2 + 2 \, \alpha_6 \, \abs{A_1}^2 \, A_2 + 2 \, \alpha_7 \, \abs{A_1}^2 \, A_1^2 \, \bar A_2 \label{A_2}
            \\
            + 3 \, \alpha_7 \, \abs{A_1}^4 \, A_2 + \alpha_{1, 2} \, A_{1_{X \! X}} + i \, \alpha_{2, 2} \, A_{1_X} + i \, \alpha_{3, 2} \, \abs{A_1}^2 \, A_{1_X} + i \, \alpha_{4, 2} \, A_1^2 \, \bar A_{1_X}
            \\
            + \alpha_{5, 2} \, A_1 + \alpha_{6, 2} \, \abs{A_1}^2 \, A_1 + \alpha_{7, 2} \, \abs{A_1}^4 \, A_1 = 0,
        \end{multline}
        $\alpha_{p, 2} = \left \langle \mbf Q_p^{[6]}, \bs \psi\right \rangle$, for $p = 1, \ldots, 7$, with
        {
        \allowdisplaybreaks

        }
        Here, we highlight that $\alpha_{p, 2} = 0$ for all $p = 1, \ldots, 7$ if $a_{2 \, q + 1} = b_{2 \, q + 1} = 0$ for $q = 0, 1, 2$, which implies that $A_2 = 0$. If this was not the case, then the expansion $A_2 = \left(R_2 + i \, R_1 \, \varphi_2\right) \, e^{i \, \varphi_1}$ can be used to solve \eqref{A_2}. To simplify calculations and close up with key clarifications, we assume that $a_{2 \, q + 1} = b_{2 \, q + 1} = 0$ for $q = 0, 1, 2$.
        
        After finding a solution for \eqref{A_2} and using the same idea developed after \eqref{resonant-order-five}, we obtain that
        \begin{align*}
            \mathcal M_1 \, \mbf W_1^{[6]} = - \mbf Q_1^{[6]} + \frac{\alpha_{1, 2}}{\left \langle \bs \psi, \bs \psi\right \rangle} \, \bs \psi,
        \end{align*}
        \begin{align*}
            \mathcal M_1 \, \mbf W_{1, p}^{[6]} = - \mbf Q_p^{[6]} + \frac{\alpha_{p, 2}}{\left \langle \bs \psi, \bs \psi\right \rangle} \, \bs \psi \quad \text{for } p = 2, 3, 4, 5, 6, 7,
        \end{align*}        
        Moreover,
        {\allowdisplaybreaks
        \begin{align*}
            \mathcal M_2 \, \mbf W_2^{[6]} &= - a_1 \, \jac \mbf f_{1, 0}(\mbf 0) \, \mbf W_2^{[5]} - b_1 \, \jac \mbf f_{0, 1}(\mbf 0) \, \mbf W_2^{[5]} - a_2 \, \jac \mbf f_{1, 0}(\mbf 0) \, \mbf W_2^{[4]}
            \\
            & \quad - b_2 \, \jac \mbf f_{0, 1}(\mbf 0) \, \mbf W_2^{[4]} - a_3 \, \jac \mbf f_{1, 0}(\mbf 0) \, \mbf W_2^{[3]} - b_3 \, \jac \mbf f_{0, 1}(\mbf 0) \, \mbf W_2^{[3]} - a_4 \, \jac \mbf f_{1, 0}(\mbf 0) \, \mbf W_2^{[2]}
            \\
            & \quad - b_4 \, \jac \mbf f_{0, 1}(\mbf 0) \, \mbf W_2^{[2]} - a_1^2 \, \jac \mbf f_{2, 0}(\mbf 0) \, \mbf W_2^{[4]} - a_1 \, b_1 \, \jac \mbf f_{1, 1}(\mbf 0) \, \mbf W_2^{[4]} - b_1^2 \, \jac \mbf f_{0, 2}(\mbf 0) \, \mbf W_2^{[4]}
            \\
            & \quad - 2 \, a_1 \, a_2 \, \jac \mbf f_{2, 0}(\mbf 0) \, \mbf W_2^{[3]} - \left(a_1 \, b_2 + a_2 \, b_1\right) \, \jac \mbf f_{1, 1}(\mbf 0) \, \mbf W_2^{[3]} - 2 \, b_1 \, b_2 \, \jac \mbf f_{0, 2}(\mbf 0) \, \mbf W_2^{[3]}
            \\
            & \quad - 2 \, a_1 \, a_3 \, \jac \mbf f_{2, 0}(\mbf 0) \, \mbf W_2^{[2]} - \left(a_1 \, b_3 + a_3 \, b_1\right) \, \jac \mbf f_{1, 1}(\mbf 0) \, \mbf W_2^{[2]} - 2 \, b_1 \, b_3 \, \jac \mbf f_{0, 2}(\mbf 0) \, \mbf W_2^{[2]}
            \\
            & \quad - a_2^2 \, \jac \mbf f_{2, 0}(\mbf 0) \, \mbf W_2^{[2]} - a_2 \, b_2 \, \jac \mbf f_{1, 1}(\mbf 0) \, \mbf W_2^{[2]} - b_2^2 \, \jac \mbf f_{0, 2}(\mbf 0) \, \mbf W_2^{[2]} - a_1^3 \, \jac \mbf f_{3, 0}(\mbf 0) \, \mbf W_2^{[3]}
            \\
            & \quad - a_1^2 \, b_1 \, \jac \mbf f_{2, 1}(\mbf 0) \, \mbf W_2^{[3]} - a_1 \, b_1^2 \, \jac \mbf f_{1, 2}(\mbf 0) \, \mbf W_2^{[3]} - b_1^3 \, \jac \mbf f_{0, 3}(\mbf 0) \, \mbf W_2^{[3]}
            \\
            & \quad - 3 \, a_1^2 \, a_2 \, \jac \mbf f_{3, 0}(\mbf 0) \, \mbf W_2^{[2]} - 2 \, a_1 \, a_2 \, b_1 \, \jac \mbf f_{2, 1}(\mbf 0) \, \mbf W_2^{[2]} - a_2 \, b_1^2 \, \jac \mbf f_{1, 2}(\mbf 0) \, \mbf W_2^{[2]}
            \\
            & \quad - a_1^2 \, b_2 \, \jac \mbf f_{2, 1}(\mbf 0) \, \mbf W_2^{[2]} - 2 \, a_1 \, b_1 \, b_2 \, \jac \mbf f_{1, 2}(\mbf 0) \, \mbf W_2^{[2]} - 3 \, b_1^2 \, b_2 \, \jac \mbf f_{0, 3}(\mbf 0) \, \mbf W_2^{[2]}
            \\
            & \quad - a_1^4 \, \jac \mbf f_{4, 0}(\mbf 0) \, \mbf W_2^{[2]} - a_1^3 \, b_1 \, \jac \mbf f_{3, 1}(\mbf 0) \, \mbf W_2^{[2]} - a_1^2 \, b_1^2 \, \jac \mbf f_{2, 2}(\mbf 0) \, \mbf W_2^{[2]}
            \\
            & \quad - a_1 \, b_1^3 \, \jac \mbf f_{1, 3}(\mbf 0) \, \mbf W_2^{[2]} - b_1^4 \, \jac \mbf f_{0, 4}(\mbf 0) \, \mbf W_2^{[2]} - 2 \, \mbf F_{2, 0, 0}\left(\bs \phi_1^{[1]}, \mbf W_{1, 5}^{[5]}\right)
            \\
            & \quad - 2 \, \mbf F_{2, 0, 0}\left(\mbf W_1^{[2]}, \mbf W_1^{[4]}\right) - \mbf F_{2, 0, 0}\left(\mbf W_1^{[3]}, \mbf W_1^{[3]}\right) - 2 \, a_1 \, \mbf F_{2, 1, 0}\left(\bs \phi_1^{[1]}, \mbf W_1^{[4]}\right)
            \\
            & \quad - 2 \, b_1 \, \mbf F_{2, 0, 1}\left(\bs \phi_1^{[1]}, \mbf W_1^{[4]}\right) - 2 \, a_1 \, \mbf F_{2, 1, 0}\left(\mbf W_1^{[2]}, \mbf W_1^{[3]}\right) - 2 \, b_1 \, \mbf F_{2, 0, 1}\left(\mbf W_1^{[2]}, \mbf W_1^{[3]}\right)
            \\
            & \quad - 2 \, a_2 \, \mbf F_{2, 1, 0}\left(\bs \phi_1^{[1]}, \mbf W_1^{[3]}\right) - 2 \, b_2 \, \mbf F_{2, 0, 1}\left(\bs \phi_1^{[1]}, \mbf W_1^{[3]}\right) - a_2 \, \mbf F_{2, 1, 0}\left(\mbf W_1^{[2]}, \mbf W_1^{[2]}\right)
            \\
            & \quad - b_2 \, \mbf F_{2, 0, 1}\left(\mbf W_1^{[2]}, \mbf W_1^{[2]}\right) - 2 \, a_3 \, \mbf F_{2, 1, 0}\left(\bs \phi_1^{[1]}, \mbf W_1^{[2]}\right) - 2 \, b_3 \, \mbf F_{2, 0, 1}\left(\bs \phi_1^{[1]}, \mbf W_1^{[2]}\right)
            \\
            & \quad - a_4 \, \mbf F_{2, 1, 0}\left(\bs \phi_1^{[1]}, \bs \phi_1^{[1]}\right) - b_4 \, \mbf F_{2, 0, 1}\left(\bs \phi_1^{[1]}, \bs \phi_1^{[1]}\right) - 2 \, a_1^2 \, \mbf F_{2, 2, 0}\left(\bs \phi_1^{[1]}, \mbf W_1^{[3]}\right)
            \\
            & \quad - 2 \, a_1 \, b_1 \, \mbf F_{2, 1, 1}\left(\bs \phi_1^{[1]}, \mbf W_1^{[3]}\right) - 2 \, b_1^2 \, \mbf F_{2, 0, 2}\left(\bs \phi_1^{[1]}, \mbf W_1^{[3]}\right) - a_1^2 \, \mbf F_{2, 2, 0}\left(\mbf W_1^{[2]}, \mbf W_1^{[2]}\right)
            \\
            & \quad - a_1 \, b_1 \, \mbf F_{2, 1, 1}\left(\mbf W_1^{[2]}, \mbf W_1^{[2]}\right) - b_1^2 \, \mbf F_{2, 0, 2}\left(\mbf W_1^{[2]}, \mbf W_1^{[2]}\right) - 4 \, a_1 \, a_2 \, \mbf F_{2, 2, 0}\left(\bs \phi_1^{[1]}, \mbf W_1^{[2]}\right)
            \\
            & \quad - 2 \, \left(a_1 \, b_2 + a_2 \, b_1\right) \, \mbf F_{2, 1, 1}\left(\bs \phi_1^{[1]}, \mbf W_1^{[2]}\right) - 4 \, b_1 \, b_2 \, \mbf F_{2, 0, 2}\left(\bs \phi_1^{[1]}, \mbf W_1^{[2]}\right)
            \\
            & \quad - 2 \, a_1 \, a_3 \, \mbf F_{2, 2, 0}\left(\bs \phi_1^{[1]}, \bs \phi_1^{[1]}\right) - \left(a_1 \, b_3 + a_3 \, b_1\right) \, \mbf F_{2, 1, 1}\left(\bs \phi_1^{[1]}, \bs \phi_1^{[1]}\right)
            \\
            & \quad - 2 \, b_1 \, b_3 \, \mbf F_{2, 0, 2}\left(\bs \phi_1^{[1]}, \bs \phi_1^{[1]}\right) - a_2^2 \, \mbf F_{2, 2, 0}\left(\bs \phi_1^{[1]}, \bs \phi_1^{[1]}\right) - a_2 \, b_2 \, \mbf F_{2, 1, 1}\left(\bs \phi_1^{[1]}, \bs \phi_1^{[1]}\right)
            \\
            & \quad - b_2^2 \, \mbf F_{2, 0, 2}\left(\bs \phi_1^{[1]}, \bs \phi_1^{[1]}\right) - 2 \, a_1^3 \, \mbf F_{2, 3, 0}\left(\bs \phi_1^{[1]}, \mbf W_1^{[2]}\right) - 2 \, a_1^2 \, b_1 \, \mbf F_{2, 2, 1}\left(\bs \phi_1^{[1]}, \mbf W_1^{[2]}\right)
            \\
            & \quad - 2 \, a_1 \, b_1^2 \, \mbf F_{2, 1, 2}\left(\bs \phi_1^{[1]}, \mbf W_1^{[2]}\right) - 2 \, b_1^3 \, \mbf F_{2, 0, 3}\left(\bs \phi_1^{[1]}, \mbf W_1^{[2]}\right) - 3 \, a_1^2 \, a_2 \, \mbf F_{2, 3, 0}\left(\bs \phi_1^{[1]}, \bs \phi_1^{[1]}\right)
            \\
            & \quad - 2 \, a_1 \, a_2 \, b_1 \, \mbf F_{2, 2, 1}\left(\bs \phi_1^{[1]}, \bs \phi_1^{[1]}\right) - a_2 \, b_1^2 \, \mbf F_{2, 1, 2}\left(\bs \phi_1^{[1]}, \bs \phi_1^{[1]}\right) - a_1^2 \, b_2 \, \mbf F_{2, 2, 1}\left(\bs \phi_1^{[1]}, \bs \phi_1^{[1]}\right)
            \\
            & \quad - 2 \, a_1 \, b_1 \, b_2 \, \mbf F_{2, 1, 2}\left(\bs \phi_1^{[1]}, \bs \phi_1^{[1]}\right) - 3 \, b_1^2 \, b_2 \, \mbf F_{2, 0, 3}\left(\bs \phi_1^{[1]}, \bs \phi_1^{[1]}\right) - a_1^4 \, \mbf F_{2, 4, 0}\left(\bs \phi_1^{[1]}, \bs \phi_1^{[1]}\right)
            \\
            & \quad - a_1^3 \, b_1 \, \mbf F_{2, 3, 1}\left(\bs \phi_1^{[1]}, \bs \phi_1^{[1]}\right) - a_1^2 \, b_1^2 \, \mbf F_{2, 2, 2}\left(\bs \phi_1^{[1]}, \bs \phi_1^{[1]}\right) - a_1 \, b_1^3 \, \mbf F_{2, 1, 3}\left(\bs \phi_1^{[1]}, \bs \phi_1^{[1]}\right)
            \\
            & \quad - b_1^4 \, \mbf F_{2, 0, 4}\left(\bs \phi_1^{[1]}, \bs \phi_1^{[1]}\right),
        \end{align*}
        }
        {\allowdisplaybreaks
        \begin{align*}
            \mathcal M_2 \, \mbf W_{2, 2}^{[6]} &= - a_1 \, \jac \mbf f_{1, 0}(\mbf 0) \, \mbf W_{2, 2}^{[5]} - b_1 \, \jac \mbf f_{0, 1}(\mbf 0) \, \mbf W_{2, 2}^{[5]} - a_2 \, \jac \mbf f_{1, 0}(\mbf 0) \, \mbf W_{2, 2}^{[4]} - b_2 \, \jac \mbf f_{0, 1}(\mbf 0) \, \mbf W_{2, 2}^{[4]}
            \\
            & \quad - a_1^2 \, \jac \mbf f_{2, 0}(\mbf 0) \, \mbf W_{2, 2}^{[4]} - a_1 \, b_1 \, \jac \mbf f_{1, 1}(\mbf 0) \, \mbf W_{2, 2}^{[4]} - b_1^2 \, \jac \mbf f_{0, 2}(\mbf 0) \, \mbf W_{2, 2}^{[4]}
            \\
            & \quad - 2 \, \mbf F_{2, 0, 0}\left(\bs \phi_1^{[1]}, \mbf W_{1, 6}^{[5]} + \mbf W_3^{[5]}\right) - 4 \, \mbf F_{2, 0, 0}\left(\mbf W_0^{[2]}, \mbf W_2^{[4]}\right)
            \\
            & \quad - 2 \, \mbf F_{2, 0, 0}\left(\mbf W_1^{[2]}, \mbf W_{1, 2}^{[4]} + \mbf W_3^{[4]}\right) - 4 \, \mbf F_{2, 0, 0}\left(\mbf W_2^{[2]}, \mbf W_0^{[4]}\right) - 4 \, \mbf F_{2, 0, 0}\left(\mbf W_0^{[3]}, \mbf W_2^{[3]}\right)
            \\
            & \quad - 2 \, \mbf F_{2, 0, 0}\left(\mbf W_1^{[3]}, \mbf W_{1, 2}^{[3]} + \mbf W_3^{[3]}\right) - 2 \, a_1 \, \mbf F_{2, 1, 0}\left(\bs \phi_1^{[1]}, \mbf W_{1, 2}^{[4]} + \mbf W_3^{[4]}\right)
            \\
            & \quad - 2 \, b_1 \, \mbf F_{2, 0, 1}\left(\bs \phi_1^{[1]}, \mbf W_{1, 2}^{[4]} + \mbf W_3^{[4]}\right) - 4 \, a_1 \, \mbf F_{2, 1, 0}\left(\mbf W_0^{[2]}, \mbf W_2^{[3]}\right)
            \\
            & \quad - 4 \, b_1 \, \mbf F_{2, 0, 1}\left(\mbf W_0^{[2]}, \mbf W_2^{[3]}\right) - 2 \, a_1 \, \mbf F_{2, 1, 0}\left(\mbf W_1^{[2]}, \mbf W_{1, 2}^{[3]} + \mbf W_3^{[3]}\right)
            \\
            & \quad - 2 \, b_1 \, \mbf F_{2, 0, 1}\left(\mbf W_1^{[2]}, \mbf W_{1, 2}^{[3]} + \mbf W_3^{[3]}\right) - 4 \, a_1 \, \mbf F_{2, 1, 0}\left(\mbf W_2^{[2]}, \mbf W_0^{[3]}\right)
            \\
            & \quad - 4 \, b_1 \, \mbf F_{2, 0, 1}\left(\mbf W_2^{[2]}, \mbf W_0^{[3]}\right) - 2 \, a_2 \, \mbf F_{2, 1, 0}\left(\bs \phi_1^{[1]}, \mbf W_{1, 2}^{[3]} + \mbf W_3^{[3]}\right)
            \\
            & \quad - 2 \, b_2 \, \mbf F_{2, 0, 1}\left(\bs \phi_1^{[1]}, \mbf W_{1, 2}^{[3]} + \mbf W_3^{[3]}\right) - 4 \, a_2 \, \mbf F_{2, 1, 0}\left(\mbf W_0^{[2]}, \mbf W_2^{[2]}\right)
            \\
            & \quad - 4 \, b_2 \, \mbf F_{2, 0, 1}\left(\mbf W_0^{[2]}, \mbf W_2^{[2]}\right) - 2 \, a_1^2 \, \mbf F_{2, 2, 0}\left(\bs \phi_1^{[1]}, \mbf W_{1, 2}^{[3]} + \mbf W_3^{[3]}\right)
            \\
            & \quad - 2 \, a_1 \, b_1 \, \mbf F_{2, 1, 1}\left(\bs \phi_1^{[1]}, \mbf W_{1, 2}^{[3]} + \mbf W_3^{[3]}\right) - 2 \, b_1^2 \, \mbf F_{2, 0, 2}\left(\bs \phi_1^{[1]}, \mbf W_{1, 2}^{[3]} + \mbf W_3^{[3]}\right)
            \\
            & \quad  - 4 \, a_1^2 \, \mbf F_{2, 2, 0}\left(\mbf W_0^{[2]}, \mbf W_2^{[2]}\right) - 4 \, a_1 \, b_1 \, \mbf F_{2, 1, 1}\left(\mbf W_0^{[2]}, \mbf W_2^{[2]}\right)
            \\
            & \quad - 4 \, b_1^2 \, \mbf F_{2, 0, 2}\left(\mbf W_0^{[2]}, \mbf W_2^{[2]}\right) - 6 \, \mbf F_{3, 0, 0}\left(\bs \phi_1^{[1]}, \bs \phi_1^{[1]}, \mbf W_0^{[4]} + \mbf W_2^{[4]}\right)
            \\
            & \quad - 12 \, \mbf F_{3, 0, 0}\left(\bs \phi_1^{[1]}, \mbf W_1^{[2]}, \mbf W_0^{[3]} + \mbf W_2^{[3]}\right) - 12 \, \mbf F_{3, 0, 0}\left(\bs \phi_1^{[1]}, \mbf W_1^{[3]}, \mbf W_0^{[2]} + \mbf W_2^{[2]}\right)
            \\
            & \quad - 6 \, \mbf F_{3, 0, 0}\left(\mbf W_1^{[2]}, \mbf W_1^{[2]}, \mbf W_0^{[2]} + \mbf W_2^{[2]}\right) - 6 \, a_1 \, \mbf F_{3, 1, 0}\left(\bs \phi_1^{[1]}, \bs \phi_1^{[1]}, \mbf W_0^{[3]} + \mbf W_2^{[3]}\right)
            \\
            & \quad - 6 \, b_1 \, \mbf F_{3, 0, 1}\left(\bs \phi_1^{[1]}, \bs \phi_1^{[1]}, \mbf W_0^{[3]} + \mbf W_2^{[3]}\right) - 12 \, a_1 \, \mbf F_{3, 1, 0}\left(\bs \phi_1^{[1]}, \mbf W_1^{[2]}, \mbf W_0^{[2]} + \mbf W_2^{[2]}\right)
            \\
            & \quad - 12 \, b_1 \, \mbf F_{3, 0, 1}\left(\bs \phi_1^{[1]}, \mbf W_1^{[2]}, \mbf W_0^{[2]} + \mbf W_2^{[2]}\right) - 6 \, a_2 \, \mbf F_{3, 1, 0}\left(\bs \phi_1^{[1]}, \bs \phi_1^{[1]}, \mbf W_0^{[2]} + \mbf W_2^{[2]}\right)
            \\
            & \quad - 6 \, b_2 \, \mbf F_{3, 0, 1}\left(\bs \phi_1^{[1]}, \bs \phi_1^{[1]}, \mbf W_0^{[2]} + \mbf W_2^{[2]}\right) - 6 \, a_1^2 \, \mbf F_{3, 2, 0}\left(\bs \phi_1^{[1]}, \bs \phi_1^{[1]}, \mbf W_0^{[2]} + \mbf W_2^{[2]}\right)
            \\
            & \quad - 6 \, a_1 \, b_1 \, \mbf F_{3, 1, 1}\left(\bs \phi_1^{[1]}, \bs \phi_1^{[1]}, \mbf W_0^{[2]} + \mbf W_2^{[2]}\right) - 6 \, b_1^2 \, \mbf F_{3, 0, 2}\left(\bs \phi_1^{[1]}, \bs \phi_1^{[1]}, \mbf W_0^{[2]} + \mbf W_2^{[2]}\right)
            \\
            & \quad - 16 \, \mbf F_{4, 0, 0}\left(\bs \phi_1^{[1]}, \bs \phi_1^{[1]}, \bs \phi_1^{[1]}, \mbf W_1^{[3]}\right) - 24 \, \mbf F_{4, 0, 0}\left(\bs \phi_1^{[1]}, \bs \phi_1^{[1]}, \mbf W_1^{[2]}, \mbf W_1^{[2]}\right)
            \\
            & \quad - 16 \, a_1 \, \mbf F_{4, 1, 0}\left(\bs \phi_1^{[1]}, \bs \phi_1^{[1]}, \bs \phi_1^{[1]}, \mbf W_1^{[2]}\right) - 16 \, b_1 \, \mbf F_{4, 0, 1}\left(\bs \phi_1^{[1]}, \bs \phi_1^{[1]}, \bs \phi_1^{[1]}, \mbf W_1^{[2]}\right)
            \\
            & \quad - 4 \, a_2 \, \mbf F_{4, 1, 0}\left(\bs \phi_1^{[1]}, \bs \phi_1^{[1]}, \bs \phi_1^{[1]}, \bs \phi_1^{[1]}\right) - 4 \, b_2 \, \mbf F_{4, 0, 1}\left(\bs \phi_1^{[1]}, \bs \phi_1^{[1]}, \bs \phi_1^{[1]}, \bs \phi_1^{[1]}\right)
            \\
            & \quad - 4 \, a_1^2 \, \mbf F_{4, 2, 0}\left(\bs \phi_1^{[1]}, \bs \phi_1^{[1]}, \bs \phi_1^{[1]}, \bs \phi_1^{[1]}\right) - 4 \, a_1 \, b_1 \, \mbf F_{4, 1, 1}\left(\bs \phi_1^{[1]}, \bs \phi_1^{[1]}, \bs \phi_1^{[1]}, \bs \phi_1^{[1]}\right)
            \\
            & \quad - 4 \, b_1^2 \, \mbf F_{4, 0, 2}\left(\bs \phi_1^{[1]}, \bs \phi_1^{[1]}, \bs \phi_1^{[1]}, \bs \phi_1^{[1]}\right),
        \end{align*}
        }
        \begin{align*}
            \mathcal M_2 \, \mbf W_{2, 3}^{[6]} &= - 2 \, \mbf F_{2, 0, 0}\left(\bs \phi_1^{[1]}, \mbf W_{1, 7}^{[5]} + \mbf W_{3, 2}^{[5]}\right) - 4 \, \mbf F_{2, 0, 0}\left(\mbf W_0^{[2]}, \mbf W_{2, 2}^{[4]}\right)
            \\
            & \quad - 2 \, \mbf F_{2, 0, 0}\left(\mbf W_2^{[2]}, 2 \, \mbf W_{0, 2}^{[4]} + \mbf W_4^{[4]}\right) - \mbf F_{2, 0, 0}\left(\mbf W_{1, 2}^{[3]}, \mbf W_{1, 2}^{[3]} \evs{+ 2 \, \mbf W_3^{[3]}}\right)
            \\
            & \quad - 3 \, \mbf F_{3, 0, 0}\left(\bs \phi_1^{[1]}, \bs \phi_1^{[1]}, 2 \, \mbf W_{0, 2}^{[4]} + 2 \, \mbf W_{2, 2}^{[4]} + \mbf W_4^{[4]}\right)
            \\
            & \quad - 12 \, \mbf F_{3, 0, 0}\left(\bs \phi_1^{[1]}, \mbf W_0^{[2]}, \mbf W_{1, 2}^{[3]} + \mbf W_3^{[3]}\right) - 6 \, \mbf F_{3, 0, 0}\left(\bs \phi_1^{[1]}, \mbf W_2^{[2]}, 2 \, \mbf W_{1, 2}^{[3]} + \mbf W_3^{[3]}\right)
            \\
            & \quad - 12 \, \mbf F_{3, 0, 0}\left(\mbf W_0^{[2]}, \mbf W_0^{[2]}, \mbf W_2^{[2]}\right) - 3 \, \mbf F_{3, 0, 0}\left(\mbf W_2^{[2]}, \mbf W_2^{[2]}, \mbf W_2^{[2]}\right)
            \\
            & \quad - 4 \, \mbf F_{4, 0, 0}\left(\bs \phi_1^{[1]}, \bs \phi_1^{[1]}, \bs \phi_1^{[1]}, 4 \, \mbf W_{1, 2}^{[3]} + 3 \, \mbf W_3^{[3]}\right) - 24 \, \mbf F_{4, 0, 0}\left(\bs \phi_1^{[1]}, \bs \phi_1^{[1]}, \mbf W_0^{[2]}, \mbf W_0^{[2]}\right)
            \\
            & \quad - 6 \, \mbf F_{4, 0, 0}\left(\bs \phi_1^{[1]}, \bs \phi_1^{[1]}, \mbf W_2^{[2]}, 8 \, \mbf W_0^{[2]} + 3 \, \mbf W_2^{[2]}\right)
            \\
            & \quad - 5 \, \mbf F_{5, 0, 0}\left(\bs \phi_1^{[1]}, \bs \phi_1^{[1]}, \bs \phi_1^{[1]}, \bs \phi_1^{[1]}, 8 \, \mbf W_0^{[2]} + 7 \, \mbf W_2^{[2]}\right)
            \\
            & \quad - 15 \, \mbf F_{6, 0, 0}\left(\bs \phi_1^{[1]}, \bs \phi_1^{[1]}, \bs \phi_1^{[1]}, \bs \phi_1^{[1]}, \bs \phi_1^{[1]}, \bs \phi_1^{[1]}\right),
        \end{align*}
        \begin{align*}
            \mathcal M_2 \, \mbf W_{2, 4}^{[6]} = \mbf F_{2, 0, 0}\left(\mbf W_{1, 3}^{[3]}, \mbf W_{1, 3}^{[3]}\right) + 4 \, k^2 \, \hat D \, \mbf W_{2, 3}^{[4]} - 2 \, k^2 \, \hat D \, \mbf W_2^{[2]},
        \end{align*}
        \begin{align*}
            \mathcal M_2 \, \mbf W_{2, 5}^{[6]} &= - a_1 \, \jac \mbf f_{1, 0}(\mbf 0) \, \mbf W_{2, 3}^{[5]} - b_1 \, \jac \mbf f_{0, 1}(\mbf 0) \, \mbf W_{2, 3}^{[5]} - a_2 \, \jac \mbf f_{1, 0}(\mbf 0) \, \mbf W_{2, 3}^{[4]}
            \\
            & \quad - b_2 \, \jac \mbf f_{0, 1}(\mbf 0) \, \mbf W_{2, 3}^{[4]} - a_1^2 \, \jac \mbf f_{2, 0}(\mbf 0) \, \mbf W_{2, 3}^{[4]} - a_1 \, b_1 \, \jac \mbf f_{1, 1}(\mbf 0) \, \mbf W_{2, 3}^{[4]}
            \\
            & \quad - b_1^2 \, \jac \mbf f_{0, 2}(\mbf 0) \, \mbf W_{2, 3}^{[4]} - 2 \, \mbf F_{2, 0, 0}\left(\bs \phi_1^{[1]}, \mbf W_{1, 2}^{[5]}\right) - 2 \, \mbf F_{2, 0, 0}\left(\mbf W_1^{[2]}, \mbf W_{1, 3}^{[4]}\right)
            \\
            & \quad - 2 \, \mbf F_{2, 0, 0}\left(\mbf W_1^{[3]}, \mbf W_{1, 3}^{[3]}\right) - 2 \, a_1 \, \mbf F_{2, 1, 0}\left(\bs \phi_1^{[1]}, \mbf W_{1, 3}^{[4]}\right) - 2 \, b_1 \, \mbf F_{2, 0, 1}\left(\bs \phi_1^{[1]}, \mbf W_{1, 3}^{[4]}\right)
            \\
            & \quad - 2 \, a_1 \, \mbf F_{2, 1, 0}\left(\mbf W_1^{[2]}, \mbf W_{1, 3}^{[3]}\right) - 2 \, b_1 \, \mbf F_{2, 0, 1}\left(\mbf W_1^{[2]}, \mbf W_{1, 3}^{[3]}\right) - 2 \, a_2 \, \mbf F_{2, 1, 0}\left(\bs \phi_1^{[1]}, \mbf W_{1, 3}^{[3]}\right)
            \\
            & \quad - 2 \, b_2 \, \mbf F_{2, 0, 1}\left(\bs \phi_1^{[1]}, \mbf W_{1, 3}^{[3]}\right) - 2 \, a_1^2 \, \mbf F_{2, 2, 0}\left(\bs \phi_1^{[1]}, \mbf W_{1, 3}^{[3]}\right) - 2 \, a_1 \, b_1 \, \mbf F_{2, 1, 1}\left(\bs \phi_1^{[1]}, \mbf W_{1, 3}^{[3]}\right)
            \\
            & \quad - 2 \, b_1^2 \, \mbf F_{2, 0, 2}\left(\bs \phi_1^{[1]}, \mbf W_{1, 3}^{[3]}\right) - 8 \, k^2 \, \hat D \, \mbf W_2^{[4]},
        \end{align*}
        \begin{align*}
            \mathcal M_2 \, \mbf W_{2, 6}^{[6]} &= - 2 \, \mbf F_{2, 0, 0}\left(\bs \phi_1^{[1]}, \mbf W_{1, 3}^{[5]} + \mbf W_{3, 3}^{[5]}\right) - 4 \, \mbf F_{2, 0, 0}\left(\mbf W_0^{[2]}, \mbf W_{2, 3}^{[4]}\right) - 2 \, \mbf F_{2, 0, 0}\left(\mbf W_2^{[2]}, \mbf W_{0, 3}^{[4]}\right)
            \\
            & \quad - 2 \, \mbf F_{2, 0, 0}\left(\mbf W_{1, 2}^{[3]}, \mbf W_{1, 3}^{[3]}\right) - 3 \, \mbf F_{3, 0, 0}\left(\bs \phi_1^{[1]}, \bs \phi_1^{[1]}, \mbf W_{0, 3}^{[4]} + 2 \, \mbf W_{2, 3}^{[4]}\right)
            \\
            & \quad - 6 \, \mbf F_{3, 0, 0}\left(\bs \phi_1^{[1]}, \mbf W_{1, 3}^{[3]}, 2 \, \mbf W_0^{[2]} + \mbf W_2^{[2]}\right) - 12 \, \mbf F_{4, 0, 0}\left(\bs \phi_1^{[1]}, \bs \phi_1^{[1]}, \bs \phi_1^{[1]}, \mbf W_{1, 3}^{[3]}\right)
            \\
            & \quad - 12 \, k^2 \, \hat D \, \mbf W_{2, 2}^{[4]},
        \end{align*}
        \begin{align*}
            \mathcal M_2 \, \mbf W_{2, 7}^{[6]} &= - 2 \, \mbf F_{2, 0, 0}\left(\bs \phi_1^{[1]}, \mbf W_{1, 4}^{[5]}\right) + 2 \, \mbf F_{2, 0, 0}\left(\mbf W_2^{[2]}, \mbf W_{0, 3}^{[4]}\right) + 2 \, \mbf F_{2, 0, 0}\left(\mbf W_{1, 3}^{[3]}, \mbf W_3^{[3]}\right)
            \\
            & \quad + 3 \, \mbf F_{3, 0, 0}\left(\bs \phi_1^{[1]}, \bs \phi_1^{[1]}, \mbf W_{0, 3}^{[4]}\right) + 6 \, \mbf F_{3, 0, 0}\left(\bs \phi_1^{[1]}, \mbf W_2^{[2]}, \mbf W_{1, 3}^{[3]}\right)
            \\
            & \quad + 4 \, \mbf F_{4, 0, 0}\left(\bs \phi_1^{[1]}, \bs \phi_1^{[1]}, \bs \phi_1^{[1]}, \mbf W_{1, 3}^{[3]}\right) - 4 \, k^2 \, \hat D \, \mbf W_{2, 2}^{[4]},
        \end{align*}
        and
        \begin{align*}
            \mathcal M_2 \, \mbf W_{2, 8}^{[6]} = - 2 \, \mbf F_{2, 0, 0}\left(\bs \phi_1^{[1]}, \mbf W_1^{[5]}\right) + 4 \, k^2 \, \hat D \, \mbf W_{2, 3}^{[4]} - 2 \, k^2 \, \hat D \, \mbf W_2^{[2]}.
        \end{align*}
        \paragraph{Order $\mathcal O\left(\varepsilon^7\right)$.} Finally, at order seven, \eqref{eqtoexpand} becomes
        {
        \allowdisplaybreaks

        }
        Here, we are only interested in the solvability condition coming out of this equation, which is given by
        {\allowdisplaybreaks
        \begin{multline}
            \alpha_1 \, A_{3_{X \! X}} + i \, \alpha_2 \, A_{3_X} + i \, \alpha_3 \, \abs{A_1}^2 \, A_{3_X} + i \, \alpha_3 \, A_1 \, A_{1_X} \, \bar A_3 + i \, \alpha_3 \, \bar A_1 \, A_{1_X} \, A_3 + i \, \alpha_4 \, A_1^2 \, \bar A_{3_X}
            \\
            + 2 \, i \, \alpha_4 \, A_1 \, \bar A_{1_X} \, A_3 + \alpha_5 \, A_3 + \alpha_6 \, A_1^2 \, \bar A_3 + 2 \, \alpha_6 \, \abs{A_1}^2 \, A_3 + 2 \, \alpha_7 \, \abs{A_1}^2 \, A_1^2 \, \bar A_3 + 3 \, \alpha_7 \, \abs{A_1}^4 \, A_3
            \\
            + \alpha_{1, 2} \, A_{2_{X \! X}} + i \, \alpha_{2, 2} \, A_{2_X} + i \, \alpha_{3, 2} \, \abs{A_1}^2 \, A_{2_X} + i \, \alpha_{3, 2} \, A_1 \, A_{1_X} \, \bar A_2 + i \, \alpha_{3, 2} \, \bar A_1 \, A_{1_X} \, A_2 + i \, \alpha_{4, 2} \, A_1^2 \, \bar A_{2_X}
            \\
            + 2 \, i \, \alpha_{4, 2} \, A_1 \, \bar A_{1_X} \, A_2 + \alpha_{5, 2} \, A_2 + \alpha_{6, 2} \, A_1^2 \, \bar A_2 + 2 \, \alpha_{6, 2} \, \abs{A_1}^2 \, A_2 + 2 \, \alpha_{7, 2} \, \abs{A_1}^2 \, A_1^2 \, \bar A_2 + 3 \, \alpha_{7, 2} \, \abs{A_1}^4 \, A_2
            \\
            + i \, \alpha_3 \, \bar A_1 \, A_2 \, A_{2_X} + i \, \alpha_3 \, A_1 \, \bar A_2 \, A_{2_X} + i \, \alpha_3 \, A_{1_X} \, \abs{A_2}^2 + i \, \alpha_4 \, \bar A_{1_X} \, A_2^2 + 2 \, i \, \alpha_4 \, A_1 \, A_2 \, \bar A_{2_X}
            \\
            + \alpha_6 \, \bar A_1 \, A_2^2 + 2 \, \alpha_6 \, A_1 \, \abs{A_2}^2 + \alpha_7 \, A_1^3 \, \bar A_2^2 + 3 \, \alpha_7 \, \abs{A_1}^2 \, \bar A_1 \, A_2^2 + 6 \, \alpha_7 \, \abs{A_1}^2 \, A_1 \, \abs{A_2}^2
            \\
            + i \, \alpha_{1, 3} \, A_{1_{X \! X \! X}} + \alpha_{2, 3} \, A_{1_{X \! X}} + \alpha_{3, 3} \, \abs{A_1}^2 \, A_{1_{X \! X}} + \alpha_{4, 3} \, A_1^2 \, \bar A_{1_{X \! X}} + i \, \alpha_{5, 3} \, A_{1_X} + i \, \alpha_{6, 3} \, \abs{A_1}^2 \, A_{1_X}
            \\
            + i \, \alpha_{7, 3} \, \abs{A_1}^4 \, A_{1_X} + i \, \alpha_{8, 3} \, A_1^2 \, \bar A_{1_X} + i \, \alpha_{9, 3} \, \abs{A_1}^2 \, A_1^2 \, \bar A_{1_X} + \alpha_{10, 3} \, \bar A_1 \, \left(A_{1_X}\right)^2 + \alpha_{11, 3} \, A_1 \, \abs{A_{1_X}}^2
            \\
            + \alpha_{12, 3} \, A_1 + \alpha_{13, 3} \, \abs{A_1}^2 \, A_1 + \alpha_{14, 3} \, \abs{A_1}^4 \, A_1 + \alpha_{15, 3} \, \abs{A_1}^6 \, A_1 + 2 \, \alpha_1 \, A_{1_{X \Xi}} + i \, \alpha_2 \, A_{1_\Xi} + i \, \alpha_3 \, \abs{A_1}^2 \, A_{1_\Xi}
            \\
            + i \, \alpha_4 \, A_1^2 \, \bar A_{1_\Xi} = 0, \label{A_3}
        \end{multline}
        }
        where $\alpha_{p, 3} = \left \langle \mbf Q_p^{[7]}, \bs \psi\right \rangle$ for $p = 1, \ldots, 15$, with
        {\allowdisplaybreaks

        }
        Finally, as we stated at the previous order, we assume that $A_2 = 0$, which implies that \eqref{A_3} becomes
        \begin{multline}
            \alpha_1 \, A_{3_{X \! X}} + i \, \alpha_2 \, A_{3_X} + i \, \alpha_3 \, \abs{A_1}^2 \, A_{3_X} + i \, \alpha_3 \, A_1 \, A_{1_X} \, \bar A_3 + i \, \alpha_3 \, \bar A_1 \, A_{1_X} \, A_3 + i \, \alpha_4 \, A_1^2 \, \bar A_{3_X}
            \\
            + 2 \, i \, \alpha_4 \, A_1 \, \bar A_{1_X} \, A_3 + \alpha_5 \, A_3 + \alpha_6 \, A_1^2 \, \bar A_3 + 2 \, \alpha_6 \, \abs{A_1}^2 \, A_3 + 2 \, \alpha_7 \, \abs{A_1}^2 \, A_1^2 \, \bar A_3 + 3 \, \alpha_7 \, \abs{A_1}^4 \, A_3
            \\
            + i \, \alpha_{1, 3} \, A_{1_{X \! X \! X}} + \alpha_{2, 3} \, A_{1_{X \! X}} + \alpha_{3, 3} \, \abs{A_1}^2 \, A_{1_{X \! X}} + \alpha_{4, 3} \, A_1^2 \, \bar A_{1_{X \! X}} + i \, \alpha_{5, 3} \, A_{1_X} + i \, \alpha_{6, 3} \, \abs{A_1}^2 \, A_{1_X}
            \\
            + i \, \alpha_{7, 3} \, \abs{A_1}^4 \, A_{1_X} + i \, \alpha_{8, 3} \, A_1^2 \, \bar A_{1_X} + i \, \alpha_{9, 3} \, \abs{A_1}^2 \, A_1^2 \, \bar A_{1_X} + \alpha_{10, 3} \, \bar A_1 \, \left(A_{1_X}\right)^2 + \alpha_{11, 3} \, A_1 \, \abs{A_{1_X}}^2
            \\
            + \alpha_{12, 3} \, A_1 + \alpha_{13, 3} \, \abs{A_1}^2 \, A_1 + \alpha_{14, 3} \, \abs{A_1}^4 \, A_1 + \alpha_{15, 3} \, \abs{A_1}^6 \, A_1 + 2 \, \alpha_1 \, A_{1_{X \! \Xi}} + i \, \alpha_2 \, A_{1_\Xi} + i \, \alpha_3 \, \abs{A_1}^2 \, A_{1_\Xi}
            \\
            + i \, \alpha_4 \, A_1^2 \, \bar A_{1_\Xi} = 0. \label{actual-A_3}
        \end{multline}        
        Now, from the amplitude equation for $A_1$, \eqref{firstamplitudeeq}, we note that
        \begin{align*}
            A_1'' = - \frac{1}{\alpha_1} \, \left(i \, \alpha_2 \, A_1' + i \, \alpha_3 \, \abs{A_1}^2 \, A_1' + i \, \alpha_4 \, A_1^2 \, \bar A_1' + \alpha_5 \, A_1 + \alpha_6 \, \abs{A_1}^2 \, A_1 + \alpha_7 \, \abs{A_1}^4 A_1\right),
        \end{align*}
        which can be used to simplify \eqref{actual-A_3} into
        \begin{multline}
            \alpha_1 \, A_{3_{X \! X}} + i \, \alpha_2 \, A_{3_X} + i \, \alpha_3 \, \abs{A_1}^2 \, A_{3_X} + i \, \alpha_4 \, A_1^2 \, \bar A_{3_X} + \alpha_5 \, A_3 + 2 \, \alpha_6 \, \abs{A_1}^2 \, A_3 + 3 \, \alpha_7 \, \abs{A_1}^4 \, A_3
            \\
            + i \, \alpha_3 \, \bar A_1 \, A_{1_X} \, A_3 + 2 \, i \, \alpha_4 \, A_1 \, \bar A_{1_X} \, A_3 + \alpha_6 \, A_1^2 \, \bar A_3 + 2 \, \alpha_7 \, \abs{A_1}^2 \, A_1^2 \, \bar A_3 + i \, \alpha_3 \, A_1 \, A_{1_X} \, \bar A_3
            \\
            + i \, \alpha_{1, 4} \, A_{1_X} + i \, \alpha_{2, 4} \, \abs{A_1}^2 \, A_{1_X} + i \, \alpha_{3, 4} \, \abs{A_1}^4 \, A_{1_X} + i \, \alpha_{4, 4} \, A_1^2 \, \bar A_{1_X} + i \, \alpha_{5, 4} \, \abs{A_1}^2 \, A_1^2 \, \bar A_{1_X}
            \\
            + \alpha_{6, 4} \, \bar A_1 \, \left(A_{1_X}\right)^2 + \alpha_{7, 4} \, A_1 \, \abs{A_{1_X}}^2 + \alpha_{8, 4} \, A_1 + \alpha_{9, 4} \, \abs{A_1}^2 \, A_1 + \alpha_{10, 4} \, \abs{A_1}^4 \, A_1 + \alpha_{11, 4} \, \abs{A_1}^6 \, A_1
            \\
            + 2 \, \alpha_1 \, A_{1_{X \! \Xi}} + i \, \alpha_2 \, A_{1_\Xi} + i \, \alpha_3 \, \abs{A_1}^2 \, A_{1_\Xi} + i \, \alpha_4 \, A_1^2 \, \bar A_{1_\Xi} = 0, \label{messy-A_3}
        \end{multline}
        where
        \begin{align*}
            \alpha_{1, 4} &= \alpha_{5, 3} - \frac{\alpha_5}{\alpha_1} \, \alpha_{1, 3} - \frac{\alpha_2}{\alpha_1} \, \left(\alpha_{2, 3} + \frac{\alpha_2}{\alpha_1} \, \alpha_{1, 3}\right),
            \\
            \alpha_{2, 4} &= \alpha_{6, 3} - \frac{\alpha_2}{\alpha_1} \, \alpha_{3, 3} - 2 \, \frac{\alpha_6}{\alpha_1} \, \alpha_{1, 3} - \frac{\alpha_3}{\alpha_1} \, \left(\alpha_{2, 3} + 2 \, \frac{\alpha_2}{\alpha_1} \, \alpha_{1, 3}\right),
            \\
            \alpha_{3, 4} &= \alpha_{7, 3} - 3 \, \frac{\alpha_7}{\alpha_1} \, \alpha_{1, 3} - \frac{\alpha_3}{\alpha_1} \, \left(\alpha_{3, 3} + \frac{\alpha_3}{\alpha_1} \, \alpha_{1, 3}\right) + \frac{\alpha_4}{\alpha_1} \, \left(\alpha_{4, 3} + \frac{\alpha_4}{\alpha_1} \, \alpha_{1, 3}\right),
            \\
            \alpha_{4, 4} &= \alpha_{8, 3} + \frac{\alpha_2}{\alpha_1} \, \alpha_{4, 3} - \frac{\alpha_4}{\alpha_1} \, \alpha_{2, 3} - \frac{\alpha_6}{\alpha_1} \, \alpha_{1, 3},
            \\
            \alpha_{5, 4} &= \alpha_{9, 3} + \frac{\alpha_3}{\alpha_1} \, \alpha_{4, 3} - \frac{\alpha_4}{\alpha_1} \, \alpha_{3, 3} - 2 \, \frac{\alpha_7}{\alpha_1} \, \alpha_{1, 3},
            \\
            \alpha_{6, 4} &= \alpha_{10, 3} + \frac{\alpha_3}{\alpha_1} \, \alpha_{1, 3},
            \\
            \alpha_{7, 4} &= \alpha_{11, 3} + \frac{\left(\alpha_3 + 2 \, \alpha_4\right)}{\alpha_1} \, \alpha_{1, 3},
            \\
            \alpha_{8, 4} &= \alpha_{12, 3} - \frac{\alpha_5}{\alpha_1} \, \left(\alpha_{2, 3} + \frac{\alpha_2}{\alpha_1} \, \alpha_{1, 3}\right),
            \\
            \alpha_{9, 4} &= \alpha_{13, 3} - \frac{\alpha_5}{\alpha_1} \, \left(\alpha_{3, 3} + \alpha_{4, 3} + \frac{\left(\alpha_3 + \alpha_4\right)}{\alpha_1} \, \alpha_{1, 3}\right) - \frac{\alpha_6}{\alpha_1} \, \left(\alpha_{2, 3} + \frac{\alpha_2}{\alpha_1} \, \alpha_{1, 3}\right),
            \\
            \alpha_{10, 4} &= \alpha_{14, 3} - \frac{\alpha_6}{\alpha_1} \, \left(\alpha_{3, 3} + \alpha_{4, 3} + \frac{\left(\alpha_3 + \alpha_4\right)}{\alpha_1} \, \alpha_{1, 3}\right) - \frac{\alpha_7}{\alpha_1} \, \left(\alpha_{2, 3} + \frac{\alpha_2}{\alpha_1} \, \alpha_{1, 3}\right),
            \\
            \alpha_{11, 4} &= \alpha_{15, 3} - \frac{\alpha_7}{\alpha_1} \, \left(\alpha_{3, 3} + \alpha_{4, 3} + \frac{\left(\alpha_3 + \alpha_4\right)}{\alpha_1} \, \alpha_{1, 3}\right).
        \end{align*}
        Now, we proceed to solve equation \eqref{messy-A_3}. First, we set $A_3 = \left(R_3 + i \, R_1 \, \varphi_3\right) \, e^{i \, \varphi_1}$, where $R_3 = R_3(X, \Xi)$, and $\varphi_3 = \varphi(X, \Xi)$ are real functions. This implies that we obtain two equations, one corresponding to the real and another to the imaginary part of the resulting equation. The first one, after using equations \eqref{varphi} and \eqref{diffusionequation} to simplify it, becomes
        \begin{multline*}
            R_1^2 \, \varphi_{3_{X \! X}} + 2 \, R_1 \, R_{1_X} \, \varphi_{3_X} = - \frac{\alpha_{1, 4} \, R_1 \, R_{1_X}}{\alpha_1} - \frac{\left(\alpha_3 + \alpha_4\right) \, \left(R_1^3 \, R_3\right)_X}{2 \, \alpha_1}
            \\
            + \left(\frac{\alpha_2 \, \alpha_{6, 4}}{\alpha_1^2} - \frac{\left(\alpha_{2, 4} + \alpha_{4, 4}\right)}{\alpha_1}\right) \, R_1^3 \, R_{1_X} + \left(\frac{\left(\alpha_3 + \alpha_4\right) \, \alpha_{6, 4}}{2 \, \alpha_1^2} - \frac{\left(\alpha_{3, 4} + \alpha_{5, 4}\right)}{\alpha_1}\right) \, R_1^5 \, R_{1_X},
        \end{multline*}
        which implies
        \begin{multline*}
            \varphi_{3_X} = - \frac{\alpha_2}{2 \, \alpha_1} - \frac{\alpha_{1, 4}}{2 \, \alpha_1} - \zeta_\Xi - \frac{\left(\alpha_3 + \alpha_4\right) \, R_1 \, R_3}{2 \, \alpha_1} + \left(\frac{\alpha_2 \, \alpha_{6, 4}}{4 \, \alpha_1^2} - \frac{\alpha_{2, 4} + \alpha_{4, 4}}{4 \, \alpha_1}\right) \, R_1^2
            \\
            + \left(\frac{\left(\alpha_3 + \alpha_4\right) \, \alpha_{6, 4}}{12 \, \alpha_1^2} - \frac{\alpha_{3, 4} + \alpha_{5, 4}}{6 \, \alpha_1}\right) \, R_1^4 + \frac{4 \, \alpha_1 \, \omega_3}{R_1^2}.
        \end{multline*}
        where $\omega_3 \in \mathbb R$ is a constant of integration. With this, the second equation becomes
        \begin{align}
            \alpha_1 \, R_{3_{X \! X}} - \alpha_1 \, \left(\beta_1 + 6 \, \beta_3 \, R_1^2 + 15 \, \beta_5 \, R_1^4\right) \, R_3 = 2 \, \alpha_1 \, \beta_{\dagger, 3} \, R_1 + \alpha_1 \, \frac{\dd p}{\dd R_1}\left(R_1\right) + \alpha_1 \, \beta_{*, 3} \, R_1 \, \left(R_1'\right)^2, \label{re-R_3}
        \end{align}
        where
        \begin{align*}
            \beta_{\dagger, 3} &= \left(\alpha_3 - 3 \, \alpha_4\right) \, \omega_3 - \frac{\alpha_2 \, \left(\alpha_2 + \alpha_{1, 4}\right)}{4 \, \alpha_1^2} - \frac{\alpha_{8, 4}}{2 \, \alpha_1} - \frac{\alpha_2}{2 \, \alpha_1} \, \zeta_\Xi,
            \\
            \alpha_1 \, \frac{\dd p}{\dd R_1}\left(R_1\right) &= - \beta_{3, 3} \, R_1^3 - \beta_{5, 3} \, R_1^5 - \beta_{7, 3} \, R_1^7,
            \\
            \beta_{3, 3} &= \frac{\alpha_2 \, \left(\alpha_3 - \alpha_4 + \alpha_{2, 4} - \alpha_{4, 4}\right)}{2 \, \alpha_1} + \frac{\alpha_{1, 4} \, \left(\alpha_3 - \alpha_4\right)}{2 \, \alpha_1} + \frac{\alpha_2^2 \, \left(\alpha_{7, 4} - \alpha_{6, 4}\right)}{4 \, \alpha_1^2} + \alpha_{9, 4},
            \\
            \beta_{5, 3} &= \frac{\alpha_2 \, \left(\alpha_{3, 4} - \alpha_{5, 4}\right)}{2 \, \alpha_1} + \frac{\alpha_3 \, \left(3 \, \alpha_{2, 4} - \alpha_{4, 4}\right)}{8 \, \alpha_1} - \frac{\alpha_4 \, \left(\alpha_{2, 4} + 5 \, \alpha_{4, 4}\right)}{8 \, \alpha_1}
            \\
            & \quad + \frac{\alpha_2 \, \alpha_3 \, \left(2 \, \alpha_{7, 4} - 3 \, \alpha_{6, 4}\right)}{8 \, \alpha_1^2} + \frac{\alpha_2 \, \alpha_4 \, \left(\alpha_{6, 4} + 2 \, \alpha_{7, 4}\right)}{8 \, \alpha_1^2} + \alpha_{10, 4},
            \\
            \beta_{7, 3} &= \frac{\alpha_3 \, \left(2 \, \alpha_{3, 4} - \alpha_{5, 4}\right)}{6 \, \alpha_1} - \frac{\alpha_4 \, \alpha_{5, 4}}{2 \, \alpha _1} + \frac{\alpha_3^2 \, \left(3 \, \alpha_{7, 4} - 5 \, \alpha_{6, 4}\right)}{48 \, \alpha_1^2} + \frac{\alpha_3 \, \alpha_4 \, \left(3 \, \alpha_{7, 4} - \alpha_{6, 4}\right)}{24 \, \alpha_1^2}
            \\
            & \quad + \frac{\alpha_4^2 \, \left(\alpha_{6, 4} + \alpha_{7, 4}\right)}{16 \, \alpha_1^2} + \alpha_{11, 4},
            \\
            \beta_{*, 3} &= - \frac{1}{\alpha_1} \, \left(\alpha_{6, 4} + \alpha_{7, 4}\right).
        \end{align*}
        Now, note that \eqref{re-R_3} is a linear equation with respect to $R_3$ that has a homogeneous solution given by
        \begin{align*}
            R_{3, h} = R_1' \, \left(\omega_1 + \omega_2 \, \int^X \frac{\dd s}{\left(R_1'\right)^2}\right),
        \end{align*}
        which implies that a particular solution is given by
        \begin{multline*}
            R_{3, p} = R_1' \, \left(- \beta_{\dagger, 3} \, \left(R_1^2 \, \int^X \frac{\dd s}{\left(R_1'\right)^2} - \int^X \frac{R_1^2}{\left(R_1'\right)^2} \, \dd \sigma\right)\right.
            \\
            - \left(p\left(R_1\right) \, \int^X \frac{\dd s}{\left(R_1'\right)^2} - \int^X \frac{p\left(R_1\right)}{\left(R_1'\right)^2} \, \dd \sigma\right) - \beta_{*, 3} \, \int^X \, \left(R_1 \, \left(R_1'\right)^3 \, \int^\sigma \frac{\dd s}{\left(R_1'\right)^2}\right) \, \dd \sigma
            \\
            \left. + \left(\beta_{\dagger, 3} \, R_1^2 + p\left(R_1\right) + \beta_{*, 3} \, \int^X R_1 \, \left(R_1'\right)^3 \, \dd \sigma\right) \, \int^X \frac{\dd s}{\left(R_1'\right)^2}\right)
            \\
            = R_1' \, \left(\beta_{\dagger, 3} \, \int^X \frac{R_1^2}{\left(R_1'\right)^2} \, \dd \sigma + \int^X \frac{p\left(R_1\right)}{\left(R_1'\right)^2} \, \dd \sigma\right.
            \\
            \left. + \beta_{*, 3} \, \left(\int^X R_1 \, \left(R_1'\right)^3 \, \dd \sigma\right) \, \left(\int^X \frac{\dd s}{\left(R_1'\right)^2}\right) - \beta_{*, 3} \, \int^X \, \left(R_1 \, \left(R_1'\right)^3 \, \int^\sigma \frac{\dd s}{\left(R_1'\right)^2}\right) \, \dd \sigma\right).
        \end{multline*}
        Therefore,
        \begin{align*}
            R_3 &= R_{3, h} + R_{3, p}
            \\
            &= \begin{multlined}[t]
                R_1' \, \left(\omega_1 + \int^X \frac{\left(\omega_2 + \beta_{\dagger, 3} \, R_1^2 + p\left(R_1\right)\right)}{\left(R_1'\right)^2} \, \dd \sigma\right.
                \\
                \left. + \beta_{*, 3} \, \left(\int^X R_1 \, \left(R_1'\right)^3 \, \dd \sigma\right) \, \left(\int^X \frac{\dd s}{\left(R_1'\right)^2}\right) - \beta_{*, 3} \, \int^X \, \left(R_1 \, \left(R_1'\right)^3 \, \int^\sigma \frac{\dd s}{\left(R_1'\right)^2}\right) \, \dd \sigma\right),
            \end{multlined}
        \end{align*}
        which lets us note that, as $X \to - \infty$, we have
        \begin{align*}
            A_3 &\sim \left(-\frac{\beta_3}{\beta_1}\right)^{3/2} \, \frac{ \left(- \beta_1^3 \, \beta_{\star, 3} + 6 \, \beta_3^2 \, \omega_2 + 24 \, i \, \alpha_1 \, \sqrt{\beta_1} \, \beta_3^2 \, \omega_3\right)}{12 \, \sqrt{2} \, \beta_3^3} \, e^{- \sqrt{\beta_1} \, X} \, e^{i \, \varphi_1},
        \end{align*}
        implying that we must set
        \begin{align*}
            \omega_2 = - \frac{\beta_1^3 \, \left(\alpha_{6, 4} + \alpha_{7, 4}\right)}{6 \, \alpha_1 \, \beta_3^2}, \qquad \text{and} \qquad \omega_3 = 0,
        \end{align*}
        to keep $A_3$ bounded as $X \to - \infty$. On the other hand, using these variables, as $X \to \infty$ we have
        \begin{align*}
            A_3 &\sim \frac{\left(- \frac{\beta_3}{\beta_1}\right)^{3/2}}{24 \, \sqrt{2} \, \alpha_1^2 \, \beta_3^6} \, \left(2 \, \alpha_1 \, \beta_3 + i \, \sqrt{\beta_1} \, \left(\alpha_3 + \alpha_4\right)\right) \, \left(- 6 \, \alpha_1 \, \beta_1 \, \beta_3^3 \, \beta_{\dagger, 3} + 3 \, \alpha_1 \, \beta_3^4 \, \omega_2 - 3 \, \beta_1^2 \, \beta_3^2 \, \beta_{3, 3}\right.
            \\
            & \quad \left. + 4 \, \beta_1^3 \, \beta_3 \, \beta_{5, 3} - 6 \, \beta_1^4 \, \beta_{7, 3}\right) \, e^{2 \, \sqrt{\beta_1} \, X} \, e^{i \, \varphi_1}
            \\
            & \quad - i \, \frac{\beta_1^2 \, \left(- \frac{\beta_3}{\beta_1}\right)^{3/2}}{12 \, \sqrt{2} \, \alpha_1^3 \, \beta_3^6} \, \left(3 \, \alpha_2 \, \beta_3^3 \, \left(4 \, \alpha_1 \, \beta_1 \, \alpha_{6, 4} + \left(\alpha_3 + \alpha_4\right) \, \alpha_{1, 4} + 4 \, \alpha_1^2 \, \beta_3\right)\right.
            \\
            & \quad + 2 \, \alpha_1 \, \left(2 \, \alpha_1 \, \beta_3^2 \, \left(4 \, \beta_1^2 \, \left(\alpha_{3, 4} + \alpha_{5, 4}\right) - 3 \, \beta_1 \, \beta_3 \, \left(\alpha_{2, 4} + \alpha_{4, 4}\right) + 3 \, \beta_3^2 \, \alpha_{1, 4}\right)\right.
            \\
            & \quad \left. + \left(\alpha_3 + \alpha_4\right) \, \left(- \beta_1^2 \, \beta_3 \, \left(\beta_3 \, \left(5 \, \alpha_{6, 4} + \alpha_{7, 4}\right) + 4 \, \beta_{5, 3}\right) + 3 \, \beta_3^3 \, \alpha_{8, 4} + 12 \, \beta_1^3 \, \beta_{7, 3}\right)\right)
            \\
            & \quad \left. + 3 \, \alpha_2^2 \, \left(\alpha_3 + \alpha_4\right) \, \beta_3^3 + 6 \, \alpha_1 \, \beta_3^3 \, \left(\alpha_2 \, \left(\alpha_3 + \alpha_4\right) + 4 \, \alpha_1^2 \, \beta_3\right) \, \zeta_\Xi\right) \, X \, e^{2 \, \sqrt{\beta_1} \, X} \, e^{i \, \varphi_1}.
        \end{align*}
        Once again, this expression must equal zero for $A_3$ to remain finite as $X \to \infty$. To obtain explicit conditions that make this quantity equal to zero, note that if $\alpha_2 = 0$, then we can choose
        \begin{align}
            \beta_{7, 3} &= - \frac{\beta_3 \, \left(\beta_1^2 \, \left(\beta_3 \, \left(\alpha_{6, 4} + \alpha_{7, 4}\right) - 8 \, \beta_{5, 3}\right) - 6 \, \beta_3^2 \, \alpha_{8, 4} + 6 \, \beta_1 \, \beta_3 \, \beta_{3, 3}\right)}{12 \, \beta_1^3}, \label{Maxwell_correctionalpha_20}
            \\
            \zeta_\Xi &= \frac{1}{12 \, \alpha_1^2 \, \beta_3^3} \, \left(\left(\alpha_3 + \alpha_4\right) \, \left(2 \, \beta_1^2 \, \left(\beta_3 \, \left(3 \, \alpha_{6, 4} + \alpha_{7, 4}\right) - 2 \, \beta_{5, 3}\right) - 9 \, \beta_3^2 \, \alpha_{8, 4} + 6 \, \beta_1 \, \beta_3 \, \beta_{3, 3}\right)\right. \notag
            \\
            & \quad \left. - 2 \, \alpha_1 \, \beta_3 \, \left(4 \, \beta_1^2 \, \left(\alpha_{3, 4} + \alpha_{5, 4}\right) - 3 \, \beta_1 \, \beta_3 \, \left(\alpha_{2, 4} + \alpha_{4, 4}\right) + 3 \, \beta_3^2 \, \alpha_{1, 4}\right)\right), \label{zetaprimedefalpha_20}
        \end{align}
        whilst, if $\alpha_2 \neq 0$, then we choose
        \begin{align}
            \alpha_{2, 4} &= \frac{1}{12 \, \alpha_1^2 \, \beta_1 \, \beta_3^3} \, \left(3 \, \alpha_2 \, \beta_3^3 \, \left(2 \, \alpha_1 \, \left(2 \, \beta_1 \, \alpha_{6, 4} + 2 \, \alpha_1 \, \beta_3 + \left(\alpha_3 + \alpha_4\right) \, \zeta_\Xi\right) + \left(\alpha_3 + \alpha_4\right) \, \alpha_{1, 4}\right)\right. \notag
            \\
            & \quad + 2 \, \alpha_1 \, \left(2 \, \alpha_1 \, \beta_3^2 \, \left(4 \, \beta_1^2 \, \left(\alpha_{3, 4} + \alpha_{5, 4}\right) - 3 \, \beta_3 \, \beta_1 \, \alpha_{4, 4} + 3 \, \beta_3^2 \, \alpha_{1, 4}\right) + 3 \, \alpha_2^2 \, \left(\alpha_3 + \alpha_4\right) \, \beta_3^3\right. \notag
            \\
            & \quad + \left(\alpha_3 + \alpha_4\right) \, \left(- \beta_1^2 \, \beta_3 \, \left(\beta_3 \, \left(5 \, \alpha_{6, 4} + \alpha_{7, 4}\right) + 4 \, \beta_{5, 3}\right) + 3 \, \beta_3^3 \, \alpha_{8, 4} + 12 \, \beta_1^3 \, \beta_{7, 3}\right) \notag
            \\
            & \quad \left.\left. + 12 \, \alpha_1^2 \, \beta_3^4 \, \zeta_\Xi\right)\right), \label{Maxwell_correctionalpha_2n0}
            \\
            \zeta_\Xi &= \begin{multlined}[t]
                \frac{1}{6} \, \left(\frac{\beta_1^2 \, \beta_3 \, \left(\beta_3 \, \left(\alpha_{6, 4} + \alpha_{7, 4}\right) - 8 \, \beta_{5, 3}\right) - 6 \, \beta_3^3 \, \alpha_{8, 4} + 12 \, \beta_1^3 \, \beta_{7, 3} + 6 \, \beta_1 \, \beta_3^2 \, \beta_{3, 3}}{\alpha_2 \, \beta_3^3}\right.
                \\
                \left. - \frac{3 \, \left(\alpha_2 + \alpha_{1, 4}\right)}{\alpha_1}\right),
            \end{multlined} \label{zetaprimedefalpha_2n0}
        \end{align}
        where \eqref{Maxwell_correctionalpha_20} and \eqref{Maxwell_correctionalpha_2n0} correspond to equations that must be solved in terms of the parameters of the expansion in order to obtain a correction to the Maxwell point, and \eqref{zetaprimedefalpha_20} and \eqref{zetaprimedefalpha_2n0} correspond to the values of $\zeta_\Xi$ in these two cases, where the equalities \eqref{Maxwell_correctionalpha_20} and \eqref{Maxwell_correctionalpha_2n0} have been used to explicitly state the value of $\zeta_\Xi$ in each case. We highlight that, in both cases, $\zeta_\Xi \in \mathbb R$ is a constant, which implies that $\zeta$ is a linear function on $\Xi$.

        \begin{remark} \label{rem:imaginary integrals}
            The integrals in this appendix were defined as $\ds{\int^X \dd s}$ just for simplicity, as the integrals for the homogeneous solution do converge when computing $\ds{\int_{X_0}^X \dd s}$, but the same integrals for the particular solution do not.
        \end{remark}
        In any case, we highlight that what is highlighted in Remark \ref{rem:imaginary integrals} does not change the value of the leading-order coefficient of $A_3$ at $X = X_0$, which is given by
        \begin{multline}
            - \frac{e^{\pi \, \xi + \zeta \, i}}{192 \, \alpha_1^3 \, \beta_1^{5/4} \, \left(- \beta_3\right)^{9/2} \, \sqrt{-\frac{1}{X_0 - X}} \, \left(X_0 - X\right)^{2 - \eta \, i}}
            \\
            \times \left(- 3 \, (2 \, \pi + 3 \, i) \, \alpha_2 \, \beta_3^3 \, \left(\left(\alpha_3 + \alpha_4\right) \, \sqrt{\beta_1} + 2 \, i \, \alpha_1 \, \beta_3\right) \, \left(\alpha_{1, 4} + 2 \, \alpha_1 \, \zeta_\Xi\right)\right.
            \\
            + i \, \alpha_1 \, \left(2 \, \alpha_1 \, \beta_3 \, \left(\beta_3 \, \left(- 16 \, \beta_1^{5/2} \, \left(\alpha_{3, 4} + \alpha_{5, 4}\right) + i \, \beta_1^2 \, \left(5 \, \beta_3 \, \left(\alpha_{6, 4} + \alpha_{7, 4}\right) + 8 \, \beta_{5, 3}\right)\right.\right.\right.
            \\
            \left.\left. - 6 \, (2 \, \pi + 3 \, i) \, \beta_3^2 \, \alpha_{8, 4} + 6 \, i \, \beta_1 \, \beta_3 \, \beta_{3, 3}\right) + 12 \, (2 \, \pi - 3 \, i) \, \beta_1^3 \, \beta_{7, 3}\right)
            \\
            + \left(\alpha_3 + \alpha_4\right) \, \sqrt{\beta_1} \, \left(\beta_1^2 \, \beta_3 \, \left(\beta_3 \, \left(15 \, \alpha_{6, 4} - \alpha_{7, 4}\right) + 8 \, \beta_{5, 3}\right) + 6 \, (- 3 + 2 \, \pi \, i) \, \beta_3^3 \, \alpha_{8, 4}\right.
            \\
            \left.\left. + 12 \, (- 5 - 2 \, \pi \, i) \, \beta_1^3 \, \beta_{7, 3} + 6 \, \beta_3^2 \, \beta_1 \, \beta_{3, 3}\right)\right)
            \\
            \left. - 3 \, (2 \, \pi + 3 \, i) \, \alpha_2^2 \, \beta_3^3 \, \left(\left(\alpha_3 + \alpha_4\right) \, \sqrt{\beta_1} + 2 \, \alpha_1 \, \beta_3 \, i\right)\right). \label{C_3}
        \end{multline}
        
    \section{Equation for the remainder up to order 7} \label{app:rem_order7}
        In this appendix, in order to be complete, we state the key expressions of each component of the remainder up to order six, to obtain its amplitude equation at order seven. We highlight that the expressions below will depend on the vectors that were obtained in \ref{app:7expansion}. In particular, we have that
        {\allowdisplaybreaks
        \begin{align*}
            \mbf R_N^{[5]} &= 2 \, A_1 \, \bar B_1 \, \mbf W_0^{[5]} + 4 \, \abs{A_1}^2 \, A_1 \, \bar B_1 \, \mbf W_{0, 2}^{[5]} + i \, A_{1_X} \, \bar B_1 \, \mbf W_{0, 3}^{[5]} + i \, \bar A_1 \, B_{1_X} \, \mbf W_{0, 3}^{[5]} + B_{1_{X \! X}} \, e^{i x} \, \mbf W_1^{[5]}
            \\
            & \quad + i \, B_{1_X} \, e^{i x} \, \mbf W_{1, 2}^{[5]} + i \, \abs{A_1}^2 \, B_{1_X} \, e^{i x} \, \mbf W_{1, 3}^{[5]} + i \, A_1 \, A_{1_X} \, \bar B_1 \, e^{i x} \, \mbf W_{1, 3}^{[5]} + i \, \bar A_1 \, A_{1_X} \, B_1 \, e^{i x} \, \mbf W_{1, 3}^{[5]}
            \\
            & \quad + i \, A_1^2 \, \bar B_{1_X} \, e^{i x} \, \mbf W_{1, 4}^{[5]} + 2 \, i \, A_1 \, \bar A_{1_X} \, B_1 \, e^{i x} \, \mbf W_{1, 4}^{[5]} + B_1 \, e^{i x} \, \mbf W_{1, 5}^{[5]} + A_1^2 \, \bar B_1 \, e^{i x} \, \mbf W_{1, 6}^{[5]}
            \\
            & \quad + 2 \, \abs{A_1}^2 \, B_1 \, e^{i x} \, \mbf W_{1, 6}^{[5]} + 2 \, \abs{A_1}^2 \, A_1^2 \, \bar B_1 \, e^{i x} \, \mbf W_{1, 7}^{[5]} + 3 \, \abs{A_1}^4 \, B_1 \, e^{i x} \, \mbf W_{1, 7}^{[5]} + i \, B_{1_\Xi} \, e^{i x} \, \mbf W_{1, 3}^{[3]}
            \\
            & \quad + 2 \, A_1 \, B_1 \, e^{2 i x} \, \mbf W_2^{[5]} + A_1^3 \, \bar B_1 \, e^{2 i x} \, \mbf W_{2, 2}^{[5]} + 3 \, \abs{A_1}^2 \, A_1 \, B_1 \, e^{2 i x} \, \mbf W_{2, 2}^{[5]} + i \, A_{1_X} \, B_1 \, e^{2 i x} \, \mbf W_{2, 3}^{[5]}
            \\
            & \quad + i \, A_1 \, B_{1_X} \, e^{2 i x} \, \mbf W_{2, 3}^{[5]} + 3 \, A_1^2 \, B_1 \, e^{3 i x} \, \mbf W_3^{[5]} + A_1^4 \, \bar B_1 \, e^{3 i x} \, \mbf W_{3, 2}^{[5]} + 4 \, \abs{A_1}^2 \, A_1^2 \, B_1 \, e^{3 i x} \, \mbf W_{3, 2}^{[5]}
            \\
            & \quad + i \, A_1^2 \, B_{1_X} \,  e^{3 i x} \, \mbf W_{3, 3}^{[5]} + 2 \, i \, A_1 \, A_{1_X} \, B_1 \, e^{3 i x} \, \mbf W_{3, 3}^{[5]} + 2 \, A_2 \, \bar B_1 \, \mbf W_0^{[4]} + 2 \, A_1 \, \bar B_2 \, \mbf W_0^{[4]}
            \\
            & \quad + 4 \, \abs{A_1}^2 \, A_1 \, \bar B_2 \, \mbf W_{0, 2}^{[4]} + 4 \, A_1^2 \, \bar A_2 \, \bar B_1 \, \mbf W_{0, 2}^{[4]} + 8 \, \abs{A_1}^2 \, A_2 \, \bar B_1 \, \mbf W_{0, 2}^{[4]} + i \, A_{1_X} \, \bar B_2 \, \mbf W_{0, 3}^{[4]}
            \\
            & \quad + i \, A_{2_X} \, \bar B_1 \, \mbf W_{0, 3}^{[4]} + i \, \bar A_2 \, B_{1_X} \, \mbf W_{0, 3}^{[4]} + i \, \bar A_1 \, B_{2_X} \, \mbf W_{0, 3}^{[4]} + 2 \, A_2 \, \bar B_2 \, \mbf W_0^{[3]} + 2 \, A_3 \, \bar B_1 \, \mbf W_0^{[3]}
            \\
            & \quad + 2 \, A_1 \, \bar B_3 \, \mbf W_0^{[3]} + 2 \, A_1 \, \bar B_4 \, \mbf W_0^{[2]} + 2 \, A_3 \, \bar B_2 \, \mbf W_0^{[2]} + 2 \, A_4 \, \bar B_1 \, \mbf W_0^{[2]} + 2 \, A_2 \, \bar B_3 \, \mbf W_0^{[2]}
            \\
            & \quad + B_2 \, e^{i x} \, \mbf W_1^{[4]} + A_1^2 \, \bar B_2 \, e^{i x} \, \mbf W_{1, 2}^{[4]} + 2 \, \abs{A_1}^2 \, B_2 \, e^{i x} \, \mbf W_{1, 2}^{[4]} + 2 \, \bar A_1 \, A_2 \, B_1 \, e^{i x} \, \mbf W_{1, 2}^{[4]}
            \\
            & \quad + 2 \, A_1 \, A_2 \, \bar B_1 \, e^{i x} \, \mbf W_{1, 2}^{[4]} + 2 \, A_1 \, \bar A_2 \, B_1 \, e^{i x} \, \mbf W_{1, 2}^{[4]} + i \, B_{2_X} \, e^{i x} \, \mbf W_{1, 3}^{[4]} + B_3 \, e^{i x} \, \mbf W_1^{[3]}
            \\
            & \quad + A_2^2 \, \bar B_1 \, e^{i x} \, \mbf W_{1, 2}^{[3]} + A_1^2 \, \bar B_3 \, e^{i x} \, \mbf W_{1, 2}^{[3]} + 2 \, A_1 \, \bar A_2 \, B_2 \, e^{i x} \, \mbf W_{1, 2}^{[3]} + 2 \, \bar A_1 \, A_2 \, B_2 \, e^{i x}  \, \mbf W_{1, 2}^{[3]}
            \\
            & \quad + 2 \, A_1 \, A_2 \, \bar B_2 \, e^{i x} \, \mbf W_{1, 2}^{[3]} + 2 \, \abs{A_2}^2 \, B_1 \, e^{i x} \, \mbf W_{1, 2}^{[3]} + 2 \, \abs{A_1}^2 \, B_3 \, e^{i x} \, \mbf W_{1, 2}^{[3]} + 2 \, A_1 \, \bar A_3 \, B_1 \, e^{i x} \, \mbf W_{1, 2}^{[3]}
            \\
            & \quad + 2 \, \bar A_1 \, A_3 \, B_1 \, e^{i x} \, \mbf W_{1, 2}^{[3]} + 2 \, A_1 \, A_3 \, \bar B_1 \, e^{i x} \, \mbf W_{1, 2}^{[3]} + i \, B_{3_X} \, e^{i x} \, \mbf W_{1, 3}^{[3]} + B_4 \, e^{i x} \, \mbf W_1^{[2]} + B_5 \, e^{i x} \, \bs \phi_1^{[1]}
            \\
            & \quad + 2 \, A_1 \, B_2 \, e^{2 i x} \, \mbf W_2^{[4]} + 2 \, A_2 \, B_1 \, e^{2 i x} \, \mbf W_2^{[4]} + A_1^3 \, \bar B_2 \, e^{2 i x} \, \mbf W_{2, 2}^{[4]} + 3 \, \abs{A_1}^2 \, A_1 \, B_2 \, e^{2 i x} \, \mbf W_{2, 2}^{[4]}
            \\
            & \quad + 3 \, A_1^2 \, \bar A_2 \, B_1 \, e^{2 i x} \, \mbf W_{2, 2}^{[4]} + 3 \, A_1^2 \, A_2 \, \bar B_1 \, e^{2 i x} \, \mbf W_{2, 2}^{[4]} + 6 \, \abs{A_1}^2 \, A_2 \, B_1 \, e^{2 i x} \, \mbf W_{2, 2}^{[4]}
            \\
            & \quad + i \, A_1 \, B_{2_X} \, e^{2 i x} \, \mbf W_{2, 3}^{[4]} + i \, A_{1_X} \, B_2 \, e^{2 i x} \, \mbf W_{2, 3}^{[4]} + i \, A_2 \, B_{1_X} \, e^{2 i x} \, \mbf W_{2, 3}^{[4]} + i \, A_{2_X} \, B_1 \, e^{2 i x} \, \mbf W_{2, 3}^{[4]}
            \\
            & \quad + 2 \, A_2 \, B_2 \, e^{2 i x} \, \mbf W_2^{[3]} + 2 \, A_1 \, B_3 \, e^{2 i x} \, \mbf W_2^{[3]} + 2 \, A_3 \, B_1 \, e^{2 i x} \, \mbf W_2^{[3]} + 2 \, A_3 \, B_2 \, e^{2 i x} \, \mbf W_2^{[2]}
            \\
            & \quad + 2 \, A_2 \, B_3 \, e^{2 i x} \, \mbf W_2^{[2]} + 2 \, A_1 \, B_4 \, e^{2 i x} \, \mbf W_2^{[2]} + 2 \, A_4 \, B_1 \, e^{2 i x} \, \mbf W_2^{[2]} + 3 \, A_1^2 \, B_2 \, e^{3 i x} \, \mbf W_3^{[4]}
            \\
            & \quad + 6 \, A_1 \, A_2 \, B_1 \, e^{3 i x} \, \mbf W_3^{[4]} + 3 \, A_2^2 \, B_1 \, e^{3 i x} \, \mbf W_3^{[3]} + 3 \, A_1^2 \, B_3 \, e^{3 i x} \, \mbf W_3^{[3]} + 6 \, A_1 \, A_2 \, B_2 \, e^{3 i x} \, \mbf W_3^{[3]}
            \\
            & \quad + 6 \, A_1 \, A_3 \, B_1 \, e^{3 i x} \, \mbf W_3^{[3]} + \ldots + c.c.,
        \end{align*}
        }
        where '$\ldots$' represents terms that are multiples of $e^{4ix}$ or $e^{5ix}$. Moreover,
        {\allowdisplaybreaks
        \begin{align*}
            \mbf R_N^{[6]} &= 2 \, A_1 \, \bar B_1 \, \mbf W_0^{[6]} + 4 \, \abs{A_1}^2 \, A_1 \, \bar B_1 \, \mbf W_{0, 2}^{[6]} + 6 \, \abs{A_1}^4 \, A_1 \, \bar B_1 \, \mbf W_{0, 3}^{[6]} + 2 \, A_{1_X} \, \bar B_{1_X} \, \mbf W_{0, 4}^{[6]}
            \\
            & \quad + i \, \bar A_1 \, B_{1_X} \, \mbf W_{0, 5}^{[6]} + i \, A_{1_X} \, \bar B_1 \, \mbf W_{0, 5}^{[6]} + i \, \abs{A_1}^2 \, \bar A_1 \, B_{1_X} \, \mbf W_{0, 6}^{[6]} + i \, \bar A_1^2 \, A_{1_X} \, B_1 \, \mbf W_{0, 6}^{[6]}
            \\
            & \quad + 2 \, i \, \abs{A_1}^2 \, A_{1_X} \, \bar B_1 \, \mbf W_{0, 6}^{[6]} + A_1 \, \bar B_{1_{X \! X}} \, \mbf W_{0, 7}^{[6]} + A_{1_{X \! X}} \, \bar B_1 \, \mbf W_{0, 7}^{[6]} + i \, \bar A_1 \, B_{1_\Xi} \, \mbf W_{0, 3}^{[4]}
            \\
            & \quad + 2 \, A_1 \, \bar B_2 \, \mbf W_0^{[5]} + 2 \, A_2 \, \bar B_1 \, \mbf W_0^{[5]} + 4 \, \abs{A_1}^2 \, A_1 \, \bar B_2 \, \mbf W_{0, 2}^{[5]} + 4 \, A_1^2 \, \bar A_2 \, \bar B_1 \, \mbf W_{0, 2}^{[5]}
            \\
            & \quad + 8 \, \abs{A_1}^2 \, A_2 \, \bar B_1 \, \mbf W_{0, 2}^{[5]} + i \, \bar A_1 \, B_{2_X} \, \mbf W_{0, 3}^{[5]} + i \, A_{1_X} \, \bar B_2 \, \mbf W_{0, 3}^{[5]} + i \, \bar A_2 \, B_{1_X} \, \mbf W_{0, 3}^{[5]}
            \\
            & \quad + i \, A_{2_X} \, \bar B_1 \, \mbf W_{0, 3}^{[5]} + 2 \, A_2 \, \bar B_2 \, \mbf W_0^{[4]} + 2 \, A_1 \, \bar B_3 \, \mbf W_0^{[4]} + 2 \, A_3 \, \bar B_1 \, \mbf W_0^{[4]} + 4 \, A_1^2 \, \bar A_2 \, \bar B_2 \, \mbf W_{0, 2}^{[4]}
            \\
            & \quad + 4 \, A_1 \, \bar A_2^2 \, B_1 \, \mbf W_{0, 2}^{[4]} + 4 \, \abs{A_1}^2 \, A_1 \, \bar B_3 \, \mbf W_{0, 2}^{[4]}+ 4 \, A_1^2 \, \bar A_3 \, \bar B_1 \, \mbf W_{0, 2}^{[4]} + 8 \, \abs{A_1}^2 \,  A_2 \, \bar B_2 \, \mbf W_{0, 2}^{[4]}
            \\
            & \quad + 8 \, \abs{A_2}^2 \, A_1 \, \bar B_1 \, \mbf W_{0, 2}^{[4]} + 8 \, \abs{A_1}^2 \, A_3 \, \bar B_1 \, \mbf W_{0, 2}^{[4]} + i \, \bar A_2 \, B_{2_X} \, \mbf W_{0, 3}^{[4]} + i \, A_{2_X} \, \bar B_2 \, \mbf W_{0, 3}^{[4]}
            \\
            & \quad + i \, \bar A_1 \, B_{3_X} \, \mbf W_{0, 3}^{[4]} + i \, \bar A_3 \, B_{1_X} \, \mbf W_{0, 3}^{[4]} + i \, A_{1_X} \, \bar B_3 \, \mbf W_{0, 3}^{[4]} + i \, A_{1_\Xi} \, \bar B_1 \, \mbf W_{0, 3}^{[4]}
            \\
            & \quad + i \, A_{3_X} \, \bar B_1 \, \mbf W_{0, 3}^{[4]} + 2 \, A_3 \, \bar B_2 \, \mbf W_0^{[3]} + 2 \, A_2 \, \bar B_3 \, \mbf W_0^{[3]} + 2 \, A_1 \, \bar B_4 \, \mbf W_0^{[3]} + 2 \, A_4 \, \bar B_1 \, \mbf W_0^{[3]}
            \\
            & \quad + 2 \, A_4 \, \bar B_2 \, \mbf W_0^{[2]} + 2 \, A_3 \, \bar B_3 \, \mbf W_0^{[2]} + 2 \, A_2 \, \bar B_4 \, \mbf W_0^{[2]} + 2 \, A_1 \, \bar B_5 \, \mbf W_0^{[2]} + 2 \, A_5 \, \bar B_1 \, \mbf W_0^{[2]}
            \\
            & \quad + B_{1_{X \! X}} \, e^{i x} \, \mbf W_1^{[6]} + i \, B_{1_X} \, e^{i x} \, \mbf W_{1, 2}^{[6]} + i \, \abs{A_1}^2 \, B_{1_X} \, e^{i x} \, \mbf W_{1, 3}^{[6]} + i \, A_1 \, A_{1_X} \, \bar B_1 \, e^{i x} \, \mbf W_{1, 3}^{[6]}
            \\
            & \quad + i \, \bar A_1 \, A_{1_X} \, B_1 \,  e^{i x} \, \mbf W_{1, 3}^{[6]} + i \, A_1^2 \, \bar B_{1_X} \, e^{i x} \, \mbf W_{1, 4}^{[6]} + 2 \, i \, A_1 \, \bar A_{1_X} \, B_1 \, e^{i x} \, \mbf W_{1, 4}^{[6]} + B_1 \, e^{i x} \, \mbf W_{1, 5}^{[6]}
            \\
            & \quad + A_1^2 \, \bar B_1 \, e^{i x} \, \mbf W_{1, 6}^{[6]} + 2 \, \abs{A_1}^2 \, B_1 \, e^{i x} \, \mbf W_{1, 6}^{[6]} + 2 \, \abs{A_1}^2 \, A_1^2 \, \bar B_1 \, e^{i x} \, \mbf W_{1, 7}^{[6]} + 3 \, \abs{A_1}^4 \, B_1 \, e^{i x} \, \mbf W_{1, 7}^{[6]}
            \\
            & \quad + i \, B_{1_\Xi} \, e^{i x} \, \mbf W_{1, 3}^{[4]} + i \, B_{2_\Xi} \, e^{i x} \, \mbf W_{1, 3}^{[3]} + B_{2_{X \! X}} \, e^{i x} \, \mbf W_1^{[5]} + i \, B_{2_X} \, e^{i x} \, \mbf W_{1, 2}^{[5]}
            \\
            & \quad + i \, \abs{A_1}^2 \, B_{2_X} \, e^{i x} \, \mbf W_{1, 3}^{[5]} + i \, \bar A_1 \, A_{1_X} \, B_2 \, e^{i x} \, \mbf W_{1, 3}^{[5]} + i \, A_1 \, A_{1_X} \, \bar B_2 \, e^{i x} \, \mbf W_{1, 3}^{[5]}
            \\
            & \quad + i \, \bar A_1 \, A_{2_X} \, B_1 \, e^{i x} \, \mbf W_{1, 3}^{[5]} + i \, A_{1_X} \, \bar A_2 \, B_1 \, e^{i x} \, \mbf W_{1, 3}^{[5]} + i \, A_1 \, A_{2_X} \, \bar B_1 \, e^{i x} \, \mbf W_{1, 3}^{[5]}
            \\
            & \quad + i \, \bar A_1 \, A_2 \, B_{1_X} \, e^{i x} \, \mbf W_{1, 3}^{[5]} + i \, A_1 \, \bar A_2 \, B_{1_X} \, e^{i x} \, \mbf W_{1, 3}^{[5]} + i \, A_{1_X} \, A_2 \, \bar B_1 \, e^{i x} \, \mbf W_{1, 3}^{[5]}
            \\
            & \quad + i \, A_1^2 \, \bar B_{2_X} \, e^{i x} \, \mbf W_{1, 4}^{[5]} + 2 \, i \, A_1 \, \bar A_{1_X} \, B_2 \, e^{i x} \, \mbf W_{1, 4}^{[5]} + 2 \, i \, A_1 \, A_2 \, \bar B_{1_X} \, e^{i x} \, \mbf W_{1, 4}^{[5]}
            \\
            & \quad + 2 \, i \, A_1 \, \bar A_{2_X} \, B_1 \, e^{i x} \, \mbf W_{1, 4}^{[5]} + 2 \, i \, \bar A_{1_X} \, A_2 \, B_1 \, e^{i x} \, \mbf W_{1, 4}^{[5]} + B_2 \, e^{i x} \, \mbf W_{1, 5}^{[5]} + A_1^2 \, \bar B_2 \, e^{i x} \, \mbf W_{1, 6}^{[5]}
            \\
            & \quad + 2 \, \abs{A_1}^2 \, B_2 \, e^{i x} \, \mbf W_{1, 6}^{[5]} + 2 \, \bar A_1 \, A_2 \, B_1 \, e^{i x} \, \mbf W_{1, 6}^{[5]} + 2 \, A_1 \,  \bar A_2 \, B_1 \, e^{i x} \, \mbf W_{1, 6}^{[5]}
            \\
            & \quad + 2 \, A_1 \, A_2 \, \bar B_1 \, e^{i x} \, \mbf W_{1, 6}^{[5]} + 2 \, \abs{A_1}^2 \, A_1^2 \, \bar B_2 \, e^{i x} \, \mbf W_{1, 7}^{[5]} + 2 \, A_1^3 \, \bar A_2 \, \bar B_1 \, e^{i x} \, \mbf W_{1, 7}^{[5]}
            \\
            & \quad + 3 \, \abs{A_1}^4 \, B_2 \, e^{i x} \, \mbf W_{1, 7}^{[5]} + 6 \, \abs{A_1}^2 \, \bar A_1 \, A_2 \, B_1 \, e^{i x} \, \mbf W_{1, 7}^{[5]} + 6 \, \abs{A_1}^2 \, A_1 \, \bar A_2 \, B_1 \, e^{i x} \, \mbf W_{1, 7}^{[5]}
            \\
            & \quad + 6 \, \abs{A_1}^2 \, A_1 \, A_2 \, \bar B_1 \, e^{i x} \, \mbf W_{1, 7}^{[5]} + B_3 \, e^{i x} \, \mbf W_1^{[4]} + A_1^2 \, \bar B_3 \, e^{i x} \, \mbf W_{1, 2}^{[4]} + A_2^2 \, \bar B_1 \, e^{i x} \, \mbf W_{1, 2}^{[4]}
            \\
            & \quad + 2 \, \bar A_1 \, A_2 \, B_2 \, e^{i x} \, \mbf W_{1, 2}^{[4]} + 2 \, A_1 \, \bar A_2 \, B_2 \, e^{i x} \, \mbf W_{1, 2}^{[4]} + 2 \, A_1 \, A_2 \, \bar B_2 \, e^{i x} \, \mbf W_{1, 2}^{[4]}
            \\
            & \quad + 2 \, \abs{A_1}^2 \, B_3 \, e^{i x} \, \mbf W_{1, 2}^{[4]} + 2 \, \abs{A_2}^2 \, B_1 \, e^{i x} \, \mbf W_{1, 2}^{[4]} + 2 \, A_1 \, \bar A_3 \, B_1 \, e^{i x} \, \mbf W_{1, 2}^{[4]}
            \\
            & \quad + 2 \, A_1 \, A_3 \, \bar B_1 \, e^{i x} \, \mbf W_{1, 2}^{[4]} + 2 \, \bar A_1 \, A_3 \, B_1 \, e^{i x} \, \mbf W_{1, 2}^{[4]} + i \, B_{3_X} \, e^{i x} \, \mbf W_{1, 3}^{[4]} + B_4 \, e^{i x} \, \mbf W_1^{[3]}
            \\
            & \quad + A_2^2 \, \bar B_2 \, e^{i x} \, \mbf W_{1, 2}^{[3]} + A_1^2 \, \bar B_4 \, e^{i x} \, \mbf W_{1, 2}^{[3]} + 2 \, \abs{A_2}^2 \, B_2 \, e^{i x} \, \mbf W_{1, 2}^{[3]}
            \\
            & \quad + 2 \, A_1 \, A_3 \, \bar B_2 \, e^{i x} \, \mbf W_{1, 2}^{[3]} + 2 \, A_1 \, \bar A_3 \, B_2 \, e^{i x} \, \mbf W_{1, 2}^{[3]} + 2 \, \bar A_1 \, A_3 \, B_2 e^{i x} \, \mbf W_{1, 2}^{[3]}
            \\
            & \quad + 2 \, \bar A_1 \, A_2 \, B_3 \, e^{i x} \, \mbf W_{1, 2}^{[3]} + 2 \, A_1 \, \bar A_2 \, B_3 \, e^{i x} \, \mbf W_{1, 2}^{[3]} + 2 \, A_1 \, A_2 \, \bar B_3 \, e^{i x} \, \mbf W_{1, 2}^{[3]}
            \\
            & \quad + 2 \, \abs{A_1}^2 \, B_4 \, e^{i x} \, \mbf W_{1, 2}^{[3]} + 2 \, \bar A_1 \, A_4 \, B_1 \, e^{i x} \, \mbf W_{1, 2}^{[3]} + 2 \, A_1 \, \bar A_4 \, B_1 \, e^{i x} \, \mbf W_{1, 2}^{[3]}
            \\
            & \quad + 2 \, A_2 \, \bar A_3 \, B_1 \, e^{i x} \, \mbf W_{1, 2}^{[3]} + 2 \, \bar A_2 \, A_3 \, B_1 \, e^{i x} \, \mbf W_{1, 2}^{[3]} + 2 \, A_2 \, A_3 \, \bar B_1 \, e^{i x} \, \mbf W_{1, 2}^{[3]}
            \\
            & \quad + 2 \, A_1 \, A_4 \, \bar B_1 \, e^{i x} \, \mbf W_{1, 2}^{[3]} + i \, B_{4_X} \, e^{i x} \, \mbf W_{1, 3}^{[3]} + B_5 \, e^{i x} \, \mbf W_1^{[2]} + B_6 \, e^{i x} \, \bs \phi_1^{[1]} + 2 \, A_1 \, B_1 \, e^{2 i x} \, \mbf W_2^{[6]}
            \\
            & \quad + A_1^3 \, \bar B_1 \, e^{2 i x} \, \mbf W_{2, 2}^{[6]} + 3 \, \abs{A_1}^2 \, A_1 \, B_1 \, e^{2 i x} \, \mbf W_{2, 2}^{[6]} + 2 \, \abs{A_1}^2 \, A_1^3 \, \bar B_1 \, e^{2 i x} \, \mbf W_{2, 3}^{[6]}
            \\
            & \quad + 4 \, \abs{A_1}^4 \, A_1 \, B_1 \, e^{2 i x} \, \mbf W_{2, 3}^{[6]} + 2 \, A_{1_X} \, B_{1_X} \, e^{2 i x} \, \mbf W_{2, 4}^{[6]} + i \, A_1 \, B_{1_X} \,  e^{2 i x} \, \mbf W_{2, 5}^{[6]}
            \\
            & \quad + i \, A_{1_X} \, B_1 \, e^{2 i x} \, \mbf W_{2, 5}^{[6]} + i \, \abs{A_1}^2 \, A_1 \, B_{1_X} \, e^{2 i x} \, \mbf W_{2, 6}^{[6]} + i \, A_1^2 \, A_{1_X} \, \bar B_1 \, e^{2 i x} \, \mbf W_{2, 6}^{[6]}
            \\
            & \quad + 2 \, i \, \abs{A_1}^2 \, A_{1_X} \, B_1 \, e^{2 i x} \, \mbf W_{2, 6}^{[6]} + i \, A_1^3 \, \bar B_{1_X} \, e^{2 i x} \, \mbf W_{2, 7}^{[6]} + 3 \, i \, A_1^2 \, \bar A_{1_X} \, B_1 \, e^{2 i x} \, \mbf W_{2, 7}^{[6]}
            \\
            & \quad + A_1 \, B_{1_{X \! X}} \, e^{2 i x} \, \mbf W_{2, 8}^{[6]} + A_{1_{X \! X}} \, B_1 \, e^{2 i x} \, \mbf W_{2, 8}^{[6]} + i \, A_1 \, B_{1_\Xi} \, e^{2 i x} \, \mbf W_{2, 3}^{[4]} + i \, A_{1_\Xi} \, B_1 \, e^{2 i x} \, \mbf W_{2, 3}^{[4]}
            \\
            & \quad + 2 \, A_1 \, B_2 \, e^{2 i x} \, \mbf W_2^{[5]} + 2 \, A_2 \, B_1 \, e^{2 i x} \, \mbf W_2^{[5]} + A_1^3 \, \bar B_2 \, e^{2 i x} \, \mbf W_{2, 2}^{[5]} + 3 \, \abs{A_1}^2 \, A_1 \, B_2 \, e^{2 i x} \, \mbf W_{2, 2}^{[5]}
            \\
            & \quad + 3 \, A_1^2 \, A_2 \, \bar B_1 \, e^{2 i x} \, \mbf W_{2, 2}^{[5]} + 3 \, A_1^2 \, \bar A_2 \, B_1 \, e^{2 i x} \, \mbf W_{2, 2}^{[5]} + 6 \, \abs{A_1}^2 \, A_2 \, B_1 \, e^{2 i x} \, \mbf W_{2, 2}^{[5]}
            \\
            & \quad + i \, A_1 \, B_{2_X} \, e^{2 i x} \, \mbf W_{2, 3}^{[5]} + i \, A_{1_X} \, B_2 \, e^{2 i x} \, \mbf W_{2, 3}^{[5]} + i \, A_2 \, B_{1_X} \, e^{2 i x} \, \mbf W_{2, 3}^{[5]} + i \, A_{2_X} \, B_1 \, e^{2 i x} \, \mbf W_{2, 3}^{[5]}
            \\
            & \quad + 2 \, A_2 \, B_2 \, e^{2 i x} \, \mbf W_2^{[4]} + 2 \, A_1 \, B_3 \, e^{2 i x} \, \mbf W_2^{[4]} + 2 \, A_3 \, B_1 \, e^{2 i x} \, \mbf W_2^{[4]} + A_1^3 \, \bar B_3 \, e^{2 i x} \, \mbf W_{2, 2}^{[4]}
            \\
            & \quad + 3 \, A_1^2 \, \bar A_2 \, B_2 \, e^{2 i x} \, \mbf W_{2, 2}^{[4]} + 3 \, A_1^2 \, A_2 \, \bar B_2 \, e^{2 i x} \, \mbf W_{2, 2}^{[4]} + 3 \, \abs{A_1}^2 \, A_1 \, B_3 \, e^{2 i x} \, \mbf W_{2, 2}^{[4]}
            \\
            & \quad + 3 \, \bar A_1 \, A_2^2 \, B_1 \, e^{2 i x} \, \mbf W_{2, 2}^{[4]} + 3 \, A_1 \, A_2^2 \, \bar B_1 \, e^{2 i x} \, \mbf W_{2, 2}^{[4]} + 3 \, A_1^2 \, \bar A_3 \, B_1 \, e^{2 i x} \, \mbf W_{2, 2}^{[4]}
            \\
            & \quad + 3 \, A_1^2 \, A_3 \, \bar B_1 \, e^{2 i x} \, \mbf W_{2, 2}^{[4]} + 6 \, \abs{A_1}^2 \, A_2 \, B_2 \, e^{2 i x} \, \mbf W_{2, 2}^{[4]} + 6 \, \abs{A_2}^2 \, A_1 \, B_1 \, e^{2 i x} \, \mbf W_{2, 2}^{[4]}
            \\
            & \quad + 6 \, \abs{A_1}^2 \, A_3 \, B_1 \, e^{2 i x} \, \mbf W_{2, 2}^{[4]} + i \, A_2 \, B_{2_X} \, e^{2 i x} \, \mbf W_{2, 3}^{[4]} + i \, A_{2_X} \, B_2 \, e^{2 i x} \, \mbf W_{2, 3}^{[4]}
            \\
            & \quad + i \, A_1 \, B_{3_X} \, e^{2 i x} \, \mbf W_{2, 3}^{[4]} + i \, A_{1_X} \, B_3 \, e^{2 i x} \, \mbf W_{2, 3}^{[4]} + i \, A_3 \, B_{1_X} \, e^{2 i x} \, \mbf W_{2, 3}^{[4]} + i \, A_{3_X} \, B_1 \, e^{2 i x} \, \mbf W_{2, 3}^{[4]}
            \\
            & \quad + 2 \, A_3 \, B_2 \, e^{2 i x} \, \mbf W_2^{[3]} + 2 \, A_2 \, B_3 \, e^{2 i x} \, \mbf W_2^{[3]} + 2 \, A_1 \, B_4 \, e^{2 i x} \, \mbf W_2^{[3]} + 2 \, A_4 \, B_1 \, e^{2 i x} \, \mbf W_2^{[3]}
            \\
            & \quad + 2 \, A_4 \, B_2 \, e^{2 i x} \, \mbf W_2^{[2]} + 2 \, A_3 \, B_3 \, e^{2 i x} \, \mbf W_2^{[2]} + 2 \, A_2 \, B_4 \, e^{2 i x} \, \mbf W_2^{[2]} + 2 \, A_1 \, B_5 \, e^{2 i x} \, \mbf W_2^{[2]}
            \\
            & \quad + 2 \, A_5 \, B_1 \, e^{2 i x} \, \mbf W_2^{[2]} + \ldots + c.c.
        \end{align*}
        }
        where '$\ldots$' represents terms that are multiples of $e^{3ix}$, $e^{4ix}$, $e^{5ix}$, or $e^{6ix}$.
        
        Last but not least, the solvability condition at order $\mathcal O\left(\varepsilon^7\right)$ is given by
        {\allowdisplaybreaks
        \begin{multline}
            \alpha_1 \, B_{3_{X \! X}} + i \, \alpha_2 \, B_{3_X} + i \, \alpha_3 \, \abs{A_1}^2 \, B_{3_X} + i \, \alpha_3 \, \bar A_1 \, A_{1_X} \, B_3 + i \, \alpha_3 \, A_1 \, A_{1_X} \, \bar B_3 + i \, \alpha_4 \, A_1^2 \, \bar B_{3_X}
            \\
            + 2 \, i \, \alpha_4 \, A_1 \, \bar A_{1_X} \, B_3 + \alpha_5 \, B_3 + \alpha_6 \, A_1^2 \, \bar B_3 + 2 \, \alpha_6 \, \abs{A_1}^2 \, B_3 + 2 \, \alpha_7 \, \abs{A_1}^2 \, A_1^2 \, \bar B_3 + 3 \, \alpha_7 \, \abs{A_1}^4 \, B_3
            \\
            + \alpha_{1, 2} \, B_{2_{X \! X}} + i \, \alpha_{2, 2} \, B_{2_X} + i \, \alpha_{3, 2} \, \abs{A_1}^2 \, B_{2_X} + i \, \alpha_{3, 2} \, \bar A_1 \, A_{1_X} \, B_2 + i \, \alpha_{3, 2} \, A_1 \, A_{1_X} \, \bar B_2
            \\
            + i \, \alpha_{4, 2} \, A_1^2 \, \bar B_{2_X} + 2 \, i \, \alpha_{4, 2} \, A_1 \, \bar A_{1_X} \, B_2 + \alpha_{5, 2} \, B_2 + \alpha_{6, 2} \, A_1^2 \, \bar B_2 + 2 \, \alpha_{6, 2} \, \abs{A_1}^2 \, B_2
            \\
            + 2 \, \alpha_{7, 2} \, \abs{A_1}^2 \, A_1^2 \, \bar B_2 + 3 \, \alpha_{7, 2} \, \abs{A_1}^2 \, A_1 \, B_2 + i \, \alpha_3 \, \bar A_1 \, A_2 \, B_{2_X} + i \, \alpha_3 \, A_1 \, \bar A_2 \, B_{2_X} + i \, \alpha_3 \, \bar A_1 \, A_{2_X} \, B_2
            \\
            + i \, \alpha_3 \, A_{1_X} \, \bar A_2 \, B_2 + i \, \alpha_3 \, A_1 \, A_{2_X} \, \bar B_2 + i \, \alpha_3 \, A_{1_X} \, A_2 \, \bar B_2 + 2 \, i \, \alpha_4 \, \bar A_{1_X} \, A_2 \, B_2 + 2 \, i \, \alpha_4 \, A_1 \, A_2 \, \bar B_{2_X}
            \\
            + 2 \, i \, \alpha_4 \, A_1 \, \bar A_{2_X} \, B_2 + 2 \, \alpha_6 \, \bar A_1 \, A_2 \, B_2 + 2 \, \alpha_6 \, A_1 \, \bar A_2 \, B_2 + 2 \, \alpha_6 \, A_1 \, A_2 \, \bar B_2 + 2 \, \alpha_7 \, A_1^3 \, \bar A_2 \, \bar B_2
            \\
            + 6 \, \alpha_7 \, \abs{A_1}^2 \, \bar A_1 \, A_2 \, B_2 + 6 \, \alpha_7 \, \abs{A_1}^2 \, A_1 \, \bar A_2 \, B_2 + 6 \, \alpha_7 \, \abs{A_1}^2 \, A_1 \, A_2 \, \bar B_2 + i \, \alpha_{1, 3} \, B_{1_{X \! X \! X}} + \alpha_{2, 3} \, B_{1_{X \! X}}
            \\
            + \alpha_{3, 3} \, \abs{A_1}^2 \, B_{1_{X \! X}} + \alpha_{3, 3} \, \bar A_1 \, A_{1_{X \! X}} \, B_1 + \alpha_{3, 3} \, A_1 \, A_{1_{X \! X}} \, \bar B_1 + \alpha_{4, 3} \, A_1^2 \, \bar B_{1_{X \! X}} + 2 \, \alpha_{4, 3} \, A_1 \, \bar A_{1_{X \! X}} \, B_1
            \\
            + i \, \alpha_{5, 3} \, B_{1_X} + i \, \alpha_{6, 3} \, \abs{A_1}^2 \, B_{1_X} + i \, \alpha_{6, 3} \, \bar A_1 \, A_{1_X} \, B_1 + i \, \alpha_{6, 3} \, A_1 \, A_{1_X} \, \bar B_1 + i \, \alpha_{7, 3} \, \abs{A_1}^4 \, B_{1_X}
            \\
            + 2 \, i \, \alpha_{7, 3} \, \abs{A_1}^2 \, \bar A_1 \, A_{1_X} \, B_1 + 2 \, i \, \alpha_{7, 3} \, \abs{A_1}^2 \, A_1 \, A_{1_X} \, \bar B_1 + i \, \alpha_{8, 3} \, A_1^2 \, \bar B_{1_X} + 2 \, i \, \alpha_{8, 3} \, A_1 \, \bar A_{1_X} \, B_1
            \\
            + i \, \alpha_{9, 3} \, \abs{A_1}^2 \, A_1^2 \, \bar B_{1_X} + i \, \alpha_{9, 3} \, A_1^3 \, \bar A_{1_X} \, \bar B_1 + 3 \, i \, \alpha_{9, 3} \, \abs{A_1}^2 \, A_1 \, \bar A_{1_X} \, B_1 + \alpha_{10, 3} \, A_{1_X}^2 \, \bar B_1
            \\
            + 2 \, \alpha_{10, 3} \, \bar A_1 \, A_{1_X} \, B_{1_X} + \alpha_{11, 3} \, A_1 \, \bar A_{1_X} \, B_{1_X} + \alpha_{11, 3} \, A_1 \, A_{1_X} \, \bar B_{1_X} + \alpha_{11, 3} \, \abs{A_{1_X}}^2 \, B_1 + \alpha_{12, 3} \, B_1
            \\
            + \alpha_{13, 3} \, A_1^2 \, \bar B_1 + 2 \, \alpha_{13, 3} \, \abs{A_1}^2 \, B_1 + 2 \, \alpha_{14, 3} \, \abs{A_1}^2 \, A_1^2 \, \bar B_1 + 3 \, \alpha_{14, 3} \, \abs{A_1}^4 \, B_1 + 3 \, \alpha_{15, 3} \, \abs{A_1}^4 \, A_1^2 \, \bar B_1
            \\
            + 4 \, \alpha_{15, 3} \, \abs{A_1}^6 \, B_1 + 2 \, \alpha_1 \, B_{1_{X \! \Xi}} + i \, \alpha_2 \, B_{1_\Xi} + i \, \alpha_3 \, \abs{A_1}^2 \, B_{1_\Xi} + i \, \alpha_3 \, \bar A_1 \, A_{1_\Xi} \, B_1 + i \, \alpha_3 \, A_1 \, A_{1_\Xi} \, \bar B_1
            \\
            + i \, \alpha_4 \, A_1^2 \, \bar B_{1_\Xi} + 2 \, i \, \alpha_4 \, A_1 \, \bar A_{1_\Xi} \, B_1 + i \, \alpha_3 \, \bar A_1 \, A_3 \, B_{1_X} + i \, \alpha_3 \, A_1 \, \bar A_3 \, B_{1_X} + i \, \alpha_3 \, \abs{A_2}^2 \, B_{1_X}
            \\
            + i \, \alpha_3 \, \bar A_1 \, A_{3_X} \, B_1 + i \, \alpha_3 \, A_1 \, A_{3_X} \, \bar B_1 + i \, \alpha_3 \, A_{1_X} \, \bar A_3 \, B_1 + i \, \alpha_3 \, \bar A_2 \, A_{2_X} \, B_1 + i \, \alpha_3 \, A_{1_X} \, A_3 \, \bar B_1
            \\
            + i \, \alpha_3 \, A_2 \, A_{2_X} \, \bar B_1 + i \, \alpha_4 \, A_2^2 \, \bar B_{1_X} + 2 \, i \, \alpha_4 \, A_1 \, A_3 \, \bar B_{1_X} + 2 \, i \, \alpha_4 \, A_1 \, \bar A_{3_X} \, B_1 + 2 \, i \, \alpha_4 \, \bar A_{1_X} \, A_3 \, B_1
            \\
            + 2 \, i \, \alpha_4 \, A_2 \, \bar A_{2_X} \, B_1 + \alpha_6 \, A_2^2 \, \bar B_1 + 2 \, \alpha_6 \, \bar A_1 \, A_3 \, B_1 + 2 \, \alpha_6 \, A_1 \, \bar A_3 \, B_1 + 2 \, \alpha_6 \, A_1 \, A_3 \, \bar B_1 + 2 \, \alpha_6 \, \abs{A_2}^2 \, B_1
            \\
            + 2 \, \alpha_7 \, A_1^3 \, \bar A_3 \, \bar B_1 + 3 \, \alpha_7 \, \bar A_1^2 \, A_2^2 \, B_1 + 3 \, \alpha_7 \, A_1^2 \, \bar A_2^2 \, B_1 + 6 \, \alpha_7 \, \abs{A_1}^2 \, A_1 \, \bar A_3 \, B_1 + 6 \, \alpha_7 \, \abs{A_1}^2 \, \bar A_1 \, A_3 \, B_1
            \\
            + 6 \, \alpha_7 \, \abs{A_1}^2 \, A_1 \, A_3 \, \bar B_1 + 6 \, \alpha_7 \, \abs{A_1}^2 \, A_2^2 \, \bar B_1 + 6 \, \alpha_7 \, A_1^2 \, \abs{A_2}^2 \, \bar B_1 + 12 \, \alpha_7 \, \abs{A_1}^2 \, \abs{A_2}^2 \, B_1
            \\
            + i \, \alpha_{3, 2} \, \bar A_1 \, A_2 \, B_{1_X} + i \, \alpha_{3, 2} \, A_1 \, \bar A_2 \, B_{1_X} + i \, \alpha_{3, 2} \, \bar A_1 \, A_{2_X} \, B_1 + i \, \alpha_{3, 2} \, A_{1_X} \, \bar A_2 \, B_1 + i \, \alpha_{3, 2} \, A_1 \, A_{2_X} \, \bar B_1
            \\
            + i \, \alpha_{3, 2} \, A_{1_X} \, A_2 \, \bar B_1 + 2 \, i \, \alpha_{4, 2} \, A_1 \, A_2 \, \bar B_{1_X} + 2 \, i \, \alpha_{4, 2} \, A_1 \, \bar A_{2_X} \, B_1 + 2 \, i \, \alpha_{4, 2} \, \bar A_{1_X} \, A_2 \, B_1 + 2 \, \alpha_{6, 2} \, \bar A_1 \, A_2 \, B_1
            \\
            + 2 \, \alpha_{6, 2} \, A_1 \, \bar A_2 \, B_1 + 2 \, \alpha_{6, 2} \, A_1 \, A_2 \, \bar B_1 + 2 \, \alpha_{7, 2} \, A_1^3 \, \bar A_2 \, \bar B_1 + 6 \, \alpha_{7, 2} \, \abs{A_1}^2 \, \bar A_1 \, A_2 \, B_1
            \\
            + 6 \, \alpha_{7, 2} \, \abs{A_1}^2 \, A_1 \, \bar A_2 \, B_1 + 6 \, \alpha_{7, 2} \, \abs{A_1}^2 \, A_1 \, A_2 \, \bar B_1 = 0. \label{rem_7_eq}
        \end{multline}
        }

    \section{Logic of the code} \label{app:code}
        The GitHub repository, \cite{beyond-code}, contains two types of code, a Python code that computes the amplitude equations and finds the asymptotic expansion of the Maxwell point up to order 7, which can be located within the folder `demos'; and a Mathematica code that performs the asymptotic expansion up to a high order (see \eqref{DM}) to obtain $c_r^{[0]}$ (see \eqref{c0-def}). These codes were developed in different platforms due to the capabilities of each one (Python because of the existence of classes, and Mathematica because of the symbolic manipulation).
        
        \subsection{Python code}
            To run the main Python code, \texttt{Beyond-all-order asymptotics.py}, one requires \texttt{Python} 3.13.5, together with the following packages: \texttt{IPython} 8.30.0, \texttt{sympy} 1.13.3, and \texttt{mpmath} 1.3.0.

            Apart from that, the code requires a Python file \texttt{foo.py} within a folder called \texttt{foo} with the definitions of
            \begin{itemize}
                \item a dictionary, \texttt{parameters}, which specifies all the parameters of the system that are not being expanded, together with their values;
                \item a list, \texttt{unevaluatedparameters}, which specifies all the parameters that are required not to be evaluated by Python (the recommendation is to make this list empty; otherwise, the code will take longer to run);
                \item the strings \texttt{par1} and \texttt{par2}, which are the names of the parameters $a$ and $b$ to be expanded;
                \item a string, \texttt{muval}, which is the value of the square of the wavenumber, $k^2$;
                \item a list of strings, \texttt{equilibrium}, that contains the algebraic coordinates of the steady state of the system;
                \item a list of strings, \texttt{var}, which is the vector of coordinates of the system;
                \item a list of lists \texttt{diffmatrix}, which encodes the diffusion matrix of the system;
                \item a vector of strings, \texttt{kinetics}, which corresponds to the reaction of the system;
                \item a real number, \texttt{tol}, which corresponds to the threshold that is accepted to check for Turing bifurcation and codimension-two conditions;
                \item a string, \texttt{phiunit}, that is a boolean variable to state whether the vector $\bs \phi_1^{[1]}$ is required to be a unit vector;
                \item a boolean variable, \texttt{simp} that states whether the algebraic expressions of vectors and coefficients are to be simplified or not (if the system is algebraically complicated, it might take a long time to simplify these expressions);
                \item a boolean variable, \texttt{simple}, which states whether the system is `simple' or not, on order for the equations for $\alpha_{i, 2}$ to be solved explicitly or not (this variable is relevant only when the coefficients $a_1, b_1, a_3, b_3, a_5, b_5$ have not been set to be zero from the start);
                \item and a function \texttt{extrapars(extraparvals)}, that edits the dictionary that evaluates the parameters of the expansion following, for example, the definition \texttt{extraparvals[aNF[1]]=0} to define $a_1 = 0$, for the first order expansion of the parameter $a$ of your system.
            \end{itemize}
            For example, for the Brusselator system, the file \texttt{Brusselator.py} is set up as follows:
            \begin{verbatim}
    parameters = {'delta': Mul(Add(21, sqrt(313)), Pow(48, - 1))}
    
    unevaluatedparameters = []
    
    par1 = 'a'
    par2 = 'c'
    
    muval = '(21 + sqrt(313))/(16 * delta**2)'
    
    translate = True
    
    equilibrium = ['a', '(c - 1)/a']
    
    var = ['u',
           'v']
    
    diffmatrix=[['delta^2', 0],
                [0, 1]]                
    
    kinetics=['a - c * u + u**2*v',
              '(c - 1)*u - u**2*v']
    
    tol = 1e-7
    
    parameter_functions = {}
    
    phiunit = 'n'
    
    simp = 'y'
    
    simple = 'y'
    
    def extrapars(extraparvals):
        extraparvals[aNF[0]] = 3
        extraparvals[aNF[2]] = 1
        extraparvals[aNF[3]] = 0
        extraparvals[aNF[4]] = 0
        extraparvals[aNF[5]] = 0
        extraparvals[aNF[6]] = 0
        extraparvals[bNF[0]] = Mul(Add(969, Mul(37, sqrt(313))),
                                Pow(128, - 1))
        extraparvals[bNF[3]] = 0
        return extraparvals,
            \end{verbatim}
            which can be read together with example \eqref{Brus} to have an idea of how to set up the code. \evss{After this has been set up, one just needs to run \texttt{Beyond-all-order asymptotics.py} on a terminal and, after the prompt `Enter the name of the system you would like to analyze: ' appears, one must enter \texttt{foo}.}
            
            In particular, after importing the packages, the main Python script shifts the variables and the parameters of the system so that the steady state under analysis is $\mbf 0$ and the codimension-two bifurcation point occurs at $(a, b) = (0, 0)$. Then, the code computes the derivatives of the system up to order 7 to compute the Taylor expansions required to compute the vectors corresponding to the solutions at each order.

            After the problem has been set up in Python, the code starts by obtaining the vectors $\bs \phi_1^{[1]}$ and $\bs \psi$, which span the kernels of $\mathcal M_1$ and $\mathcal M_1^\intercal$, respectively. To do this, Python runs loops to find a row and a column that can be removed from $\mathcal M_1$ to obtain an invertible submatrix, as explained in \cite{TOMS}. The indices of these rows and columns are saved as \texttt{criticalrow} and \texttt{criticalcol}, respectively. After this, the \texttt{criticalrow} coordinate of $\bs \phi_1^{[1]}$ is set as 1, and the other coordinates are found by solving the corresponding linear system of equations that now has an invertible matrix of coefficients. The same thing is done for $\psi$, but using the transpose of $\mathcal M_1$, and reversing \texttt{criticalrow} and \texttt{criticalcol}.

            Next, to solve systems of equations in the case of having an invertible matrix of coefficients, as in the case of $\mbf W_0^{[2]}$, defined as the solution to
            \begin{align*}
                \mathcal M_0 \, \mbf W_0^{[2]} = - \mbf F_{2, 0, 0}\left(\bs \phi_1^{[1]}, \bs \phi_1^{[1]}\right),
            \end{align*}
            the code performs
            \begin{verbatim}
                DS_phiphi_00 = second_order(0, 0, phiNF, phiNF)

                negativeRHS.actualcoord = DS_phiphi_00
                
                W02NF = linearsolver(W02NF, negativeRHS, coefmat0),
            \end{verbatim}
            where \texttt{DS\_phiphi\_00} represents $\mbf F_{2, 0, 0}\left(\bs \phi_1^{[1]}, \bs \phi_1^{[1]}\right)$, and has an intuitive name so that the code is `easy-to-follow'. Then, as most of the terms on the right-hand side of the equations defining the vectors of the expansion have a negative sign in front, we need to pass the negative right-hand side to the function \texttt{linearsolver}, which takes care of the sign automatically.

            All the vectors are found in the same way. We highlight that the code solves the equations `quickly', as it does not work with the explicit expressions of $\mathcal M_0$ and \texttt{negativeRHS.actualcoord} (as they can be messy) when solving the equations. Instead, it solves the equation with dummy symbolic variables that are replaced after the solution has been obtained.

            One last thing to highlight here has to do with the extra terms that appear in the equation for the remainder due to the separation of $b$ from the Maxwell point. This can be included near the end of the Python script. Unfortunately, this step is model-dependent, as it can be seen in Section \ref{sec:examples}, so this step requires some thinking in terms of the dominant nonlinearities of the system. For example, for the Brusselator, we have that the extra term in the amplitude equation is given by
            \begin{align*}
                \frac{1}{66} \, \left(37 - \sqrt{313}\right) \, \frac{\Delta c}{\varepsilon^4} \, A_1,
            \end{align*}
            and this was obtained with the following lines of code, which are present in the main Python script:
            \begin{verbatim}
    elif modelname=='Brusselator':
        deltagamma = symbols('deltagamma', real = True)
        
        quantity = psiNF.actualcoord.dot(Mul(deltagamma, SS_phi_01)) \
                                            .subs(extraparvals) # Linear
        
        quantity = simplify(quantity.subs(phiNF_eval).subs(muNF, muval) \
                                            .subs(parameters))
            \end{verbatim}
            where the variable \texttt{quantity} equals
            \begin{align*}
                \frac{\delta \gamma}{66} \, \left(37 - \sqrt{313}\right).
            \end{align*}

            \evss{The output of the Python code comprises three key things:}
            \begin{enumerate}
                \item \evss{First, the code shows you the order of the equation it is solving. It goes from order 1 to order 7 (this last order is not reported in the output, as the code does other calculations at that order). After order 4, if everything goes correctly and the code solves all the solvability conditions, then the output `The expansion to meet solvability conditions has been found.' will be shown. Otherwise, it will show `print('Something went wrong. Review extra parameters.')'. If this output is shown, the code will keep running, but it is recommended that one stops it and reviews the system and parameter values provided, and the solutions that Python has found for the solvability conditions, which are present in the variable \texttt{extraparvals}.}

                \evss{Specifically, the solvability conditions \eqref{solvcond1}, \eqref{solvcond13}, \eqref{solvcond14}, \eqref{solvcond24}, \eqref{solvcond34} are coded in the variables \texttt{equation1}--\texttt{equation5}, respectively. If the code shows the error mentioned above, one can see the value of each of these variables. If any of them is not zero, then that solvability condition is not being met. Some extra input might be needed. We recommend checking the output and fixing the Python code `foo.py' accordingly.}

                \evss{Also, at order 7, if all the calculations are carried out successfully, then the code will show the output `A Maxwell point has been found for the parameter values provided.' If not, the code will show `Something went wrong. Review extra parameters.' One can check what the problem is by checking the value of the variable \texttt{Maxwell\_equation}, which should equal zero, together with the signs of the variables \texttt{beta1}, \texttt{beta3}, and \texttt{beta5}, which should be positive, negative, and positive, respectively. If any of these conditions do not hold, it means that the code was not able to find a Maxwell point with the input provided.}

                \item \evss{After that, Python will show 4 values that are needed to run the Mathematica code explained below. Specifically, Python outputs the values of \texttt{ANF1}, \texttt{ANF3}, \texttt{eta}, and \texttt{xi}, which correspond to \eqref{C_A}, \eqref{C_3}, and the terms defined in \eqref{etaxi}, respectively. Also, these values correspond to the variables \texttt{C11NF}, \texttt{A3\_coef.subs(omega1NF, 0)}, \texttt{etaNF}, and \texttt{xiNF}, respectively.}

                \item \evss{The last key part of the output corresponds to a file named \texttt{objs.pkl} that is saved within the folder \texttt{foo}. This file contains the variables \texttt{alpha1}--\texttt{alpha7}, and \texttt{ANF1}, \texttt{ANF3}, \texttt{eta}, \texttt{xi}, and \texttt{extraparvals} that are saved after running the script for the model within the folder \texttt{foo}. These variables can be obtained by running the four lines commented at the end of the code. These variables can be used to check that the code's output is consistent with what is shown in this paper in Tables \ref{tab:valSH23}--\ref{tab:valBru4}.}
            \end{enumerate}
            \evss{Some of the information shown in this paper is also saved in Python, but they are not needed for running the Mathematica code. For example, the variables \texttt{alpha1}--\texttt{alpha7} correspond to the coefficients of the normal form \eqref{firstamplitudeeq}. The same type of notation was used to define in Python the coefficients of the amplitude equations \eqref{A_2}, and \eqref{A_3}. Also, the variable \texttt{ratio} has the value of the ratio $\frac{h_1}{k_1}$ defined in \eqref{h1-k1ratio}.}

        \subsection{Mathematica code}
            After running the Python code, one needs to run the Mathematica code to perform the iteration to find an approximation of $c_r^{[0]}$ (see \eqref{c0-def}). Apart from setting up the system itself, one needs to give the values of $\eta$, $\xi$, $A_1$, and $A_3$, where the latter correspond to the approximations of the first and cubic order amplitudes at the singularity $X_0$. Specifically, the code requires the following variables:
            \begin{itemize}
                \item \texttt{nmax}, which represents the maximum order to be computed by the code;
                \item \texttt{var}, which represents an array with the variables of the system;
                \item \texttt{par}, which represents an array with the parameters of the system;
                \item \texttt{nvar}, which represents the number of variables of the system and is computed automatically by using the length of the array \texttt{var};
                \item \texttt{fixedpar}, which represents an array to evaluate the parameters of the system that are fixed;
                \item \texttt{parval}, which represents an array to evaluate the values of the parameters of the expansion at the codimension-two bifurcation point;
                \item \texttt{$\eta$val}, which represents the value of $\eta$ obtained from the Python code \evss{see \eqref{etaxi}};
                \item \texttt{$\xi$val}, which represents the value of $\xi$ obtained from the Python code \evss{see \eqref{etaxi}};
                \item \texttt{$\kappa$val}, which represents the value of $\kappa$ to be considered in the computation of the limit defining $c_r^{[0]}$ (see \eqref{c0-def});
                \item \texttt{C11eval}, \evss{which represents an array to evaluate the first order expansion of $A_1$ at $X_0$ (see \eqref{C_A})};
                \item \texttt{A3eval}, \evss{which represents an array to evaluate the first order expansion of $A_3$ at $X_0$ (see \eqref{C_3});}
                \item \texttt{kinetics}, which represents an array with the vector field, $\mbf f$, defined by the system \eqref{geneq};
                \item \texttt{diffmatrix}, which represents a 2-dimensional array with the diffusion matrix, $\hat D$, defined by the system \eqref{geneq};
            \end{itemize}
            For example, in the case of the Swift-Hohenberg 2-3 equation, we have
            \begin{lstlisting}
    nmax = 181;
    var = {u, v};
    par = {c, e};
    nvar = Length[var];
    fixedpar = {};
    parval = {c -> 0, e -> 0.280975743474508} // N;
    $\eta$val = {$\eta$NF -> -0.0184553367631391} // N;
    $\xi$val = {$\xi$NF -> 0} // N;
    $\kappa$val = {$\kappa$NF^2 -> I} // N;
    C11eval = {C11 -> 0.592158254607366*I} // N;
    A3eval = {ANF_3 -> -0.149272058209306
            + 0.0309130877820341*I} // N;
    kinetics = {- (1 + c)*u - 3*e*u^2 - u^3 - 2*v, v};
    diffmatrix = ({{0, - 1}, {- 1, 0}});
            \end{lstlisting}

            After providing all this information, the code obtains the vectors that span the kernels of $\mathcal M_1$ and $\mathcal M_1^\intercal$, respectively, following the same logic as the Python code. After this, it makes the change
            \begin{align*}
                \mbf u \to \sum_{p = 0}^{N_{\max}} \varepsilon^p \mbf u_p, 
            \end{align*}
            and it starts a loop where it computes the $p$-th order of the Taylor expansion of the kinetics of the system by differentiating the kinetics $p$ times and dividing them by $p!$. After obtaining that expression, for each order, the code replaces
            \begin{align*}
                \mbf u_p \to \sum_{r = - p}^p e^{i r k x} \, \mbf W_q,
            \end{align*}
            and generates a second loop to solve the corresponding equation for each value of $r$ between $0$ and $p$ at each order, $p$ (the code does not solve the equations from $- p$ to $0$ because those vectors are known due to \eqref{conjrule}, which the code also uses). An extra step is needed only for solving the equation corresponding to the equation for $r = 1$, which is the resonant one. For that equation, a solvability condition is determined, which lets the code obtain the amplitudes we are after in order to obtain the limit in \eqref{c0-def}.

            Finally, the argument of the limit in \eqref{c0-def} is used to compute an approximation of $c_r^{[0]}$ for large values of $n$. The code also plots the result if one runs the cells below the iteration itself.

            As stated in \ref{app:symmetric_case}, some cases need extra care and that depends on whether the dominant submodes are $r = 0, \pm 2$ or, in the critical case in which $c_0^{[0]} = c_{\pm 2}^{[0]} = 0$, then one needs to change some indices appropriately (see the code within the folders Iteration/Swift-Hohenberg 3-5 to see how the formula changes).

\end{document}